\documentclass[10pt,draftcls, onecolumn]{IEEEtran}
\usepackage[pdftex]{graphics}
\usepackage{amsmath}
\usepackage{times}
\usepackage{amsmath,dsfont}
\usepackage{amssymb,amsthm,bm}
\usepackage{epsfig,verbatim}
\usepackage{subfig}
\usepackage{comment}
\usepackage{bm}
\usepackage{xr}
\usepackage[final]{hyperref}
\usepackage[final]{pdfpages}
\usepackage{refcount}

\newtheorem{definition}{Definition}
\newtheorem{theorem}{Theorem}

\newtheorem{proposition}{Proposition}
\newtheorem{lemma}{Lemma}

\newtheorem{remark}{Comment}

\newcommand{\Tgood}{\mbox{Tx}_1}
\newcommand{\Tbad}{\mbox{Tx}_2}

\newcommand{\E}{\mathds{E}}
\newcommand{\cov}{\mbox{cov}}

\newcommand{\opt}{\mbox{\tiny Opt}}
\newcommand{\Amat}{\mathds{A}}
\newcommand{\Bmat}{\mathds{B}}
\newcommand{\Cmat}{\mathds{C}}
\newcommand{\Dmat}{\mathds{D}}
\newcommand{\Kmat}{\mathds{K}}
\newcommand{\Hmat}{\mathds{H}}
\newcommand{\Omat}{\mathds{O}}
\newcommand{\hmat}{\mathds{h}}
\newcommand{\Imat}{\mathds{I}}
\newcommand{\Jmat}{\mathds{J}}
\newcommand{\Tmat}{\mathds{T}}
\newcommand{\Vmat}{\mathds{V}}
\newcommand{\Umat}{\mathds{U}}
\newcommand{\Qmat}{\mathds{Q}}
\newcommand{\Smat}{\mathds{S}}
\newcommand{\Lmat}{\mathds{L}}
\newcommand{\Gmat}{\mathds{G}}
\newcommand{\Zmat}{\mathds{Z}}

\newcommand{\dsE}{\mathds{E}}

\newcommand{\Yvec}{\mathbf{Y}}
\newcommand{\bY}{\bar{Y}}
\newcommand{\Qvec}{\mathbf{Q}}
\newcommand{\yvec}{\mathbf{y}}

\newcommand{\bXvec}{\bar{\mathbf{X}}}

\newcommand{\Xvec}{\mathbf{X}}
\newcommand{\xvec}{\mathbf{x}}
\newcommand{\Zvec}{\mathbf{Z}}
\newcommand{\bZ}{\bar{Z}}
\newcommand{\bD}{\bar{D}}
\newcommand{\zvec}{\mathbf{z}}
\newcommand{\Wvec}{\mathbf{W}}
\newcommand{\wvec}{\mathbf{w}}
\newcommand{\Vvec}{\mathbf{V}}

\newcommand{\Hvec}{\mathbf{H}}
\newcommand{\Svec}{\mathbf{S}}

\newcommand{\fnsz}{\mbox{\footnotesize z}}
\newcommand{\avec}{\mathbf{a}}

\newcommand{\Cset}{\mathfrak{C}}
\newcommand{\Sset}{\mathcal{S}}

\newcommand{\Rset}{\mathfrak{R}}
\newcommand{\hvec}{\tilde{h}}

\newcommand{\bX}{\bar{X}}
\newcommand{\oG}{\bar{G}}
\newcommand{\oX}{\bar{X}}
\newcommand{\oW}{\bar{W}}
\newcommand{\oV}{\bar{V}}
\newcommand{\oN}{\bar{N}}
\renewcommand{\th}{\tilde{h}}
\newcommand{\uth}{\underline{\tilde{h}}}
\newcommand{\utH}{\tilde{\underline{H}}}

\newcommand{\CN}{\mathcal{CN}}
\newcommand{\N}{\mathcal{N}}

\newcommand{\M}{\mathcal{M}}
\newcommand{\typ}{\mathcal{A}^{(n)}_\epsilon}

\newcommand{\eig}{\makebox{eig}}
\newcommand{\tH}{\tilde{H}}

\newcommand{\thvec}{\mathbf{\tilde{h}}}

\newcommand{\var}{\makebox{var}}
\newcommand{\tr}{\makebox{tr}}
\newcommand{\SNR}{\makebox{\sffamily SNR}}
\newcommand{\SNRA}{\makebox{\sffamily  SNR}}
\newcommand{\CORR}{\upsilon}

\newcommand{\CORRN}{\tilde{\upsilon}}

\renewcommand{\P}{\makebox{P}}
\newcommand{\Real}{\mathfrak{Re}}
\newcommand{\Img}{\mathfrak{Im}}
\newcommand{\RealS}{\Real}
\newcommand{\ImagS}{\Img}

\long\def\symbolfootnote[#1]#2{\begingroup\def\thefootnote{\fnsymbol{footnote}}\footnote[#1]{#2}\endgroup}
\newcommand{\tend}{\hfill$\blacksquare$}
\newif\ifextended
\extendedtrue % comment out for shortened version

\IEEEoverridecommandlockouts
\includecomment{versiona}

\title{On Cooperation and Interference in the Weak Interference Regime\\ {\huge (Full Version with Detailed Proofs)}
\thanks{
\noindent
This work was supported by the Israeli science foundation under grant 396/11. Parts of this work were presented at IEEE International Symposium on Information Theory (ISIT) 2013 in Istanbul, Turkey, and at ISIT 2014 in Honolulu, HI.}
}

\author{Daniel~Zahavi and Ron~Dabora,~\IEEEmembership{Senior~Member,~IEEE}\\Department of Electrical and Computer Engineering \\
Ben-Gurion University, Israel}

\begin{document}

\maketitle

\begin{abstract}
   Handling interference is one of the main challenges in the design of wireless networks. In this paper we study the application of cooperation for interference management in the weak interference (WI) regime,
focusing on the Z-interference channel with a causal relay (Z-ICR),
 in which the channel coefficients are subject to ergodic phase fading, all transmission powers are finite, and the relay is full-duplex. The phase fading model represents many practical communications systems in which the transmission path impairments mainly affect the phase of the signal, such as
%	orthogonal frequency-division multiplexing (OFDM) communications, and single-carrier line-of-sight microwave communications.
non-coherent wireless communications and  fiber optic channels.
In order to provide a comprehensive understanding of the benefits of cooperation in the WI regime,
we characterize, for the first time, two major performance measures for the ergodic phase fading Z-ICR in the WI regime:
The sum-rate capacity and the maximal generalized degrees-of-freedom (GDoF).
In the capacity analysis, we obtain conditions on the channel coefficients, subject to which the sum-rate capacity of the ergodic phase fading Z-ICR is achieved by treating interference as noise at each receiver, and explicitly state the corresponding sum-rate capacity. In the GDoF analysis, we derive conditions on the exponents of the magnitudes of the channel coefficients, under which treating interference as noise achieves the maximal GDoF, which is explicitly characterized as well.
It is shown that under certain conditions on the channel coefficients, {\em relaying strictly increases}
 both the sum-rate capacity and the maximal GDoF of the ergodic phase fading Z-interference channel in the WI regime.
Our results demonstrate {\em for the  first time} the gains from relaying in the presence of interference,
{\em when interference is weak and the relay power is finite}, both in increasing the sum-rate capacity and in increasing the maximal GDoF, compared
to the channel without a relay.
\end{abstract}

\section{Introduction}
\label{sec:intro1}
The interference channel (IC) \cite{Shannon:61} models communications scenarios in which two source-destination pairs communicate over a shared medium. The capacity region of the IC is generally unknown, but capacity characterizations exist for some special scenarios. For example, the capacity region of the IC with additive white Gaussian noise (AWGN), for the scenario in which the interference between the communicating pairs is very strong, was characterized in \cite{Carleial:75}, and the capacity region for the case of strong interference (SI) was characterized in \cite{Sato:81}. In both works it was shown that in order to achieve capacity, each receiver should
{\em decode} both the interfering message as well as the desired message. Additional performance measures commonly used for characterizing the performance of ICs are the degrees-of-freedom (DoF) and the generalized DoF (GDoF). The GDoF for the IC was first analyzed in \cite{etkin:08}, where it was also shown that in the very strong interference regime, the maximal GDoF of the Gaussian IC is achieved by letting each receiver decode both the interfering message as well as the intended message. It thus follows that when interference is sufficiently strong, jointly decoding both messages at each receiver is the optimal strategy from both the sum-rate and the GDoF perspectives.

\phantomsection
\label{phn:state11}
The weak interference (WI) regime is the opposite regime to the SI regime. In this regime, since the interference is weak, then decoding the interfering message cannot be done without constraining the rates of the desired information at each receiver.
In \cite{etkin:08} it was shown that when interference is sufficiently weak, treating interference as noise at the receivers achieves the maximal GDoF of the Gaussian IC in the WI regime; In \cite{Shang:09}-\cite{Motahari:09} it was shown that this strategy is also
		sum-rate optimal in the WI regime for finite SNRs.
As treating interference as noise is implemented via a low complexity, simple, point-to-point (PtP) decoding strategy, there is a strong motivation for identifying additional scenarios in which treating interference as noise at the receivers carries optimality.

In this work, we study the impact of cooperation on the communications performance in the WI regime by considering the IC with an additional relay node (ICR). The objective of the relay node in the general ICR is to simultaneously assist communications from both sources to their corresponding destinations \cite{Zahavi:12}, \cite{Zahavi:15}. The optimal transmission strategy for the relay node in this channel is not known in general. One of the main difficulties in the design of transmission schemes is that when the relay assists one pair, it may degrade the performance of the other pair. In \cite{Sahin:07}, the authors derived an achievable rate region for Gaussian ICRs by using the rate
splitting technique (see, e.g., \cite{Han:81}) at the sources, and by employing the decode-and-forward (DF) strategy at the relay. Additional inner bounds and outer bounds on the capacity region of the ICR were derived in \cite{Sahin:09} and \cite{Maric:12}. The capacity region of ergodic fading ICRs in the strong interference (SI) regime was studied in \cite{Dabora:12}, for both Rayleigh fading and phase fading scenarios. In \cite{Dabora:12} it was shown that when relay reception is good and the interference is strong, then, similarly to the IC, the optimal strategy at each receiver is to jointly decode both the desired message and the interfering message, while the optimal strategy at the relay node is to employ the DF scheme. The sum-rate capacity of the Gaussian IC with a {\em potent} relay in the WI regime was characterized in \cite{Tian:11}, in which it was shown that in such a scenario, compress-and-forward (CF) at the relay together with treating interference as noise at the destinations is sum-rate optimal. The sum-rate capacity of the ICR in the WI regime when all nodes have finite powers {\em remains unknown to date}. The ergodic sum-rate capacity of interference networks without relays, subject to phase fading, was studied in \cite{Jafar:2011}, and explicit sum-capacity expressions based on ergodic interference alignment (which requires channel state information (CSI) at the transmitters) were derived for networks with a finite number of users. The work \cite{Jafar:2011}
also derived an asymptotic sum-rate capacity expression when the number of users increases to infinity.
ICs with time-varying/frequency-selective channel coefficients, in which global CSI is available at all nodes, and in addition, the magnitudes of all links have the same exponential scaling as a function of the signal-to-noise ratio (SNR), were studied in \cite{Cadambe:09}. Under these conditions, \cite{Cadambe:09} showed that adding a relay {\em does not increase the DoF region}, and that the achievable DoF for each pair in the ICR is upper bounded by $1$. On the other hand, it was shown in \cite{Chaaban:12} that relaying can increase the {\em GDoF} for symmetric Gaussian ICRs. This follows since differently from the DoF analysis, in GDoF analysis the magnitudes of different links may have different SNR scaling exponents. In \cite{Chaaban:12}, several GDoF upper bounds were derived for Gaussian ICRs by using the cut-set theorem and the genie-aided approach, for the case in which the source-destination, source-relay, and relay-destination links scale differently as a function of the SNR. Additionally, \cite{Chaaban:12} showed that in the WI regime, {\em when the source-relay links are weaker than the interfering links} in the sense that their SNR scaling exponent is smaller, then the Han-Kobayashi (HK) scheme \cite{Han:81} achieves the maximal GDoF. The complementing scenario, i.e., {\em GDoF analysis when the interfering links are weaker than the source-relay links,} was considered in \cite{Chaaban:15}. The GDoF analysis in \cite{Chaaban:15} was based on deriving upper bounds on the sum-rate capacity of the {\em linear deterministic} ICR. Lastly, we note that the GDoF of the Gaussian IC with a broadcasting relay, in which the relay-destination links are {\em noiseless}, finite-capacity links, which are {\em orthogonal} to the other links in the channel, was studied in \cite{Zhou:13}. From the GDoF characterization, \cite{Zhou:13} concludes that in the WI regime, each bit per channel use transmitted by the relay can improve the sum-rate capacity by $2$ bits per channel use.

To date, there has been no work that characterized the sum-rate capacity and the maximal GDoF of  ergodic phase fading ICs with a causal relay in the WI regime, for scenarios in which the power of the relay is finite. In this work, we partially fill this gap by considering a special case of the ergodic phase fading ICR, in which one of the interfering links is missing, e.g., as a result of shadowing in the channel. Furthermore, we consider the scenario in which the relay node receives transmissions from only {\em one of the two sources}, but is received at both destinations. We refer to this channel configuration as {\em Z-interference channel with a relay} (Z-ICR).

\subsection*{Main Contributions}
\phantomsection
\label{phn:maincont}
In this paper, we characterize for the first time the sum-rate capacity (i.e., finite-SNR performance) and the maximal GDoF
(i.e., asymptotically high SNR performance) of the ergodic phase fading Z-ICR in the WI regime, when the relay is causal, has a finite transmission power,
and operates in full-duplex mode. Performance gain from cooperation in the WI regime is demonstrated in both the sum-rate capacity and the maximal GDoF.
In contrast to \cite{Tian:11}, which showed the optimality of CF for memoryless ICRs with AWGN and time-invariant link coefficients, we study the {\em ergodic phase fading} (also referred to as
fast phase fading) scenario and demonstrate the optimality of DF. Throughout this paper it is assumed that the nodes have {\em causal} CSI only on their {\em incoming links} (Rx-CSI); no transmitter CSI (Tx-CSI) is assumed. The links are all subject to i.i.d. phase fading (see, e.g., \cite[Section II]{Jafar:2011} and \cite[Section VII]{Kramer:05}) which can be applied to modeling many practical scenarios. One such example is non-coherent wireless communication \cite{Katz:2004}, in which phase fading occurs due to the lack of perfect frequency synchronization between the oscillators at the transmitter and at the receiver. Phase fading channel models also apply to systems which use dithering to decorrelate signals, as well as to optic fiber channels \cite{Katz:2004}. In this work, it is assumed that the relay receives transmissions from only one of the sources, while relay transmissions are received at both destinations. Thus, differently from previous works, {\em the relay cannot forward desired information to one of the destinations}.

Our main contributions are summarized as follows:
\begin{itemize}
\phantomsection
\label{phn:maincont-item1}
    \item We derive an upper bound on the achievable sum-rate of the ergodic phase fading Z-ICR by using the genie-aided approach.
					The upper bound requires a novel design of the genie signals as well as  the introduction of novel tools for proving that the
					bound is maximized by mutually independent, i.i.d., complex Normal channel inputs.
\phantomsection
\label{phn:maincont-item2}
    \item We derive a lower bound on the achievable sum-rate of the ergodic phase fading Z-ICR by using DF at the relay, and by treating interference as noise at each receiver.
    We also identify conditions on the {\em magnitudes of the channel coefficients} under which the sum-rate of our lower bound coincides with the sum-rate upper bound.
				This results in {\em the characterization of the sum-rate capacity of the ergodic phase fading Z-ICR in the WI regime.
				This is the first time capacity is characterized for a cooperative interference network in the WI regime, when all powers are finite.}
\phantomsection
\label{phn:maincont-item4}
    \item We derive two upper bounds on the achievable GDoF of the ergodic phase fading Z-ICR, as well, as a lower bound on the achievable GDoF.\\
				Note that while capacity analysis is common for fading scenarios, GDoF analysis was previously applied only to {\em time-invariant} AWGN
				channels. This follows since, when the channel coefficients vary with time (e.g., a fading channel), then the GDoF
				generally becomes a random variable.
				For the ergodic phase fading model, however, as the squared magnitude of each channel
				coefficient is a {\em constant}, then GDoF analysis is relevant despite the random temporal nature of the channel coefficients.
				To the best of our knowledge, this is the {\em first time} that GDoF analysis is carried out for a fading scenario.

    \item We identify conditions on the {\em scaling of the links' magnitudes} (i.e., SNR exponents) under which our GDoF lower bound coincides with the GDoF upper bound. {\em This characterizes the maximal GDoF of the phase fading Z-ICR in the WI regime}.
\end{itemize}
\phantomsection
\label{phn:maincont-item6}
Our results show that when certain conditions on the channel coefficients are satisfied, then adding a relay to the ergodic phase
			fading Z-IC {\em strictly increases} both the sum-rate capacity and the maximal GDoF of the channel in the WI regime.
We note that the sum-rate capacity analysis in this paper has two major differences from the work of \cite{Tian:11}: First, we consider a fading scenario while \cite{Tian:11} considered the time-invariant AWGN case, and second, we assume that the power of the relay is {\em finite} while \cite{Tian:11} considered a {\em potent} relay.

In the GDoF analysis, similarly to \cite{etkin:08}, \cite{Zahavi:15}, \cite{Chaaban:12}-\cite{Zhou:13}, we consider a general setup in which the different links scale differently as a function of the SNR, which facilitates characterizing the impact of the relative link strengthes on the SNR scaling of the sum-rate. The GDoF analysis in this paper has several fundamental differences from the works \cite{Chaaban:12}-\cite{Zhou:13}: First, note that \cite{Chaaban:12}-\cite{Zhou:13} studied the common time-invariant Gaussian channel while we consider an ergodic fading channel; Second, unlike \cite{Chaaban:12}-\cite{Zhou:13}, the channel configuration studied in this work is {\em not symmetric} and the relay cannot forward desired information to one of the destinations. We further note that, unlike \cite{Chaaban:12}-\cite{Zhou:13}, GDoF optimality in the present work is achieved only in a non-symmetric scenario in which the link from the relay to one receiver scales differently than the link from the relay to the other receiver; We also emphasize that while in our work we consider a non-orthogonal scenario, the work \cite{Zhou:13} considered noiseless, orthogonal relay-destination links, and thus, relay transmissions in \cite{Zhou:13} do not interfere with the reception of the desired signal at each receiver. It therefore follows that the GDoF of the ergodic phase fading Z-ICR, studied in the present work, cannot be derived as a special case of GDoF results for Gaussian ICRs derived in \cite{Chaaban:12}-\cite{Zhou:13}.

The rest of this paper is organized as follows: In Section \ref{sec:Model}, we define the system model and describe the notation used throughout this paper. In section \ref{sec:sum-capacity}, we characterize the sum-rate capacity of the ergodic phase fading Z-ICR in the WI regime, and in section \ref{sec:FULLGDoF}, we characterize the maximal GDoF of this channel in the WI regime. Finally, concluding remarks are provided in Section \ref{sec:conclusions}.

\section{Notation and System Model}
\label{sec:Model}
We denote random variables (RVs) with upper-case letters, e.g., $X,Y$, and their realizations with lower-case letters, e.g., $x,y$. We denote the probability density function (p.d.f.) of a continuous RV $X$ with $f_{X}(x)$. Double-stroke letters are used for denoting matrices, e.g., $\Amat$, $\hmat$, with the exception that $\E\{X\}$ denotes the stochastic expectation of $X$.
The element at the $k$'th row and $l$'th column of the matrix $\Amat$ is denoted with $\big[\Amat\big]_{k,l}$. Bold-face letters, e.g., $\Xvec$, denote column vectors, the $i$'th element of a vector $\Xvec$, $i>0,$ is denoted with $X_i$, and $X^j$ denotes the vector $(X_1, X_{2},...,X_j)^T$. Given a complex number $x$, we denote the real and the imaginary parts of $x$ with $\Real\{x\}$ and $\Img\{x\}$, respectively. $x^*$ denotes the conjugate of $x$, $\Xvec^T$ denotes the transpose of $\Xvec$, $\Amat^H$ denotes the Hermitian transpose of $\Amat$, $|\Amat|$ denotes the determinant of $\Amat$, and $\mathds{I}_n$ denotes the $n\times n$ identity matrix. For a complex vector $X^n$, we define an associated real vector by stacking its real and imaginary parts:
 $\bX^{2n}=\big(\Real\{X^n\}^T,\Img\{X^n\}^T\big)^T$. $\Rset$ and $\Cset$ denote the sets of real and of complex numbers, respectively.
 Given two $n\times n$ Hermitian matrices, $\Amat, \Bmat$, we write $\Bmat\preceq\Amat$ if $\Amat-\Bmat$ is positive semidefinite (p.s.d.) and $\Bmat\prec\Amat$ if $\Amat-\Bmat$ is positive definite (p.d.). $\typ(X,Y)$ denotes the set of weakly jointly typical sequences with respect to $f_{X,Y}(x,y)$. We denote the Gaussian distribution with mean $\mu$ and variance $\sigma^2$ with $\N(\mu,\sigma^2)$, and similarly, we denote the circularly symmetric complex Normal distribution with variance $\sigma^2$ with $\CN(0,\sigma^2)$.
For a complex random vector, the covariance matrix and the pseudo-covariance matrix are defined as in \cite[Section II]{Massey:93}.
Given an RV $X$ with $\E\{X\}=0$, $X_{G}$ (i.e., adding a subscript ``G" to the RV) denotes an RV which is distributed according to a circularly symmetric, complex Normal distribution with the same variance as the indicated RV, i.e., $X_G\sim\CN\big(0,\var\{X\}\big)$; similarly, subscript ``$\bar{\mbox{G}}$" is used to denote an RV which is distributed according to a complex Normal distribution with the same  variance as the indicated RV, where the mean is explicitly specified.
We emphasize that RVs with subscript ``$\bar{\mbox{G}}$" {\em are not necessarily circularly symmetric}. We denote $f(\SNR)\doteq \SNR^c$ if $\lim_{{\scriptsize \SNR}\to\infty}\frac{\log f({\scriptsize \SNR})}{\log{\scriptsize \SNR}}=c$, and given $f(\SNR)\doteq \SNR^c$ and $g(\SNR)\doteq \SNR^d$, we write $f(\SNR)\dot{\le}g(\SNR)$ if $c\le d$. Lastly, we note that all logarithms are of base $2$.

The Z-ICR consists of two transmitters, Tx$_1$, Tx$_2$, two receivers, Rx$_1$, Rx$_2$ and a full-duplex relay node. Tx$_k$ sends messages to Rx$_k$, $k\in\{1,2\}$. The relay node receives only the signal transmitted from Tx$_1$ but is received at both destinations simultaneously. The signal received at Rx$_1$ is a combination of the transmissions of Tx$_1$ and of the relay along with interference from Tx$_2$, while the signal received at Rx$_2$ is a combination of the transmissions of Tx$_2$ and of the relay {\em without} interference from Tx$_1$. This channel model is depicted in Fig. \ref{fig:ICR-WI model}. The received signals at Rx$_1$, Rx$_2$ and the relay at time $i$ are denoted with $Y_{1,i}$, $Y_{2,i}$, and $Y_{3,i}$, respectively; the channel inputs from $\Tgood$, $\Tbad$ and the relay at time $i$ are denoted with $X_{1,i}$, $X_{2,i}$ and $X_{3,i}$, respectively. Finally, $H_{lk,i}$ denotes the channel coefficient for the link with input $X_{l,i}$ and output $Y_{k,i}$ at time instance $i$. The relationship between the channel inputs and its outputs can be written as:
\begin{subequations}
\label{eqn:ZICR_model}
\begin{eqnarray}
    \label{eqn:Rx_1_sig}
    Y_{1,i} & = & H_{11,i} X_{1,i} + H_{21,i} X_{2,i} + H_{31,i} X_{3,i} + Z_{1,i}\\
    \label{eqn:Rx_2_sig}
    Y_{2,i} & = & H_{22,i} X_{2,i} + H_{32,i} X_{3,i} + Z_{2,i}\\
    \label{eqn:Relay_sig}
    Y_{3,i} & = & H_{13,i} X_{1,i} + Z_{3,i},
\end{eqnarray}
\end{subequations}
$i = 1, 2, ..., n$, where $Z_1$, $Z_2$ and $Z_3$ are mutually independent RVs, each independent and identically distributed (i.i.d.) over time according to $\CN(0,1)$,
and all noises are independent of the channel inputs and of the channel coefficients. The channel input signals are subject to per-symbol average power constraints: $P_{k,i}\triangleq\E\big\{|X_{k,i}|^2\big\}\le 1$, $k\in\{1,2,3\}$. The receivers and the relay node have instantaneous {\em causal} Rx-CSI on their incoming links, but the transmitters and the relay do not have Tx-CSI on their outgoing links. Under the ergodic phase fading model, the channel coefficients are given by $H_{lk,i} = \sqrt{\SNRA_{lk}} e^{j\Theta_{lk,i}}$, where $\SNRA_{lk} \in \Rset_+$ is a non-negative constant which corresponds to the signal-to-noise ratio for the link $H_{lk,i}$, and  $\Theta_{lk,i}$ is an RV uniformly distributed over $[0,2\pi)$, i.i.d. in time, independent of the other $\Theta_{lk,i}$'s, and independent of the additive noises $Z_{k,i}$ as well as of the transmitted signals, $X_{k,i}$, $k\in\{ 1,2,3\}$. The independence of the $\Theta_{lk,i}$'s implies that the coefficients $H_{lk,i}$'s are also mutually independent, and in addition they are independent in time, and independent of the other parameters of the scenario.

The channel coefficients causally available at Rx$_1$ are represented by $\tH_1\! =\! \big( H_{11}, H_{21}, H_{31}\big)^T\!\!\in\!\Cset^3\triangleq\tilde{\mathfrak{H}}_1$, at Rx$_2$ they are represented by $\tH_2\! =\! \big( H_{22}, H_{32}\big)^T\!\!\in\!\Cset^2\triangleq\tilde{\mathfrak{H}}_2$, and at the relay they are represented by $\tH_{3}=H_{13}\in\Cset\triangleq\tilde{\mathfrak{H}}_3$. Let $\utH=(\tH_1^T,\tH_2^T,\tH_3)^T\!\!\in\!\Cset^6$ be the vector of all channel coefficients, and let $\underline{\SNR}\triangleq\big(\SNRA_{11},\SNRA_{21},\SNRA_{31},\SNRA_{22},\SNRA_{32},\SNRA_{13}\big)$. We now state several definitions:

\begin{figure}
    \centering
  \includegraphics[scale=0.29]{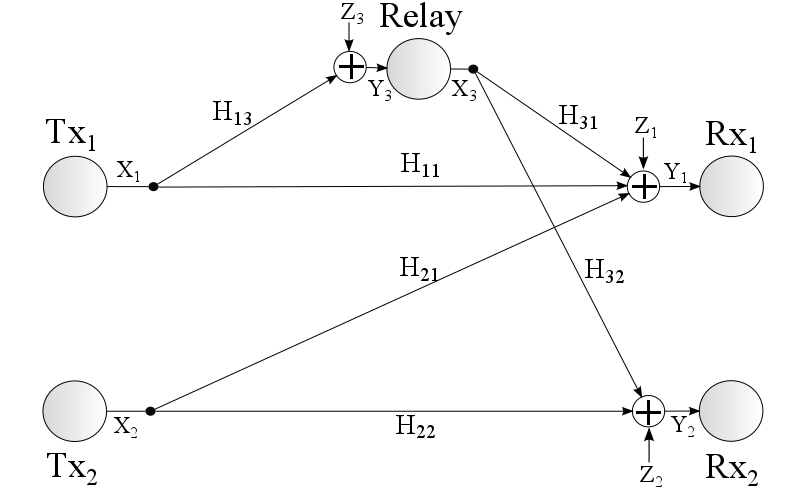}
\caption{\small The ergodic phase fading Z-ICR. The relay node receives  transmissions only from Tx$_1$, but is received at both destinations simultaneously.}
\label{fig:ICR-WI model}
\end{figure}

\begin{definition}
\label{def:code}
    \em{}An $(R_1,R_2,n)$ code for the Z-ICR consists of two message sets, $\M_k \triangleq \big\{1,2,...,2^{n R_k}\big\}$, $k = 1,2$, two encoders at the sources, $e_1, e_2$, employing deterministic mappings; \quad $e_k: \M_k \mapsto \Cset^n, k\in\{1,2\}$, and two decoders at the destinations, $g_1,g_2$;\quad $g_k: \tilde{\mathfrak{H}}_k^n \times \Cset^n \mapsto \M_k$, $k=1,2$.
		Since the relay receives transmissions only from Tx$_1$, the transmitted signal at the relay at time $i$ is generated via a set of $n$ functions $\{t_i(\cdot)\}_{i=1}^n$, such that $x_{3,i} = t_i\big(y_{3}^{i-1},h_{13}^{i-1}\big) \in \Cset$, $i=1,2,...,n$.
\end{definition}

\begin{remark}
    \label{rem:IndiSignals}
    \em{} Note that since the messages at the transmitters are independent and there is no feedback, then the signals transmitted from Tx$_1$ and from Tx$_2$ are necessarily independent as well. Additionally, since the relay receives transmissions only from Tx$_1$, then its transmitted signal is independent of the signal transmitted from Tx$_2$. Combining both observations we can write $f_{\Xvec_1,\Xvec_2,\Xvec_3}(\xvec_1,\xvec_2,\xvec_3)=f_{\Xvec_1,\Xvec_3}(\xvec_1,\xvec_3)\cdot f_{\Xvec_2}(\xvec_2)$. We denote the correlation coefficient between channel inputs $X_1$ and $X_3$ at time index $i$ with $\CORR_i$: $\CORR_i\triangleq\frac{\E\{X_{1,i}X_{3,i}^*\}}{\sqrt{P_{1,i}P_{3,i}}}, 0\le|\CORR_i|\le 1$.
\end{remark}

\begin{definition}
\em{}The average probability of error on an $(R_1, R_2, n)$ code is defined as $\P_e^{(n)} \triangleq \Pr\Big(g_1(\tH_1^n, Y_1^n) \ne M_1 \mbox{ or } g_2(\tH_2^n, Y_2^n) \ne M_2\Big)$, where each message is selected independently and uniformly from its message set.
\end{definition}

\begin{definition}
\em{} A rate pair $(R_1, R_2)$ is called achievable if, for any $\epsilon >0$ and $\delta >0$ there exists some blocklength $n_0(\epsilon,\delta)$, such that for every $n > n_0(\epsilon,\delta)$ there exists an  $(R_1 - \delta, R_2 - \delta ,n)$ code with $\P_e^{(n)} < \epsilon$.
\end{definition}

\begin{definition}
\label{def:capacity}
\em{}The capacity region is defined as the convex hull of all achievable rate pairs.
\end{definition}

\phantomsection
\label{stat1}
The objective of this work is to characterize two performance measures for the ergodic phase fading Z-ICR in the WI regime: The sum-rate capacity, which characterizes the performance at finite SNRs, and the maximal GDoF, which characterizes the performance at asymptotically high SNRs. As these cases are fundamentally different in nature, the WI regime is defined for each performance measure in accordance with the relevant notions in the literature.
In the following we briefly overview the WI conditions for each performance measure, leaving the detailed discussions to the relevant sections.
      \begin{itemize}
          \item In Section \ref{sec:sum-capacity} the sum-rate capacity is characterized for the WI regime, defined in Eqns. \eqref{eq:WI_con2-Thm}.
				Generally speaking, WI in this case occurs when the SNRs of the interfering links, $\SNRA_{32}$ and $\SNRA_{21}$,
				are sufficiently low compared to the SNRs for carrying the desired information. This is in accordance with the acceptable notion of
				the WI regime for sum-rate capacity analysis at finite SNRs, see, e.g., \cite{Shang:09}-\cite{Motahari:09}.
				
          \item  In Section \ref{sec:FULLGDoF} the maximal GDoF is characterized for the WI regime,
				defined in Eqn. \eqref{eq:optimalGDoFCon1}. Generally speaking, WI in this case occurs when the exponents of the SNRs of the interfering links
				are sufficiently smaller than the exponents of the SNRs of the information paths. In section \ref{sec:FULLGDoF}, in the statement of Thm. \ref{thm:WI-GDoF},
				this condition is expressed as $\lambda=\alpha \le \frac{1}{2}$, where $\alpha$ and $\lambda$ denote the exponential scalings of the
				Tx$_2$-Rx$_1$ link and of the relay-Rx$_2$ link, respectively. This definition is in accordance with the acceptable notion of the WI regime
                       for DoF and GDoF analysis, see, e.g., \cite{etkin:08} and  \cite{Chaaban:12}. We note that in \cite{etkin:08} and
                      \cite{Chaaban:12} the WI regime is characterized by $\alpha \le 1$, while
				GDoF optimality for the communications scheme described in the current paper requires a stricter notion of WI, characterized by
				$\lambda=\alpha \le \frac{1}{2}$. However, note that {\em in \cite{etkin:08}, treating interference as noise is GDoF optimal only for $\alpha\le\frac{1}{2}$},
                      which is in agreement with our result.
      \end{itemize}

%------------------------------------

%------------------------------------

\section{Finite SNR Analysis: The Sum-Rate Capacity in the WI Regime}
\label{sec:sum-capacity}

\subsection{Preliminaries}
\label{sec:preliminaries}
We begin by presenting several lemmas used in the derivation of the main result of this section. We note that while some of the following lemmas appeared in previous works for {\em real} variables, in the following we extend these lemmas to complex variables. Accordingly, in the appendices we include explicit proofs only for those lemmas whose proofs do not follow directly from the original proofs for real RVs.

%------------------------------------

\begin{lemma}
\label{lem:lemma1}
    \em Let $\Zvec_1$ and $\Zvec_2$ be a pair of $n$-dimensional, circularly symmetric, complex Normal random vectors, and let $\Xvec$ be an $n$-dimensional complex random vector whose p.d.f. is denoted by $f_{\Xvec}(\xvec)$. Consider the following optimization problem:
    \begin{eqnarray}
        \label{eq:Lemma1OP}
        &&\max_{f_{\Xvec}(\xvec)}\mbox{ } h(\Xvec+\Zvec_1)-h(\Xvec+\Zvec_2)\\
        &&\mbox{subject to:  \phantom{xx} } \tr\big(\cov(\Xvec)\big)\le nP.\nonumber
    \end{eqnarray}
    Then, a circularly symmetric, complex Normal random vector $\Xvec^{\mbox{\tiny Opt}}_G\sim\CN({\bf 0},\mathds{C}_X^{\opt})$ is an optimal solution to the optimization problem in \eqref{eq:Lemma1OP}. Additionally, if $\Zvec_1$ and $\Zvec_2$ have i.i.d. entries, i.e., $\Zvec_k\sim\CN({\bf 0},\gamma_k\mathds{I}_n), \gamma_k\in\Rset^+, k\in\{1,2\}$,  and if it holds that $\gamma_1\le \gamma_2$, then the optimal solution is distributed according to $\Xvec^{\opt}_G\sim\CN\Big({\bf 0},P\cdot\mathds{I}_n\Big)$.
\end{lemma}

\vspace{-0.3cm}

\begin{proof}
\ifextended
The proof is based on \cite[Theorem 1]{Liu:07} and \cite[Corollary 2]{Shang:09}. A detailed proof is provided in Appendix \ref{app:ProofOfLemma1FULL}.
\fi
\end{proof}

%------------------------------------

\begin{lemma}
\label{lem:lemma2}
    \em Let $\Zvec$ and $\Wvec$ be a pair of $n$-dimensional, zero-mean, jointly circularly symmetric complex Normal random vectors with i.i.d. entries, s.t. their joint distribution can be written as
    \begin{equation}
    \label{eq:WZJD}
        f_{\Wvec,\Zvec}(\wvec,\zvec)=\prod_{i=1}^n f_{W,Z}(w_i, z_i).
    \end{equation}
    Denote the cross-covariance matrix between $Z_i$ and $W_i$ with
    \begin{eqnarray*}
        \cov(Z_i,W_i)=\left[\begin{array}{cc}
                        \sigma_1^2  & \;\; \CORRN_{12}\\
                        \CORRN_{12}^* & \;\; \sigma_2^2\end{array} \right], \qquad\qquad i\in\{1,2,...,n\},
    \end{eqnarray*}
    where $\sigma_1^2>0$, and $\sigma_2^2>0$. Let $\Vvec$ be an $n$-dimensional, zero-mean, circularly symmetric complex Normal random vector with i.i.d. entries, whose covariance matrix is given by $\E\{\Vvec\Vvec^H\}=\Big(\sigma_1^2-\frac{|\CORRN_{12}|^2}{\sigma_2^2}\Big)\cdot\Imat_n$. If $\Xvec$ is independent of $(\Zvec,\Wvec,\Vvec)$, then
    \begin{equation*}
        h(\Xvec+\Zvec|\Wvec)=h(\Xvec+\Vvec).
    \end{equation*}
\end{lemma}

\vspace{-0.3cm}

\begin{proof}
The proof follows similar steps as in the proof of \cite[Lemma 3]{Shang:09}. A detailed proof is provided in Appendix \ref{app:ProofOfLemma2FULL}.
\end{proof}

%------------------------------------

\begin{lemma}
\label{lem:lemma22}
\em
Let $Z$ and $W$ be a pair of possibly correlated, zero-mean, jointly circularly symmetric complex Normal RVs, and let $\Hvec_{Y}$ and $\Hvec_{S}$ be two $n\times 1$ complex random vectors. Additionally, let $\Xvec$ be an $n\times 1$ complex random vector, and let $Y$ and $S$ be noisy observations of $\Xvec$, s.t.
\begin{subequations}
\label{eq:Lemma22-DefEqs}
\begin{eqnarray}
    Y&=&\Hvec_{Y}^T\Xvec+Z\\
    S&=&\Hvec_{S}^T\Xvec+W.
\end{eqnarray}
\end{subequations}
Consider the sequence of random vectors  $\Xvec^n=(\Xvec_1^T, \Xvec_2^T..., \Xvec_n^T)^T$ and let $\Qmat_{\Xvec_i}$ denote the covariance matrix of the $n \times 1$ vector $\Xvec_i$.
Furthermore, let $Y^n$ and $S^n$ be the corresponding observations when the noise sequences $(Z^n,W^n)$ are i.i.d. in the sense of \eqref{eq:WZJD}. Define $\Hvec=(\Hvec_{Y}^T, \Hvec_{S}^T)^T$, and let
$\Hvec^n=(\Hvec_1^T, \Hvec_2^T,..., \Hvec_n^T)^T$ be an i.i.d. sequence of random vectors, in which each $2n\times 1$ vector element is distributed according to the distribution of the $2n\times 1$ random vector $\Hvec$. Then, we can bound
\begin{equation}
    \label{eq:Lemma22EQ}
    h(Y^n|S^n,\Hvec^n)\le n\cdot h(Y_G|S_G,\Hvec),
\end{equation}
where $Y_G$ and $S_G$ denote the RVs $Y$ and $S$ defined in \eqref{eq:Lemma22-DefEqs}, obtained with $\Xvec$ replaced with $\Xvec_G\sim\CN({\bf 0},\frac{1}{n}\sum_{i=1}^n\Qmat_{\Xvec_i})$.
\end{lemma}

\vspace{-0.3cm}

\begin{proof}
\ifextended
%{\large \bf MAKE SURE PROOF Uses Complex X, Hy and Hs}
The proof follows similar steps as of the proof of \cite[Lemma 1]{Annapureddy:09}. A detailed proof is provided in Appendix \ref{app:ProofOfLemma3FULL}.
\fi
\end{proof}

%------------------------------------

\begin{lemma}
\label{lem:lemma3}
    \em Let $X_1,X_2,Z_1$ and $Z_2$ be zero mean, {\em jointly circularly symmetric} complex Normal RVs s.t. $(X_1,X_2)$ is independent of $(Z_1,Z_2)$\footnote{Joint circular symmetry of $(Z_1,Z_2)$ implies that $\E\{\Real\{Z_1\}\Img\{Z_2\}\} = -\E\{\Img\{Z_1\}\Real\{Z_2\}\}$ and $\E\{\Real\{Z_1\}\Real\{Z_2\}\} = \E\{\Img\{Z_1\}\Img\{Z_2\}\}$.}.
    Let $c_1$ and $c_2$ be a pair of complex constants, and let $Y_1$ and $Y_2$ be defined via
    \begin{eqnarray*}
        Y_1&=&c_1\cdot X_1+c_2\cdot X_2+Z_1\\
        Y_2&=&c_1\cdot X_1+c_2\cdot X_2+Z_2.
    \end{eqnarray*}
    Then, $I(X_1,X_2;Y_1|Y_2)=0$ if and only if $\E\{Z_1Z_2^*\}=\E\{|Z_2|^2\}$.
\end{lemma}

\vspace{-0.3cm}

\begin{proof}
The proof is provided in Appendix \ref{app:ProofOfLemma3}.
\end{proof}

%------------------------------------

\begin{lemma}
\label{lem:lemma6}
\em Let $Z^{n+m}\triangleq\Big(\big({Z^n_1}\big)^T, \big({Z^m_2}\big)^T\Big)^T$, where $Z^n_1$ and $Z^m_2$ are two mutually independent, circularly symmetric complex Normal random vectors of lengths $n$ and $m$, respectively, each with independent entries distributed according to $Z_{1,i}\sim\CN(0,a_{1,i}\gamma_{1}), i\in\{1,2,...,n\}$ and $Z_{2,i}\sim\CN(0,a_{2,i}\gamma_{2}), i\in\{1,2,...,m\}$, where $\gamma_1$, $\gamma_2$, and $a_{k,i}, k\in\{1,2\}$ are positive, real, and finite constants. Let $X^{n+m}\triangleq\Big(\big({X^n_1}\big)^T, \big({X^m_2}\big)^T\Big)^T$ where $X^n_1$ and $X^m_2$ are two complex random vectors of lengths $n$ and $m$, respectively, with finite covariance matrices,
and further let $X^{n+m}$ be mutually independent of $Z^{n+m}$. Then, we have the following limit:
\begin{equation*}
    \lim_{\gamma_2\rightarrow\infty} I\big(X^{n+m};X^{n+m}+Z^{n+m}\big) = I\big(X^n_1;X^n_1+Z^n_1\big).
\end{equation*}
\end{lemma}

\vspace{-0.3cm}

\begin{proof}
The proof is provided in Appendix \ref{app:ProofOfLemma6}.
\end{proof}

%------------------------------------

\begin{lemma}
\label{lem:lemma8}
\em Let $\tilde{Z}_{1}^n$ and $\tilde{Z}_{2}^n$ be a pair of possibly correlated, $n$-dimensional circularly symmetric complex Normal random vectors, each with independent entries, i.e., $\tilde{Z}_{k}^n\sim\CN(0,\tilde{\Dmat}^{\fnsz}_{k}), k\in\{1,2\}$, where $\tilde{\Dmat}^{\fnsz}_{k}, k\in\{1,2\}$ are two $n\times n$ diagonal matrices with real and positive entries on their main diagonals. Let $\tilde{\Vmat}_{1}$ and $\tilde{\Vmat}_{2}$ be two $2n \times n$ deterministic complex matrices, s.t. $\tilde{\Vmat}_{k}^H\tilde{\Vmat}_{k}=\tilde{\Dmat}_{k}^{-1}$, where $\tilde{\Dmat}_{k}, k\in\{1,2\}$ are two $n\times n$ diagonal matrices with real and positive entries on their main diagonals. Let $X^{2n}$ be a $2n \times 1$ complex random vector with distribution $f_{X^{2n}}(x^{2n})$, independent of $(\tilde{Z}_{1}^n,\tilde{Z}_{2}^n)$, and let $\bar{X}^{4n}$ be the stacking of the real and imaginary parts of $X^{2n}$, i.e., $\bar{X}^{4n}\triangleq \Big(\big(\Real\{X^{2n}\}\big)^T,\big(\Img\{X^{2n}\}\big)^T\Big)^T$. Consider the following optimization problem:
\begin{equation}
    \label{eq:eq6}
    \max_{f(x^{2n}):\quad\!\!\!\!  \cov(\bar{X}^{4n})\preceq \mathds{S}}\mbox{ } h(\tilde{\Vmat}_{1}^H\cdot X^{2n}+\tilde{Z}_{1}^n)-h(\tilde{\Vmat}_{2}^H\cdot X^{2n}+\tilde{Z}_{2}^n).
    \end{equation}
Then, a {\em zero-mean} complex Normal random vector, $X^{2n}_{{\oG}}$, is an optimal solution for \eqref{eq:eq6}.
\end{lemma}

\vspace{-0.3cm}

\begin{proof}
The proof is provided in Appendix \ref{app:ProofOfLemma8}.
\end{proof}

\vspace{0.5cm}

%------------------------------------

\subsection{Sum-Rate Capacity in the WI Regime}

Let $\mathcal{C}(\underline{\SNR})$ denote the capacity region of the Z-ICR, for a given $\underline{\SNR}$.
The sum-rate capacity of the ergodic phase fading Z-ICR in the WI regime is characterized in the following theorem:
\begin{theorem}
\label{thm:WI}
Consider the ergodic phase fading Z-ICR with only Rx-CSI, defined in Section \ref{sec:Model}.
%Let the additive noises be i.i.d. circularly symmetric complex Normal processes, $\CN(0,1)$, and let the sources be subject to per-symbol power constraints $\E\big\{|X_k|^2\big\} \le 1$, $k\in\{1,2,3\}$.
If $\underline{\SNR}$ satisfies
    \begin{equation}
         \label{eq:Relay_con2Phase}
         \frac{\SNRA_{11}+\SNRA_{31}}{1+\SNRA_{21}} \le \SNRA_{13},
    \end{equation}
and if there exist two real scalars $\beta_1$ and $\beta_2$ which satisfy $0 \le \beta_1, \beta_2 \le 1$, and
    \begin{subequations}
    \label{eq:WI_con2-Thm}
    \begin{eqnarray}
        \label{eq:WI_con2-Thm-1}
        \SNRA_{32}(1+\SNRA_{21})^2 &\le& \beta_1\bigg(\SNRA_{31}\big(1-\beta_2\big)-2\SNRA_{32}\SNRA_{11}\bigg)\\
        \label{eq:WI_con2-Thm-2}
        \SNRA_{21}(1+\SNRA_{32})^2 &\le& \beta_2 \cdot \SNRA_{22}\big(1-\beta_1\big),
    \end{eqnarray}
    \end{subequations}
   then, the sum-rate capacity of the channel is given by
   \begin{equation}
        \label{eq:sumcapacityPhase}
        \sup_{(R_1,R_2)\in\mathcal{C}(\underline{\footnotesize\SNR})}\big(R_1+R_2\big)=\log\bigg(1+\frac{\SNRA_{11}+\SNRA_{31}}{1+\SNRA_{21}}\bigg)+\log\bigg(1+\frac{\SNRA_{22}}{1+\SNRA_{32}}\bigg),
    \end{equation}
    and it is achieved by $X_k\sim\CN(0,1), k\in\{1,2,3\}$, mutually independent.
\end{theorem}

\begin{remark}
\em Observe that the conditions in \eqref{eq:WI_con2-Thm} are satisfied if $\SNRA_{32}$ and $\SNRA_{21}$ are small compared to $\SNRA_{31}$ and $\SNRA_{22}$, respectively. %In particular, \eqref{eq:WI_con2-Thm-1} and \eqref{eq:MAXPOWCONTh-1} correspond to the interference at Rx$_1$ and \eqref{eq:WI_con2-Thm-2} and \eqref{eq:MAXPOWCONTh-2} correspond to the interference at Rx$_2$.
As $\SNRA_{32}$ and $\SNRA_{21}$ correspond to the strengths of the interfering links, conditions \eqref{eq:WI_con2-Thm} correspond to the WI regime.
To make this point more explicit, note that $\SNRA_{32}(1+\SNRA_{21})^2 \ge \SNRA_{32}$, hence \eqref{eq:WI_con2-Thm-1} implies that
$\SNRA_{32} \le {\beta_1\cdot \SNRA_{31}\big(1-\beta_2\big)}$, and similarly \eqref{eq:WI_con2-Thm-2} implies that  $\SNRA_{21} \le \beta_2 \cdot \SNRA_{22}\big(1-\beta_1\big)$.
\end{remark}

\begin{remark}
\label{rem:con7}
\em         Note that condition \eqref{eq:Relay_con2Phase} corresponds to good reception at the relay, in the sense that decoding the message sent by Tx$_1$ at the relay does not
			constrain the information rate from Tx$_1$ to Rx$_1$.
			This condition facilitates the sum-rate optimality of DF, as the constraints on the achievable rates are now only due to
			the rate constraints for reliable decoding at the destinations.
\end{remark}

\begin{remark}
\label{rem:ergodic_compare}
\em{}
        Note that in the ergodic phase fading case,
		the magnitudes of the channel coefficients are constants while the phases of the channel coefficients vary i.i.d. over time and
		are mutually independent across the fading links.  Thus, in the ergodic phase fading model, the channel coefficients induce
		randomly varying phases upon the components of the received signal arriving at each receiver after traveling across the different links.
		Intuitively, having mutually independent and uniformly distributed i.i.d. phases does not allow achieving non-zero correlation between the
		components of the received signal, and consequently implies that there is no loss of optimality in transmitting uncorrelated codewords. In
		particular, if the optimal input distribution is complex Normal, then the absence of correlation between the codebooks implies that
		the optimal codebooks are generated independently of each other. Indeed, in the derivations in the manuscript, it is rigorously proved
		that the optimal channel inputs for the ergodic phase fading Z-ICR are generated according to mutually independent
		complex Normal random variables. The optimality of mutually
		independent channel inputs is one of the fundamental advantages of the communications scheme we use in this manuscript, since it means that
		there is no need for coordinated transmission to optimally benefit from the relay. As will be clarified later, this fact greatly simplifies
        both the achievability scheme as well as the practical incorporation of cooperative transmission in interference networks.
        In contrast, for the no-fading case (commonly referred to as the AWGN channel) both
		the magnitudes and the phases of the channel coefficients are constants. Consequently, in the no-fading channel the correlation between
		the channel inputs is maintained at the received signal components, and hence, the optimal codebooks may be correlated. This fact
		greatly complicates the optimal achievability scheme as well as makes the derivations for the upper bounds significantly more complicated.

\end{remark}

\begin{proof}
The proof of Thm. \ref{thm:WI} consists of the following three steps:
\begin{enumerate}
    \item We derive an upper bound on the sum-rate of the ergodic phase fading Z-ICR by letting each receiver observe an appropriate genie signal.
        In particular, we show that the upper bound is maximized by mutually independent, zero-mean circularly symmetric complex Normal channel inputs, i.i.d. in time3
    \item We characterize an achievable rate region for the Z-ICR by using codebooks generated according to a mutually independent circularly symmetric complex Normal distribution, i.i.d. in time, and by employing the DF scheme at the relay, together with treating the interfering signal as noise at each receiver.
    \item %We derive conditions on the channel coefficients $\underline{\SNR}$ which guarantee that the upper bound on the sum-rate coincides with the achievable sum-rate.
        Combining the conditions for the upper bound and for the lower bound we obtain the conditions for characterizing the sum-rate capacity of the ergodic phase fading Z-ICR in the WI regime, and explicitly state the corresponding expressions.\end{enumerate}

In the following subsections, we provide a detailed proof for the above steps: Step 1 is carried out in Section \ref{sec:SUMCAUB}, Step 2 is carried out in Section \ref{sec:achievableRegionIID}, and finally, Step 3 is detailed in Section \ref{sec:SUMCADER}.

\subsection{Step 1 : An Upper Bound on the Sum-Rate Capacity}
\label{sec:SUMCAUB}
The upper bound on the sum-rate capacity of the ergodic phase fading Z-ICR is summarized in the following theorem:
\begin{theorem}
    \label{thm:SRCA-UB}
    Consider the phase fading Z-ICR with only Rx-CSI, defined in Section \ref{sec:Model}. If there are two real scalars $\beta_1$ and $\beta_2$ which satisfy $0\le \beta_1, \beta_2 \le 1$, and
    \begin{subequations}
    \label{eq:WI_con2-Thm2}
    \begin{eqnarray}
        \SNRA_{32}(1+\SNRA_{21})^2 &\le& \beta_1\bigg(\SNRA_{31}\big(1-\beta_2\big)-2\SNRA_{32}\SNRA_{11}\bigg)\\
        \SNRA_{21}(1+\SNRA_{32})^2 &\le& \beta_2 \cdot \SNRA_{22}\big(1-\beta_1\big),
    \end{eqnarray}
    \end{subequations}
    then, the sum-rate capacity is upper bounded by
\begin{equation}
\label{eq:sumrate_OB2Th}
    \sup_{(R_1,R_2)\in\mathcal{C}(\underline{\footnotesize\SNR})}\big(R_1+R_2\big)\le I(X_{1},X_{3};Y_{1}|\tH_1) + I(X_{2};Y_{2}|\tH_2),
\end{equation}
where the mutual information expressions are evaluated with mutually independent, zero mean, circularly symmetric complex Normal channel inputs, distributed according to $X_k\sim\CN(0,1)$.
\end{theorem}

\begin{proof}
We use a genie to provide additional information to the receivers. Let $W^n_{1}$ and $W^n_{2}$ be two arbitrarily correlated, circularly symmetric, complex Normal random vectors, each with i.i.d. elements distributed $\CN(0,1)$, such that $f_{W^n_{1},W^n_{2}}(w^n_{1},w^n_{2})=\prod_{i=1}^{n}f_{W_{1},W_{2}}(w_{1,i},w_{2,i})$. In addition, $(W^n_{1},W^n_{2})$ are independent of $(X_1^n, X_2^n, X_3^n)$. For $i\in\{1,2,...,n\}$, and we further let $W_{k,i}$ and $Z_{k,i}$, $k=1,2$ be jointly circularly symmetric with correlation matrix\footnote{Joint circular symmetry of $W_{k,i}$ and $Z_{k,i}$ implies that $\E\{\Real\{W_{k,i}\}\Img\{Z_{k,i}\}\} = -\E\{\Img\{W_{k,i}\}\Real\{Z_{k,i}\}\}$ and $\E\{\Real\{W_{k,i}\}\Real\{Z_{k,i}\}\} = \E\{\Img\{W_{k,i}\}\Img\{Z_{k,i}\}\}$.}:
    \begin{eqnarray*}
        \cov(W_{k,i},Z_{k,i})=\E\left\{\left[\begin{array}{c}
                        W_{k,i}\\
                        Z_{k,i}\end{array} \right]\big[W_{k,i}^*\quad Z_{k,i}^*\big]\right\}=\left[\begin{array}{cc}
                        1 & \CORRN_k\\
                        \CORRN_k^* & 1\end{array} \right], \qquad k\in\{1,2\}.
    \end{eqnarray*}
Note that since $\var(W_k)=\var(Z_k)=1, k=1,2$, then $|\CORRN_k|\le\sqrt{\var(W_k)\cdot\var(Z_k)}=1$. Define the signals
\begin{subequations}
\label{eq:Genie1RV1}
\begin{eqnarray}
\label{eq:S1def}
    S_{1,i}&\triangleq& H_{11,i} X_{1,i}+H_{31,i} X_{3,i}+ \eta_1W_{1,i}\\
\label{eq:S2def}
    S_{2,i}&\triangleq& H_{22,i} X_{2,i}+ \eta_2W_{2,i},
\end{eqnarray}
\end{subequations}
$i\in\{1,2,...,n\}$, where $\eta_1$ and $\eta_2$ are two complex-valued constants determined by the genie. Assume that at time $i$, the genie provides the signals $S_{1,i}$ and $S_{2,i}$
%and the associated Rx-CSI, $\tH_{1,i}$ and $\tH_{2,i}$,
to Rx$_1$ and Rx$_2$, respectively. For an achievable rate pair $(R_1,R_2)$, let $(\Xvec_1,\Xvec_2,\Xvec_3)$ be random vectors of length $n$ representing the statistics of the achievable codebook, and define
\begin{equation}
    \label{eq:eq28}
    q_{k,i}\triangleq \cov\big(X_{k,i}\big) \equiv \var(X_{k,i}), k\in\{1,2,3\}, \qquad \Qmat^{(i)}_{\Xvec}\triangleq\cov(X_{1,i},X_{2,i},X_{3,i}), \qquad i\in\{1,2,...,n\},
\end{equation}
and
\[
        P_k = \frac{1}{n}\sum_{i=1}^n q_{k,i}, \qquad k\in\{1,2,3\}.
\]
It is emphasized that at this point, properness of $\big(\Xvec_1, \Xvec_2, \Xvec_3\big)$ is not assumed.
Next, recall that in Comment \ref{rem:IndiSignals} we concluded that $\Xvec_2$ is independent of $(\Xvec_1,\Xvec_3)$, while $\Xvec_1$ and $\Xvec_3$ may be statistically dependent. It follows that the $n$-letter input distribution for $(\Xvec_1,\Xvec_2,\Xvec_3)$ must satisfy $f_{\Xvec_1,\Xvec_2,\Xvec_3}(\xvec_1,\xvec_2,\xvec_3)=f_{\Xvec_1,\Xvec_3}(\xvec_1,\xvec_3)\cdot f_{\Xvec_2}(\xvec_2)$.
From this observation, we conclude that for $i\in\{1,2,...,n\}$
\begin{subequations}
  \label{eq:QG}
\begin{eqnarray}
\label{eqn:QG_1}
    \Qmat^{(i)}_{\Xvec} & = & \left[\begin{array}{ccc}
                        q_{1,i}\hspace{2mm} & 0\hspace{2mm} &\CORR_i\sqrt{q_{1,i}q_{3,i}}\\
                        0\hspace{2mm}   & q_{2,i}\hspace{2mm} & 0\\
                        \CORR_i^*\sqrt{q_{1,i}q_{3,i}\hspace{2mm}} & 0\hspace{2mm} & q_{3,i} \end{array} \right], \\
    \Qmat_G & \triangleq & \frac{1}{n}\sum_{i=1}^{n}\Qmat^{(i)}_{\Xvec} \nonumber\\
\label{eqn:QG_2}
     & \equiv & \left[\begin{array}{ccc}
                        P_1\hspace{2mm} & 0\hspace{2mm} &\CORR\sqrt{P_{1}P_{3}}\\
                        0\hspace{2mm}   & P_{2}\hspace{2mm} & 0\\
                        \CORR^*\sqrt{P_{1}P_{3}}\hspace{2mm} & 0\hspace{2mm} & P_{3} \end{array} \right],
\end{eqnarray}
\end{subequations}
where $\big|\CORR_i\big| \le 1$  in \eqref{eqn:QG_1}.
Note that by the Cauchy-Schwartz inequality \cite{Steele}
\begin{eqnarray*}
    \left|\frac{1}{n}\sum_{i=1}^n \CORR_i \sqrt{q_{1,i}q_{3,i}}\right|  & \le & \sqrt{\frac{1}{n}\sum_{i=1}^n \left|\CORR_i \sqrt{q_{1,i}}\right|^2}\sqrt{\frac{1}{n}\sum_{i=1}^n \left| \sqrt{q_{3,i}}\right|^2}\\
     & \le &  \sqrt{\frac{1}{n}\sum_{i=1}^n  q_{1,i}}\sqrt{\frac{1}{n}\sum_{i=1}^n q_{3,i}}\\
     & = & \sqrt{P_1P_3},
\end{eqnarray*}
thus $|\CORR|\le 1$ in \eqref{eqn:QG_2}.
Lastly, define the random vector $(X_{1G},X_{2G},X_{3G})^T\sim\CN({\bf 0},\Qmat_G)$.

Let $M_1$ be the transmitted message at Tx$_1$,  $\hat{M}_1$ denote the decoded message at Rx$_1$, and $P_{e,1}^{(n)}$, denote the probability of error in
 decoding $M_1$ at Rx$_1$. The rate $R_1$ can be upper bounded as follows:
\begin{eqnarray*}
    nR_1&=&H(M_1)\\
    &=&H(M_1|Y_1^n,\tH_1^n)+I(M_1;Y_1^n,\tH_1^n)\\
    &\stackrel{(a)}{\le}& 1 + P_{e,1}^{(n)}nR_1+I(M_1;Y_1^n,\tH_1^n)\\
    &\stackrel{(b)}{\le}& 1 + P_{e,1}^{(n)}nR_1+I(X_1^n;Y_1^n,\tH_1^n)\\
    &\stackrel{(c)}{=}& 1 + P_{e,1}^{(n)}nR_1+I(X_1^n;Y_1^n|\tH_1^n)\\
    &\stackrel{(d)}{\le}& 1 + P_{e,1}^{(n)}nR_1+I(X_1^n,X_3^n;Y_1^n,S_1^n|\tH_1^n).
\end{eqnarray*}
Here, (a) follows from Fano's inequality \cite[Thm. 2.10.1]{cover-thomas:it-book}; (b) follows from the data processing inequality \cite[Thm. 2.8.1]{cover-thomas:it-book}, since $M_1-X_1^n-(Y_1^n,\tH_1^n)$ form a Markov chain; (c) follows since the transmitted symbols $X_{1}^n$ are independent of the channel coefficients $\tH_1^n$; and (d) follows from the chain rule of mutual information and since mutual information is non-negative. Next, define $n\epsilon_{1n}\triangleq1 + P_{e,1}^{(n)}nR_1$, and observe that since
$(R_1,R_2)$ is achievable, then $P_{e,1}^{(n)}\rightarrow 0$ for $n\rightarrow\infty$, and therefore $\epsilon_{1n}\rightarrow 0$ as $n\rightarrow\infty$. Hence, we obtain
\begin{eqnarray}
    \hspace{-0.5 cm} n(R_1-\epsilon_{1n})\!\!&\le&\!I(X_1^n,X_3^n;Y_1^n,S_1^n|\tH_1^n)\nonumber\\
    &=&\!I(X_1^n,X_3^n;S_1^n|\tH_1^n)+I(X_1^n,X_3^n;Y_1^n|S_1^n,\tH_1^n)\nonumber\\
    &=&\!h(S_1^n|\tH_1^n)-h(S_1^n|X_1^n,X_3^n,\tH_1^n)+h(Y_1^n|S_1^n,\tH_1^n)-h(Y_1^n|X_1^n,X_3^n,S_1^n,\tH_1^n)\nonumber\\
    &\stackrel{(a)}{\le}&\!h(S_1^n|\tH_1^n)-h(S_1^n|X_1^n,X_3^n,\tH_1^n)+n\cdot h(Y_{1G}|S_{1G},\tH_1)-h(Y_1^n|X_1^n,X_3^n,S_1^n,\tH_1^n)\nonumber\\
    &\stackrel{(b)}{=}&\!h(S_1^n|\tH_1^n)-n\cdot h(S_{1G}|X_{1G},X_{3G},\tH_1)+n\cdot h(Y_{1G}|S_{1G},\tH_1)-h(Y_1^n|X_1^n,X_3^n,S_1^n,\tH_1^n),\label{eq:R1_OB1}
\end{eqnarray}
where $(X_{1G},X_{2G},X_{3G})^T\sim\CN({\bf 0},\Qmat_G)$ with $\Qmat_G$ defined in \eqref{eqn:QG_2}, and $S_{1G}$ is obtained from \eqref{eq:S1def} by replacing $X_{1}$ and $X_{3}$ with $X_{1G}$ and $X_{3G}$, respectively. In the above transitions, (a) follows from Lemma \ref{lem:lemma22} using the assignment $\Xvec=(X_1,X_2,X_3)^T$, $\Hvec_{Y}=(H_{11},H_{21},H_{31})^T$, $\Hvec_{S}=(H_{11},0,H_{31})^T$, and with $\Qmat_{\Xvec_i} \equiv \Qmat^{(i)}_{\Xvec}$ defined in \eqref{eqn:QG_1}, and (b) follows from the transitions detailed below:
\begin{eqnarray*}
    h\big(S_1^n|X_1^n,X_3^n,\tH_1^n\big)
    &=&h\big(\eta_1\cdot(W_1^n)|X_1^n,X_3^n,\tH_1^n\big)\\
    &\stackrel{(c)}{=}&h\big(\eta_1\cdot(W_1^n)\big)\\
    &\stackrel{(d)}{=}&n\cdot h\big(\eta_1W_1\big)\\
    &\stackrel{(e)}{=}&n\cdot h\big(\eta_1W_1|X_{1G},X_{3G},\tH_1\big)\\
    &=&n\cdot h\big(S_{1G}|X_{1G},X_{3G},\tH_1\big),
\end{eqnarray*}
where step (c) follows since $W_1^n$ is independent of $(X_1^n,X_3^n,\tH_1^n)$; step (d) follows since $W_1^n$ has i.i.d. entries, and step (e) is valid for any joint distribution on $(X_{1G},X_{3G})$ independent of $W_1$, and thus, in this step we let $(X_{1G},X_{3G})$ be distributed according to the joint distribution $\CN({\bf 0},\Qmat_{G13})$,
where
\begin{equation*}
\Qmat_{G13}\triangleq\left[\begin{array}{cc}
                        P_1\hspace{2mm} &\CORR\sqrt{P_{1}P_{3}}\\
                        \CORR^*\sqrt{P_{1}P_{3}}\hspace{2mm}  & P_{3} \end{array} \right].
\end{equation*}
 Applying similar steps and identical arguments for $R_2$, we obtain the upper bound:
%\begin{equation}
%    n(R_2-\epsilon_{2n}) \le
%    h(S_2^n|\tH_2^n)-n\cdot\!h(S_{2G}|X_{2G},\tH_2)+n\cdot h(Y_{2G}|S_{2G},\tH_2)-h(Y_2^n|X_2^n,S_2^n,\tH_2^n).\label{eq:R1_OB2}
%\end{equation}
    \begin{eqnarray}
    \hspace{-0.5 cm}n(R_2-\epsilon_{2n}) &\le& I(X_2^n;Y_2^n,S_2^n|\tH_2^n)\nonumber\\
    &=&I(X_2^n;S_2^n|\tH_2^n)+I(X_2^n;Y_2^n|S_2^n,\tH_2^n)\nonumber\\
    &=&h(S_2^n|\tH_2^n)-h(S_2^n|X_2^n,\tH_2^n)+h(Y_2^n|S_2^n,\tH_2^n)-h(Y_2^n|X_2^n,S_2^n,\tH_2^n)\nonumber\\
    &\le&h(S_2^n|\tH_2^n)-h(S_2^n|X_2^n,\tH_2^n)+n\cdot h(Y_{2G}|S_{2G},\tH_2)-h(Y_2^n|X_2^n,S_2^n,\tH_2^n)\nonumber\\
    &=&h(S_2^n|\tH_2^n)-n\cdot\!h(S_{2G}|X_{2G},\tH_2)+n\cdot h(Y_{2G}|S_{2G},\tH_2)-h(Y_2^n|X_2^n,S_2^n,\tH_2^n).\label{eq:R1_OB2}
\end{eqnarray}
Since the maximizing complex Normal distributions in \eqref{eq:R1_OB1} and \eqref{eq:R1_OB2} are identical and equal to $(X_{1G},X_{2G},X_{3G})^T\sim\CN(\mathbf{0},\Qmat_G)$, then \eqref{eq:R1_OB1} and \eqref{eq:R1_OB2} can be combined into  a single bound on the sum-rate:
\begin{eqnarray}
     &&n(R_1+R_2-\epsilon_{1n}-\epsilon_{2n})\le\nonumber\\
     &&\qquad n\cdot\bigg(h(Y_{1G}|S_{1G},\tH_1)-h(S_{1G}|X_{1G},X_{3G},\tH_1)+h(Y_{2G}|S_{2G},\tH_2)-h(S_{2G}|X_{2G},\tH_2)\bigg)\nonumber\\
     &&\qquad\qquad\qquad\qquad +\bigg(h(S_1^n|\tH_1^n)-h(Y_1^n|X_1^n,X_3^n,S_1^n,\tH_1^n) + h(S_2^n|\tH_2^n)-h(Y_2^n|X_2^n,S_2^n,\tH_2^n)\bigg).
    \label{eq:sumrate_OB1}
\end{eqnarray}

In the following proposition, we identify the maximizing distribution for the first brackets in the right-hand side of \eqref{eq:sumrate_OB1}:
\begin{proposition}
    \label{prop:proposition1}
    The expression in the first brackets in the right-hand side of \eqref{eq:sumrate_OB1} is maximized with mutually independent circularly symmetric complex Normal channel inputs distributed according to $X_k\sim\CN(0,1), k\in\{1,2,3\}$.
\end{proposition}
\begin{proof}
    The proof is provided in Appendix \ref{app:appA}.
\end{proof}

Next, we show that the expression in the second brackets in the right-hand side of \eqref{eq:sumrate_OB1} is also maximized by mutually independent and i.i.d. in time channel inputs, distributed according to $X_{k,i}\sim\CN(0,1), k\in\{1,2,3\}, i\in\{1,2,...,n\}$. Assume that there exists a pair of complex scalars $\CORRN_1$ and $\eta_2$ s.t. $|\CORRN_1|\le 1$ and  $\SNRA_{21}|\eta_2|^2 \le \SNRA_{22}(1-|\CORRN_1|^2)$, and let $V^n_1$ be an $n$-dimensional random vector with i.i.d. elements distributed according to $V_{1,i}\sim\CN(0,1-|\CORRN_1|^2), i\in\{1,2,...,n\}$. Additionally, let $\Hmat_{h_{lk}}^{(n)}$ be an $n\times n$ diagonal matrix s.t. $\big[\Hmat_{h_{lk}}^{(n)}\big]_{i,i}=h_{lk,i}$. Then
\begin{eqnarray}
    &&\hspace{-2 cm}h(S_2^n|\tH_2^n)-h(Y_1^n|X_1^n,X_3^n,S_1^n,\tH_1^n)\nonumber\\
    &=&   \E_{\tH_1^n,\tH_2^n}\Big\{h\big(\Hmat_{h_{22}}^{(n)}X^n_2+\eta_2W^n_2|\tH_2^n=\th_2^n\big)-h\big(\Hmat_{h_{21}}^{(n)}X^n_{2}+Z^n_{1}|W^n_{1},\tH_1^n=\th_1^n\big)\Big\}\nonumber\\
    &\stackrel{(a)}{=}&   \E_{\tH_1^n,\tH_2^n}\Big\{h\big(\Hmat_{h_{22}}^{(n)}X^n_2+\eta_2W^n_2|\tH_2^n=\th_2^n\big)-h\big(\Hmat
    _{h_{21}}^{(n)}X^n_{2}+V^n_{1}|\tH_1^n=\th_1^n\big)\Big\}\nonumber\\
    &\stackrel{(b)}{\le}& \E_{\tH_1,\tH_2}\Big\{n\cdot h\big(h_{22}X_{2G}+\eta_2W_2|\tH_2=\th_2\big)-n\cdot h\big(h_{21}X_{2G}+V_{1}|\tH_1=\th_1\big)\Big\}\nonumber\\
    &\stackrel{(c)}{\le}& n\cdot\bigg(h(H_{22}X_{2G}+\eta_2W_2|\tH_2)-h(H_{21}X_{2G}+V_{1}|\tH_1)\bigg)\bigg|_{\substack{P_1=P_2=P_3=1\\\CORR=0}},\label{eq:Leg1ofI}
\end{eqnarray}
where (a) follows from Lemma \ref{lem:lemma2}; for step (b) we first use \cite[Eq. (13)]{Massey:93}\footnote{For a complex random vector $\Xvec$ and any complex matrix $\Amat$ it holds that $h\big(\Amat\cdot\Xvec\big)= h\big(\Xvec\big)+2\cdot\log\big|\det(\Amat)\big|.$} and obtain that since $\th_1^n$ and $\th_2^n$ are given, then we can write
\begin{eqnarray*}
    h\big(\Hmat_{h_{22}}^{(n)}X^n_2+\eta_2W^n_2|\tH_2^n=\th_2^n\big)&=&h\left(X^n_{2}+\Big(\Hmat_{h_{22}}^{(n)}\Big)^{-1}\eta_2W^n_2\Big|\tH_2^n=\th_2^n\right)+\log\Big(\big(\SNRA_{22}\big)^{n}\Big)\\
    h\big(\Hmat_{h_{21}}^{(n)}X^n_2+V^n_{1}|\tH_1^n=\th_1^n\big)&=&h\left(X^n_{2}+\Big(\Hmat_{h_{21}}^{(n)}\Big)^{-1}V^n_{1}\Big|\tH_1^n=\th_1^n\right)+\log\Big(\big(\SNRA_{21}\big)^{n}\Big).
\end{eqnarray*}
Next, we note that since the magnitudes of channel coefficients for each link are equal, then the noise vectors $\Big(\Hmat_{h_{22}}^{(n)}\Big)^{-1}\eta_2W^n_2$ and $\Big(\Hmat_{h_{21}}^{(n)}\Big)^{-1}V^n_{1}$, each has i.i.d. elements. Step (b) now follows from Lemma \ref{lem:lemma1} which states that if $\SNRA_{21}|\eta_2|^2 \le \SNRA_{22}\cdot(1-|\CORRN_1|^2)$, then, subject to the trace constraint $\tr\big(\cov(X_2^n)\big)\le \sum_{i=1}^n q_{2,i}\equiv nP_2$, we have that $h\left(X_{2}^n+\Big({\Hmat_{h_{22}}^{(n)}}\Big)^{-1}\eta_2(W_2)^n\Big|\tH_2^n=\th_2^n\right)- h\left(X_{2}^n+\Big({\Hmat_{h_{21}}^{(n)}}\Big)^{-1}V_{1}^n\Big|\tH_1^n=\th_1^n\right)$ is maximized by $X_2^n$ distributed according to a circularly symmetric complex Normal
distribution, with i.i.d. elements, each distributed according to $X_{2G}\sim\CN(0,P_2)$. To prove Step (c) recall that $0\le P_1,P_2,P_3\le1$; Step (c) then follows since $\Big(h(H_{22}X_{2G}+\eta_2W_2|\tH_2)-h(H_{21}X_{2G}+V_{1}|\tH_1)\Big)$ does not depend on $(\CORR,P_1,P_3)$ and thus, we can set $\CORR=0, P_1=P_3=1$. Additionally, if $\SNRA_{21}|\eta_2|^2 \le \SNRA_{22}(1-|\CORRN_1|^2)$, then the derivative of the expression in step (b) with respect to $P_2$ is non-negative\footnote{
The derivative is $\frac{\partial}{\partial P_2} \log\left( \frac{{\scriptsize \SNRA}_{22}P_2+|\eta_2|^2}{{\scriptsize \SNRA}_{21}P_2 + (1-|\CORRN_2|^2)} \right)
   =   \frac{{\scriptsize \SNRA}_{21}P_2 + (1-|\CORRN_2|^2)}{ {\scriptsize \SNRA}_{22}P_2+|\eta_2|^2} \cdot \frac{    {\scriptsize \SNRA}_{22}\left({\scriptsize \SNRA}_{21}P_2 + (1-|\CORRN_2|^2)\right)    - {\scriptsize \SNRA}_{21}\left({\scriptsize \SNRA}_{22}P_2+|\eta_2|^2\right)}{\left({\scriptsize \SNRA}_{21}P_2 + (1-|\CORRN_2|^2)\right)^2} =
%   \frac{\SNRA_{21}P_2 + (1-|\CORRN_2|^2)}{ \SNRA_{22}P_2+|\eta_2|^2} \cdot \frac{    \SNRA_{22}(1-|\CORRN_2|^2)    - \SNRA_{21}|\eta_2|^2}{\left(\SNRA_{21}P_2 + (1-|\CORRN_2|^2)\right)^2} $.
    \frac{    {\scriptsize \SNRA}_{22}(1-|\CORRN_2|^2)    - {\scriptsize \SNRA}_{21}|\eta_2|^2}{\left({\scriptsize \SNRA}_{22}P_2+|\eta_2|^2\right)\left({\scriptsize \SNRA}_{21}P_2 + (1-|\CORRN_2|^2)\right)} $.
},
and thus, this expression is a non-decreasing function of $P_2$, from which we conclude that it is maximized with $P_2=1$.

Next, assume that there exists a pair of complex scalars $\CORRN_2$ and $\eta_1$ s.t.
\begin{subequations}
\label{eqn:derivation_iii inputs_part9_0}
\begin{equation}
\label{eqn:derivation_iii inputs_part9_1}
|\CORRN_2|\le 1
\end{equation}
and
\begin{equation}
\label{eqn:derivation_iii inputs_part9_2}
 \SNRA_{32}|\eta_1|^2 \le \SNRA_{31}\big(1-|\CORRN_2|^2\big)-2\SNRA_{32}\SNRA_{11},
\end{equation}
\end{subequations}
and let $V^n_2$ be an $n$-dimensional random vector with i.i.d. elements, each distributed according to $V_{2,i}\sim\CN(0,1-|\CORRN_2|^2), i\in\{1,2,...,n\}$. It now follows that
\begin{eqnarray}
    &&\hspace{-2cm}h(S_1^n|\tH_1^n)-h(Y_2^n|X_2^n,S_2^n,\tH_2^n)\nonumber\\
    &\stackrel{(a)}{=}& \E_{\tH_1^n,\tH_2^n}\Big\{ h\big(\Hmat_{h_{11}}^{(n)}X^n_1+\Hmat_{h_{31}}^{(n)}X^n_3+\eta_1W^n_{1}|\tH_1^n=\th_1^n\big)
                -h\big(\Hmat_{h_{32}}^{(n)}X^n_3+Z^n_{2}|W^n_{2},\tH_2^n=\th_2^n\big)\Big\}\nonumber\\
    &\stackrel{(b)}{=}& \E_{\tH_1^n,\tH_2^n}\Big\{ h\big(\Hmat_{h_{11}}^{(n)}X^n_1+\Hmat_{h_{31}}^{(n)}X^n_3+\eta_1W^n_{1}|\tH_1^n=\th_1^n\big)
                -h\big(\Hmat_{h_{32}}^{(n)}X^n_3+V^n_{2}|\tH_2^n=\th_2^n\big)\Big\}\nonumber\\
    &\stackrel{(c)}{\le}& \E_{\tH_1^n,\tH_2^n}\Big\{ h\big(\Hmat_{h_{11}}^{(n)}X^n_{1\bar{G}}+\Hmat_{h_{31}}^{(n)}X^n_{3\bar{G}}+\eta_1W^n_{1}|\tH_1^n=\th_1^n\big)
                -h\big(\Hmat_{h_{32}}^{(n)}X^n_{3\bar{G}}+V^n_{2}|\tH_2^n=\th_2^n\big)\Big\}\nonumber\\
    &\stackrel{(d)}{\le}& n\cdot\bigg( h(H_{11}X_{1G}+H_{31}X_{3G}+\eta_1W_{1}|\tH_1)-h(H_{32}X_{3G}+V_{2}|\tH_2)\bigg)\bigg|_{\substack{P_1=P_2=P_3=1\\\CORR=0}}.\label{eq:Leg33ofI}
\end{eqnarray}
In the above transitions (a) follows from the fact that the relay receives transmissions only from Tx$_1$, which makes $(X^n_1,X^n_3)$ necessarily independent of $X^n_2$, and (b) follows from Lemma \ref{lem:lemma2}.
For (c) we apply Lemma \ref{lem:lemma8} by first setting
        \begin{equation}
            \label{eq:eq13}
            X^{2n}=\Big((X^n_1)^T,(X^n_3)^T\Big)^T,\;\;\; \tilde{\Vmat}_{1}^H=\Big(\Hmat_{h_{11}}^{(n)},\Hmat_{h_{31}}^{(n)}\Big), \;\;\; \tilde{\Vmat}_{2}^H=\Big(\mathds{O}^{(n)},\Hmat_{h_{32}}^{(n)}\Big), \;\;\; \tilde{Z}_1^n\triangleq \eta_1W^n_{1},\;\;\;  \tilde{Z}_2^n\triangleq V^n_{2},
        \end{equation}
        where $\mathds{O}^{(n)}$ is an $n\times n$ matrix in which all entries are equal to zero.
        To determine the matrix $\mathds{S}$ for the application of Lemma \ref{lem:lemma8} we consider
         the random vectors $(\Xvec_1,{\Xvec}_2,{\Xvec}_3)$ {\em corresponding to the achievable code},
        and for $k\in\{1,3\}$, we let $\Xvec_{kR}\triangleq\Real\{\Xvec_{k}\}$ and $\Xvec_{kI}\triangleq\Img\{\Xvec_{k}\}$ be two $n\times 1$ real random vectors,  and $\bar{\Xvec}_k\triangleq(\Xvec_{kR}^T,\Xvec_{kI}^T)^T$ be an $2n \times 1$ real random vector, where $k\in\{1,3\}$. Lastly, we define the random vectors $\bar{\Xvec}_R\triangleq \big(\Xvec_{1R}^T,\Xvec_{3R}^T\big)^T$, $\bar{\Xvec}_I\triangleq \big(\Xvec_{1I}^T,\Xvec_{3I}^T\big)^T$, and $\bar{\Xvec} \triangleq \left((\bar{\Xvec}_R)^T, (\bar{\Xvec}_I)^T\right)^T$, all corresponding to the achievable code.  The matrix $\mathds{S}$ for the application of Lemma \ref{lem:lemma8} is determined via
        \begin{equation*}
          \mathds{S}\triangleq \cov(\bar{\Xvec})=\left[\begin{array}{cc}
                            \Qmat_{RR}    & \Qmat_{RI}\\
                            \Qmat_{RI}^T & \Qmat_{II}\end{array} \right], \qquad \Qmat_{RR}= \cov\big(\bar{\Xvec}_R\big), \qquad \Qmat_{II}= \cov\big(\bar{\Xvec}_I\big), \qquad \Qmat_{RI}=\cov\big(\bar{\Xvec}_R,\bar{\Xvec}_I\big).
        \end{equation*}
       Lastly we note that $\tilde{\Vmat}_{1}^H\cdot\tilde{\Vmat}_{1}=(\SNRA_{11}+\SNRA_{31})\cdot\Imat_{n}$ and $\tilde{\Vmat}_{2}^H\cdot\tilde{\Vmat}_{2}=\SNRA_{32}\cdot\Imat_{n}$, which satisfies the conditions of Lemma \ref{lem:lemma8}. It thus follows that $\Big(\big({X^n_{1\bar{G}}}\big)^T,\big({X^n_{3\bar{G}}}\big)^T\Big)^T=X^{2n}_{\bar{G}}$ is a $2n \times 1$
      zero mean {\em complex} Normal random vector whose covariance matrix satisfies $\cov\bigg(\Big(\big(\Real\{X^{2n}_{\bar{G}}\}\big)^T,\big(\Img\{X^{2n}_{\bar{G}}\}\big)^T\Big)^T\bigg)\preceq\mathds{S}$.
      Consequently, we have that the maximizing $\left(X_{1\oG},X_{3\oG} \right)$,
      obtained from the optimal $X_{\oG}^{2n}$, satisfy for $k = 1, 3$:
      \begin{subequations}\label{eqn:rels_for_stp_e}
        \begin{eqnarray}
        \E\big\{|X_{k\oG,i}|^2\big\} & \le & \E\big\{|X_{k,i}|^2\big\} =1\\
            \tr\big\{\cov(X_{k\oG}^n)\big\} & = & \tr\Big\{\cov\big(\Real\left\{X_{k\oG}^n\right\}\big)\Big\}+\tr\Big\{\cov\big(\Img\left\{X_{k\oG }^n\right\}\big)\Big\}\nonumber\\
            & \le & \tr\big\{\cov(\Xvec_{kR})\big\}+\tr\big\{\cov(\Xvec_{kI})\big\}\nonumber\\
            & = & \sum_{i=1}^{n} q_{k,i}\nonumber\\
            & \le &  n,
        \end{eqnarray}
      \end{subequations}
        where $q_{k,i}$ denotes the variance of {\em complex} symbol $X_{k}$ at time index $i\in\{1,2,...,n\}$, $ k\in\{1,3\}$, in the achievable code.
        Lastly, Step (d) is proved in Appendix \ref{app:appB} using relationships \eqref{eqn:rels_for_stp_e}.

Plugging \eqref{eq:Leg1ofI} and \eqref{eq:Leg33ofI} into the second line of \eqref{eq:sumrate_OB1}, we conclude that if it is possible to choose $\CORRN_1,\CORRN_2,\eta_1$ and $\eta_2$ s.t.
\begin{subequations}
\label{eq:WI-Firsthalf}
\begin{eqnarray}
\label{eq:WI-Firsthalf_1stcond}
    \SNRA_{32}|\eta_1|^2 &\le& \SNRA_{31}\big(1-|\CORRN_2|^2\big)-2\SNRA_{32}\SNRA_{11}\\
    \SNRA_{21}|\eta_2|^2 &\le& \SNRA_{22}\big(1-|\CORRN_1|^2\big),
\end{eqnarray}
\end{subequations}
then the sum-rate is upper bounded by
\begin{eqnarray*}
     &&\!\!\!\!\!n(R_1+R_2-\epsilon_{1n}-\epsilon_{2n})\le\nonumber\\
     &&\qquad n\cdot\bigg(h(Y_{1G}|S_{1G},\tH_1)-h(S_{1G}|X_{1G},X_{3G},\tH_1)+h(Y_{2G}|S_{2G},\tH_2)-h(S_{2G}|X_{2G},\tH_2)\bigg)\nonumber\\
     &&\qquad\qquad\qquad +n\cdot\bigg( h(H_{22}X_{2G}+\eta_2W_2|\tH_2) - h(H_{21}X_{2G}+V_{1}|\tH_1) \nonumber\\
     && \qquad\qquad \qquad \qquad + h(H_{11}X_{1G}+H_{31}X_{3G}+\eta_1W_{1}|\tH_1) - h(H_{32}X_{3G}+V_{2}|\tH_2)\bigg)\nonumber\\
	%%%%%%%%%%%%%%%%
     &&\qquad  \stackrel{(a)}{=} n\cdot\bigg(h(Y_{1G}|S_{1G},\tH_1)-h(S_{1G}|X_{1G},X_{3G},\tH_1)+h(Y_{2G}|S_{2G},\tH_2)- h(S_{2G}|X_{2G},\tH_2)\nonumber\\
     && \qquad \qquad \qquad + h(S_{2G}|\tH_2) - h(H_{11}X_{1G} + H_{21}X_{2G} + H_{31}X_{3G} + Z_{1}|X_{1G}, X_{3G}, W_1,\tH_1) + h(S_{1G}|\tH_1) \nonumber\\
     && \qquad\qquad \qquad \qquad  - h(H_{22}X_{2G} + H_{32}X_{3G}+Z_{2}|X_{2G}, W_2,\tH_2)\bigg)\nonumber\\
		%%%%%%%%%%%%%%%%
     &&\qquad  = n\cdot\bigg(h(Y_{1G}|S_{1G},\tH_1) +I(X_{1G},X_{3G};S_{1G}|\tH_1)+h(Y_{2G}|S_{2G},\tH_2)+I(X_{2G};S_{2G}|\tH_2)\nonumber\\
     && \qquad \qquad\qquad \qquad  - h(Y_{1G}|X_{1G}, X_{3G}, S_{1G},\tH_1)  - h(Y_{2G}|X_{2G}, S_{2G},\tH_2)\bigg)\nonumber\\
    %%%%%%%%%%%%%%%%%
     &&\qquad  = n\cdot\Big(I(X_{1G},X_{3G};Y_{1G},S_{1G}|\tH_1)+I(X_{2G};Y_{2G},S_{2G}|\tH_2)\Big),
\end{eqnarray*}
 s.t. all the expressions are evaluated with circularly symmetric complex Normal channel inputs $(X_{1G},X_{2G},X_{3G})^T\sim\CN({\bm 0},\Qmat^{\opt}_{G})$ where
\begin{equation*}
    \Qmat^{\opt}_{G}=\left[\begin{array}{ccc}
                        1\hspace{2mm}   & 0\hspace{2mm} &0\\
                        0\hspace{2mm}   & 1\hspace{2mm} &0\\
                        0\hspace{2mm}   & 0\hspace{2mm} &1\end{array} \right].
\end{equation*}
In the above transitions, (a) follows from Lemma \ref{lem:lemma2}.
It thus follows that for $n$ large enough, the sum-rate capacity is upper-bounded by:
\begin{equation}
\label{eq:sumrate_OB2}
    \sup_{(R_1,R_2)\in\mathcal{C}(\underline{\footnotesize\SNR})} (R_1+R_2) \le I(X_{1G},X_{3G};Y_{1G},S_{1G}|\tH_1)+I(X_{2G};Y_{2G},S_{2G}|\tH_2),
\end{equation}
where $(X_{1G},X_{2G},X_{3G})^T\sim\CN({\bm 0},\Qmat^{\opt}_{G})$. To complete the proof of Theorem \ref{thm:SRCA-UB}, note that
    \begin{eqnarray*}
        I(X_{1G},X_{3G};Y_{1G},S_{1G}|\tH_1)&=&I(X_{1G},X_{3G};Y_{1G}|\tH_1) + I(X_{1G},X_{3G};S_{1G}|Y_{1G},\tH_1)\\        I(X_{2G};Y_{2G},S_{2G}|\tH_2)&=&I(X_{2G};Y_{2G}|\tH_2)+I(X_{2G};S_{2G}|Y_{2G},\tH_2).
    \end{eqnarray*}
Hence, if we find conditions under which
    \begin{subequations}
    \label{eq:sumrate_1}
    \begin{eqnarray}
        \label{eq:sumrate_cond1}
        I(X_{1G},X_{3G};S_{1G}|Y_{1G},\tH_1)&=&0\\
        \label{eq:sumrate_cond2}
        I(X_{2G};S_{2G}|Y_{2G},\tH_2)&=&0,
    \end{eqnarray}
    \end{subequations}
then, when these conditions are satisfied, an upper bound on the sum-rate capacity is given by
\begin{equation}
\label{eq:sumrate_OB223}
    \sup_{(R_1,R_2)\in\mathcal{C}(\underline{\footnotesize\SNR})} (R_1+R_2) \le I(X_{1G},X_{3G};Y_{1G}|\tH_1)+I(X_{2G};Y_{2G}|\tH_2),
\end{equation}
where $(X_{1G},X_{2G},X_{3G})^T\sim\CN({\bm 0},\Qmat^{\opt}_{G})$. To that aim, note that from \eqref{eq:sumrate_cond1} we obtain
    \begin{eqnarray}
        &&I(X_{1G},X_{3G};S_{1G}|Y_{1G},\tH_1)=0\nonumber\\
        &&\Leftrightarrow \E_{\tH_1}\Big\{I(X_{1G},X_{3G};h_{11}X_{1G}+h_{31}X_{3G}+\eta_1W_{1}|h_{11}X_{1G}+h_{31}X_{3G}+h_{21}X_{2G}+Z_1,\tH_1=\th_1)\Big\}=0\nonumber.
    \end{eqnarray}
From Lemma \ref{lem:lemma3} we conclude that this is satisfied if for all values of $h_{lk}$ it holds that
    \begin{eqnarray}
        &&\E\big\{(\eta_1W_{1})(h_{21}X_{2G}+Z_1)^*\big\}=\E\big\{|h_{21}X_{2G}+Z_1|^2\big\}\nonumber\\
        &&\Leftrightarrow \eta_1\CORRN_1=1+\SNRA_{21}.\label{eq:eta1DefRev1}
    \end{eqnarray}
Applying the same arguments to \eqref{eq:sumrate_cond2}, we obtain
    \begin{eqnarray}
        &&I(X_{2G};S_{2G}|Y_{2G},\tH_2)=0\nonumber\\
        &&\Leftrightarrow  \E_{\tH_2}\Big\{I(X_{2G};h_{22}X_{2G}+\eta_2W_{2}|h_{22}X_{2G}+h_{32}X_{3G}+Z_2,\tH_2=\th_2)\Big\}=0\nonumber\\
        &&\stackrel{(a)}{\Leftrightarrow} \E\big\{(\eta_2W_{2})(h_{32}X_{3G}+Z_2)^*\big\}=\E\big\{|h_{32}X_{3G}+Z_2|^2\big\}\nonumber\\
        &&\Leftrightarrow \eta_2\CORRN_2=1+\SNRA_{32},\label{eq:eta2DefRev1}
    \end{eqnarray}
where (a) follows again from Lemma \ref{lem:lemma3}.

\noindent
Combining \eqref{eq:eta1DefRev1} and \eqref{eq:eta2DefRev1} with \eqref{eq:WI-Firsthalf}, we conclude that \eqref{eq:sumrate_OB223} constitutes an upper-bound on the sum-rate capacity if it is possible to construct a genie signal with parameters $\CORRN_1,\CORRN_2,\eta_1$ and $\eta_2$ s.t.
    \begin{eqnarray*}
        \SNRA_{32}|\eta_1|^2 &\le& \SNRA_{31}\big(1-|\CORRN_2|^2\big)-2\SNRA_{32}\SNRA_{11}\\
        \SNRA_{21}|\eta_2|^2 &\le& \SNRA_{22}\big(1-|\CORRN_1|^2\big)\\
        \eta_1\CORRN_1&=&1+\SNRA_{21}\\
        \eta_2\CORRN_2&=&1+\SNRA_{32}.
    \end{eqnarray*}
We note that this can be done if
    \begin{subequations}
    \label{eq:WI_con2}
    \begin{eqnarray}
        \SNRA_{32}(1+\SNRA_{21})^2 &\le& |\CORRN_1|^2\bigg(\SNRA_{31}\big(1-|\CORRN_2|^2\big)-2\SNRA_{32}\SNRA_{11}\bigg)\\
        \SNRA_{21}(1+\SNRA_{32})^2 &\le& |\CORRN_2|^2 \bigg(\SNRA_{22}\big(1-|\CORRN_1|^2\big)\bigg).
    \end{eqnarray}
    \end{subequations}
In conclusion, if we can find two complex scalars $\CORRN_1$ and $\CORRN_2$ s.t. $0\le|\CORRN_1|,|\CORRN_2|\le1$, for which \eqref{eq:WI_con2} is satisfied, then an upper-bound on the sum-rate capacity is given by
\begin{equation}
\label{eq:sumrate_OB2232}
    \sup_{(R_1,R_2)\in\mathcal{C}(\underline{\footnotesize\SNR})} (R_1+R_2) \le  \bigg\{I(X_{1G},X_{3G};Y_{1G}|\tH_1)+I(X_{2G};Y_{2G}|\tH_2)\bigg\},
\end{equation}
where $X_{kG}\sim\CN(0,1), k\in\{1,2,3\}$, mutually independent. The proof of Theorem \ref{thm:SRCA-UB} is completed by identifying
 $\beta_1\equiv |\CORRN_1|^2$ and $\beta_2\equiv |\CORRN_2|^2$.
\end{proof}

\subsection{Step 2: An Achievable Rate Region}
\label{sec:achievableRegionIID}
We next characterize an achievable rate region for the ergodic phase fading Z-ICR. This region is stated in the following proposition:

\begin{proposition}
\label{thm:achievable_region}
Consider the ergodic phase fading Z-ICR with only Rx-CSI, defined in Section \ref{sec:Model}.
Let the channel inputs be generated i.i.d. in time according to $X_k\sim\CN(0,1), k\in\{1,2,3\}$, mutually independent. If it holds that
    \begin{eqnarray}
        \label{eq:Relay_con1}
        I(X_1,X_3;Y_1|\tH_1)\le I(X_1;Y_3|X_3,\tH_3),
    \end{eqnarray}
    then an achievable rate region for the Z-ICR is given by all the non-negative rate pairs $(R_1,R_2)$ satisfying
    \begin{subequations}
    \label{eq:ICR_region}
    \begin{eqnarray}
        R_1 &\le& I(X_1,X_3;Y_1|\tH_1)\\
        R_2 &\le& I(X_2;Y_2|\tH_2).
    \end{eqnarray}
    \end{subequations}
\end{proposition}

\begin{proof}
The achievability is based on the DF strategy at the relay. Fix the blocklength $n$ and the input distribution
$f_{X_1,X_2,X_3}(x_1,x_2,x_3)=f_{X_1}(x_1)\cdot f_{X_2}(x_2)\cdot f_{X_3}(x_3)$, with $X_k\sim\CN(0,1), k=1,2,3$. We employ a transmission scheme in which $B-1$ messages are transmitted using $nB$ channel symbols:

\paragraph{Code Construction}
\label{sec:ICR Code Book}
For each message $m_k \in \mathcal{M}_k, k\in\{1,2\}$ select a codeword $\xvec_k(m_k)$ according to the p.d.f.
$ f_{\Xvec_k}\big(\xvec_k(m_k)\big) =\prod_{i=1}^n f_{X_k}\big(x_{k,i}(m_k)\big)$.
For each $\tilde{m}_1 \in \mathcal{M}_1$ select a codeword $\xvec_3(\tilde{m}_1)$ according to the p.d.f.
$f_{\Xvec_3}\big(\xvec_3(\tilde{m}_1)\big) =\prod_{i=1}^n f_{X_3}\big(x_{3,i}(\tilde{m}_1)\big)$.

\paragraph{Encoding at Block $b$}
\label{sec:ICR Encoding}
At block $b$, Tx$_k$ transmits the message $m_{k,b}$ via the codeword $\xvec_k(m_{k,b}), k\in\{1,2\}$. Let $\hat{m}_{1,b-1}$ denote
the decoded message at the relay at block $b-1$. At block $b$, the relay transmits the codeword $\xvec_3(\hat{m}_{1,b-1})$.
At block $b=1$ the relay transmits the codeword $\xvec_3(1)$, and at block $b=B$, Tx$_1$ and Tx$_2$ transmit the codewords $\xvec_1(1)$ and $\xvec_2(1)$, respectively.

\paragraph{Decoding at the Relay}
\label{sec:ICR relay decoding}
The decoding process at the relay is similar to the one used in \cite[Section VII-D]{Kramer:05}. For decoding $m_{1,b}$, the decoder at the relay looks for a unique, $m_1\in\mathcal{M}_1$ that satisfies:
\begin{equation*}
\label{eq:Decoing at Relay}
    \Big(\xvec_1(m_1),\xvec_3(\hat{m}_{1,b-1}),\yvec_3(b),{\thvec_3}(b)\Big) \in \typ(X_1,X_3,Y_3,\tH_3).
\end{equation*}
From \cite[Eq. (15)]{Kramer:05} it directly follows that the relay can decode reliably if $n$ is large enough, as long as
\begin{equation*}
\label{eq:rate_at_relay}
    R_1 \le I(X_1;Y_3,\tH_3|X_3)\stackrel{(a)}{=}I(X_1;Y_3|X_3,\tH_3),
\end{equation*}
where (a) holds since the channel coefficients are independent of the transmitted symbols.

\paragraph{Decoding at Rx$_1$}
  Rx$_1$ uses a backward block decoding scheme as in \cite[Appendix A]{Kramer:05} while treating the signal from Tx$_2$ as additive noise (recall that the codebooks are generated independently).
  Assume the relay has correctly decoded all the messages $\{m_{1,b}\}_{b=1}^{B-1}$.
  Assuming that Rx$_1$ has correctly decoded $m_{1,b+1}$, then, in order to decode $m_{1,b}$, Rx$_1$ generates the sets:
    \begin{eqnarray*}
        \mathcal{E}_{0,b}  &\triangleq& \Big \{\hat{m}_1\in\mathcal{M}_1: \big(\xvec_1(m_{1,b+1}),\xvec_3(\hat{m}_1),\mathbf{y}_1(b+1), \mathbf{\hvec}_1(b+1)\big) \in \typ(X_1,X_3,Y_1,\tilde{H}_1)\Big\}.\\
        \mathcal{E}_{1,b}  &\triangleq& \Big\{\hat{m}_1\in\mathcal{M}_1: \big( \xvec_1(\hat{m}_1),\mathbf{y}_1(b),\mathbf{\hvec}_1(b)\big) \in \typ(X_1,Y_1,\tilde{H}_1)\Big\}.
    \end{eqnarray*}
Rx$_1$ then decodes $m_{1,b}$ by finding a unique $m_{1}\in \mathcal{E}_{0,b} \cap \mathcal{E}_{1,b}$. Note that since the codewords are independent of each other, error events associated with $\mathcal{E}_{0,b}$ are independent of error events associated with $\mathcal{E}_{1,b}$. Thus, by using standard joint-typicality arguments \cite[Theorem 7.6.1]{cover-thomas:it-book}, it follows that decoding can be done reliably by taking $n$ large enough as long as
    \begin{eqnarray*}
        \label{eq:rate_at_Rx1}
         R_1 & \le & I(X_1;Y_1,\tH_1)+I(X_3;Y_1,\tH_1|X_1)\\
          & = & I(X_1,X_3;Y_1,\tH_1)\\
          & = &  I(X_1,X_3;Y_1|\tH_1).
    \end{eqnarray*}

\paragraph{Decoding at Rx$_2$}
Rx$_2$ treats the signal  from the relay as additive noise. This can be done since the codebooks are generated independently. The decoder at Rx$_2$ is therefore the decoder for PtP channels:
At block $b$ the decoder looks for a unique message $m_{2}\in\mathcal{M}_2$ that satisfies
\[
    \Big(\xvec_2(m_{2}),\yvec_2(b), \tilde{\bf{h}}_2(b)\Big)\in \typ(X_2,Y_2,\tH_2).
\]
It thus follows
from \cite[Thm. 9.1.1]{cover-thomas:it-book} that Rx$_2$ can reliably decode $m_{2,b}$ if $n$ is large enough, as long as
\begin{equation*}
    \label{eq:rate_at_Rx2}
    R_2 \le I(X_2;Y_2,\tH_2)=I(X_2;Y_2|\tH_2).
\end{equation*}
Finally, we observe that if \eqref{eq:Relay_con1} is satisfied, i.e., if $I(X_1,X_3;Y_1|\tH_1)\le I(X_1;Y_3|X_3,\tH_3)$, then the decoder at the relay can reliably decode the signal from Tx$_1$ whenever Rx$_1$ can. Consequently, we conclude that any rate pair inside the region specified in \eqref{eq:ICR_region} is achievable.
\end{proof}

\subsection{Step 3: The Sum-Rate Capacity in the WI Regime}
\label{sec:SUMCADER}
Note that when the conditions in \eqref{eq:WI_con2-Thm2} and \eqref{eq:Relay_con1}  hold (corresponding to conditions \eqref{eq:WI_con2-Thm} and \eqref{eq:Relay_con2Phase} in Thm. \ref{thm:WI}, respectively), then the {\em upper bound} on the sum-rate in \eqref{eq:sumrate_OB2Th} coincides with the {\em achievable} sum-rate obtained from \eqref{eq:ICR_region},
where both sum-rate expressions are evaluated with mutually independent channel inputs, distributed according to $X_k\sim\CN(0,1)$, $k\in\{1,2,3\}$.
This results in a characterization of  the {\em sum-rate capacity} for the ergodic phase fading Z-ICR.
The proof of Theorem \ref{thm:WI} is completed by observing that for mutually independent channel inputs distributed according to $X_k\sim\CN(0,1), k\in\{1,2,3\}$, the mutual information
expressions in \eqref{eq:sumrate_OB2Th} and \eqref{eq:Relay_con1} are explicitly written as
\begin{eqnarray*}
    I(X_{1},X_{3};Y_{1}|\tH_1)&=& \log\bigg(1+\frac{\SNRA_{11}+\SNRA_{31}}{1+\SNRA_{21}}\bigg)\\
    I(X_{2};Y_{2}|\tH_2)&=&\log\bigg(1+\frac{\SNRA_{22}}{1+\SNRA_{32}}\bigg)\\
    I(X_1;Y_3|X_3,\tH_3) &=& \log(1+\SNRA_{13}),
\end{eqnarray*}
from which the explicit expressions \eqref{eq:Relay_con2Phase} and \eqref{eq:sumcapacityPhase} are obtained.
\end{proof}

\subsection{Comments}
\begin{remark}
\label{rem:ZICSRC}
\em{}Consider the scenario in which the relay is off, referred to as the Z-IC  \cite[Theorem 2]{Sason:04}. This scenario can be obtained from the Z-ICR by letting $\SNRA_{31}=\SNRA_{32}=0$. In this case, since decoding at the relay does not constrain the rates, from Theorem \ref{thm:WI} we conclude that if $\SNRA_{21}\le \SNRA_{22}$, then the sum-rate capacity of the
ergodic phase fading Z-IC is given by
\begin{equation}
    \label{eq:ZICSUMRATE}
    \sup_{\left\{\substack{(R_1,R_2)\in\mathcal{C}({\footnotesize\underline{\SNR}),} \\ {\footnotesize \mbox{s.t.}\phantom{x}\SNRA_{31}=\SNRA_{32}=0}}\right\}}\!\!\!\!\!\! (R_1+R_2) =\log\bigg(1+\frac{\SNRA_{11}}{1+\SNRA_{21}}\bigg)+\log\Big(1+\SNRA_{22}\Big)\triangleq C_{\mbox{\footnotesize sum}}^{\mbox{\scriptsize PF-Z-IC}}(\underline{\SNR}),
\end{equation}
which is similar to the sum-rate capacity expression  for the AWGN Z-IC in the WI regime characterized in \cite[Theorem 2]{Sason:04} (although in the current work the channel is
subject to ergodic phase fading).
\end{remark}

\begin{remark}
\label{rem:sum-cap-comp-ZIC}
\em{} An interesting question that arises is whether adding a relay node to the Z-IC increases the sum-rate in the WI regime,  when the interfering signal is treated as noise at each receiver.
In Fig. \ref{fig:sum-rate} we show that the answer to this question is positive.
\begin{figure}[h]
    \centering
    \includegraphics[scale=0.6]{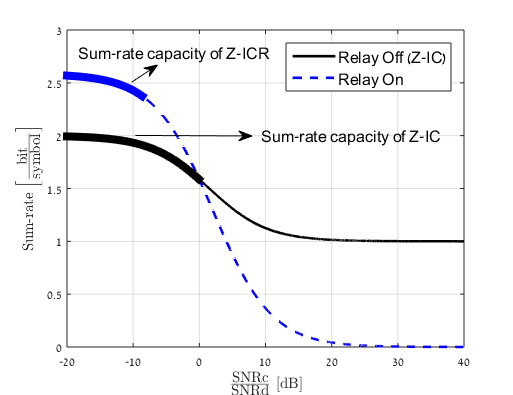}
    \vspace{-0.3 cm}
    \caption{\footnotesize The sum-rates of Proposition \ref{thm:achievable_region} and of \eqref{eq:ZICSUMRATE} for the scenario in which  $\mbox{\footnotesize\sffamily SNR}_{11}=\mbox{\footnotesize\sffamily SNR}_{22}=\mbox{\footnotesize\sffamily SNR}_{31}\triangleq\mbox{\footnotesize\sffamily SNR}_{\mbox{d}}=1$, and $\mbox{\footnotesize\sffamily SNR}_{21}=\mbox{\footnotesize\sffamily SNR}_{32}\triangleq\mbox{\footnotesize\sffamily SNR}_{\mbox{c}}$.}
    \label{fig:sum-rate}
\end{figure}

%${\SNRA}_{11}={\SNRA}_{22}={\SNRA}_{31}=10,{\SNRA}_{21}=h^2$

Fig. \ref{fig:sum-rate} depicts the sum-rate of \eqref{eq:ICR_region} with and without a relay (as discussed in Comment \ref{rem:ZICSRC}, turning off the relay is achieved by setting $\SNRA_{31}=\SNRA_{32}=0$ in Eqns. \eqref{eq:WI_con2-Thm} and \eqref{eq:sumcapacityPhase}),  for Z-ICR scenarios in which condition \eqref{eq:Relay_con1}, or equivalently \eqref{eq:Relay_con2Phase}, is satisfied. We consider a symmetric setting  by letting $\mbox{\sffamily SNR}_{11}=\mbox{\sffamily SNR}_{22}=\mbox{\sffamily SNR}_{31}=\mbox{\sffamily SNR}_{\mbox{\scriptsize d}}$, and $\mbox{\sffamily SNR}_{21}=\mbox{\sffamily SNR}_{32}=\mbox{\sffamily SNR}_{\mbox{\scriptsize c}}$. Thus, $\mbox{\sffamily SNR}_{\mbox{\scriptsize c}}$ and $\mbox{\sffamily SNR}_{\mbox{\scriptsize d}}$ denote the strengths of the interfering links and of the links carrying desired information, respectively, and hence, the relative strength of the interference is given as $\frac{\mbox{\sffamily SNR}_{\mbox{\scriptsize c}}}{\mbox{\sffamily SNR}_{\mbox{\scriptsize d}}}$. It can be seen from the figure that when the interference is sufficiently weak, the relay increases the sum-rate, which follows as the rate increase for Tx$_1$-Rx$_1$ is greater than the rate decrease for Tx$_2$-Rx$_2$.

At interference levels, $\frac{\small\SNRA_{\mbox{\scriptsize c}}}{\small\SNRA_{\mbox{\scriptsize d}}}$, which correspond to the thick lines in each plot, treating the interfering signal as noise is sum-rate optimal, and the resulting achievable sum-rates from \eqref{eq:ICR_region} and \eqref{eq:ZICSUMRATE} correspond to the sum-rate capacities
for the Z-ICR and for the Z-IC, respectively.
%coincide with the sum-rate capacity of each channel. Namely, \eqref{eq:ICR_region} coincides with the sum-rate stated in \eqref{eq:sumcapacityPhase}, and the achievable sum-rate of the Z-IC is given by $C_{\mbox{\footnotesize sum}}^{\mbox{\footnotesize Z-IC}}(\underline{\footnotesize\SNR})$ in \eqref{eq:ZICSUMRATE}.
Thus, it is evident from Fig. \ref{fig:sum-rate} that in some scenarios, adding a relay node and employing the communications scheme described in the proof of Proposition \ref{thm:achievable_region}, {\em strictly increases} the sum-rate capacity of the ergodic phase fading Z-IC in the WI regime, $C_{\mbox{\footnotesize sum}}^{\mbox{\scriptsize PF-Z-IC}}(\underline{\SNR})$. In particular, we observe that for the symmetric scenario of Fig. \ref{fig:sum-rate},  adding a relay strictly increases the sum-rate {\em capacity} as long as $\frac{\small\SNRA_{\mbox{\scriptsize c}}}{\small\SNRA_{\mbox{\scriptsize d}}}< 0\hspace{1 mm}[\mbox{dB}]$, and for sufficiently weak interference, e.g., $\frac{\small\SNRA_{\mbox{\scriptsize c}}}{\small\SNRA_{\mbox{\scriptsize d}}}< -10\hspace{1 mm}[\mbox{dB}]$, DF at the relay achieves the sum-rate  capacity of the Z-ICR.
\end{remark}

\begin{remark}
\em{}Note that for the set of channel coefficients satisfying \eqref{eq:Relay_con2Phase} and \eqref{eq:WI_con2-Thm}, the sum-rate capacity stated in \eqref{eq:sumcapacityPhase} is an upper bound on the sum-rate capacity of the ergodic phase fading ICR ({\em with both interfering links active}) in the weak interference regime, when the relay node receives transmissions only from Tx$_1$ (as is the case in Theorem \ref{thm:WI}). If, in addition, the relay node receives the transmissions of Tx$_2$, then a new coding strategy must be developed for the WI regime in order to facilitate simultaneous enhancement of the desired signal at both destinations. Finding the optimal scheme and the corresponding sum-rate capacity is currently an open issue that requires further research.
\end{remark}

\begin{remark}
\em{}Fig. \ref{fig:WI1} shows the region of relay locations in the 2D-plane in which DF at the relay achieves the sum-rate capacity of the ergodic phase fading Z-ICR in the WI regime.
This figure was obtained using a channel model in which the  attenuation $\sqrt{\SNRA_{ij}}$ is linked to the distance from node $i$ to node $j$, $d_{ij}$, via $\SNRA_{ij}=\frac{1}{d_{ij}^4}$. This attenuation model corresponds to the two-ray propagation model.
\begin{figure}[h]
    \centering
    \includegraphics[scale=0.5]{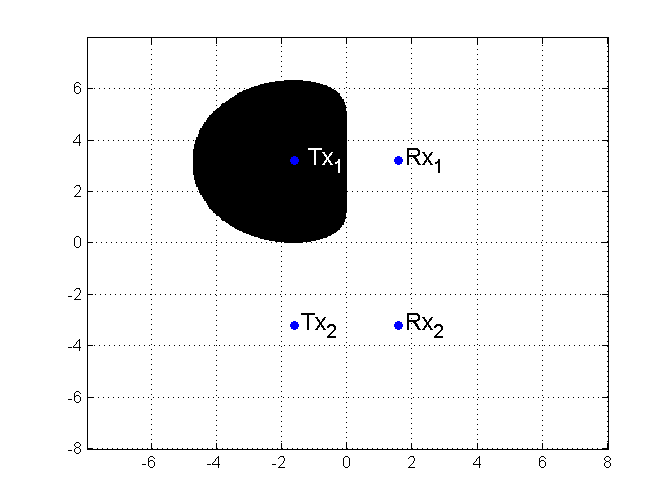}
    \vspace{-0.3 cm}
    \caption{\footnotesize The geographical position of the relay in a 2D-plane where the conditions of Theorem \ref{thm:WI} are satisfied.}
    \label{fig:WI1}
\end{figure}
Note that since the signal from the relay is desired at Rx$_1$ and is treated as noise at Rx$_2$, then for the WI conditions to hold, the relay should be closer to Rx$_1$ (to strengthen the desired signal at Rx$_1$) and farther away from Rx$_2$ (to decrease the interference at Rx$_2$). However, the relay should remain relatively close to Tx$_1$ to allow reliable decoding of the messages from Tx$_1$ at the relay.% and yet not too close to Rx$_1$ in order to satisfy \eqref{eq:Relay_con2Phase}.
\end{remark}

\section{Asymptotic SNR Analysis: The Optimal GDoF in the WI Regime}
\label{sec:FULLGDoF}
In this section, we
%provide a complementing performance measure to the sum-rate performance of the phase fading Z-ICR in the WI regime, by characterizing
characterize the maximal GDoF of the ergodic phase fading Z-ICR in the WI regime. Since GDoF analysis characterizes the performance in the asymptotically high SNR regime, i.e., $\SNR_{lk}\rightarrow\infty$ for all links, then, in order to analyze the effect of different link conditions, we consider a scenario in which the magnitudes of the channel coefficients scale differently as a function of the SNR. Letting $\alpha, \beta, \gamma$ and $\lambda$ be four non-negative real numbers, in this section we consider a  model in which
\begin{subequations}
\begin{eqnarray}
    \label{eq:SNRLine1}
    \SNRA_{11}= &\SNR, \qquad\;\; \SNRA_{22}&=\SNR,\\
    \label{eq:SNRLine2}
    \SNRA_{21}= &\SNR^{\alpha}, \qquad \SNRA_{32}&=\SNR^{\lambda}\\
    \label{eq:SNRLine3}
    \SNRA_{13}= &\SNR^{\gamma}, \qquad \SNRA_{31}&=\SNR^{\beta}.
\end{eqnarray}
\end{subequations}
Observe that the direct links scale as $\SNR$, the interfering links from Tx$_2$ to Rx$_1$, and from the relay to Rx$_2$ scale as $\SNR^{\alpha}$ and $\SNR^{\lambda}$, respectively, and the links on the cooperation path from Tx$_1$ to the relay, and from the relay to Rx$_1$ scale as $\SNR^{\gamma}$ and $\SNR^{\beta}$, respectively. Let $\mathcal{C}(\SNR)$ denote the capacity region of the Z-ICR for a given value of the parameter $\SNR$, and define $C_{\mbox{\footnotesize sum}}(\SNR)\triangleq \max_{(R_1,R_2)\in\mathcal{C}({\scriptsize \SNR})} \big(R_1+R_2\big)$. Then, the GDoF is defined as (see also \cite[Def. 1]{Chaaban:12}):
\begin{equation*}
    \label{eq:gdofDef}
    \mbox{GDoF}\triangleq \lim_{\scriptsize{\SNR}\rightarrow\infty}\frac{C_{\mbox{\scriptsize sum}}(\SNR)}{\log\big(\SNR\big)}.
\end{equation*}
In the following theorem, we characterize the maximal GDoF of the ergodic phase fading Z-ICR in the WI regime:

\begin{theorem}
\label{thm:WI-GDoF}
   Consider the ergodic phase fading Z-ICR with only Rx-CSI, defined in Section II.
   %Let the additive noises be i.i.d. circularly symmetric complex Normal processes, $\CN(0,1)$, and let the sources be subject to power constraints $\E\big\{|X_k|^2\big\} \le 1$, $k\in\{1,2,3\}$.
   If the interference is symmetric and weak in the sense of
\begin{subequations}
\label{eq:optimalGDoFCon}
\begin{equation}
\label{eq:optimalGDoFCon1}
    \lambda=\alpha \le \frac{1}{2},
\end{equation}
and it also holds that
\begin{equation}
\label{eq:optimalGDoFCon2}
    1+2\alpha < \beta \le \gamma+\alpha,
\end{equation}
\end{subequations}
then the maximal GDoF of the channel is
\begin{equation}
\label{eqn:GDofMaxValue}
    \mbox{\em GDoF}_{\mbox{\em\footnotesize max}}=1+\beta-2\alpha,
\end{equation}
and it is achieved with mutually independent, zero mean complex Normal channel inputs with  positive powers satisfying $0<P_k\le 1$, $k\in\{1,2,3\}$.
\end{theorem}

\begin{remark}
\label{rem:WIexplainFintieSNR} \em
In the following we intuitively explain the conditions \eqref{eq:optimalGDoFCon} in Thm. \ref{thm:WI-GDoF}.
        Note that \eqref{eq:optimalGDoFCon1} corresponds to the weak interference regime in the sense of  \cite{etkin:08} and  \cite{Chaaban:12},
		namely, that the interfering links are {\em exponentially weaker} than the direct links in the sense that
		$\lim_{{\scriptsize \SNR}\rightarrow\infty}\frac{\SNR^{\alpha}}{\SNR} = 0$. We note that while the results of  \cite{etkin:08} and  \cite{Chaaban:12} for the
		symmetric scenario hold as long as the scaling exponent of the interfering links satisfies $\alpha\le 1$,
		the GDoF optimality result of Thm. \ref{thm:WI-GDoF} requires $\lambda = \alpha \le \frac{1}{2}$, i.e., Thm. \ref{thm:WI-GDoF} requires
		a smaller exponential scaling of the interference strength, compared to the minimal exponential scaling of the interference
        required for WI in \cite{etkin:08} and \cite{Chaaban:12}.
		It follows that the WI regime for the GDoF result of Thm. \ref{thm:WI-GDoF} corresponds to a subset of the WI regime applicable for
		the results of \cite{etkin:08} and  \cite{Chaaban:12}.
        Yet, we note that in \cite{etkin:08}, GDoF optimality of {\em treating interference as noise} was shown to hold only for $\alpha \le \frac{1}{2}$, which is
        in agreement with our characterization (see \cite[Section V-B]{etkin:08}).
		Next, consider \eqref{eq:optimalGDoFCon2}: Observe that \eqref{eq:optimalGDoFCon2} can be written as
		$1 + \alpha < \beta-\alpha \le \gamma$, which is equivalent to the inequality $\SNR^{1+\alpha} < \frac{\SNR^{\beta}}{\SNR^\alpha}\le\SNR^{\gamma}$. Note that
$\frac{\SNR^{\beta}}{\SNR^\alpha}\le\SNR^{\gamma}$ implies that the relay reception is good enough such that the SNR on the incoming link at the relay, $\big($i.e., $\SNR^{\gamma}\big)$ is higher than the  SNR on the link from the relay to Rx$_1$, when interference is treated as additive noise at Rx$_1$ $\big($i.e., $\frac{\SNR^{\beta}}{\SNR^\alpha}\big)$.
The inequality $\SNR^{1+\alpha}<\frac{\SNR^{\beta}}{\SNR^\alpha}$ implies that interference should be weak enough s.t. the SNR on the link from the relay to Rx$_1$, achieved by treating interference as additive noise at Rx$_1$ $\big($i.e., $\frac{\SNR^{\beta}}{\SNR^\alpha}\big)$, will be higher than the SNR of the direct link from
Tx$_1$ to Rx$_1$ augmented by the interference at Rx$_1$, (i.e., $\SNR^{1+\alpha}$). Hence, the second inequality represents an additional weak interference condition.
\end{remark}
%\begin{remark}
%          \em{} Note that in the study of the GDoF we consider the asymptotically high SNR regime, hence, from the channel model of Section \ref{sec:Model}, we conclude that the communications scheme described in this paper indeed achieves the maximal GDoF of the phase fading Z-ICR in the WI regime as long as the transmission powers are finite and greater than zero. However, while the technical derivations of the GDoF upper bound require $P_1=P_2=P_3=1$, which corresponds to transmission at maximal power, the achievability scheme can obtain the maximal GDoF at any finite transmission power, as long as the conditions of Theorem \ref{thm:WI-GDoF} are satisfied.
%\end{remark}

\begin{proof}
The proof of Thm. \ref{thm:WI-GDoF} consists of the following steps:
\begin{enumerate}
    \item We derive an upper bound on the GDoF of the ergodic phase fading Z-ICR by combining two bounds: A bound derived using a genie, and a bound obtained by following
    the derivations of the cut-set bound theorem \cite[Thm. 15.10.1]{cover-thomas:it-book}.
    \item We derive a lower bound on the GDoF by considering the communications scheme used in Section \ref{sec:achievableRegionIID}.
    \item We derive conditions on the SNR exponents of the channel coefficients under which our lower bound coincides with the upper bound, thereby, characterizing the maximal GDoF of the
        ergodic phase fading Z-ICR in the weak interference regime, subject to these conditions.
\end{enumerate}
In the following subsections, we provide a detailed proof for the above steps. Specifically, Steps 1 is carried out in Subsection \ref{sec:UB}, Step 2 is carried out in Subsection \ref{sec:IB},  and finally, Step 3 is detailed in Subsection \ref{sec:GDoFOptSec}.

\subsection{An Upper Bound on the Achievable GDoF}
\label{sec:UB}
An upper bound on the achievable GDoF of the Z-ICR is stated in the following theorem:
\begin{theorem}
    \label{thm:GDOFUB}
    Consider the ergodic phase fading Z-ICR with only Rx-CSI, stated in Section \ref{sec:Model}. If $\beta>2\lambda+1$, then an upper bound on the achievable GDoF is given by
    \begin{equation}
        \label{eq:GDoFOB}
%        \mbox{\em GDoF}^+= \min\Big\{\max\big\{2,1+\min\{\beta,\gamma\}\big\},\max\{\alpha,1,\beta-\lambda\}+\max\{\lambda,1-\alpha\}\Big\}.
%        \mbox{\em GDoF}^+= \min\Big\{\max\big\{2,1+\min\{\beta,\gamma\}\big\},\max\{\alpha,\beta-\lambda\}+\max\{\lambda,1-\alpha\}\Big\}.
	        \mbox{\em GDoF}^+= \min\Big\{\max\big\{2,1+\min\{\beta,\gamma\}\big\},\max\{\alpha+\lambda,\beta,  1+\beta -\alpha - \lambda\}\Big\}.
    \end{equation}
\end{theorem}
\begin{proof}
The upper bound is obtained as a combination of two bounds: The first bound is derived using a genie, and the second bound is derived by following the derivation of the cut-set theorem \cite[Thm. 15.10.1]{cover-thomas:it-book}.
\subsubsection{An Upper Bound Using a Genie}
Consider the following genie signals:
\begin{eqnarray*}
        S_{1,i}&=&H_{32,i} X_{3,i} + Z_{2,i}\\
        S_{2,i}&=&H_{21,i} X_{2,i} + Z_{1,i},
\end{eqnarray*}
$i\in\{1,2,...,n\}$. Suppose that a genie provides $\{S_{1,i}\}_{i=1}^n$ to Rx$_1$ and $\{S_{2,i}\}_{i=1}^n$ to Rx$_2$, i.e., the genie provides to Rx$_2$ an interference-free, noisy version of its desired signal as it is received at Rx$_1$, and to Rx$_1$ it provides a noisy version of the relay signal component observed at Rx$_2$. Let $M_k$ denote the message transmitted from Tx$_k$, and let $\hat{M}_k$ denote the decoded message at Rx$_k$. Additionally, let $P_{e,k}^{(n)}$ denote the probability of error in the estimation of $M_k$ at Rx$_k$ and define $n\epsilon_{kn}\triangleq 1+P_{e,k}^{(n)}nR_k, k\in\{1,2\}$. Then, for an achievable rate pair $(R_1,R_2)$, we obtain:
\begin{eqnarray}
    nR_1&=&H(M_1)\nonumber\\
    &=&H(M_1)-H(M_1|Y_1^n,\utH{}^n)+H(M_1|Y_1^n,\utH{}^n)\nonumber\\
    &\stackrel{(a)}{\le}&I(M_1;Y_1^n,\utH{}^n)+n\epsilon_{1n}\nonumber\\
    &\stackrel{(b)}{\le}&I(X_1^n;Y_1^n,\utH{}^n)+n\epsilon_{1n}\nonumber\\
    &=&I(X_1^n;\utH{}^n)+I(X_1^n;Y_1^n|\utH{}^n)+n\epsilon_{1n}\nonumber\\
    &\stackrel{(c)}{=}&I(X_1^n;Y_1^n|\utH{}^n)+n\epsilon_{1n}\nonumber\\
    &\stackrel{(d)}{\le}&I(X_1^n;Y_1^n|\utH{}^n)+I(X_3^n;Y_1^n|\utH{}^n,X_1^n)+I(X_1^n,X_3^n;S_1^n|\utH{}^n,Y_1^n)+n\epsilon_{1n}\nonumber\\
    &=&I(X_1^n,X_3^n;Y_1^n,S_1^n|\utH^n)+n\epsilon_{1n}\nonumber\\
    &=&I(X_1^n,X_3^n;S_1^n|\utH{}^n)+I(X_1^n,X_3^n;Y_1^n|S_1^n,\utH{}^n)+n\epsilon_{1n}\nonumber\\
    &=&h(S_1^n|\utH{}^n)-h(S_1^n|X_1^n,X_3^n,\utH{}^n)+h(Y_1^n|S_1^n,\utH{}^n)-h(Y_1^n|X_1^n,X_3^n,S_1^n,\utH{}^n)+n\epsilon_{1n}\nonumber\\
    &\stackrel{(e)}{=}&h(S_1^n|\utH{}^n)-h(Z_2^n)+h(Y_1^n|S_1^n,\utH{}^n)-h(S_2^n|\utH{}^n)+n\epsilon_{1n},\label{eq:eq4}
\end{eqnarray}
where (a) follows from Fano's inequality \cite[Thm. 2.10.1]{cover-thomas:it-book}, (b) follows from the data processing inequality \cite[Thm. 2.8.1]{cover-thomas:it-book} as $M_1-X_1^n-(Y_1^n,\tilde{H}_1^n)$ forms a Markov chain, (c) follows since channel inputs $X_1^n$ are independent of the channel coefficients $\tilde{H}_1^n$, (d) follows since mutual information is nonnegative, and (e) follows since
\begin{eqnarray*}
    h(Y_1^n|X_1^n,X_3^n,S_1^n,\utH{}^n)&=&h(Y_1^n|X_1^n,X_3^n,Z_2^n,\utH{}^n)\\
    &=&h\big(\{H_{21,i}X_{2,i}+ Z_{1,i}\}_{i=1}^n|X_1^n,X_3^n,Z_2^n,\utH{}^n\big)\\
    &\stackrel{(f)}{=}&h\big(\{H_{21,i}X_{2,i}+ Z_{1,i}\}_{i=1}^n|\utH{}^n\big)\\
    &\equiv& h(S_2^n|\utH{}^n),
\end{eqnarray*}
where (f) follows since $X_1^n$ and $X_3^n$ are independent of $X_2^n$, which follows since the message sets at the sources are mutually independent, and since the relay receives transmissions only from Tx$_1$. Similarly, for $R_2$ we have
\begin{eqnarray}
    nR_2&\le&I(X_2^n;Y_2^n,S_2^n|\utH{}^n)+n\epsilon_{2n}\nonumber\\
    &=&I(X_2^n;S_2^n|\utH{}^n)+I(X_2^n;Y_2^n|S_2^n,\utH{}^n)+n\epsilon_{2n}\nonumber\\
    &=&h(S_2^n|\utH{}^n)-h(S_2^n|X_2^n,\utH{}^n)+h(Y_2^n|S_2^n,\utH{}^n)-h(Y_2^n|X_2^n,S_2^n,\utH{}^n)+n\epsilon_{2n}\nonumber\\
    &=&h(S_2^n|\utH{}^n)-h(Z_1^n)+h(Y_2^n|S_2^n,\utH{}^n)-h(S_1^n|\utH{}^n)+n\epsilon_{2n}.\label{eq:eq5}
\end{eqnarray}
Let $R_{\mbox{\scriptsize sum}}=R_1+R_2$, denote $\CORR_i\triangleq\frac{\E\{X_{1,i}X_{3,i}^*\}}{\sqrt{P_{1,i}P_{3,i}}}$, where $|\CORR_i|\le 1$, and define $\theta_i\triangleq\arg\{h_{11,i}h_{31,i}^*\CORR_i\}$.
%Assume that the channel coefficients are s.t. $\SNRA_{31}>2\SNRA_{32}\sqrt{\SNRA_{11}\SNRA_{31}}$, which is satisfied in the high $\SNR$ regime if $\beta\ge2\lambda+1$.
 Then, by combining \eqref{eq:eq4} and \eqref{eq:eq5} we obtain
\begin{eqnarray}
    &&\hspace{-0.6 cm}n\big(R_{\mbox{\scriptsize sum}}-\epsilon_{1n}-\epsilon_{2n}\big)\nonumber\\
    &\le& h(S_1^n|\utH{}^n)-h(Z_2^n)+h(Y_1^n|S_1^n,\utH{}^n)-h(S_2^n|\utH{}^n)+ h(S_2^n|\utH{}^n)-h(Z_1^n)+h(Y_2^n|S_2^n,\utH{}^n)-h(S_1^n|\utH{}^n)\nonumber\\
    &=&h(Y_1^n|S_1^n,\utH{}^n)-h(Z_1^n)+h(Y_2^n|S_2^n,\utH{}^n)-h(Z_2^n)\nonumber\\
    &=&\sum_{i=1}^n h(Y_{1,i}|Y_1^{i-1},S_1^n,\utH{}^n)-h(Z_1^n)+ \sum_{i=1}^n h(Y_{2,i}|Y_2^{i-1},S_2^n,\utH{}^n)-h(Z_2^n)\nonumber\\
    &\stackrel{(a)}{\le}&\sum_{i=1}^n h(Y_{1,i}|S_{1,i},\utH_i)-h(Z_1^n)+ \sum_{i=1}^n h(Y_{2,i}|S_{2,i},\utH_i) -h(Z_2^n)\nonumber\\
    &\stackrel{(b)}{=}&\sum_{i=1}^n \E_{\utH_i}\big\{h(Y_{1,i}|S_{1,i},\utH_i=\uth_i)\big\}-n\cdot h(Z_1)+ \sum_{i=1}^n \E_{\utH_i}\big\{h(Y_{2,i}|S_{2,i},\utH_i=\uth_i)\big\}-n\cdot h(Z_2)\nonumber\\
    &\stackrel{(c)}{\le}&\sum_{i=1}^n \E_{\utH_i}\bigg\{\max_{\substack{0\le|\CORR_i|\le 1\\ P_{k,i}\le1, k\in\{1,2,3\}}} h(Y_{1G,i}|S_{1G,i},\utH_i=\uth_i)\bigg\}-n\cdot h(Z_1)\nonumber\\
    &&\qquad\qquad\qquad\qquad\qquad+ \sum_{i=1}^n  \E_{\utH_i}\bigg\{\max_{\substack{0\le|\CORR_i|\le 1\\ P_{k,i}\le1, k\in\{1,2,3\}}} h(Y_{2G,i}|S_{2G,i},\utH_i=\uth_i)\bigg\}-n\cdot h(Z_2)\nonumber\\
    &=&\sum_{i=1}^n \E_{\utH_i}\bigg\{\max_{\substack{0\le|\CORR_i|\le 1\\ P_{k,i}\le1, k\in\{1,2,3\}}} \log\big(2\pi\cdot\cov(Y_{1G,i}|S_{1G,i},\uth_i)\big)\bigg\}-\sum_{i=1}^n\log(2\pi)\nonumber\\
    &&\qquad\qquad\qquad\qquad\qquad+\sum_{i=1}^n \E_{\utH_i}\bigg\{\max_{\substack{0\le|\CORR_i|\le 1\\ P_{k,i}\le1, k\in\{1,2,3\}}} \log\big(2\pi\cdot\cov(Y_{2G,i}|S_{2G,i},\uth_i)\big)\bigg\}-\sum_{i=1}^n\log(2\pi)\nonumber\\
    %%%%%%%%%%%%%%%%%%%%%%%%%%%%%%%%%%%%%%
    &\stackrel{(d)}{=}&\sum_{i=1}^n \E_{\utH_i}\Bigg\{\max_{\substack{0\le|\CORR_i|\le 1\\ P_{k,i}\le1, k\in\{1,2,3\}}} \log\bigg(\var(Y_{1G,i}|\uth_i)-\frac{\big|\E\{Y_{1G,i}S_{1G,i}^*|\uth_i\}\big|^2}{\var(S_{1G,i}|\uth_i)}\bigg)\Bigg\}\nonumber\\
    &&\qquad\qquad\qquad\qquad\qquad+\sum_{i=1}^n \E_{\utH_i}\Bigg\{\max_{\substack{0\le|\CORR_i|\le 1\\ P_{k,i}\le1, k\in\{1,2,3\}}} \log\bigg(\var(Y_{2G,i}|\uth_i)-\frac{\big|\E\{Y_{2G,i}S_{2G,i}^*|\uth_i\}\big|^2}{\var(S_{2G,i}|\uth_i)}\bigg)\Bigg\}\nonumber\\
    %%%%%%%%%%%%%%%%%%%%%%
    &=&\sum_{i=1}^n  \E_{\utH_i}\Bigg\{\max_{\substack{0\le|\CORR_i|\le 1\\ P_{k,i}\le1, k\in\{1,2,3\}}} \!\!\!\! \log\bigg(1\!+\!P_{2,i}|h_{21,i}|^2+\frac{P_{1,i}|h_{11,i}|^2\!+\!P_{3,i}|h_{31,i}|^2\!+\!P_{1,i}P_{3,i}|h_{11,i}|^2|h_{32,i}|^2(1\!-\!|\CORR_i|^2)}{1+P_{3,i}|h_{32,i}|^2}\nonumber\\
&&\qquad\qquad\qquad\qquad\qquad\qquad\qquad\qquad\qquad\qquad\qquad\qquad\qquad\qquad +\frac{2|h_{11,i}||h_{31,i}^*|\sqrt{P_{1,i}P_{3,i}}|\CORR_i|\cos(\theta_i)}{1+P_{3,i}|h_{32,i}|^2}\bigg)\Bigg\}\nonumber\\
&& \qquad\qquad\qquad\qquad\qquad\qquad\qquad + \sum_{i=1}^n \E_{\utH_i}\Bigg\{ \max_{P_{k,i}\le1, k\in\{1,2,3\}} \log\bigg(1+P_{3,i}|h_{32,i}|^2+\frac{P_{2,i}|h_{22,i}|^2}{1+P_{2,i}|h_{21,i}|^2}\bigg)\Bigg\}\nonumber
\end{eqnarray}
\begin{eqnarray}
    \hspace{-0.8 cm}&\stackrel{(e)}{\le}&\sum_{i=1}^n \E_{\utH_i}\Bigg\{\! \max_{P_{k,i}\le1, k\in\{1,2,3\}} \log\!\bigg(1\!+\!P_{2,i}|h_{21,i}|^2\!+\!\frac{P_{1,i}|h_{11,i}|^2\!+\!P_{3,i}|h_{31,i}|^2\!+\!P_{1,i}P_{3,i}|h_{11,i}|^2|h_{32,i}|^2\!}{1+P_{3,i}|h_{32,i}|^2}\nonumber\\
    &&\quad\qquad\qquad\qquad\qquad\qquad\qquad\qquad\qquad\qquad\qquad\qquad\qquad\qquad\qquad\qquad +\frac{2|h_{11,i}||h_{31,i}^*|\sqrt{P_{1,i}P_{3,i}}}{1+P_{3,i}|h_{32,i}|^2}\bigg)\Bigg\}\nonumber\\
    && \qquad\qquad\qquad\qquad\qquad\quad+ \! \sum_{i=1}^n \E_{\utH_i}\Bigg\{ \max_{P_{k,i}\le1, k\in\{1,2,3\}} \log\bigg(1+P_{3,i}|h_{32,i}|^2+\frac{P_{2,i}|h_{22,i}|^2}{1+P_{2,i}|h_{21,i}|^2}\bigg)\!\Bigg\}\label{eq:MaxPGDoF}\\
    &\stackrel{(f)}{\le}&\sum_{i=1}^n  \E_{\utH_i}\Bigg\{ \log\bigg(1\!+\!|h_{21,i}|^2+\frac{|h_{11,i}|^2\!+\!|h_{31,i}|^2\!+\!|h_{11,i}|^2|h_{32,i}|^2\!+\!2|h_{11,i}||h_{31,i}^*|}{1+|h_{32,i}|^2}\bigg)\Bigg\}\nonumber\\
    &&\qquad\qquad\qquad\qquad\qquad\qquad\qquad\qquad \qquad+ \sum_{i=1}^n  \E_{\utH_i}\Bigg\{ \log\bigg(1+|h_{32,i}|^2+\frac{|h_{22,i}|^2}{1+|h_{21,i}|^2}\bigg)\Bigg\}\nonumber\\
    &\stackrel{(g)}{=}&\sum_{i=1}^n\log\bigg(1+\SNRA_{21}+\frac{\SNRA_{11}+\SNRA_{31}+\SNRA_{11}\SNRA_{32}+2\sqrt{\SNRA_{11}\SNRA_{31}}}{1+\SNRA_{32}}\bigg)\nonumber\\  &&\qquad\qquad\qquad\qquad\qquad\qquad\qquad\qquad\qquad\qquad\qquad +\sum_{i=1}^n \log\bigg(1+\SNRA_{32}+\frac{\SNRA_{22}}{1+\SNRA_{21}}\bigg)\nonumber\\
    &=&n \log\bigg(1+\SNRA_{21}+\frac{\SNRA_{11}+\SNRA_{31}+\SNRA_{11}\SNRA_{32}+2\sqrt{\SNRA_{11}\SNRA_{31}}}{1+\SNRA_{32}}\bigg)\nonumber\\
    &&\qquad\qquad\qquad\qquad\qquad\qquad\qquad\qquad\qquad\qquad\qquad +n\log\bigg(1+\SNRA_{32}+\frac{\SNRA_{22}}{1+\SNRA_{21}}\bigg),\nonumber
\end{eqnarray}

\vspace{3mm}

\noindent
where step (a) follows since conditioning reduces entropy, (b) follows since $Z_{1,i}$ and $Z_{2,i}$ are i.i.d. in time, (c) follows from \cite[Lemma 2]{Zahavi:12}, which states that given the set of channel coefficients at time $i$, $\uth_i$, then $h(Y_{k,i}|S_{k,i},\utH_i=\uth_i)$ is maximized with $Y_{k,i}$ and $S_{k,i}, k\in\{1,2\}$ distributed according to the zero-mean, circularly symmetric jointly proper complex Normal distribution with the covariance matrix $\cov(Y_{kG,i},S_{kG,i}|\uth_i)=\cov(Y_{k,i},S_{k,i}|\uth_i), k\in\{1,2\}$. Note that in this step the maximizing $v_i$, $P_{1,i}$, $P_{2,i}$, $P_{3,i}$, are generally functions of $\uth_i$.
 Step (d) follows from the direct application of the expression for the conditional covariance of jointly complex Normal RVs \cite[Section. VI, Eq. (6.5)]{Valerie:81}, (e) follows since $0\le|\CORR_i|\le 1$ and $-1\le\cos(\theta_i)\le 1$ and since the logarithm function is a monotonically increasing function of its argument, (f) follows since both sums of logarithmic functions in \eqref{eq:MaxPGDoF} are maximized by $P_{1,i} = P_{2,i} = P_{3,i} = 1$, $i\in\{1,2,...,n\}$. To see this point we consider each of the sums separately:

\begin{enumerate}

\phantomsection
\label{phn:expln1}
\item  Begin by considering the first logarithmic term in \eqref{eq:MaxPGDoF}: We now show that the expression
    \begin{equation}
    \label{eqn:first_log_term_in39}
    \frac{P_{1,i}|h_{11,i}|^2\!+\!P_{3,i}|h_{31,i}|^2\!+\!P_{1,i}P_{3,i}|h_{11,i}|^2|h_{32,i}|^2\!+2|h_{11,i}||h_{31,i}^*|\sqrt{P_{1,i}P_{3,i}}}{1+P_{3,i}|h_{32,i}|^2}
    \end{equation}
    which appears in the first summation of \eqref{eq:MaxPGDoF} increases monotonically with respect to both $P_{1,i}$ and $P_{3,i}$.
    To that aim we note that from inspecting the expression \eqref{eqn:first_log_term_in39} it is evident that it increases monotonically with respect to $P_{1,i}$, for any $P_{3,i}\ge 0$.
    Next, for any fixed $0\le P_{1,i} \le 1$, we  differentiate \eqref{eqn:first_log_term_in39} with respect to $P_{3,i}$ and obtain:
    \begin{eqnarray*}
      &&\hspace{-0.5 cm}\frac{\partial}{\partial P_{3,i}} \left\{\frac{P_{1,i}|h_{11,i}|^2\!+\!P_{3,i}|h_{31,i}|^2\!+\! P_{1,i}P_{3,i}|h_{11,i}|^2|h_{32,i}|^2\!+2|h_{11,i}||h_{31,i}^*|\sqrt{P_{1,i}P_{3,i}}}{1+P_{3,i}|h_{32,i}|^2}\right\}=\\
      &&\hspace{+5 cm} \frac{|h_{31,i}|^2+|h_{11,i}||h_{31,i}^*|\frac{\sqrt{P_{1,i}}}{\sqrt{P_{3,i}}}-|h_{11,i}||h_{31,i}^*||h_{32,i}|^2\sqrt{P_{1,i}}\sqrt{P_{3,i}}}{\big(1+P_{3,i}|h_{32,i}|^2\big)^2}.
    \end{eqnarray*}
        From this expression, we note that as $0\le P_{1,i}, P_{3,i} \le 1$, then the above derivative is positive if
				$|h_{31,i}|^2 > |h_{11,i}||h_{31,i}^*||h_{32,i}|^2$, or equivalently, if
        $\SNRA_{31}>\SNRA_{32}\sqrt{\SNRA_{11}\SNRA_{31}}$, which is satisfied if $\beta > 2\lambda+1$.
 We conclude that  \eqref{eqn:first_log_term_in39} increases monotonically with respect to $0\le P_{1,i},P_{3,i}\le 1$. This conclusion,
				combined with the facts that the expression in the logarithm in the first summation in \eqref{eq:MaxPGDoF} is monotone increasing in $P_{2,i}$,
                and that the logarithm function itself is monotone increasing,
				leads to the conclusion that if $\beta > 2\lambda+1$, then the first logarithmic expression in \eqref{eq:MaxPGDoF} monotonically increases with
				respect to $P_{1,i}$, $P_{2,i}$ and $P_{3,i}$, hence it is maximized by setting $P_{1,i} = P_{2,i} = P_{3,i} = 1$.
\smallskip

\item Now consider the second term in \eqref{eq:MaxPGDoF}: The function $\frac{a_1x}{1+b_1x}$ monotonically increases with respect to $x$ as long as $a_1, b_1  > 0$ and thus, letting $a_1=|h_{22,i}|^2$ and $b_1=|h_{21,i}|^2$, we conclude that $\frac{P_{2,i}|h_{22,i}|^2}{1+P_{2,i}|h_{21,i}|^2}$ increases with respect to $P_{2,i}$.
It also immediately follows that
			\[
			 \left. \log\bigg(1+P_{3,i}|h_{32,i}|^2+\frac{P_{2,i}|h_{22,i}|^2}{1+P_{2,i}|h_{21,i}|^2}\bigg)\right|_{P_{2,i}=1}
				= \log\bigg(1+P_{3,i}|h_{32,i}|^2+\frac{|h_{22,i}|^2}{1+|h_{21,i}|^2}\bigg)
			\]
is maximized by $P_{3,i}=1.$

\end{enumerate}

\vspace{3mm}

\noindent
We conclude that if $\beta > 2\lambda+1$ then \eqref{eq:MaxPGDoF} is maximized when all nodes transmit at their maximum available power: $P_{1,i}=P_{2,i}=P_{3,i}=1, i\in\{1,2,...,n\}$. Finally, step (g) follows since in the ergodic phase fading model, the magnitudes of the channel coefficients are constants and do not depend on the time index, and therefore the expectation can be omitted. Observe that as $(R_1,R_2)$ is achievable, then for $k\in\{1,2\}$, $P_{e,k}^{(n)}\rightarrow 0$ as $n\rightarrow\infty$, and hence, $\epsilon_{kn}\rightarrow 0$ as $n\rightarrow\infty$. We therefore conclude that the sum-rate is asymptotically bounded by
\begin{eqnarray*}
    R_{\mbox{\scriptsize sum}} &\le&\log\bigg(1+\SNRA_{21}+\frac{\SNRA_{11}+\SNRA_{31}+\SNRA_{11}\SNRA_{32}+2\sqrt{\SNRA_{11}\SNRA_{31}}}{1+\SNRA_{32}}\bigg)\\
&& \hspace{8.3 cm}\quad+\log\bigg(1+\SNRA_{32}+\frac{\SNRA_{22}}{1+\SNRA_{21}}\bigg)\\
    &\stackrel{(a)}{\doteq}& \log\Big(\SNR^{\max\{\alpha,1,\beta-\lambda\}}\Big) + \log\Big(\SNR^{\max\{\lambda,1-\alpha\}}\Big),
\end{eqnarray*}
where (a) follows since $\max\{1,\beta\}\ge \frac{1+\beta}{2}$. We note that if $\beta > 2\lambda+1$,
then $\beta - \lambda > 1 + \lambda \ge 1$, hence, $\max\{\alpha,1,\beta-\lambda\} = \max\{\alpha,\beta-\lambda\}$, and
\[
  \max\{\alpha,\beta-\lambda\}+\max\{\lambda,1-\alpha\}=\max\{\alpha+\lambda,\beta, 1, 1+\beta -\alpha - \lambda\}.
\]
 Therefore, if $\beta > 2\lambda+1$  then the genie-aided GDoF upper bound is given by
\begin{equation}
    \label{eq:GDoF1}
    \mbox{GDoF}_1^+=\max\{\alpha+\lambda,\beta,  1+\beta -\alpha - \lambda\}.
		%\max\{\alpha,\beta-\lambda\}+\max\{\lambda,1-\alpha\}.
\end{equation}

\subsubsection{An Upper Bound Based on the Cut-Set Theorem}
\label{subsec:upperboundcutset}
We derive three rate bounds
following along the lines of the proof of the cut-set theorem \cite[Thm. 15.10.1]{cover-thomas:it-book}. First, we derive an upper bounds on R$_1$ by considering the
cut $\Sset=\{\mbox{Tx}_1, \mbox{Relay}, \mbox{Rx}_2\}, \Sset^c=\{\mbox{Tx}_2,\mbox{Rx}_1\}$, i.e., allowing full cooperation between
$\mbox{Tx}_1$ and the $\mbox{Relay}$. For this cut we obtain
\begin{eqnarray}
        nR_1&=& H(M_1)\nonumber\\
        &\stackrel{(a)}{\le}& I(M_1;Y_1^n,\utH{}^n|M_2)+n\epsilon_{1n}\nonumber\\
        &=& \sum_{i=1}^n \Big[h(Y_{1,i},\utH_i|Y_{1}^{i-1},\utH{}^{i-1}, M_2) - h(Y_{1,i},\utH_i|Y_{1}^{i-1},\utH{}^{i-1}, M_1,M_2)\Big]+n\epsilon_{1n\nonumber}\\
        &=& \sum_{i=1}^n \Big[h(\utH_i|Y_{1}^{i-1},\utH{}^{i-1}, M_2) + h(Y_{1,i}|Y_{1}^{i-1},\utH{}^{i}, M_2)\nonumber\\
        &&\qquad\qquad\qquad\qquad
        - h(\utH_i|Y_{1}^{i-1},\utH{}^{i-1}, M_1,M_2)-h(Y_{1,i}|Y_{1}^{i-1},\utH{}^{i}, M_1,M_2)\Big]+n\epsilon_{1n}\nonumber\\
        &\stackrel{(b)}{=}& \sum_{i=1}^n\Big[h(\utH_i) + h(Y_{1,i}|Y_{1}^{i-1},\utH{}^{i}, M_2)\nonumber\\
        &&\qquad\qquad\qquad\qquad
        - h(\utH_i)-h(Y_{1,i}|Y_{1}^{i-1},\utH{}^{i}, M_1,M_2)\Big]+n\epsilon_{1n}\nonumber\\
        &\stackrel{(c)}{\le}& \sum_{i=1}^n\Big[ h(Y_{1,i}|Y_{1}^{i-1}, \utH{}^{i},M_2,X_{2,i}) - h(Y_{1,i}|Y_{1}^{i-1},\utH{}^{i}, M_1,M_2,X_{1,i},X_{2,i},X_{3,i})\Big]+n\epsilon_{1n}\nonumber\\
        &\stackrel{(d)}{\le}& \sum_{i=1}^n \Big[h(Y_{1,i}|X_{2,i},\utH_i) - h(Y_{1,i}|X_{1,i},X_{2,i},X_{3,i},\utH_i)\Big]+n\epsilon_{1n}\nonumber\\
        &=& \sum_{i=1}^n I(X_{1,i},X_{3,i};Y_{1,i}|X_{2,i},\utH_i)+ n\epsilon_{1n}\nonumber\\
        &=& \sum_{i=1}^n \E_{\utH_i}\Big\{ I(X_{1,i},X_{3,i};Y_{1,i}|X_{2,i},\utH_i = \uth_i)\Big\}+ n\epsilon_{1n}\nonumber\\
        &\stackrel{(e)}{\le}& \sum_{i=1}^n \E_{\utH_i}\bigg\{ \max_{\substack{0\le|\CORR_i|\le 1\\P_{k,i}\le1, k\in\{1,2,3\}}} I(X_{1G,i},X_{3G,i};Y_{1G,i}|X_{2G,i},\utH_i = \uth_i)\bigg\}+ n\epsilon_{1n}\nonumber\\
        &\stackrel{(f)}{\le}& \sum_{i=1}^n \E_{\utH_i}\Big\{ \log\Big(1+|h_{11,i}|^2+|h_{31,i}|^2\Big)\Big\}+ n\epsilon_{1n}\nonumber\\
        &=& \sum_{i=1}^n \log\Big(1+\SNRA_{11}+\SNRA_{31}\Big)+ n\epsilon_{1n}\nonumber\\
        &=& n \cdot \log\Big(1+\SNRA_{11}+\SNRA_{31}\Big)+ n\epsilon_{1n},\label{eq:eq6G}
\end{eqnarray}
where (a) follows from Fano's inequality \cite[Thm. 2.10.1]{cover-thomas:it-book} and since the messages from Tx$_1$ and Tx$_2$ are drawn independently, (b) follows since channel coefficients are i.i.d. in time and are independent of the channel inputs, of the noise, and of the messages sent by the sources, (c) follows since $X_{2,i}$ is a deterministic function of $M_2$ and since adding conditioning decreases the differential entropy, (d) follows since adding conditioning can only decrease entropy, and since the channel outputs at time $i$ depend only on the channel inputs and the channel coefficients at time $i$. To prove step (e) first note that
\begin{equation*}
  h(Y_{1,i}|X_{2,i},\utH_i) - h(Y_{1,i}|X_{1,i},X_{2,i},X_{3,i},\utH_i)=h(Y_{1,i}|X_{2,i},\utH_i) - h(Z_{1,i}).
\end{equation*}
Step (e) then follows from \cite[Lemma 2]{Zahavi:12} which states that the conditional entropy is maximized by jointly circularly symmetric complex normal channel inputs with covariance matrix $\cov(X_{2,i},Y_{1,i})$. Note that we can write
\[
	(X_{2,i},Y_{1,i}) = \left[ \begin{array}{cccc}
	0 & 1 & 0 & 0\\
	H_{11,i} \;& \; H_{21,i} \; & \; H_{31,i} \; & \; 1
	\end{array}\right]
	\left[\begin{array}{c}
	X_{1,i}\\ X_{2,i} \\ X_{3,i} \\ Z_{1,i}
	\end{array}\right].
\]
As the pair $(X_{2,i},Y_{1,i})$ is a linear transformation of a random vector, and as in addition, $(X_{2,i},Y_{1,i})$ is distributed according to a zero mean, jointly complex Gaussian
distribution, we conclude that the joint distribution of $(X_{2G,i},Y_{1G,i})$ with the covariance matrix  $\cov(X_{2,i},Y_{1,i})$  is obtained by letting $(X_{1,i}, X_{2,i},X_{3,i}, Z_{1,i})$ be a jointly complex Gaussian random vector, which, in turn is obtained when $(X_{1,i}, X_{2,i},X_{3,i})$ is a jointly complex Normal vector with covariance matrix $\cov(X_{1,i}, X_{2,i},X_{3,i})$.
Finally, step (f) follows from \cite[Eqn. (A.10)]{Zahavi:12}.

Next, by using the cut $\Sset=\{\mbox{Tx}_1,\mbox{Rx}_2\}, \Sset^c=\{\mbox{Relay},\mbox{Tx}_2,\mbox{Rx}_1\}$, i.e., by allowing full cooperation between Rx$_1$ and the relay, we obtain an additional upper bound on $R_1$. This bound is expressed as:
\begin{eqnarray}
        nR_1&=& H(M_1)\nonumber\\
        &\le& I(M_1;Y_1^n,\utH^n|M_2)+n\epsilon_{1n}\nonumber\\
        &\le& I(M_1;Y_1^n,Y_3^n,\utH^n|M_2)+n\epsilon_{1n}\nonumber\\
        &=& \sum_{i=1}^n \Big[ h(Y_{1,i},Y_{3,i},\utH_i|Y_{1}^{i-1},Y_{3}^{i-1},\utH^{i-1}, M_2) - h(Y_{1,i},Y_{3,i},\utH_i|Y_{1}^{i-1},Y_{3}^{i-1},\utH^{i-1}, M_1,M_2)\Big]+n\epsilon_{1n}\nonumber\\
        &=& \sum_{i=1}^n\Big[ h(\utH_i|Y_{1}^{i-1},Y_{3}^{i-1},\utH^{i-1}, M_2) + h(Y_{1,i},Y_{3,i}|Y_{1}^{i-1},Y_{3}^{i-1},\utH^{i}, M_2)\nonumber\\
        &&\qquad\qquad
        - h(\utH_i|Y_{1}^{i-1},Y_{3}^{i-1},\utH^{i-1}, M_1,M_2)-h(Y_{1,i},Y_{3,i}|Y_{1}^{i-1},Y_{3}^{i-1},\utH^{i}, M_1,M_2)\Big]+n\epsilon_{1n}\nonumber\\
        &\stackrel{(a)}{\le}& \sum_{i=1}^n\Big[ h(Y_{1,i},Y_{3,i}|Y_{1}^{i-1}, Y_{3}^{i-1}, \utH^{i},M_2,X_{2,i},X_{3,i})\nonumber\\
        &&\qquad\qquad- h(Y_{1,i},Y_{3,i}|Y_{1}^{i-1},Y_{3}^{i-1},\utH^{i}, M_1,M_2,X_{1,i},X_{2,i},X_{3,i})\Big]+n\epsilon_{1n}\nonumber\\
        &\le& \sum_{i=1}^n\Big[ h(Y_{1,i},Y_{3,i}|X_{2,i},X_{3,i},\utH_i) - h(Y_{1,i},Y_{3,i}|X_{1,i},X_{2,i},X_{3,i},\utH_i)\Big]+n\epsilon_{1n}\nonumber\\
        &=&   \sum_{i=1}^n I(X_{1,i};Y_{1,i},Y_{3,i}|X_{2,i},X_{3,i},\utH_i)+ n\epsilon_{1n}\nonumber\\
        &=& \sum_{i=1}^n \E_{\utH_i}\Big\{ I(X_{1,i};Y_{1,i},Y_{3,i}|X_{2,i},X_{3,i},\utH_i=\uth_i)\Big\}+ n\epsilon_{1n}\nonumber\\
        &\le& \sum_{i=1}^n \E_{\utH_i}\bigg\{\max_{\substack{0\le|\CORR_i|\le 1\\P_{k,i}\le1, k\in\{1,2,3\}}} I(X_{1G,i};Y_{1G,i},Y_{3G,i}|X_{2G,i},X_{3G,i},\utH_i=\uth_i)\bigg\}+ n\epsilon_{1n}\nonumber
\end{eqnarray}
\begin{eqnarray}
        &\stackrel{(b)}{\le}& \sum_{i=1}^n \E_{\utH_i}\Big\{\log\Big(1+|h_{11,i}|^2+|h_{13,i}|^2\Big)\Big\}+ n\epsilon_{1n}\hspace{7cm}\nonumber\\
        &=& \sum_{i=1}^n \log\Big(1+\SNRA_{11}+\SNRA_{13}\Big)+ n\epsilon_{1n}\nonumber\\
        &=& n \cdot \log\Big(1+\SNRA_{11}+\SNRA_{13}\Big)+ n\epsilon_{1n},\label{eq:eq7G}
\end{eqnarray}
where (a) follows since $M_2$ deterministically determines $X_2^n$, $(Y_3^{i-1},\tH_3^{i-1})$ deterministically determine $X_{3,i}$ and since conditioning reduces entropy, and (b) follows from \cite[Eqn. (A.5)]{Zahavi:12}.

Lastly, we use the cut  $\Sset=\{\mbox{Tx}_2,\mbox{Rx}_1\}$, $\Sset^c=\{\mbox{Relay},\mbox{Tx}_1,\mbox{Rx}_2\}$ to obtain an upper bound on $R_2$:
\begin{eqnarray}
        nR_2&=& H(M_2)\nonumber\\
        &\le& I(M_2;Y_2^n,\utH^n|M_1)+n\epsilon_{2n}\nonumber\\
        &\le& I(M_2;Y_2^n,Y_3^n,\utH^n|M_1)+n\epsilon_{2n}\nonumber\\
        &=& \sum_{i=1}^n\Big[ h(Y_{2,i},Y_{3,i},\utH_i|Y_{2}^{i-1},Y_{3}^{i-1},\utH^{i-1}, M_1) - h(Y_{2,i},Y_{3,i},\utH_i|Y_{2}^{i-1},Y_{3}^{i-1},\utH^{i-1}, M_1,M_2)\Big]+n\epsilon_{2n}\nonumber\\
        &=& \sum_{i=1}^n \Big[h(\utH_i|Y_{2}^{i-1},Y_{3}^{i-1},\utH^{i-1}, M_1) + h(Y_{2,i},Y_{3,i}|Y_{2}^{i-1},Y_{3}^{i-1},\utH^{i}, M_1) \nonumber\\
        &&\qquad\qquad- h(\utH_i|Y_{2}^{i-1},Y_{3}^{i-1},\utH^{i-1}, M_1,M_2)-h(Y_{2,i},Y_{3,i}|Y_{2}^{i-1},Y_{3}^{i-1},\utH^{i}, M_1,M_2)\Big]+n\epsilon_{2n}\nonumber\\
        &\le& \sum_{i=1}^n\Big[ h(Y_{2,i},Y_{3,i}|Y_{2}^{i-1}, Y_{3}^{i-1} \utH^{i},M_1,X_{1,i},X_{3,i})\nonumber\\
        &&\qquad\qquad- h(Y_{2,i},Y_{3,i}|Y_{2}^{i-1},Y_{3}^{i-1}, \utH^{i}, M_1,M_2,X_{1,i},X_{2,i},X_{3,i})\Big]+n\epsilon_{2n}\nonumber\\
        &\le& \sum_{i=1}^n \Big[h(Y_{2,i},Y_{3,i}|X_{1,i},X_{3,i},\utH_i) - h(Y_{2,i},Y_{3,i}|X_{1,i},X_{2,i},X_{3,i},\utH_i)\Big]+n\epsilon_{2n}\nonumber\\
        &=& \sum_{i=1}^n I(X_{2,i};Y_{2,i},Y_{3,i}|X_{1,i},X_{3,i},\utH_i)+ n\epsilon_{2n}\nonumber\\
        &=& \sum_{i=1}^n \Big[I(X_{2,i};Y_{3,i}|X_{1,i},X_{3,i},\utH_i)+ I(X_{2,i};Y_{2,i}|X_{1,i},X_{3,i},Y_{3,i},\utH_i)\Big]+ n\epsilon_{2n}\nonumber\\
        &\stackrel{(a)}{=}& \sum_{i=1}^n I(X_{2,i};Y_{2,i}|X_{1,i},X_{3,i},Y_{3,i},\utH_i)+ n\epsilon_{2n}\nonumber\\
        &=& \sum_{i=1}^n I(X_{2,i};Y_{2,i}|X_{1,i},X_{3,i},Z_{3,i},\utH_i)+ n\epsilon_{2n}\nonumber\\
        &\stackrel{(b)}{=}& \sum_{i=1}^n I(X_{2,i};Y_{2,i}|X_{3,i},\utH_i)+ n\epsilon_{2n}\nonumber\\
        &=& \sum_{i=1}^n \E_{\utH_i}\Big\{ I(X_{2,i};Y_{2,i}|X_{3,i},\utH_i=\uth_i)\Big\}+ n\epsilon_{2n}\nonumber
\end{eqnarray}
\begin{eqnarray}
        &\le& \sum_{i=1}^n \E_{\utH_i}\bigg\{\max_{\substack{0\le|\CORR_i|\le 1\\P_{k,i}\le1, k\in\{1,2,3\}}} I(X_{2G,i};Y_{2G,i}|X_{3G,i},\utH_i=\uth_i)\bigg\}+ n\epsilon_{2n}\hspace{4cm}\nonumber\\
        &=& \sum_{i=1}^n \E_{\utH_i}\Big\{ \log\Big(1+|h_{22,i}|^2\Big)\Big\}+ n\epsilon_{2n}\nonumber\\
        &=& \sum_{i=1}^n \log\Big(1+\SNRA_{22}\Big)+ n\epsilon_{2n}\nonumber\\
        &=& n \cdot \log\Big(1+\SNRA_{22}\Big)+ n\epsilon_{2n},\label{eq:eq8G}
\end{eqnarray}
where (a) follows since the signal $X_{2,i}$ is independent of $(X_{1,i},X_{3,i},Z_{3,i},\utH_i)$, and thus
\begin{equation*}
    I(X_{2,i};Y_{3,i}|X_{1,i},X_{3,i},\utH_i)=I(X_{2,i};Z_{3,i}|X_{1,i},X_{3,i},\utH_i)=0,
\end{equation*}
and (b) follows since $Y_{2,i}$ is a function of only $X_{2,i}, X_{3,i}, H_{22,i}, H_{32,i}$, and $Z_{2,i}$, and thus, given $X_{3,i}$, $Y_{2,i}$ is independent of $(X_{1,i}, Z_{3,i})$.
Since, for $n\rightarrow\infty$ we have $\epsilon_{kn}\rightarrow 0, k\in\{1,2\}$, then by combining \eqref{eq:eq6G}-\eqref{eq:eq8G} we obtain
\begin{eqnarray*}
    R_{\mbox{\scriptsize sum}}&\le&\min\Big\{ \log\Big(1+\SNRA_{11}+\SNRA_{31}\Big),\log\Big(1+\SNRA_{11}+\SNRA_{13}\Big)\Big\} + \log\Big(1+\SNRA_{22}\Big)\\
    &=&\min\Big\{\log\Big(1+\SNR+\SNR^{\beta}\Big),\log\Big(1+\SNR+\SNR^{\gamma}\Big)\Big\}+\log\Big(1+\SNR\Big)\\
    &\doteq&\min\Big\{\log\Big(\SNR^{\max\{1,\beta\}}\Big),\log\Big(\SNR^{\max\{1,\gamma\}}\Big)\Big\}+\log\Big(\SNR\Big).
\end{eqnarray*}
Thus, the cut-set based GDoF upper bound is given by:
\begin{equation}
    \label{eq:GDoF2}
    \mbox{GDoF}_2^+=1+\min\big\{\max\{1,\beta\},\max\{1,\gamma\}\big\}=\max\big\{2,1+\min\{\beta,\gamma\}\big\}.
\end{equation}
Note that \eqref{eq:GDoF2} holds for any relationship between $\alpha, \beta, \gamma$ and $\lambda$.
We conclude that an upper bound on the GDoF of the Z-ICR is given by the minimum of \eqref{eq:GDoF1} and \eqref{eq:GDoF2}, which coincides with \eqref{eq:GDoFOB}.
\end{proof}

\subsection{A Lower Bound on the Achievable GDoF}
\label{sec:IB}
A lower bound on the achievable GDoF of the ergodic phase fading Z-ICR is stated in the following proposition:
\begin{proposition}
\label{prop:lb_GDoF}
    Consider the ergodic phase fading Z-ICR defined in Section \ref{sec:Model}. The GDoF of this channel is lower bounded by
    \begin{equation}
        \label{eq:AchGDoF}
        \mbox{\em GDoF}^-=\min\Big\{\gamma, \max\big\{(1-\alpha)^+,(\beta-\alpha)^+\big\}\Big\}+ (1-\lambda)^+.
    \end{equation}
\end{proposition}

\phantomsection
\label{phn:proofProp3}

\begin{proof}
We use a communications scheme similar to the communications scheme of Section \ref{sec:achievableRegionIID}: The transmitters use mutually independent codebooks generated according to the i.i.d. (in time) complex Normal distributions: $X_{k,i}\sim\CN(0,P_k), k\in\{1,2,3\}, i\in\{1,2,...,n\}$, $0 < P_k \le 1$, $k\in\{1,2,3\}$. Encoding is based on the DF scheme at the relay, and for decoding we use a backward decoding scheme at Rx$_1$, and a PtP decoding rule at Rx$_2$, where both receivers treat the additive interference as noise.
Repeating the analysis in the proof of Prop. \ref{thm:achievable_region} it follows that this
coding scheme results in the following achievable rate region for the Z-ICR:
\begin{subequations}
\label{eq:ZICR_region}
\begin{eqnarray}
    R_1 &\le& \min\big\{I(X_1,X_3;Y_1|\tH_1),I(X_1;Y_3|X_3,\tH_3)\big\}\\
    R_2 &\le& I(X_2;Y_2|\tH_2).
\end{eqnarray}
\end{subequations}
Explicitly evaluating the mutual information expressions in \eqref{eq:ZICR_region} for the Gaussian p.d.f. on $(X_1,X_2,X_3)$ specified above, we arrive at
     \begin{eqnarray}
        I(X_1,X_3;Y_1|\tH_1)&=& \E_{\tH_1}\big\{I(X_1,X_3;Y_1|\tH_1=\th_1)\big\}\nonumber\\
        &=&\E_{\tH_1}\bigg\{\log\Big(1+\frac{P_1|h_{11}|^2+P_3|h_{31}|^2}{1+P_2|h_{21}|^2}\Big)\bigg\}\nonumber\\
        &=& \log\Big(1+\frac{P_1\SNR+P_3\SNR^{\beta}}{1+P_2\SNR^{\alpha}}\Big)\nonumber\\
        &\doteq& \log\Big(\SNR^{\max\{(1-\alpha)^+, (\beta-\alpha)^+\}}\Big),\label{eq:eq1}\\
        I(X_1;Y_3|X_3,\tH_3)&=&\E_{\tH_3}\big\{I(X_1;Y_3|X_3,\tH_3=\th_3)\big\}\nonumber\\
        &=&\E_{\tH_3}\big\{\log(1+P_1|h_{13}|^2)\big\}\nonumber\\
        &=&\log\Big(1+P_1\cdot\SNR^{\gamma}\Big)\nonumber\\
        &\doteq&\log\Big(\SNR^{\gamma}\Big)\label{eq:eq2C}\\
        I(X_2;Y_2|\tH_2)&=& \E_{\tH_2}\big\{I(X_2;Y_2|\tH_2=\th_2)\big\}\nonumber\\
        &=&\E_{\tH_2}\bigg\{\log\Big(1+\frac{P_2|h_{22}|^2}{1+P_3|h_{32}|^2}\Big)\bigg\}\nonumber\\
        &=&\log\bigg(1+\frac{P_2\cdot\SNR}{1+P_3\cdot\SNR^{\lambda}}\bigg)\nonumber\\
        &\doteq&\log\bigg(\SNR^{(1-\lambda)^+}\bigg).\label{eq:eq3}
    \end{eqnarray}
Combining \eqref{eq:eq1}-\eqref{eq:eq3}, we obtain the achievable GDoF stated in \eqref{eq:AchGDoF}.
\end{proof}

\subsection{The Optimality of Treating Interference as Noise}
\label{sec:GDoFOptSec}
In this section, we derive conditions on the SNR exponents of the channel coefficients under which the lower bound in \eqref{eq:AchGDoF} coincides with the upper bound in \eqref{eq:GDoFOB}, thereby characterizing the maximal GDoF for the ergodic phase fading Z-ICR in the WI regime. Consider \eqref{eq:AchGDoF} and note that if $\alpha\le \frac{1}{2}$ and $1+2\alpha< \beta \le\gamma+\alpha$, then it follows that $\gamma\ge 1$, $\beta\ge 1$,
$\max\big\{(1-\alpha)^+,(\beta-\alpha)^+\big\}=\beta-\alpha$, and $\min\{\gamma,\beta-\alpha\}=\beta-\alpha$. If it additionally holds that $\lambda=\alpha$ then \eqref{eq:AchGDoF} results in $\mbox{GDoF}^-=1+\beta-2\alpha$.

Next, consider \eqref{eq:GDoFOB} and note that if $\lambda=\alpha$, $1+2\alpha < \beta,$ and $\alpha \le \frac{1}{2}$, then
%$\max\{\alpha,1,\beta-\lambda\}+\max\{\lambda,1-\alpha\}\!=\!1+\beta-2\alpha$,
\[
  \max\{\alpha+\lambda, \beta,  1+\beta -\alpha - \lambda\} = \max\{2\alpha, \beta,  1+\beta -2\alpha \}= \max\{ \beta,  1+\beta -2\alpha \} =1+\beta -2\alpha
\]
and thus, \eqref{eq:GDoFOB} specializes to
\begin{equation*}
    \mbox{GDoF}^+=\min\Big\{\max\big\{2,1+\min\{\beta,\gamma\}\big\},\hspace{1mm} 1+\beta-2\alpha \Big\}.
\end{equation*}

Next, from \eqref{eq:optimalGDoFCon2} we conclude that $\gamma \ge 1$ and $\beta > 1+2\alpha\ge1$, and hence, $\max\big\{2,1+\min\{\beta,\gamma\}\big\}=1+\min\{\beta,\gamma\}$. Lastly, the condition $\beta\le\gamma+\alpha$ implies that $\beta\le\min\{\beta,\gamma\}+2\alpha$, i.e., $1+\beta-2\alpha\le1+\min\{\beta,\gamma\}$, and thus, it follows that $\mbox{GDoF}^+=1+\beta-2\alpha$. We conclude that if \eqref{eq:optimalGDoFCon} is satisfied, then \eqref{eq:GDoFOB} coincides with \eqref{eq:AchGDoF} and both are equal to $1+\beta-2\alpha$, thereby characterizing the maximal GDoF for the ergodic phase fading Z-ICR subject to \eqref{eq:optimalGDoFCon}.
The maximizing input distribution follows directly from the input distribution used in the proof of the lower bound in Prop. \ref{prop:lb_GDoF}, namely
$X_{k}\sim\CN(0,P_k)$, $0 < P_k \le 1$, $k\in\{1,2,3\}$.
\end{proof}

\subsection{Discussion}
\begin{comment}
\begin{remark}
    \em{} Note that one of the conditions required for guaranteeing that treating interference as noise is GDoF optimal
		is $\lambda=\alpha\!\le\! 0.5$.
		This corresponds to the {\em weak interference regime} for asymptotically high SNR \cite{etkin:08}.
		Note that in \cite{etkin:08} this regime was characterized as the regime in which interference is exponentially weaker than the direct link:
		$\frac{\SNRA^{\alpha}}{\SNRA} \rightarrow 0$ as $\SNRA \rightarrow \infty$, and $\frac{\SNRA^{\lambda}}{\SNRA} \rightarrow 0$ as $\SNRA \rightarrow \infty$. Note that while this holds for $\alpha,\lambda < 1$, in the current work we characterize the GDoF for part of this regime in
		which $\alpha, \lambda < \frac{1}{2}$.
		To consider some explicit scenarios, we note that when $\beta\!=\!\gamma\!=\!2$, then the maximal GDoF is characterized for $\alpha\!\le\! 0.5$, while when $\beta\!=\!\gamma\!=\!1.2$, the maximal GDoF is characterized for $\alpha\!\le\! 0.1$. This is depicted in Fig. \ref{fig:fig1}.
\end{remark}
\end{comment}

\begin{figure}[!t]
\centering
    \hspace{-0.4cm}
    \subfloat[$\beta=\gamma=2$]{\includegraphics[scale=0.5]{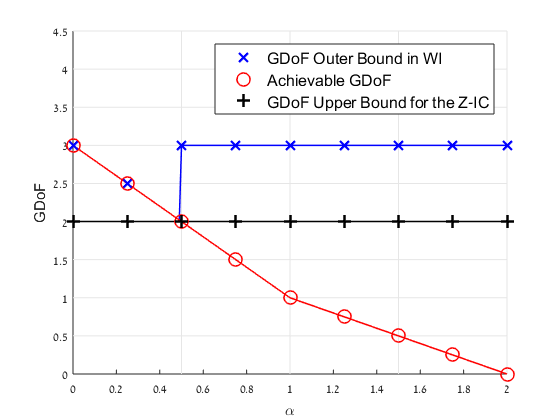}}
\hspace{-0.1cm}
    \subfloat[$\beta=\gamma=1.2$]{\includegraphics[scale=0.5]{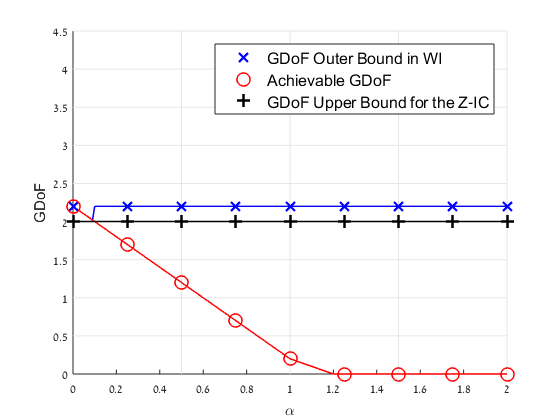}}
    \hspace{-0.4cm}
    \caption{\small The upper bound on the GDoF of \eqref{eq:GDoFOB}, and the achievable GDoF of \eqref{eq:AchGDoF} for the phase fading Z-ICR, together with the GDoF upper bound for the Z-IC given in \eqref{eq:ZICUB}}
    \label{fig:fig1}\vspace{-0.7 cm}
\end{figure}

\begin{remark}
\label{rem:GDoFImproveZIC}
    \em{} Consider the ergodic phase fading Z-IC:
%    In this model, one of the pairs suffers from interference while the second pair communicates over an interference-free channel.
    An upper bound on the achievable sum-rate for this channel is  obtained   by applying cut-set theorem \cite[Thm. 15.10.1]{cover-thomas:it-book}:
        \begin{eqnarray}
        R^{\mbox{\footnotesize PF-Z-IC}}_{\mbox{\scriptsize sum}}(\SNR) & \le & \max_{\left\{\substack{f(x_1)f(x_2):\\ \E\big\{|X_k|^2\big\}\le1, \;k\in\{1,2\}}\right\}}\Big\{\! I(X_{1};Y_{1}|X_{2},\utH)\! +\! I(X_{2};Y_{2}|\utH)\!\Big\}\nonumber\\
        \label{eq:ZICUB}
        & \doteq & \log(\SNR)+\log(\SNR).
    \end{eqnarray}
    %\begin{equation}
    %    \label{eq:ZICUB}
    %    R^{\mbox{\tiny PF Z-IC}}_{\mbox{\scriptsize sum}}(\SNR) \le \max_{f(x_1)f(x_2):\E\big\{|X_k|^2\big\}\le1, k\in\{1,2\}}\!\!\Big\{\! I(X_{1};Y_{1}|X_{2},\utH)\! +\! I(X_{2};Y_{2}|\utH)\!\Big\}\doteq \log(\SNR)+\log(\SNR).
    %\end{equation}
    It follows that the GDoF for this channel is upper bounded by $2$. Comparing the GDoF upper bound of the phase fading Z-IC with the lower bound on the GDoF of the phase fading Z-ICR stated in \eqref{eq:AchGDoF}, we note that if $1\le \beta \le \gamma+\alpha$, then for the ergodic
		phase fading Z-ICR we have $\mbox{GDoF}_{\mbox{\tiny Z-ICR}}^{-}=(\beta-\alpha)^++ (1-\lambda)^+$. Hence, when $2 < (\beta-\alpha)^++(1-\lambda)^+$ the relay node {\em strictly increases} the GDoF of the ergodic phase fading Z-IC {\em even in scenarios in which the relay receives transmissions only from one of the transmitters}, and the interference is treated as noise at both destinations. In Fig. \ref{fig:fig1}, the GDoF upper bound and the GDoF lower bound for the Z-ICR, as well as the upper bound on the GDoF of the Z-IC, are plotted vs. $\alpha$ for two sets of $(\beta,\gamma)$: $\beta=\gamma=2$ and $\beta=\gamma=1.2$, subject to ergodic phase fading.
%Figure \ref{fig:fig1} also depicts the GDoF upper bound of the ergodic phase-fading Z-IC for comparison with the GDoF bounds for the Z-ICR.
Observe that the GDoF for the Z-ICR is {\em strictly} greater than that of the Z-IC for  $\beta=\gamma=2$, when  $\alpha< 0.5$ and for $\beta=\gamma=1.2$, when $\alpha< 0.1$.
Fig. \ref{fig:fig1} also clearly demonstrates the GDoF optimality of treating interference as noise in the WI regime.
\end{remark}

\begin{remark}
\label{commt:mutualIndep}
    \em{} Note that from Theorems \ref{thm:WI} and \ref{thm:WI-GDoF} we conclude that mutually independent channel inputs achieve both the sum-rate capacity and the maximal GDoF of the ergodic phase fading Z-ICR in the weak interference regime. Hence, using the communications scheme described in Section \ref{sec:achievableRegionIID}, there is no need for coordinating the codebooks of Tx$_1$ and of the relay to achieve optimality in both perspectives (capacity and GDoF). This observation suggests that when adding a relay to the interference network considered in this manuscript,
		the transmission scheme at the sources should remain unchanged, and that only the receivers should be modified to take advantage of the relay transmissions when decoding the messages
from the sources, in order to improve performance. This conclusion substantially simplifies adding relay nodes to existing wireless communications networks, and provides a strong support for user cooperation for interference management in the weak interference regime.
\end{remark}

\begin{remark}
\label{comm:achifinit}
\em{} From the derivation of the achievable GDoF in Section \ref{sec:FULLGDoF}, it directly follows that the maximal GDoF of the ergodic phase fading Z-ICR can be achieved with channel inputs generated according to mutually independent i.i.d Gaussians {\em with any arbitrary non-zero power},
and it is not necessary to use the maximal power $P_k=1$, $k=1,2,3$, for generating the channel inputs. Note, however, that the  technical derivation of the GDoF upper bound does require $P_k=1$, $k=1,2,3$,
because we first upper bound the rate at {\em any} SNR and then take $\SNRA \rightarrow \infty$. Yet, the achievability scheme can obtain the maximal GDoF when the nodes transmit with any finite positive powers, as long as
the conditions of Thm. \ref{thm:WI-GDoF} are satisfied.
\end{remark}

\section{Conclusions}
\label{sec:conclusions}
In this paper, we studied the two major performance measures of the ergodic phase fading Z-ICR: The sum-rate capacity and the GDoF. We focused on scenarios in which the interference is weak and the relay receives transmissions only from $\Tgood$. We first characterized the sum-rate capacity of the ergodic phase fading Z-ICR in the WI regime. This is the first capacity result for the Z-ICR in the WI regime in which the relay power is finite. Next, we explained why GDoF analysis is relevant for this channel model although the fading process is ergodic, and then characterized the maximal GDoF for this channel in the weak interference regime. To the best of our knowledge, this is the first time that GDoF analysis is carried out for a fading scenario. For both performance measures, optimal performance was achieved by treating the interfering signal as additive noise at the
destination receivers, in combination with using the DF strategy at the relay.

Our results show that adding a relay to the Z-IC enhances both its sum-capacity and GDoF compared to communications without a relay. Combined with our previous results on fading ICRs in the SI regime \cite{Zahavi:12}, \cite{Dabora:12}, and \cite{Zahavi:15}, we conclude that there is a very strong motivation for employing relay nodes for interference management in both the WI regime as well as in the strong interference regime. Additionally, the fact that the optimal channel inputs are mutually independent both in the strong interference regime and in the weak interference regime, further motivates incorporating relay nodes into existing wireless networks. The results in this paper constitute a starting point for studying the combination of cooperation and interference in the WI regime.

%----------------------------------------------------------------------------------------------------------------------------
%                                                       Appendices
%----------------------------------------------------------------------------------------------------------------------------
\newpage
\appendices
\setcounter{equation}{0}
\numberwithin{equation}{section}

\section{Proof of Lemma \ref{lem:lemma1}}
\label{app:ProofOfLemma1FULL}
The proof is based on \cite[Theorem 1]{Liu:07} and \cite[Corollary 2]{Shang:09}. Note that since $\Zvec_1$ and $\Zvec_2$ are circularly symmetric complex Normal random vectors, then they can be written as $\tilde{\Zvec}_k\triangleq(\Zvec_{kR}^T,\Zvec_{kI}^T)^T, k=1,2$, where $\Zvec_{kR}$ and $\Zvec_{kI}$ are two mutually independent i.i.d., $n$-dimensional real Gaussian random vectors which represent the real and imaginary parts of $\Zvec_k$, respectively. It follows that the noise vectors $\tilde{\Zvec}_1$ and $\tilde{\Zvec}_2$ are two $2n$-dimensional Gaussian random vectors with i.i.d. entries. Similarly, consider $\tilde{\Xvec}\triangleq(\Xvec_R^T,\Xvec_I^T)^T$ where $\Xvec_R$ and $\Xvec_I$ are the real and imaginary parts of $\Xvec$, respectively. Using these new definitions, we obtain
     \begin{eqnarray*}
        h(\Xvec+\Zvec_1)-h(\Xvec+\Zvec_2) &=& h(\Xvec_R+\Zvec_{1R},\Xvec_I+\Zvec_{1I})-h(\Xvec_R+\Zvec_{2R},\Xvec_I+\Zvec_{2I})\nonumber\\
        &=& h(\tilde{\Xvec}+\tilde{\Zvec}_1)-h(\tilde{\Xvec}+\tilde{\Zvec}_2).
    \end{eqnarray*}
It follows that the maximization problem in \eqref{eq:Lemma1OP} can be rewritten in the following equivalent form:
    \begin{eqnarray}
        \label{eq:Lemma1OP2}
        &&\max_{f(\tilde{\xvec})}\mbox{ } h(\tilde{\Xvec}+\tilde{\Zvec}_1)-h(\tilde{\Xvec}+\tilde{\Zvec}_2)\\
        &&\mbox{subject to } \tr\big(\cov(\tilde{\Xvec})\big)\le nP,\nonumber
    \end{eqnarray}
where we note that the constraint in the new maximization problem is due to the fact that $\tr\big(\cov(\Xvec)\big)=\tr\big(\cov(\tilde{\Xvec})\big)$. Next, note that in \cite[Section III]{Shang:09} it is stated that a Gaussian random vector is the optimal solution to \eqref{eq:Lemma1OP2} (see the discussion beneath Eq. (26) in \cite{Shang:09}). Additionally, note that from \cite[Lemma. 1]{Zahavi:12} it follows that for a random vector $\Xvec$, the zero-mean, random vector $\Xvec^{\rm zm}\triangleq\Xvec-\E\{\Xvec\}$ has the same entropy as $\Xvec$. Thus, from these two observations, we obtain that a zero-mean complex Normal random vector is the optimal solution to the original maximization problem in \eqref{eq:Lemma1OP}. Finally, note that in \cite[Corollary 2]{Shang:09}, it is further stated that the optimal solution to \eqref{eq:Lemma1OP2} should have a diagonal covariance matrix of the form $\tilde{P}\Imat$ for some positive real scalar $\tilde{P}$. I.e., the optimal solution for the maximization problem in \eqref{eq:Lemma1OP} should be further {\em circularly symmetric}. Finally, setting $\mu=1$ in \cite[Corollary 2]{Shang:09}, we have that if $\gamma_1\le \gamma_2$, then, since we consider the trace constraint of the form $\tr\big(\cov(\tilde{\Xvec})\big)\le nP$, we obtain $\tilde{P}=\frac{1}{2n}nP$. Hence, we conclude that the optimal solution to \eqref{eq:Lemma1OP} for the scenario where $\gamma_1\le \gamma_2$ is $\Xvec^{\mbox{\tiny Opt}}_G\sim\CN\Big(0,P\cdot\Imat_n\Big)$. This completes the proof of Lemma \ref{lem:lemma1}.
\tend

\section{Proof of Lemma \ref{lem:lemma2}}
\label{app:ProofOfLemma2FULL}
We follow the same steps as in the proof of \cite[Lemma 3]{Shang:09} using circularly symmetric complex Normal RVs instead of real-valued Normal RVs. Consider a complex RV $V'$ that has the same marginal distribution as $V$ but is independent of $(Z,W)$ and $X$. Then, we obtain
    \begin{eqnarray*}
        h(X^n+Z^n|W^n)&=&\int_{\Cset^n} f_{W^n}(w^n) h(X^n+Z^n|W^n=w^n) dw^n\\
                &\stackrel{(a)}{=}&\int_{\Cset^n} f_{W^n}(w^n) h(X^n+V'^n+\frac{\CORRN_{12}}{\sigma_2^2}W^n|W^n=w^n) dw^n\\
                &=&\int_{\Cset^n} f_{W^n}(w^n) h(X^n+V'^n|W^n=w^n) dw^n\\
                &=& h(X^n+V'^n)\\
                &=& h(X^n+V^n),
    \end{eqnarray*}
where (a) follows from the fact that $(V'^n+\frac{\CORRN_{12}}{\sigma_2^2}W^n, W^n)$ has the same joint distribution as $(Z^n,W^n)$. This can be shown using the fact that $V'$ is independent of $W$, and thus, for $k\in\{1,2,...,n\}$ it follows that
    \begin{eqnarray*}
        \E\left\{\Big|V'_k+\frac{\CORRN_{12}}{\sigma_2^2}W_k\Big|^2\right\}&=&\E\{|V'_k|^2\}+\E\left\{\Big|\frac{\CORRN_{12}}{\sigma_2^2}W_k\Big|^2\right\}=
        \Big(\sigma_1^2-\frac{|\CORRN_{12}|^2}{\sigma_2^2}\Big)+\frac{|\CORRN_{12}|^2}{\sigma_2^2}=\sigma_1^2\\
        \E\left\{\Big(V'_k+\frac{\CORRN_{12}}{\sigma_2^2}W_k\Big)\cdot W_k^*\right\}&=&\CORRN_{12}.
    \end{eqnarray*}
As $Z^n,W^n$ and $V'^n$ have i.i.d. elements, then $(V'^n+\frac{\CORRN_{12}}{\sigma_2^2}W^n, W^n)$ has the same mean and the same covariance matrix as $(Z^n,W^n)$. Since $Z^n$, $W^n$, and $V'^n$ are all complex Normal RVs, then the fact the first and the second moments are identical implies that $(V'^n+\frac{\CORRN_{12}}{\sigma_2^2}W^n, W^n)$ has the same joint distribution as $(Z^n,W^n)$. This completes the proof.
\tend

\section{Proof of Lemma \ref{lem:lemma22}}
\label{app:ProofOfLemma3FULL}
We follow similar steps as in the proof of \cite[Lemma 1]{Annapureddy:09}, the only difference being that we use circularly symmetric complex Normal RVs instead of real-valued Normal RVs. Let $Q$ be a time sharing random variable taking values from $1$ to $n$ with equal probability. Let ${\Xvec}_G\sim\CN(0,\frac{1}{n}\sum_{i=1}^{n}\Qmat_{\Xvec_i})$,
and let ${\Yvec}_G$ and ${\Svec}_G$ be the corresponding $\Yvec$ and $\Svec$. Then
\begin{eqnarray*}
    h(\Yvec^n|\Svec^n,\Hmat^n)
    &=&\sum_{i=1}^n h(\Yvec_i|\Yvec^{i-1},\Svec^n,\Hmat^n)\\
    &\stackrel{(a)}{\le}&\sum_{i=1}^n h(\Yvec_i|\Svec_i,\Hmat_i)\\
    &=&n\cdot h(\Yvec_Q|\Svec_Q,\Hmat_Q,Q)\\
    &\stackrel{(b)}{\le}&n\cdot h(\Yvec_Q|\Svec_Q,\Hmat_Q)\\
    &=&n\cdot\E_{\Hmat_Q}\big\{h(\Yvec_Q|\Svec_Q,\hmat_Q)\big\}\\
    &\stackrel{(c)}{=}&n\cdot\E_{\Hmat}\big\{h(\Yvec_Q|\Svec_Q,\hmat)\big\}\\
    &\stackrel{(d)}{\le}&n\cdot\E_{\Hmat}\big\{h({\Yvec}_G|{\Svec}_G,\hmat)\big\}\\
    &=&n\cdot h({\Yvec}_G|{\Svec}_G,\Hmat),
\end{eqnarray*}
where steps (a) and (b) follow since conditioning reduces entropy, (c) follows since the distribution of $\Hmat_Q$ does not depends on $Q$. To prove step (d) note that from \cite[Lemma. 1]{Zahavi:12} it follows that for a random vector $\Xvec$, the zero-mean, random vector $\Xvec^{\rm zm}\triangleq\Xvec-\E\{\Xvec\}$ has the same entropy as $\Xvec$, and hence, since $h({\Yvec}_Q|{\Svec}_Q,\hmat)=h({\Yvec}_Q,{\Svec}_Q|\hmat)-h({\Svec}_Q|\hmat)$, then
\begin{equation*}
  h({\Yvec}_Q|{\Svec}_Q,\hmat)=h({\Yvec}_Q-\E\{{\Yvec}_Q\}|{\Svec}_Q-\E\{{\Svec}_Q\},\hmat).
\end{equation*}
Thus, we can consider only zero mean RVs to further upper bound $h(\Yvec_Q|\Svec_Q,\hmat)$ from step (c). Step (d) then follows from \cite[Lemma 2]{Zahavi:12}.
\tend

\section{Proof of Lemma \ref{lem:lemma3}}
\label{app:ProofOfLemma3}
    The proof is based on the proof of \cite[Lemma 8]{Annapureddy:09}. %and incorporates elements from the proof of \cite[Lemma 7]{Annapureddy:09}.
    Recall that a circularly symmetric, complex Normal random vector can be represented as a random vector with double the length, whose components are real, jointly Gaussian RVs. It follows that, as \cite[Lemma 7]{Annapureddy:09} is stated for real Gaussian random vectors,\footnote{\cite[Lemma 7]{Annapureddy:09} states that for the real Gaussian random vectors $\Xvec$, $\Yvec$, and $\Svec$, the following three statements are equivalent: (1) $I(\Xvec;\Svec|\Yvec)=0$, (2) $\Xvec-\Yvec-\Svec$ form a Markov chain, and (3) $\hat{\Svec}(\Xvec,\Yvec)$, the MMSE estimate of $\Svec$ given $(\Xvec,\Yvec)$, is equal to $\hat{\Svec}(\Yvec)$ the MMSE estimate of $\Svec$ given $\Yvec$.} it holds also for  circularly symmetric complex Normal random vectors.
    Letting $\Xvec\triangleq (X_1,X_2)^T$ then
     from \cite[Lemma 7]{Annapureddy:09} we obtain the following equivalence for $(\Xvec, Y_1, Y_2)$:
    \begin{equation}
    \label{eq:AppA3}
    I(\Xvec;Y_1|Y_2)=0 \Leftrightarrow \E\{Y_1|\Xvec,Y_2\}=\E\{Y_1|Y_2\}.
    \end{equation}
    %Let $\Xvec\triangleq (X_1,X_2)$ and let $\tilde{Y}_1(\Xvec,Y_2)$ denote
    The MMSE estimate of $Y_1$ based on $(\Xvec,Y_2)$is given by
    \begin{eqnarray}
        %\tilde{Y}_1(\Xvec,Y_2)&=&
        \E\{Y_1|\Xvec,Y_2\}
        &=&\E\{Y_1|\Xvec,Z_2\}\nonumber\\
        &=&c_1\cdot X_1+c_2\cdot X_2+\E\{Z_1|Z_2\}\nonumber\\
        &\stackrel{(a)}{=}&c_1\cdot X_1+c_2\cdot X_2+\frac{\E\{Z_1Z_2^*\}}{\E\{|Z_2|^2\}}Z_2\nonumber\\
        &=&\frac{\E\{Z_1Z_2^*\}}{\E\{|Z_2|^2\}}Y_2+\left(1-\frac{\E\{Z_1Z_2^*\}}{\E\{|Z_2|^2\}}\right)(c_1\cdot X_1+c_2\cdot X_2).\label{eq:AppA1}
    \end{eqnarray}
    Here (a) follows from \cite[Theorem 23.7.4]{lapidoth:book} which states that for zero-mean, jointly Normal real random vectors $\Qvec_1$ and $\Qvec_2$, the conditional expectation of $\Qvec_1$ given $\Qvec_2$ can be obtained as
	$\E\{\Qvec_1|\Qvec_2\}=\frac{\E\{\Qvec_1\Qvec_2^T\}}{\E\{\Qvec_2\Qvec_2^T\}}\Qvec_2$. Then, the formula $\E\{Z_1|Z_2\} = \frac{\E\{Z_1Z_2^*\}}{\E\{|Z_2|^2\}}Z_2$ is obtained
by letting $\Qvec_k = (\Real\{Z_k\}, \Img\{Z_k\})^T$ and noting that joint circular symmetry of $(Z_1,Z_2)$ implies that $\E\{\Real\{Z_1\}\Img\{Z_2\}\} = -\E\{\Img\{Z_1\}\Real\{Z_2\}\}$ and $\E\{\Real\{Z_1\}\Real\{Z_2\}\} = \E\{\Img\{Z_1\}\Img\{Z_2\}\}$. Computing $E\{Y_1|Y_2\}$ explicitly we obtain
    \begin{equation}
        \E\{Y_1|Y_2\}=\frac{\E\{Y_1Y_2^*\}}{\E\{|Y_2|^2\}}Y_2=\frac{\E\{(c_1\cdot X_1+c_2\cdot X_2) (c_1\cdot X_1+c_2\cdot X_2)^*\}+\E\{Z_1Z_2^*\}}{\E\{(c_1\cdot X_1+c_2\cdot X_2)(c_1\cdot X_1+c_2\cdot X_2)^*\}+\E\{|Z_2|^2\}}Y_2.\label{eq:AppA2}
    \end{equation}
    Comparing \eqref{eq:AppA1} and \eqref{eq:AppA2} we conclude that $\E\{Y_1|\Xvec,Y_2\}=\E\{Y_1|Y_2\}$, iff $\E\{Z_1Z_2^*\}=\E\{|Z_2|^2\}$. Hence, from \eqref{eq:AppA3} we obtain that $I(\Xvec;Y_1|Y_2)=0$ iff $\E\{Z_1Z_2^*\}=\E\{|Z_2|^2\}$.
\tend

The lemma also can be proved by explicitly using the real vector representation of complex Normal vectors as follows:
 Let %$Y_k = Y_{kR} + j\cdot Y_{kI}$, where
    $Y_{kR}=\Real\{Y_k\}$, $Y_{kI}=\Img\{Y_k\}$, $Z_{kR}=\Real\{Z_k\}$, $Z_{kI}=\Img\{Z_k\}$,
    $D = c_1 \cdot  X_1   +   c_2 \cdot  X_2$,  $D_R =\Real\{D\}$, and $D_I=\Img\{D\}$. Additionally, denote $\bY_k=(Y_{kR}, Y_{kI})^T$, $\bZ_k=(Z_{kR}, Z_{kI})^T$, $k=1,2$,  $\bD = (D_R, D_I)^T$, and
    $\Dmat = \left[\begin{array}{cc} D_R^2 & \;\; D_R D_I\\ D_ID_R & \;\; D_I^2 \end{array} \right]$. Finally, define $\Xvec_R = \Real\{\Xvec\}$, $\Xvec_I = \Img\{\Xvec\}$,
    $\bXvec = \big(\Xvec_R^T,\Xvec_I^T\big)^T$, and $\Zmat = \frac{1}{\var(Z_2)}\left[\begin{array}{cc} 2\E\left\{Z_{1R}Z_{2R}\right\} & \;\; 2\E\left\{Z_{1R}Z_{2I}\right\} \\ -2\E\left\{Z_{1R}Z_{2I}\right\} & \;\; 2\E\left\{Z_{1R}Z_{2R}\right\} \end{array} \right]$.
    Lastly, note that since $Z_1$ and $Z_2$ are jointly circularly symmetric complex Normal, then by definition,
    \begin{eqnarray*}
        \E\left\{ \big(Z_{1R}+jZ_{1I}\big)\big(Z_{2R}+jZ_{2I}\big) \right\} = \E\Big\{\big(Z_{1R}Z_{2R} - Z_{1I}Z_{2I}\big) +j\big(Z_{1I}Z_{2R}+Z_{1R}Z_{2I}\big)\Big\} = 0,
    \end{eqnarray*}
    hence,
    \begin{equation}
    \label{eqn:appndxeqn_jointcircsymm}
       \E\Big\{ Z_{1R}Z_{2R} \Big\} =     \E\Big\{Z_{1I}Z_{2I}\Big\}, \quad \mbox{    and    } \quad \E\Big\{Z_{1I}Z_{2R}\Big\} = -\E\Big\{Z_{1R}Z_{2I}\Big\} = 0.
    \end{equation}

    With these definitions,
    for the MMSE estimate of $Y_1$ based on $(\Xvec,Y_2)$ we write
    \begin{eqnarray}
        \E\{Y_1|\Xvec,Y_2\}
        & & \!\! = \E\{\bY_1|\bXvec,\bY_2\}\nonumber\\
        & &\!\! = \E\{\bY_1|\bXvec,\bZ_2\}\nonumber\\
        %%%%%%%%%%%%%%%%%%%%%%%%%%%%%%%%%%%%%%%%%%%%%%%%%%%%
        & & \!\! \stackrel{(a)}{=} \bD + \E\{\bZ_1|\bZ_2\}\nonumber\\
        %%%%%%%%%%%%%%%%%%%%%%%%%%%%%%%%%%%%%%%%%%%%%%%%%
        & & \!\! = \bD + \E\{\bZ_1\cdot\bZ_2^T\}\left(\E\{\bZ_2\cdot\bZ_2^T\}\right)^{-1}\bZ_2\nonumber\\
        %%%%%%%%%%%%%%%%%%%%%%%%%%%%%%%%%%%%%%%%%%%%%%%%
        & & \!\! = \bD + \E\left\{\left[\begin{array}{cc}Z_{1R}Z_{2R} & \;\; Z_{1R}Z_{2I} \\ Z_{1I}Z_{2R} & \;\; Z_{1I}Z_{2I} \end{array} \right]\right\}\cdot\left(\E\left\{\left[\begin{array}{cc}Z_{2R}^2 & \;\; Z_{2R}Z_{2I} \\ Z_{2I}Z_{2R} & \;\; Z_{2I}^2 \end{array} \right]\right\}\right)^{-1}\bZ_2\nonumber\\
        %%%%%%%%%%%%%%%%%%%%%%%%%%%%%%%%%%%%%%%%%%%%%%%%%%%%%%%
    & & \!\! \stackrel{(b)}{=} \bD + \E\left\{\left[\begin{array}{cc}Z_{1R}Z_{2R} & \;\; Z_{1R}Z_{2I} \\ -Z_{1R}Z_{2I} & \;\; Z_{1R}Z_{2R} \end{array} \right]\right\}\cdot\left(\left[\begin{array}{cc}\frac{1}{2}\var(Z_2) & \;\; 0 \\ 0 & \;\; \frac{1}{2}\var(Z_2) \end{array} \right]\right)^{-1}\bZ_2\nonumber\\
        %%%%%%%%%%%%%%%%%%%%%%%%%%%%%%%%%%%%%%%%%%%%%%%%%%%%%%%
    & & \!\! = \bD +  \frac{1}{\var(Z_2)}\left[\begin{array}{cc} 2\E\left\{Z_{1R}Z_{2R}\right\} & \;\; 2\E\left\{Z_{1R}Z_{2I}\right\} \\ -2\E\left\{Z_{1R}Z_{2I}\right\} & \;\; 2\E\left\{Z_{1R}Z_{2R}\right\} \end{array} \right]\bZ_2\nonumber\\
    %%%%%%%%%%%%%%%%%%%%%%%%%%%%%%%%%%%%%%%%%
    & & \!\!= \bD + \Zmat \cdot \bZ_2\nonumber\\
    %%%%%%%%%%%%%%%%%%%%%%%%%%%%%%%%%%%%%%%
        & & \!\!= \bD \left(\Imat_2 - \Zmat\right) + \Zmat \cdot \bY_2.\label{eq:AppA1}
    \end{eqnarray}
    Here, (a) follows from \cite[Theorem 23.7.4]{lapidoth:book} by using the real vector representation for the complex RVs, and (b) follows from the conditions on the cross-correlations in the statement of the lemma.

    Computing $E\{Y_1|Y_2\}$ explicitly we obtain
    \begin{eqnarray}
       \E\{Y_1|Y_2\}
        &  & = \E\{\bY_1|\bY_2\}\nonumber\\
        %%%%%%%%%%%%%%%%%%%%%%%%%%%%
        &  & = \E\left\{\big(\bD+\bZ_1\big)\cdot\big(\bD+\bZ_2\big)^T \right\} \cdot\left(\E\left\{ \big(\bD+\bZ_2\big)\cdot\big(\bD+\bZ_2\big)^T\right\}\right)^{-1}\bY_2\nonumber\\
        %%%%%%%%%%%%%%%%%%%%%%%%%%%%%%%%
        &  & = \E\left\{\bD \cdot \bD^T+\bZ_1\cdot\bZ_2^T \right\} \cdot\left(\E\left\{ \bD \cdot \bD^T +\bZ_2\cdot\bZ_2^T\right\}\right)^{-1}\bY_2\nonumber\\
        %%%%%%%%%%%%%%%%%%%%%%%%%%%%%%%%
        &  & = \left(\frac{1}{\var(Z_2)}2\E\left\{\Dmat\right\}+\Zmat\right)\cdot\left(\frac{1}{\var(Z_2)}2\E\left\{ \Dmat\right\} +\Imat_2\right)^{-1}\bY_2.
        \label{eq:AppA2}
    \end{eqnarray}
    Comparing \eqref{eq:AppA1} and \eqref{eq:AppA2} we conclude that $\E\{Y_1|\Xvec,Y_2\}=\E\{Y_1|Y_2\}$, if and only if
    $\Zmat=\Imat_2$, namely
    \begin{eqnarray*}
        \E\left\{Z_{1R}Z_{2R}\right\} & \stackrel{(a)}{=} & \E\left\{Z_{1I}Z_{2I}\right\} = \frac{1}{2}\var(Z_2)\\
        \E\left\{Z_{1R}Z_{2I}\right\} & \stackrel{(b)}{=} & \E\left\{Z_{2R}Z_{1I}\right\} \stackrel{(c)}{=} 0.
    \end{eqnarray*}
    Noting that relationships (a), (b), and (c) hold by joint circular symmetry, as shown in \eqref{eqn:appndxeqn_jointcircsymm}, and also
      that $\E\{|Z_2|^2\} = \var(Z_2)$, and  $\E\{Z_1\cdot Z_2^*\} = \E\big\{(Z_{1R}Z_{2R} + Z_{1I}Z_{2I}) + j\cdot ( -Z_{1R}Z_{2I} + Z_{2R}Z_{1I})\big\}$, we conclude that subject to the conditions of the lemma,  $\Zmat=\Imat_2$ is equivalent to $\E\{|Z_2|^2\} =\E\{Z_1\cdot Z_2^*\} $.

\bigskip
To complete  the derivations we show that the expression $\frac{\E\{(c_1\cdot X_1+c_2\cdot X_2) (c_1\cdot X_1+c_2\cdot X_2)^*\}+\E\{Z_1Z_2^*\}}{\E\{(c_1\cdot X_1+c_2\cdot X_2)(c_1\cdot X_1+c_2\cdot X_2)^*\}+\E\{|Z_2|^2\}}Y_2$ is equivalent to $\left(\frac{1}{\var(Z_2)}2\E\left\{\Dmat\right\}+\Zmat\right)\cdot\left(\frac{1}{\var(Z_2)}2\E\left\{ \Dmat\right\} +\Imat_2\right)^{-1}\bY_2$.
First note that $D = c_1\cdot X_1+c_2\cdot X_2$ is a circularly symmetric complex Normal scalar, hence $\E\{D_R^2\} =\E\{ D_I^2\big\}$ and $\E\{D_ID_R\} = -\E\{D_RD_I\}=0$. Thus:
\begin{eqnarray*}
   \E\{(c_1\cdot X_1+c_2\cdot X_2) (c_1\cdot X_1+c_2\cdot X_2)^*\} & = & \E\Big\{\big(D_R+jD_I\big)\big(D_R-jD_I\big)\Big\}\\
& = & \E\Big\{\big(D_R^2 + D_I^2\big) +j\big(D_ID_R-D_RD_I\big)\Big\}\\
& = & 2\E\Big\{D_R^2\Big\}.
\end{eqnarray*}
Next, from the p.d.f. of $Z_1$ we further have $\E\{|Z_1|^2\} = 1$, which implies that $\E\Big\{ Z_{1R}^2\Big\} =     \E\Big\{Z_{1I}^2\Big\} = \frac{1}{2}$, and $\E\Big\{Z_{1I}Z_{1R}\Big\} = -\E\Big\{Z_{1R}Z_{1I}\Big\} = 0$.
Thus, we obtain
%%%%%%%%%%%%%%%%%%%%%%
\begin{eqnarray*}
\E\{Z_1Z_2^*\} & = & \E\Big\{\big(Z_{1R}+jZ_{1I}\big)\big(Z_{2R}-jZ_{2I}\big)\Big\}\\
&  = &  \E\Big\{\big(Z_{1R}Z_{2R} + Z_{1I}Z_{2I}\big) +j\big(Z_{1I}Z_{2R}-Z_{1R}Z_{2I}\big)\Big\}\\
&  = &  2\E\Big\{Z_{1R}Z_{2R}\Big\}.
\end{eqnarray*}
Combining the above derivations we can evaluate
\begin{eqnarray*}
&& \frac{\E\{(c_1\cdot X_1+c_2\cdot X_2) (c_1\cdot X_1+c_2\cdot X_2)^*\}+\E\{Z_1Z_2^*\}}{\E\{(c_1\cdot X_1+c_2\cdot X_2)(c_1\cdot X_1+c_2\cdot X_2)^*\}+\E\{|Z_2|^2\}}Y_2
 =\frac{2\E\Big\{D_R^2\Big\} + 2\E\Big\{Z_{1R}Z_{2R}\Big\}}{2\E\Big\{D_R^2\Big\}  + 1}(Y_{2R} +jY_{2I}).
\end{eqnarray*}
Note that plugging $\E\Big\{ Z_{1R}^2\Big\} =     \E\Big\{Z_{1I}^2\Big\} = \frac{1}{2}$, $\E\Big\{Z_{1I}Z_{2R}\Big\} = -\E\Big\{Z_{1R}Z_{2I}\Big\} = 0$, $\E\{D_R^2\} =\E\{ D_I^2\big\}$ and $\E\{D_ID_R\} = -\E\{D_RD_I\}=0$ in
$\left(\frac{1}{\var(Z_2)}2\E\left\{\Dmat\right\}+\Zmat\right)\cdot\left(\frac{1}{\var(Z_2)}2\E\left\{ \Dmat\right\} +\Imat_2\right)^{-1}\bY_2$, we obtain that both $Y_{2R}$ and $Y_{2I}$ are multiplied with the same
coefficients, which are equal to $\frac{2\E\Big\{D_R^2\Big\} + 2\E\Big\{Z_{1R}Z_{2R}\Big\}}{2\E\Big\{D_R^2\Big\}  + 1}$, which complete the proof of equivalence of the two approaches.

%%%%%%%%%%%%%%%%%%%%%%%%%%%%%%%%%%%%%%%%%%%%%%%%%%%%%%%%%%%%%%%%%%%%%%%%%%%%%%%%%%%%%%%%%%%%%%%%%%%%%%%%%%%%%%%%%%%%%%%%%%%%%%%%%%%%%%%
%%%%%%%%%%%%%%%%%%%%%%%%%%%%%%%%%%%%%%%%%%%%%%%%%%%%%%%%%%%%%%%%%%%%%%%%%%%%%%%%%%%%%%%%%%%%%%%%%%%%%%%%%%%%%%%%%%%%%%%%%%%%%%%%%%%%%%%

\section{Proof of Lemma \ref{lem:lemma6}}
\label{app:ProofOfLemma6}
We follow the same approach as the proof of \cite[Lemma 13]{Liu:07}. Let $\Xvec_1=X_1^n$, $\Xvec_2=X_2^m$, $\Zvec_1=Z_1^n$, $\Zvec_2=Z_2^m$,
$\Xvec=(\Xvec_1^T,\Xvec_2^T)^T$, and let $\Zvec=(\Zvec_1^T,\Zvec_2^T)^T$. By the chain rule of mutual information we have
\begin{equation*}
    I(\Xvec;\Xvec+\Zvec) = I(\Xvec_1;\Xvec_1+\Zvec_1) +  I(\Xvec_1;\Xvec_2+\Zvec_2|\Xvec_1+\Zvec_1)+ I(\Xvec_2;\Xvec+\Zvec|\Xvec_1).
\end{equation*}
In the following, we will show that
\begin{equation}
\label{eqnappB:desired}
    \lim_{\gamma_2\rightarrow\infty}I(\Xvec_1;\Xvec_2+\Zvec_2|\Xvec_1+\Zvec_1)=\lim_{\gamma_2\rightarrow\infty}I(\Xvec_2;\Xvec+\Zvec | \Xvec_1)=0,
\end{equation}
which will complete the proof. Starting with $I(\Xvec_1;\Xvec_2+\Zvec_2|\Xvec_1+\Zvec_1)$, we note that
\begin{eqnarray*}
    I(\Xvec_1;\Xvec_2+\Zvec_2|\Xvec_1+\Zvec_1)&=&h(\Xvec_2+\Zvec_2|\Xvec_1+\Zvec_1)-h(\Xvec_2+\Zvec_2|\Xvec_1,\Zvec_1)\\
    &=&h(\Xvec_2+\Zvec_2|\Xvec_1+\Zvec_1)-h(\Xvec_2+\Zvec_2|\Xvec_1,\Zvec_1)\\
    &\stackrel{(a)}{\le}&h(\Xvec_2+\Zvec_2)-h(\Xvec_2+\Zvec_2|\Xvec_1,\Zvec_1)\\
    &\stackrel{(b)}{=}&  h(\Xvec_2+\Zvec_2)-h(\Xvec_2+\Zvec_2|\Xvec_1)\\
    &=&I(\Xvec_1;\Xvec_2+\Zvec_2)\\
    &\stackrel{(c)}{\le}& I(\Xvec_2;\Xvec_2+\Zvec_2),
\end{eqnarray*}
where (a) follows since conditioning does not increase entropy \cite[Theorem 2.6.5]{cover-thomas:it-book}, (b) follows since $\Zvec_1$ is independent of $(\Zvec_2, \Xvec_1, \Xvec_2)$, and (c) follows since $I(\Xvec_1;\Xvec_2+\Zvec_2)\le I(\Xvec_1,\Xvec_2 ;\Xvec_2+\Zvec_2)$, and
\begin{eqnarray}
    I(\Xvec_1,\Xvec_2 ;\Xvec_2+\Zvec_2)  &=& I(\Xvec_2;\Xvec_2+\Zvec_2) +  I(\Xvec_1;\Xvec_2+\Zvec_2 | \Xvec_2)\nonumber \\
    &=& I(\Xvec_2;\Xvec_2+\Zvec_2) +  h(\Xvec_2+\Zvec_2|\Xvec_2) -  h(\Xvec_2+\Zvec_2|\Xvec_1, \Xvec_2)\nonumber\\
    &=& I(\Xvec_2;\Xvec_2+\Zvec_2) +  h(\Zvec_2) -  h(\Zvec_2)\nonumber\\
		\label{eqnapx:equiv1a}
    &=& I(\Xvec_2;\Xvec_2+\Zvec_2).
\end{eqnarray}
Next, note that
\begin{eqnarray}
    I(\Xvec_2;\Xvec+\Zvec|\Xvec_1)  &=& I(\Xvec_2;\Xvec_2+\Zvec_2|\Xvec_1) +  I(\Xvec_2;\Xvec_1+\Zvec_1|\Xvec_1,\Xvec_2+\Zvec_2)\nonumber \\
    &=& I(\Xvec_2;\Xvec_2+\Zvec_2|\Xvec_1) +  h(\Xvec_1+\Zvec_1|\Xvec_1,\Xvec_2+\Zvec_2) -  h(\Xvec_1+\Zvec_1|\Xvec_1,\Xvec_2,\Xvec_2+\Zvec_2)\nonumber\\
    &=& I(\Xvec_2;\Xvec_2+\Zvec_2|\Xvec_1) +  h(\Zvec_1) -  h(\Zvec_1)\nonumber\\
    &=& I(\Xvec_2;\Xvec_2+\Zvec_2|\Xvec_1) \nonumber\\
    &=& h(\Xvec_2+\Zvec_2|\Xvec_1) - h(\Xvec_2+\Zvec_2|\Xvec_1,\Xvec_2) \nonumber\\
    &\le& h(\Xvec_2+\Zvec_2) - h(\Zvec_2) \nonumber\\
		\label{eqnapx:equiv1b}
    &=& I(\Xvec_2;\Xvec_2+\Zvec_2).
\end{eqnarray}
In conclusion, from \eqref{eqnapx:equiv1a} and \eqref{eqnapx:equiv1b} we obtain that
\begin{eqnarray*}
    I(\Xvec_1;\Xvec_2+\Zvec_2|\Xvec_1+\Zvec_1) &\le& I(\Xvec_2;\Xvec_2+\Zvec_2)\\
    I(\Xvec_2;\Xvec+\Zvec|\Xvec_1) &\le& I(\Xvec_2;\Xvec_2+\Zvec_2).
\end{eqnarray*}
Finally, let $\Dmat_2$ denote a diagonal matrix with $\{a_{2,i}\}_{i=1}^m, |a_{2,i}|<\infty, \forall i, 1\le i\le m$ on its diagonal, where $\{a_{2,i}\}_{i=1}^m$ are defined in the statement of the lemma. The proof is then completed by noting that for a given input covariance matrix $\Kmat_{X_2}$, a circularly symmetric complex Normal $\Xvec_2\sim\CN(0,\Kmat_{X_2})$ maximizes $I(\Xvec_2;\Xvec_2+\Zvec_2)$ \cite[Theorem 2]{Massey:93}, hence
\begin{equation*}
    \lim_{\gamma_2\rightarrow\infty}I(\Xvec_2;\Xvec_2+\Zvec_2)\le \lim_{\gamma_2\rightarrow\infty} \Big(\log(|\gamma_2\cdot\Dmat_2+\Kmat_{X_2}|) - \log(|\gamma_2\cdot\Dmat_2|)\Big) = \lim_{\gamma_2\rightarrow\infty} \log(|\Imat+\gamma_2^{-1}\cdot\Dmat_2^{-1}\cdot\Kmat_{X_2}|)=0.
\end{equation*}
It thus follows that
\begin{equation*}
    \lim_{\gamma_2\rightarrow\infty}I(\Xvec_1;\Xvec_2+\Zvec_2|\Xvec_1+\Zvec_1)= \lim_{\gamma_2\rightarrow\infty}I(\Xvec_2;\Xvec+\Zvec|\Xvec_1)=0,
\end{equation*}
which results in the desired equality \eqref{eqnappB:desired}.
\tend

%%%%%%%%%%%%%%%%%%%%%%%%%%%%%%%%%%%%%%%%%%%%%%%%%%%%%%%%%%%%%%%%%%%%%%%%%%%%%%%%%%%%%%%%%%%%%%%%%%%%%%%%%%%%%%%%%%%%%%%%%%%%%%%%%%%%%%%

\section{Proof of Lemma \ref{lem:lemma8}}
\label{app:ProofOfLemma8}
We follow steps similar to those used in the proof of \cite[Corollary 6]{Liu:07}, generalizing the derivations to hold for {\em complex vectors}. Using the notation of \cite[Section III.A]{Massey:93}, define the $4n\times 2n$ real matrix $\Vmat_{kk}$ as
\begin{equation*}
\Vmat_{kk}\triangleq\Bigg[\begin{array}{cc}
\Real\{\tilde{\Vmat}_k\} & -\Img\{\tilde{\Vmat}_k\}\\
\Img\{\tilde{\Vmat}_k\}  & \Real\{\tilde{\Vmat}_k\}\end{array}\Bigg] \qquad k\in\{1,2\}.
\end{equation*}
Next, let $\Xvec=X^{2n}$, and define the $4n\times 1$ real-valued vector $\bar{\Xvec}\triangleq(\Xvec_R^T,\Xvec_I^T)^T$, where $\Xvec_R$ and $\Xvec_I$ are the real and the imaginary parts of $\Xvec$, respectively. Note that similarly to \cite[Section III-A]{Massey:93}, by using this notation we have $\Vmat^T_{kk}\bar{\Xvec}=\Big(\Real\{\tilde{\Vmat}_k^H\Xvec\}^T, \Img\{\tilde{\Vmat}_k^H\Xvec\}^T\Big)^T$; also note that this notation preserves the orthogonality among the columns of $\Vmat_{kk}$ s.t. $\Vmat_{kk}^T\Vmat_{kk}=\Dmat_{kk}^{-1}=\bigg[\begin{array}{cc}
\tilde{\Dmat}^{-1}_{k}& \Omat_{n\times n}\\\Omat_{n \times n}& \tilde{\Dmat}^{-1}_{k}\end{array}\bigg], k\in\{1,2\}$, where $\Dmat_{kk}^{-1}\in\Rset^{2n\times 2n}$. Therefore, we can find two $4n \times 2n$ matrices $\Vmat_{12}$ and $\Vmat_{21}$ s.t. $\Vmat_{1}\triangleq(\Vmat_{11},\Vmat_{12})$ and $\Vmat_{2}\triangleq(\Vmat_{21},\Vmat_{22})$ are two $4n \times 4n$  matrices with orthogonal columns and hence, they can be written as $\Vmat_k=\Umat_k\cdot\Cmat_k$, where $\Umat_k\triangleq(\Umat_{k1},\Umat_{k2}), k\in\{1,2\}$ is a $4n\times 4n$  matrix with orthonormal columns, $\Umat_{lk}, l,k\in\{1,2\}$ are four $4n\times 2n$ matrices with orthonormal columns,
and $\Cmat_k, k\in\{1,2\}$ is a $4n\times 4n$ diagonal matrix whose elements are given by:%with real positive elements, $\big\{[\Cmat_k]_{ii}\big\}_{i=1}^{4n},
\begin{equation}
\label{EqnAppxLem6:C_k_Def}
    \Cmat_k\triangleq\left[\begin{matrix}
    \sqrt{\sum_{j=1}^{4n} \big|[\Vmat_k]_{j,1}\big|^2} &  0  & \ldots & 0\\
    0  &  \sqrt{\sum_{j=1}^{4n} \big|[\Vmat_k]_{j,2}\big|^2} & \ldots & 0\\
    \vdots & \vdots & \ddots & \vdots\\
    0  &   0       &\ldots & \sqrt{\sum_{j=1}^{4n} \big|[\Vmat_k]_{j,4n}\big|^2}
    \end{matrix}\right].
\end{equation}
Let $\Cmat_{kk}$, $k\in\{1,2\}$ be two $2n\times 2n$ diagonal matrices whose elements are given by
$\big[\Cmat_{11}\big]_{i,i}=\big[\Cmat_{1}\big]_{i,i}$, and $\big[\Cmat_{22}\big]_{i,i}=\big[\Cmat_{2}\big]_{2n+i,2n+i}, i\in\{1,2,3,...,2n\}$.
With these assignments we can represent $\Cmat_k=\left[\begin{matrix} \Cmat_{1k} & \Omat_{2n \times 2n} \\ \Omat_{2n \times 2n} & \Cmat_{2k}\end{matrix}\right]$, where
$\Cmat_{12}$  and $\Cmat_{21}$ are defined to satisfy \eqref{EqnAppxLem6:C_k_Def}.
Using the above definitions we can write $\Vmat_{kk}^T=\Cmat_{kk}^T\cdot\Umat_{kk}^T, k\in\{1,2\}$.
Recall  the definition of $\Dmat_{kk}^{-1}\in \Rset_{++}^{2n\times 2n}$:
\begin{equation*}
  \Dmat_{kk}^{-1}=\Vmat_{kk}^T\Vmat_{kk}=\Cmat_{kk}^T\cdot\Umat_{kk}^T\cdot\Umat_{kk}\cdot\Cmat_{kk}=\Cmat_{kk}^T\Cmat_{kk}=\Cmat_{kk}\Cmat_{kk}=\Cmat_{kk}^2, \quad k\in\{1,2\}.
\end{equation*}
Thus, from \cite[Proposition 8.1.2 and Lemma 8.2.1]{MatrixMathematics} it follows that we can write $\Dmat_{kk}^{-\frac{1}{2}}=\Cmat_{kk}=\Cmat_{kk}^T, k\in\{1,2\}$. Using these definitions, we obtain
\begin{equation}
\label{eq:eq16}
\Vmat_{kk}^T=\Cmat_{kk}^T\cdot\Umat_{kk}^T=\Dmat_{kk}^{-\frac{1}{2}}\cdot\Umat_{kk}^T, \qquad k\in\{1,2\}.
\end{equation}
Next, let $\Dmat_{12},\Dmat_{21}, \bar{\Dmat}_{11}$ and $\bar{\Dmat}_{22}$ be four arbitrary $2n\times 2n$ dimensional diagonal matrices with real positive elements and define the following $4n\times 4n$ matrices:
\begin{subequations}
\label{eqnC:Def_Dmatrices}
\begin{eqnarray}
    & & \Dmat_1\triangleq\left[\begin{array}{cc}
    \Dmat_{11} & \Omat_{2n \times 2n}\\
    \Omat_{2n \times 2n}     & \Dmat_{12}\end{array}\right],\quad \Dmat_2\triangleq\left[\begin{array}{cc}
    \Dmat_{21} & \Omat_{2n \times 2n}\\
    \Omat_{2n \times 2n}     & \Dmat_{22}\end{array}\right],\\
    & & \bar{\Dmat}_1\triangleq\left[\begin{array}{cc}
    \bar{\Dmat}_{11} & \Omat_{2n \times 2n}\\
    \Omat_{2n \times 2n}     & \Dmat_{12}\end{array}\right],\quad \bar{\Dmat}_2\triangleq\left[\begin{array}{cc}
    \Dmat_{21} & \Omat_{2n \times 2n}\\
    \Omat_{2n \times 2n}     & \bar{\Dmat}_{22}\end{array}\right].
\end{eqnarray}
\end{subequations}
Here, $\Dmat_{lk}$ is a diagonal matrix with real positive elements defined by the vector
 $\bm{d}_{lk}=(d_{lk,1}, d_{lk,2},..., d_{lk,2n}), lk\in\{12,21,11,22\}$ on
its diagonal. Similarly, $\bar{\Dmat}_{kk}$ is a diagonal matrix with real positive elements defined by the vector $\bar{\bm{d}}_{kk}=(\bar{d}_{kk,1}, \bar{d}_{kk,2},..., \bar{d}_{kk,2n}), k\in\{1,2\}$ on its diagonal.

For $k\in\{1,2\}$ set $\Kmat_{\bar{Z}_k} \triangleq \Umat_k\cdot(\bar{\Dmat}_{k}\cdot\Dmat_{k})\cdot\Umat_k^T$,  let $\bar{\Zvec}_k$ be a $4n\times 1$ real-valued random vector distributed according to $\N\big(0,\Kmat_{\bar{Z}_k}\big)$, and consider the following optimization problem:
\begin{eqnarray}
    \label{eq:Lemma6CovOpt}
    &&\max_{f(\bar{\xvec})}\mbox{ } h(\bar{\Xvec}+\bar{\Zvec}_1)-h(\bar{\Xvec}+\bar{\Zvec}_2)\\
    &&\mbox{subject to: \phantom{x} } \cov(\bar{\Xvec})\preceq\mathds{S},\nonumber
\end{eqnarray}
where maximization is carried out over all $4n\times 1$ real vectors $\bar{\Xvec}$. Denote $\cov(\bar{\Xvec})\equiv\Kmat_{\bar{X}}$. From \cite[Theorem 1]{Liu:07} it directly follows that a Gaussian random vector $\bar{\Xvec}$ is an optimal solution to this problem. Furthermore, from \cite[Lemma 1]{Zahavi:12}, we obtain that for any pair of complex random vectors $\bar{\Xvec}$ and $\bar{\Xvec}^{{\rm zm}}\triangleq\bar{\Xvec}-\E\{\bar{\Xvec}\}$, it holds that $h(\bar{\Xvec})=h(\bar{\Xvec}^{{\rm zm}})$. Thus, we conclude that a {\em zero-mean} Gaussian random vector is the optimal solution to \eqref{eq:Lemma6CovOpt}. Hence,
\begin{equation}
    \label{eq:eq9}
    \max_{f(\bar{\xvec}): \quad\!\!\!\!\cov(\bar{\Xvec})\preceq\mathds{S}}\mbox{ } h(\bar{\Xvec}+\bar{\Zvec}_1)-h(\bar{\Xvec}+\bar{\Zvec}_2) = \max_{\Kmat_{\bar{X}}:\quad\!\!\!\! 0\preceq\Kmat_{\bar{X}}\preceq\mathds{S}} \Big\{\frac{1}{2}\log\big(|\Kmat_{\bar{X}}+\Kmat_{\bar{Z}_1}|\big)-\frac{1}{2}\log\big(|\Kmat_{\bar{X}}+\Kmat_{\bar{Z}_2}|\big)\Big\}.
\end{equation}
Adding $h(\bar{\Zvec}_2)-h(\bar{\Zvec}_1)$ to both sides of \eqref{eq:eq9} we obtain
\begin{equation}
    \label{eq:eq9C3}
    \max_{f(\bar{\xvec}): \quad\!\!\!\!\cov(\bar{\Xvec})\preceq\mathds{S}}\mbox{ } I(\bar{\Xvec};\bar{\Xvec}+\bar{\Zvec}_1)-I(\bar{\Xvec};\bar{\Xvec}+\bar{\Zvec}_2) = \max_{\Kmat_{\bar{X}}: \quad\!\!\!\!0\preceq\Kmat_{\bar{X}}\preceq\mathds{S}} \Big\{\frac{1}{2}\log\big(|\Imat_{4n}+\Kmat^{-1}_{\bar{Z}_1}\Kmat_{\bar{X}}|\big)-\frac{1}{2}\log\big(|\Imat_{4n}+\Kmat^{-1}_{\bar{Z}_2}\Kmat_{\bar{X}}|\big)\Big\}.
\end{equation}

Next, define $\Umat^T_k\bar{\Zvec}_k\triangleq \hat{\Zvec}_k \equiv (\hat{\Zvec}_{k1}^T,\hat{\Zvec}_{k2}^T)^T$, where $\hat{\Zvec}_{kl}$, $(k, l)\in\{1, 2\}\times\{1, 2\}$ are four $2n\times 1$ vectors. Note that $\Umat_k$ is an orthogonal matrix, i.e., $\Umat^T_k\Umat_k=\Imat_{4n}$, and thus the covariance matrix of $\hat{\Zvec}_k$ is given by
\begin{equation}
    \label{eq:hatZCovMat}
    \Kmat_{\hat{Z}_k}=\Umat^T_k\cdot\Kmat_{\bar{Z}_k}\cdot\Umat_k=\Umat^T_k\cdot\Umat_k\cdot(\bar{\Dmat}_{k}\cdot\Dmat_k)\cdot\Umat^T_k\cdot\Umat_k=(\bar{\Dmat}_{k}\cdot\Dmat_k), \qquad k\in\{1,2\}.
\end{equation}
Since $\hat{\Zvec}_{k}$, is a linear transformation of a Gaussian vector $\bar{\Zvec}_k$, it is a Gaussian random vector. Hence,
it follows from \eqref{eq:hatZCovMat} that $\hat{\Zvec}_{k1}$ and $\hat{\Zvec}_{k2}$ are mutually independent for $k\in\{1,2\}$. Next, note that for any real random vector $\Xvec$ and any real matrix $\Amat$ it holds that (see \cite[Eqn. (13)]{Massey:93}):
\begin{equation}
\label{eq:eq22}
h\big(\Amat\cdot\Xvec\big)= h\big(\Xvec\big)+\log\big|\det(\Amat)\big|.
\end{equation}
Hence, since  any orthogonal matrix is invertible, then we can write
\begin{subequations}
\label{EqnApx_Lem2:MI_equiv}
\begin{eqnarray}
    I(\bar{\Xvec};\bar{\Xvec}+\bar{\Zvec}_1)&=& I(\Umat^T_1\bar{\Xvec};\Umat^T_1\bar{\Xvec}+\Umat^T_1\bar{\Zvec}_1)\\
    I(\bar{\Xvec};\bar{\Xvec}+\bar{\Zvec}_2)&=& I(\Umat^T_2\bar{\Xvec};\Umat^T_2\bar{\Xvec}+\Umat^T_2\bar{\Zvec}_2).
\end{eqnarray}
\end{subequations}
Using Lemma \ref{lem:lemma6}, it follows that
\begin{eqnarray*}
    \lim_{\bm{d}_{12}\rightarrow\infty} I(\bar{\Xvec};\bar{\Xvec}+\bar{\Zvec}_1)&=&\lim_{\bm{d}_{12}\rightarrow\infty} I(\Umat^T_1\bar{\Xvec};\Umat^T_1\bar{\Xvec}+\Umat^T_1\bar{\Zvec}_1)=I(\Umat^T_{11}\bar{\Xvec};\Umat^T_{11}\bar{\Xvec}+\hat{\Zvec}_{11})\\
    \lim_{\bm{d}_{21}\rightarrow\infty} I(\bar{\Xvec};\bar{\Xvec}+\bar{\Zvec}_2)&=&\lim_{\bm{d}_{21}\rightarrow\infty} I(\Umat^T_2\bar{\Xvec};\Umat^T_2\bar{\Xvec}+\Umat^T_2\bar{\Zvec}_2)=I(\Umat^T_{22}\bar{\Xvec};\Umat^T_{22}\bar{\Xvec}+\hat{\Zvec}_{22}),
\end{eqnarray*}
where we use  $\bm{d}_{lk}\rightarrow\infty$ to imply that all the elements of $\bm{d}_{lk}$ go to infinity, i.e.,  $d_{lk,1}, d_{lk,2},..., d_{lk,2n}\rightarrow\infty$. Thus,
\begin{equation}
    \label{eq:eq7}
    \lim_{\substack{\bm{d}_{12}\rightarrow\infty\\ \bm{d}_{21}\rightarrow\infty}}\Big(I(\bar{\Xvec};\bar{\Xvec}+\bar{\Zvec}_1)-I(\bar{\Xvec};\bar{\Xvec}+\bar{\Zvec}_2)\Big)= I(\Umat^T_{11}\bar{\Xvec};\Umat^T_{11}\bar{\Xvec}+\hat{\Zvec}_{11})-I(\Umat^T_{22}\bar{\Xvec};\Umat^T_{22}\bar{\Xvec}+\hat{\Zvec}_{22}).
\end{equation}
Note that for an $n\times m$ matrix $\Amat$ and an $m\times n$ matrix $\Bmat$, Sylvester's determinant theorem \cite[Page 271]{Pozrikidis:14} states that $|\Imat_n+\Amat\Bmat|=|\Imat_m+\Bmat\Amat|$, and that given two square matrices $\Amat$ and $\Bmat$, it holds that $(\Amat\cdot\Bmat)^{-1}=\Bmat^{-1}\Amat^{-1}$, thus, due to the continuity of $\log(\Imat+\Amat)$ over the semidefinite $\Amat$ we obtain from \eqref{eq:eq9C3} (see, e.g., \cite[Eq. (164)]{Liu:07})
\begin{eqnarray}
&&\hspace{-0.9cm}\lim_{\substack{\bm{d}_{12}\rightarrow\infty\\\bm{d}_{21}\rightarrow\infty}} \Big(\log\big(|\Imat_{4n}+\Kmat^{-1}_{\bar{Z}_1}\Kmat_{\bar{X}}|\big)-\log\big(|\Imat_{4n}+\Kmat^{-1}_{\bar{Z}_2}\Kmat_{\bar{X}}|\big)\Big)\nonumber\\
%%%%%%%%%%%%%%%%%%%%%%%%%%%%%%%%%%%%%%%%%%%%%%%%%
&=&\lim_{\substack{\bm{d}_{12}\rightarrow\infty\\\bm{d}_{21}\rightarrow\infty}} \Big(\log\big(|\Imat_{4n}+\Umat_1\cdot (\bar{\Dmat}_{1}\cdot\Dmat_{1})^{-1}\cdot\Umat_1^T\cdot\Kmat_{\bar{X}}|\big)-\log\big(|\Imat_{4n}+\Umat_2\cdot (\bar{\Dmat}_{2}\cdot\Dmat_{2})^{-1}\cdot\Umat_2^T\cdot\Kmat_{\bar{X}}|\big)\Big)\nonumber\\
%%%%%%%%%%%%%%%%%%%%%%%%%%%%%%%%%%
&=&\log\left(\left|\Imat_{4n}+\Umat_1\cdot
\left[\begin{matrix}
(\bar{\Dmat}_{11}\cdot\Dmat_{11})^{-1}\;\; & \Omat_{2n \times 2n}\\
\Omat_{2n \times 2n} & \Omat_{2n \times 2n} \end{matrix}\right]
\cdot\Umat_1^T\cdot\Kmat_{\bar{X}}\right|\right)\nonumber\\
&& \qquad \qquad\qquad \qquad
-\log\left(\left|\Imat_{4n}
	+\Umat_2\cdot \left[\begin{matrix}
\Omat_{2n \times 2n} & \Omat_{2n \times 2n}\\
\Omat_{2n \times 2n} & \;\;(\bar{\Dmat}_{11}\cdot\Dmat_{11})^{-1}  \end{matrix}\right]
\cdot\Umat_2^T\cdot\Kmat_{\bar{X}}\right|\right)\nonumber\\
%%%%%%%%%%%%%%%%%%%%%%%%%%%%%%%%%%%%%%%%%%%%%%%%%%%
& =&\log\Big(\big|\Imat_{4n}+\Umat_{11}\big(\bar{\Dmat}_{11}\cdot\Dmat_{11}\big)^{-1}\Umat^T_{11}\Kmat_{\bar{X}}\big|\Big)-\log\Big(\big|\Imat_{4n}+\Umat_{22}\big(\bar{\Dmat}_{22}\cdot\Dmat_{22}\big)^{-1}\Umat^T_{22}\Kmat_{\bar{X}}\big|\Big)\nonumber\\
%%%%%%%%%%%%%%%%%%%%%%%%%%%%%%%%%%%
\label{eq:eq24}&=&
\log\Big(\big|\Imat_{2n}+\big(\bar{\Dmat}_{11}\cdot\Dmat_{11}\big)^{-1}\Umat^T_{11}\Kmat_{\bar{X}}\Umat_{11}\big|\Big)-\log\Big(\big|\Imat_{2n}+\big(\bar{\Dmat}_{22}\cdot\Dmat_{22}\big)^{-1}\Umat^T_{22}\Kmat_{\bar{X}}\Umat_{22}\big|\Big).
\end{eqnarray}
Moreover, the convergence of \eqref{eq:eq24} is uniform in $\Kmat_{\bar{X}}$, because the continuity of
$\log\big(|\Imat+\Amat|\big)$ over $\Amat$ is uniform, and $\Umat_k^T\cdot\Kmat_{\bar{X}}\cdot\Umat_k, k\in\{1,2\}$ is bounded for $0\preceq \Kmat_{\bar{X}}\preceq \mathds{S}$, in the sense that $0\preceq \Umat_k^T\cdot \Kmat_{\bar{X}}\cdot \Umat_k\preceq \Umat_k^T\cdot \mathds{S}\cdot \Umat_k$.
We thus have\footnote{For uniform convergence $\lim_{y\rightarrow 0} f_y(x) = f(x)$ (see {\tt http://www2.math.umd.edu/~czaja/chap1.pdf} ): For any $\epsilon <0$,  $\exists \delta > 0$ s.t.  $|f_y(x) - f(x)|<\epsilon,\;\; \forall y < \delta, \;\; \forall x\in \Rset$. Hence, $\exists
\delta>0$ s.t. $\max_{x \in \Rset} |f_y(x) -  f(x)|\le \epsilon, \;\; \forall y<\delta$. Thus, $\exists \delta>0$, s.t. $\forall y<\delta$,
 $|\max_{x \in \Rset} f_y(x) - \max_{x\in \Rset} f(x)| \le \max_{x \in \Rset} |f_y(x) -  f(x)| \le\epsilon$. Thus,
$\lim_{y\rightarrow 0} \max_{x \in \Rset} f_y(x) = \max_{x\in \Rset} f(x)$.} (see e.g. \cite[Eq. (165)]{Liu:07})
%{EqnApx_Lem2:MI_equiv}
\begin{eqnarray}
&&\hspace{-0.9cm}\lim_{\substack{\bm{d}_{12}\rightarrow\infty\\\bm{d}_{21}\rightarrow\infty}} \Bigg( \max_{\Kmat_{\bar{X}}:\;0\preceq\Kmat_{\bar{X}}\preceq\mathds{S}} \bigg\{\log\Big(\big|\Imat_{4n}+\Kmat^{-1}_{\bar{Z}_1}\Kmat_{\bar{X}}\big|\Big)-\log\Big(\big|\Imat_{4n}+\Kmat^{-1}_{\bar{Z}_2}\Kmat_{\bar{X}}\big|\Big)\bigg\} \Bigg) =\nonumber\\
\label{eq:eq8}&&\hspace{-0.3cm}\max_{\Kmat_{\bar{X}}:\; 0\preceq\Kmat_{\bar{X}}\preceq\mathds{S}}
\bigg\{\log\Big(\big|\Imat_{2n}+\big(\bar{\Dmat}_{11}\cdot\Dmat_{11}\big)^{-1}\Umat^T_{11}\Kmat_{\bar{X}}\Umat_{11}\big|\Big)
        -\log\Big(\big|\Imat_{2n}+\big(\bar{\Dmat}_{22}\cdot\Dmat_{22}\big)^{-1}\Umat^T_{22}\Kmat_{\bar{X}}\Umat_{22}\big|\Big)\bigg\}.
\end{eqnarray}
Now, %taking the limits ${\bm d}_{12}\rightarrow\infty$ and ${\bm d}_{12}\rightarrow\infty$ in \eqref{eq:eq8}, then,
using \eqref{eq:eq9C3}, \eqref{eq:eq7} %\eqref{eq:eq9},
%\eqref{EqnApx_Lem2:MI_equiv}
and \eqref{eq:eq8}   we obtain
\begin{eqnarray*}
&&\max_{f(\bar{\xvec}): \quad\!\!\!\!\cov(\bar{\Xvec})\preceq\mathds{S}} I(\Umat^T_{11}\bar{\Xvec};\Umat^T_{11}\bar{\Xvec}+\hat{\Zvec}_{11})-I(\Umat^T_{22}\bar{\Xvec};\Umat^T_{22}\bar{\Xvec}+\hat{\Zvec}_{22})
=\\
&&\qquad\qquad\max_{\Kmat_{\bar{X}}: \quad\!\!\!\!0\preceq\Kmat_{\bar{X}}\preceq\mathds{S}}
\bigg\{\frac{1}{2}\log\Big(\big|\Imat_{2n}+\big(\bar{\Dmat}_{11}\cdot\Dmat_{11}\big)^{-1}\Umat^T_{11}\Kmat_{\bar{X}}\Umat_{11}\big|\Big)
-\frac{1}{2}\log\Big(\big|\Imat_{2n}+\big(\bar{\Dmat}_{22}\cdot\Dmat_{22}\big)^{-1}\Umat^T_{22}\Kmat_{\bar{X}}\Umat_{22}\big|\Big)\bigg\},
\end{eqnarray*}
or equivalently,
\begin{eqnarray*}
&&\max_{f(\bar{\xvec}): \quad\!\!\!\!\cov(\bar{\Xvec})\preceq\mathds{S}} h(\Umat^T_{11}\bar{\Xvec}+\hat{\Zvec}_{11})-h(\Umat^T_{22}\bar{\Xvec}+\hat{\Zvec}_{22})
=\\
&&\qquad\qquad\max_{\Kmat_{\bar{X}}: \quad\!\!\!\!0\preceq\Kmat_{\bar{X}}\preceq\mathds{S}}
\bigg\{\frac{1}{2}\log\Big((\pi e)^n\big|\bar{\Dmat}_{11}\cdot\Dmat_{11}+\Umat^T_{11}\Kmat_{\bar{X}}\Umat_{11}\big|\Big)-\frac{1}{2}\log\Big((\pi e)^n\big|\bar{\Dmat}_{22}\cdot\Dmat_{22}+\Umat^T_{22}\Kmat_{\bar{X}}\Umat_{22}\big|\Big)\bigg\}.
\end{eqnarray*}
Recall that $\Dmat_{kk}, k=1,2$ is a diagonal matrix with real and positive entries. Thus, from \cite[Proposition 8.1.2 and Lemma 8.2.1]{MatrixMathematics} we have that $\Dmat_{kk}^{-1}, k\in\{1,2\}$ can be written as $\Dmat_{kk}^{-1}=\Amat^2$, where $\Amat$ is a p.d. matrix defined as $\Amat\triangleq\Dmat_{kk}^{-\frac{1}{2}}$. Using \eqref{eq:eq22} once more, we obtain for $k = 1,2$
\begin{equation*}
h(\Umat^T_{kk}\bar{\Xvec}+\hat{\Zvec}_{kk})=h(\Dmat_{kk}^{-\frac{1}{2}}\Umat^T_{kk}\bar{\Xvec}+\Dmat_{kk}^{-\frac{1}{2}}\hat{\Zvec}_{kk})-\log(|\Dmat_{kk}^{-\frac{1}{2}}|).
\end{equation*}
Since both ${\Dmat}_{kk}$ and $\bar{\Dmat}_{kk}$ are diagonal matrices then, $\Dmat^{-\frac{1}{2}}_{kk}\cdot\bar{\Dmat}_{kk}\cdot\Dmat_{kk}\cdot\Dmat^{-\frac{1}{2}}_{kk}=\bar{\Dmat}_{kk}\cdot\Dmat^{-\frac{1}{2}}_{kk}\cdot\Dmat_{kk}\cdot\Dmat^{-\frac{1}{2}}_{kk}=\bar{\Dmat}_{kk}$, and hence, we obtain that
\begin{eqnarray*}
&&\hspace{-0.3cm}\max_{f(\bar{\xvec}): \quad\!\!\!\!\cov(\bar{\Xvec})\preceq\mathds{S}} h(\Dmat_{11}^{-\frac{1}{2}}\Umat^T_{11}\bar{\Xvec}+\Dmat_{11}^{-\frac{1}{2}}\hat{\Zvec}_{11})-h(\Dmat_{22}^{-\frac{1}{2}}\Umat^T_{22}\bar{\Xvec}+\Dmat_{22}^{-\frac{1}{2}}\hat{\Zvec}_{22})=\\
&&\quad\max_{\Kmat_{\bar{X}}: \quad\!\!\!\!0\preceq\Kmat_{\bar{X}}\preceq\mathds{S}}
\bigg\{\frac{1}{2}\log\Big((\pi e)^n\big|\bar{\Dmat}_{11}+\Dmat_{11}^{-\frac{1}{2}}\Umat^T_{11}\Kmat_{\bar{X}}\Umat_{11}\Dmat_{11}^{-\frac{1}{2}}\big|\Big)-\frac{1}{2}\log\Big((\pi e)^n\big|\bar{\Dmat}_{22}+\Dmat_{22}^{-\frac{1}{2}}\Umat^T_{22}\Kmat_{\bar{X}}\Umat_{22}\Dmat_{22}^{-\frac{1}{2}}\big|\Big)\bigg\}.
\end{eqnarray*}
Recall that from \eqref{eq:eq16} we have $\Vmat^T_{ll}=\Dmat_{ll}^{-\frac{1}{2}}\Umat^T_{ll}$, and define $\Zvec_{ll}\triangleq\Dmat_{ll}^{-\frac{1}{2}}\hat{\Zvec}_{ll}, l\in\{1,2\}$. Additionally, from \eqref{eq:hatZCovMat} we conclude that $\hat{\Zvec}_{ll}\sim\N({\bm 0},\bar{\Dmat}_{ll}\cdot\Dmat_{ll}), l\in\{1,2\}$. Hence, since $\bar{\Dmat}_{ll}$ and $\Dmat_{ll}$ are positive and real diagonal matrices, then $\Zvec_{ll}$ is distributed according to $\Zvec_{ll}\sim\N({\bm 0},\bar{\Dmat}_{ll})$. Thus,
\begin{eqnarray}
&&\hspace{-1cm}\max_{f(\bar{\xvec}): \quad\!\!\!\!\cov(\bar{\Xvec})\preceq\mathds{S}} h(\Vmat^T_{11}\bar{\Xvec}+\Zvec_{11})-h(\Vmat^T_{22}\bar{\Xvec}+\Zvec_{22})=\nonumber\\
&&\qquad\max_{\Kmat_{\bar{X}}: \quad\!\!\!\!0\preceq\Kmat_{\bar{X}}\preceq\mathds{S}}
\bigg\{\frac{1}{2}\log\Big((\pi e)^n\big|\bar{\Dmat}_{11}+\Vmat^T_{11}\Kmat_{\bar{X}}\Vmat_{11}\big|\Big)-\frac{1}{2}\log\Big((\pi e)^n\big|\bar{\Dmat}_{22}+\Vmat^T_{22}\Kmat_{\bar{X}}\Vmat_{22}\big|\Big)\bigg\}.      \label{eq:eq21}
\end{eqnarray}
The proof of Lemma \ref{lem:lemma8} is completed by recalling that the elements of $\bar{\Dmat}_{kk}, k\in\{1,2\}$ are chosen arbitrarily and thus, \eqref{eq:eq21} holds for any p.d. $\bar{\Dmat}_{kk}$,
and in particular for
$\bar{\Dmat}_{kk} \triangleq \left[\begin{array}{cc} \frac{1}{2}\tilde{\Dmat}^{\fnsz}_{k} & \;\Omat_{n \times n}\; \\ \;\Omat_{n \times n}\; & \frac{1}{2}\tilde{\Dmat}^{\fnsz}_{k}\end{array}\right], \;\; k\in\{1,2\}$,
and by noting that a Gaussian random vector $\bar{\Xvec}$ achieves the r.h.s. of \eqref{eq:eq21} with equality, i.e., a zero-mean complex Normal $\Xvec$ is an optimal solution to \eqref{eq:eq6}. This completes the proof of Lemma \ref{lem:lemma8}.
\tend

\phantomsection
\label{eqn:endC}

%%%%%%%%%%%%%%%%%%%%%%%%%%%%%%%%%%%%%%%%%%%%%%%%%%%%%%%%%%%%%%%%%%%%%%%%%%%%%%%%%%%%%%%%%%%%%%%%%%%%%%%%%%%%%%%%%%%%%%%%%%%%%%%%%%%

\section{Proof of Proposition \ref{prop:proposition1}}
\label{app:appA}
First note that from the construction of the genie signals, it follows that the entropy expressions $h(S_{1G}|X_{1G},X_{3G},\tH_1)$ and $h(S_{2G}|X_{2G},\tH_2)$ do not depend on $\CORR$ and on $(P_1,P_2,P_3)$. Next, recall that $(X_{1G},X_{2G},X_{3G})^T\sim\CN({\bf 0},\Qmat_G)$ and consider $h(Y_{2G}|S_{2G},\tH_2)$: Defining $\theta_{\eta_2^*\CORRN_2^*}\triangleq\arg\{\eta_2^*\CORRN_2^*\}$ we can write:
\begin{eqnarray*}
  &&\hspace{-1.2 cm} h(Y_{2G}|S_{2G},\tH_2)-\log(\pi e)\\
  &\stackrel{(a)}{=}&\E_{\tH_2}\bigg\{\log\Big(\var(Y_{2G})-\frac{|\E\{Y_{2G}S_{2G}^*\}|^2}{\var(S_{2G})}\Big)\bigg\}\\
  &=&\E_{\tH_2}\bigg\{\log\Big(1+\SNRA_{22}P_2+\SNRA_{32}P_3-\frac{|\SNRA_{22}P_2+\eta_2^*\CORRN_2^*|^2}{\SNRA_{22}P_2+|\eta_2|^2}\Big)\bigg\}\\
  &=&\E_{\tH_2}\bigg\{\log\Big(1+\SNRA_{22}P_2+\SNRA_{32}P_3-\frac{\SNRA_{22}^2P_2^2+|\eta_2^*|^2|\CORRN_2^*|^2+2|\eta_2^*||\CORRN_2^*|\cos(\theta_{\eta_2^*\CORRN_2^*})\SNR_{22}P_2}{\SNRA_{22}P_2+|\eta_2|^2}\Big)\bigg\}\\
  &\triangleq&\E_{\tH_2}\Big\{\log\Big(f_1(P_2)\Big)\Big\},
\end{eqnarray*}
where (a) follows from a direct calculation via $h(Y_{2G}|S_{2G},\tH_2) = h(Y_{2G},S_{2G}|\tH_2)-h(S_{2G}|\tH_2)$, followed by applying \cite[Theorem 23.7.4]{lapidoth:book}. Observe that $f_1(P_2)$ is independent of $(\CORR, P_1)$, and that it increases with respect to $P_3$. Additionally, since $(Y_{2G}, S_{2G})$ are jointly circularly symmetric complex Normal when $\tH_2=\th_2$ is given, we note that
\begin{eqnarray*}
  \frac{\partial f_1(P_2)}{\partial P_2} &=& \frac{\big(\SNRA_{22}|\eta_2|^2-2|\eta_2^*||\CORRN_2^*|\cos(\theta_{\eta_2^*\CORRN_2^*})\SNRA_{22}\big)\cdot\big(\SNRA_{22}P_2+|\eta_2|^2\big)}{\big(\SNRA_{22}P_2+|\eta_2|^2\big)^2}\\
  &&\qquad-\frac{\big(\SNRA_{22}P_2|\eta_2|^2-|\eta_2^*|^2|\CORRN_2^*|^2-2|\eta_2^*||\CORRN_2^*|\cos(\theta_{\eta_2^*\CORRN_2^*})\SNRA_{22}P_2\big)\cdot\SNRA_{22}}{\big(\SNRA_{22}P_2+|\eta_2|^2\big)^2}\\
	%%%%%%%%%%%%%%%%%%%%%%%%%%%%%%%%%%%%%%%%%%%%%%%%%
  &=& \frac{\SNRA_{22}\cdot \Big(|\eta_2|^4+|\eta^*_2|^2|\CORRN_2^*|^2-2|\eta_2|^2|\eta_2^*||\CORRN_2^*|\cos(\theta_{\eta_2^*\CORRN_2^*})\Big)}{\big(\SNRA_{22}P_2+|\eta_2|^2\big)^2}\\
	%%%%%%%%%%%%%%%%%%%%%%%%%%%%%%%%%%%%%%%%%%%%%%%%%%%
  &\ge&\frac{\SNRA_{22}\Big(|\eta_2^*|^2-|\eta_2^*||\CORRN_2^*|\Big)^2}{{\big(\SNRA_{22}P_2+|\eta_2|^2\big)^2}}\\
  &\ge& 0.
\end{eqnarray*}
We therefore conclude that $h(Y_{2G}|S_{2G},\tH_2)$ is maximized with $P_2=P_3=1$, and that setting $\CORR=0$ and $P_1=1$ does not affect the value of $h(Y_{2G}|S_{2G},\tH_2)$.

\phantomsection
\label{phn:derviationdiffP2}

Next, consider $h(Y_{1G}|S_{1G},\tH_1)$, and define
\begin{subequations}
\label{eq:eqcs}
\begin{eqnarray}
  c_1 &\triangleq& \SNRA_{11}P_1+ \SNRA_{31}P_3\\
  c_2 &\triangleq& 2\sqrt{\SNRA_{11}P_1\SNRA_{31}P_3} \\
  c_3 &\triangleq& 1+\SNRA_{21}P_2 \\
  c_4 &\triangleq& |\eta_1^*|^2 \\
  c_5 &\triangleq& |\eta_1^*||\CORRN_1^*|\\
  \theta_1 &\triangleq& \arg(h_{11}h_{31}^*v)\\
  \theta_2 &\triangleq& \arg(\eta_1^*\CORRN_1^*).
\end{eqnarray}
\end{subequations}
With these definitions we can write
\begin{eqnarray*}
  &&\hspace{-0.7 cm}h(Y_{1G}|S_{1G},\tH_1)-\log(\pi e)\\
  &=&\E_{\tH_1}\bigg\{\log\bigg(\var(Y_{1G})-\frac{|\E\{Y_{1G}S_{1G}^*\}|^2}{\var(S_{1G})}\bigg)\bigg\}\\
  &=&\frac{1}{2\pi}\int_{\theta_1=0}^{2\pi}\!\!\left(\log\bigg(2\pi\cdot\Big( c_3+\frac{c_1c_4+c_2c_4\cos(\theta_1)|v|-c_5^2-(c_1+c_2\cos(\theta_1)|v|)\cdot2c_5\cos(\theta_2)}{c_1+c_2\cos(\theta_1)|v|+c_4}\Big)\bigg)\!-\log(2\pi)\right)\mbox{d}\theta_1\\
  &\stackrel{(a)}{=}&\log\bigg(2\pi\cdot\frac{c_3c_4 + c_1(c_3 + c_4) - c_5^2 - c_2 c_3 |v| - c_2 c_4 |v| -
 2 c_5 (c_1 - c_2 |v|)\cos(\theta_2)}{c_1 + c_4 - c_2 |v|}\bigg)-\log(2\pi)\\
  &\triangleq&\log\big( f_2(|v|) \big),
\end{eqnarray*}
where (a) follows from explicit analytical calculation\footnote{\label{fnt:techrepref} See details in the Appendix on Pgs. \number\numexpr\getpagerefnumber{pginc_start}+1\relax--\number\numexpr\getpagerefnumber{pginc_end}-1\relax.}.
%the technical report \cite{DaboraZahabi:2016}.}.
Next, differentiating $f_2(|v|)$ with respect to $|v|$ we obtain
\begin{equation*}
  \frac{\partial f_2(|v|)}{\partial |v|} = -\frac{ c_2\cdot \Big(c_4^2 + c_5^2 - 2 c_4 c_5 \cos(\theta_2)\Big)}{(c_1 + c_4 - c_2 |v|)^2},
\end{equation*}
and we note that since both $c_4$ and $c_5$ are non-negative real numbers, then $0\le(c_4-c_5)^2\le \Big(c_4^2 + c_5^2 - 2 c_4 c_5 \cos(\theta_2)\Big)$. Thus, as $c_2$ is positive, then the derivative of $f_2(|v|)$ with respect to $|v|$ is non-positive, i.e., $f_2(|v|)$ is a non-increasing function of $|v|$, and hence, it is maximized at $|v|=0$. Next, setting $|v|=0$ in $f_2(|v|)$ we obtain:
\begin{equation*}
  f_2(|v|)\Big|_{|v|=0} = \frac{c_3c_4 + c_1(c_3 + c_4) - c_5^2  - 2 c_5 c_1\cos(\theta_2)}{c_1 + c_4}\triangleq f_3(P_1,P_2,P_3).
\end{equation*}
Note that $f_3(P_1,P_2,P_3)$ is a monotonically increasing function of $c_3$ and  is independent of $c_2$. Additionally, note that since both $c_4$ and $c_5$ are nonnegative real numbers, then
\begin{equation*}
  \frac{\partial f_3(P_1,P_2,P_3)}{\partial c_1} = \frac{c_4^2+c_5^2-2c_4c_5\cos(\theta_2)}{(c_1 + c_4)^2} \ge 0.
\end{equation*}
We therefore conclude that $f_3(P_1,P_2,P_3)$ is a non-decreasing function of $c_1$. From the definition of $(c_1,c_2,c_3)$ in \eqref{eq:eqcs}, we conclude that $f_3(P_1,P_2,P_3)$ is non-decreasing with respect to $P_1, P_2$ and $P_3$. In conclusion, $h(Y_{1G}|S_{1G},\tH_1)$ is maximized with $v=0$ and $P_1=P_2=P_3=1$. This completes the proof of Proposition \ref{prop:proposition1}.
\tend

%%%%%%%%%%%%%%%%%%%%%%%%%%%%%%%%%%%%%%%%%%%%%%%%%%%%%%%%%%%%%%%%%%%%%%%%%%%%%%%%%%%%%%%%%%%%%%%%%%%%%%%%%%%%%%%%%%%%%%%

\section{Proving the Maximizing Distribution for Equation \eqref{eq:Leg33ofI} in Theorem \ref{thm:SRCA-UB} is i.i.d.}
\label{app:appB}

The objective of this appendix is to derive the maximum value and  characterize the associated maximizing distribution, for the expectation:
\begin{eqnarray}
    \label{eq:NewAppEeq0}
    &&{\E}_{{\tH}^n_1,{\tH}^n_2}\left\{h\bigg({\Hmat}^{(n)}_{h_{31}}X^n_{3\oG}+{\Hmat}^{(n)}_{h_{11}}X^n_{1\oG}+{\eta }_1\cdot W^n_1\big|{\tH}^n_1= {\th}^n_1\right)
-h\left({\Hmat}^{(n)}_{h_{32}}X^n_{3\oG}+V^n_2\big|{\tH}^n_2={\th}^n_2\right)\bigg\},
\end{eqnarray}
where $W^n_1\sim \mathcal{C}\mathcal{N}\left({\bm 0},{\Imat}_n\right)$, and $V^n_2\sim \mathcal{C}\mathcal{N}\left({\bm 0},\big(1-{\left|v_2\right|}^2\big){\Imat}_n\right)$.

\noindent
Let ${\hmat}^{(n)}_{h_{kl}}$ denote a realization of ${\Hmat}^{(n)}_{h_{kl}}$, where $(k,l)\in\big\{(1,1),(3,1),(3,2)\big\}$.
Now, consider
\[
    h\left({\hmat}^{(n)}_{h_{31}}X^n_{3\oG}+{\hmat}^{(n)}_{h_{11}}X^n_{1\oG}+{\eta }_1\cdot W^n_1\big|{\tH}^n_1={\th}^n_1\right),
\]
and define
the $2n \times 2n$ real matrices
\begin{subequations}
  \begin{eqnarray}
    \label{eqn:appE_def_h31}
    {\hmat}_{31}&\triangleq& \left[ \begin{array}{cc}
    \RealS\left\{{\hmat}^{(n)}_{h_{31}}\right\} &\;\;\; -\ImagS\left\{{\hmat}^{(n)}_{h_{31}}\right\} \\
    \ImagS\left\{{\hmat}^{(n)}_{h_{31}}\right\} &\;\;\; \RealS\left\{{\hmat}^{(n)}_{h_{31}}\right\} \end{array}
    \right],\\
    \label{eqn:appE_def_h11}
     {\hmat}_{11}&\triangleq& \left[ \begin{array}{cc}
    \RealS\left\{{\hmat}^{(n)}_{h_{11}}\right\} &\;\;\; -\ImagS\left\{{\hmat}^{(n)}_{h_{11}}\right\} \\
    \ImagS\left\{{\hmat}^{(n)}_{h_{11}}\right\} &\;\;\; \RealS\left\{{\hmat}^{(n)}_{h_{11}}\right\} \end{array}
    \right],\\
    {\oX}^{2n}_{k\oG}&=&\left(\left(\RealS{\left\{X^n_{k\oG}\right\}}\right)^T,\left(\ImagS{\left\{X^n_{k\oG}\right\}}\right)^T\right)^T,\ \ k\in \left\{1,3\right\}. \nonumber
\end{eqnarray}
\end{subequations}

\noindent
Since $X^n_{1\oG}$ and $X^n_{3\oG}$ are two zero-mean  complex jointly Gaussian random vectors, then $\left({\oX}^{2n}_{1\oG},\ {\oX}^{2n}_{3\oG}\right)$ are zero mean real jointly  Gaussian vectors. Note that since ${\hmat}^{(n)}_{h_{31}}$ is a diagonal matrix with  ${\left[{\hmat}^{(n)}_{h_{31}}\right]}_{i,i}=h_{31,i}$, then $\RealS\left\{{\hmat}^{(n)}_{h_{31}}\right\}$ and $\ImagS\left\{{\hmat}^{(n)}_{h_{31}}\right\}$ are  diagonal
\ifextended
matrices, and thus $\RealS\left\{{\hmat}^{(n)}_{h_{31}}\right\}={\left(\RealS\left\{{\hmat}^{(n)}_{h_{31}}\right\}\right)}^T$, $\ImagS\left\{{\hmat}^{(n)}_{h_{31}}\right\}={\left(\ImagS\left\{{\hmat}^{(n)}_{h_{31}}\right\}\right)}^T$, and $\RealS\left\{{\hmat}^{(n)}_{h_{31}}\right\}\ImagS\left\{{\hmat}^{(n)}_{h_{31}}\right\}=\ImagS\left\{{\hmat}^{(n)}_{h_{31}}\right\}\RealS\left\{{\hmat}^{(n)}_{h_{31}}\right\}$.
\else
matrices.
\fi
Next, define ${\theta }_{31,i}\triangleq {\arg \left\{h_{31,i}\right\}\ }$, and write $h_{31,i}={\sqrt{\SNR}}_{31}\cdot e^{j{\theta }_{31,i}}$. Additionally,
for $k\in \left\{1,3\right\}$ we define the $n\times n$ real diagonal matrices ${\Cmat}_{k1}$ and ${\Smat}_{k1}$ whose diagonal elements are given by ${\left[{\Cmat}_{k1}\right]}_{i,i}={\cos \left({\theta }_{k1,i}\right)\ },\ {\left[{\Smat}_{k1}\right]}_{i,i}={\sin \left({\theta }_{k1,i}\right)\ }$. With these definitions we write
\[\RealS\left\{{\hmat}^{(n)}_{h_{k1}}\right\}=\sqrt{\SNR_{k1}}\cdot {\Cmat}_{k1},\ \ \ \ \ \ImagS\left\{{\hmat}^{(n)}_{h_{k1}}\right\}\ =\sqrt{\SNR_{k1}}\cdot {\Smat}_{k1}.\]
Lastly, we define the $n\times n$ real diagonal matrix ${\Lmat}_{31}$, with non-negative elements via
\begin{equation*}
    {\Lmat}_{31}\triangleq {\hmat}^{(n)}_{h_{31}}\cdot \left({\hmat}^{(n)}_{h_{31}}\right)^H= {\left(\RealS\left\{{\hmat}^{(n)}_{h_{31}}\right\}\right)}^2+{\left(\ImagS\left\{{\hmat}^{(n)}_{h_{31}}\right\}\right)}^2=\SNR_{31}{\left({\Cmat}_{31}\right)}^2+\SNR_{31}{\left({\Smat}_{31}\right)}^2=\SNR_{31}{\Imat}_n.
\end{equation*}

Next, recall that ${\hmat}_{31}$ is the realization of ${\Hmat}_{31}$ corresponding to ${\Hmat}^{(n)}_{h_{31}}$. With these definitions \ifextended we have:
\begin{eqnarray*}
    {\hmat}^T_{31}{\hmat}_{31}&=&{\left[ \begin{array}{cc}
    \RealS\left\{{\hmat}^{(n)}_{h_{31}}\right\} & -\ImagS\left\{{\hmat}^{(n)}_{h_{31}}\right\} \\
    \ImagS\left\{{\hmat}^{(n)}_{h_{31}}\right\} & \RealS\left\{{\hmat}^{(n)}_{h_{31}}\right\} \end{array}
    \right]}^T\left[ \begin{array}{cc}
    \RealS\left\{{\hmat}^{(n)}_{h_{31}}\right\} & -\ImagS\left\{{\hmat}^{(n)}_{h_{31}}\right\} \\
    \ImagS\left\{{\hmat}^{(n)}_{h_{31}}\right\} & \RealS\left\{{\hmat}^{(n)}_{h_{31}}\right\} \end{array}
    \right]\\
    &=&\left[ \begin{array}{cc}
    {\left(\RealS\left\{{\hmat}^{(n)}_{h_{31}}\right\}\right)}^T & {\left(\ImagS\left\{{\hmat}^{(n)}_{h_{31}}\right\}\right)}^T \\
    -{\left(\ImagS\left\{{\hmat}^{(n)}_{h_{31}}\right\}\right)}^T & {\left(\RealS\left\{{\hmat}^{(n)}_{h_{31}}\right\}\right)}^T \end{array}
    \right]\left[ \begin{array}{cc}
    \RealS\left\{{\hmat}^{(n)}_{h_{31}}\right\} & -\ImagS\left\{{\hmat}^{(n)}_{h_{31}}\right\} \\
    \ImagS\left\{{\hmat}^{(n)}_{h_{31}}\right\} & \RealS\left\{{\hmat}^{(n)}_{h_{31}}\right\} \end{array}
    \right]\\
    &\stackrel{(a)}{=}&\left[ \begin{array}{cc}
    \RealS\left\{{\hmat}^{(n)}_{h_{31}}\right\} & \ImagS\left\{{\hmat}^{(n)}_{h_{31}}\right\} \\
    -\ImagS\left\{{\hmat}^{(n)}_{h_{31}}\right\} & \RealS\left\{{\hmat}^{(n)}_{h_{31}}\right\} \end{array}
    \right]\left[ \begin{array}{cc}
    \RealS\left\{{\hmat}^{(n)}_{h_{31}}\right\} & -\ImagS\left\{{\hmat}^{(n)}_{h_{31}}\right\} \\
    \ImagS\left\{{\hmat}^{(n)}_{h_{31}}\right\} & \RealS\left\{{\hmat}^{(n)}_{h_{31}}\right\} \end{array}
    \right]\\
    &=&\left[ \begin{array}{cc}
    {\left(\RealS\left\{{\hmat}^{(n)}_{h_{31}}\right\}\right)}^2+{\left(\ImagS\left\{{\hmat}^{(n)}_{h_{31}}\right\}\right)}^2 & {\Omat}_{n\times n} \\
    {\Omat}_{n\times n} & {\left(\RealS\left\{{\hmat}^{(n)}_{h_{31}}\right\}\right)}^2+{\left(\ImagS\left\{{\hmat}^{(n)}_{h_{31}}\right\}\right)}^2 \end{array}
    \right]\\
    &\equiv& \left[ \begin{array}{cc}
    {\Lmat}_{31} & {\Omat}_{n\times n} \\
    {\Omat}_{n\times n}& {\Lmat}_{31} \end{array}
    \right]\\
    &=&\SNR_{31}\cdot {\Imat}_{2n},
\end{eqnarray*}
where (a) follows as ${\hmat}^{(n)}_{h_{31}}$ is a diagonal matrix. It\else it \fi follows that
\ifextended
\begin{equation*}
    \left[ \begin{array}{cc}
    {\left({\Lmat}_{31}\right)}^{-\frac{1}{2}} & {\Omat}_{n\times n} \\
    {\Omat}_{n\times n} & {\left({\Lmat}_{31}\right)}^{-\frac{1}{2}} \end{array}
    \right]{\hmat}^T_{31}{\hmat}_{31}\left[ \begin{array}{cc}
    {\left({\Lmat}_{31}\right)}^{-\frac{1}{2}} & {\Omat}_{n\times n} \\
    {\Omat}_{n\times n} & {\left({\Lmat}_{31}\right)}^{-\frac{1}{2}} \end{array}
    \right]={\left(\SNR_{31}\right)}^{-\frac{1}{2}}\cdot {\Imat}^T_{2n}\cdot {\hmat}^T_{31}{\hmat}_{31}{{\cdot \Imat}_{2n}\cdot \left(\SNR_{31}\right)}^{-\frac{1}{2}}={\Imat}_{2n}.
\end{equation*}
\else
$  \left[ \begin{array}{cc}
    {\left({\Lmat}_{31}\right)}^{-\frac{1}{2}} & {\Omat}_{n\times n} \\
    {\Omat}_{n\times n} & {\left({\Lmat}_{31}\right)}^{-\frac{1}{2}} \end{array}
    \right]{\hmat}^T_{31}{\hmat}_{31}\left[ \begin{array}{cc}
    {\left({\Lmat}_{31}\right)}^{-\frac{1}{2}} & {\Omat}_{n\times n} \\
    {\Omat}_{n\times n} & {\left({\Lmat}_{31}\right)}^{-\frac{1}{2}} \end{array}
    \right]={\Imat}_{2n}$.
\fi
\ifextended
Letting ${\Umat}_{31}\triangleq\frac{1}{\sqrt{\SNR_{31}}}{\hmat}_{31}$, the above equality can be written as ${\Umat}^T_{31}{\Umat}_{31}={\Imat}_{2n}$, hence, we conclude that  ${\Umat}_{31}$ is a $2n\times 2n$ orthogonal matrix, and that the matrix ${\hmat}_{31}$ can now be written as ${\hmat}_{31}=\sqrt{\SNR_{31}}\cdot {\Umat}_{31}$, where
\begin{equation*}
    {\Umat}_{31}=\left[ \begin{array}{cc}
    {\Cmat}_{31} & {-\Smat}_{31} \\
    {\Smat}_{31} & {\Cmat}_{31} \end{array}
    \right].
\end{equation*}
Similarly we define ${\Umat}_{11}=\left[ \begin{array}{cc}
{\Cmat}_{11} & {-\Smat}_{11} \\
{\Smat}_{11} & {\Cmat}_{11} \end{array}
\right]$.
\else
Letting ${\Umat}_{31}\triangleq\frac{1}{\sqrt{\SNR_{31}}}{\hmat}_{31} = \left[ \begin{array}{cc}
    {\Cmat}_{31} & {-\Smat}_{31} \\
    {\Smat}_{31} & {\Cmat}_{31} \end{array}
    \right]$, the above equality can be written as ${\Umat}^T_{31}{\Umat}_{31}={\Imat}_{2n}$, hence we conclude that  ${\Umat}_{31}$ is a $2n\times 2n$ orthogonal matrix, and that the matrix ${\hmat}_{31}$ can now be written as ${\hmat}_{31}=\sqrt{\SNR_{31}}\cdot {\Umat}_{31}$. Similarly we define ${\Umat}_{11}=\left[ \begin{array}{cc}
{\Cmat}_{11} & {-\Smat}_{11} \\
{\Smat}_{11} & {\Cmat}_{11} \end{array}
\right]$.
\fi
 Next, consider ${\eta }_1\cdot W^n_1$. Begin by writing:
\begin{eqnarray}
    \!\!\!\!\!\!\!\!\!\!\!\!\!\!\!{\eta }_1\cdot W^n_1&=&
    \big(\RealS\left\{{\eta }_1\right\}+j\cdot \ImagS\left\{{\eta }_1\right\}\big)\big(\RealS\left\{W^n_1\right\}+j\cdot \ImagS\left\{W^n_1\right\}\big)\nonumber\\
    \label{eqnE:defweta}
    &=&\big(\RealS\left\{{\eta }_1\right\}\RealS\left\{W^n_1\right\}-\ImagS\left\{{\eta }_1\right\}\ImagS\left\{W^n_1\right\}\big)+j\cdot \big(\RealS\left\{{\eta }_1\right\}\ImagS\left\{W^n_1\right\}+\ImagS\left\{{\eta }_1\right\}\RealS\left\{W^n_1\right\}\big),
\end{eqnarray}
and define ${\oW}^{2n}_{\eta}\triangleq {\left(\big(\RealS{\left\{{\eta }_1\cdot W^n_1\right\}}\big)^T,\big(\ImagS{\left\{{\eta }_1\cdot W^n_1\right\}\big)}^T\right)}^T$ and ${\oW}^{2n}_1\triangleq \left(\big(\RealS\left\{W^n_1\right\}\big)^T,\big(\ImagS\left\{W^n_1\right\}\big)^T\right)^T$.
%Note that $\RealS\left\{{\eta }_1\cdot W^n_1\right\}=\RealS\left\{{\eta }_1\right\}\RealS\left\{W^n_1\right\}-\ImagS\left\{{\eta }_1\right\}\ImagS\left\{W^n_1\right\}$.
Since $W^n_1$ is an i.i.d. circularly symmetric complex Gaussian vector, where each element has a unit variance, then we conclude that\footnote{For complex Normal $\Zvec=\Xvec+j\Yvec$, where $\Xvec$ and $\Yvec$ are real vectors, define $\mu=\E\{\Zvec\}, \Gmat=\E\{(\Zvec-\mu) (\Zvec-\mu)^H\}$, and $\Jmat=\E\{(\Zvec-\mu) (\Zvec-\mu)^T\}$. Then, $(\Xvec^T,\Yvec^T )^T$ is a real jointly Gaussian  vector with covariance matrix $\cov(\Xvec,\Yvec)=\frac{1}{2} \ImagS\{-\Gmat+\Jmat\}$. For a circularly symmetric complex Normal vector $\Zvec$ we have $\mu=0, \Jmat=\Omat$, and $\cov(\Xvec)=\cov(\Yvec)=\frac{1}{2} \RealS\{\Gmat\}$. When $\Zvec$ is circularly symmetric complex Normal with i.i.d. elements, each has a unit variance, then $\Gmat=\Imat_n$ is real, hence, $\cov(\Xvec,\Yvec)=0$, and from joint Gaussianity it follows that the real and imaginary parts of $\Zvec$ are mutually independent, see http://www.rle.mit.edu/rgallager/documents/CircSymGauss.pdf\phantom{l}.}  ${\oW}^{2n}_1$ is a real jointly Normal vector with covariance matrix
\ifextended
\cite[Lemma 5]{Massey:93}:
\begin{equation*}
    \left[\  \begin{array}{cc}
    \cov\big(\RealS\left\{W^n_1\right\},\RealS\left\{W^n_1\right\}\big) & \cov\big(\RealS\left\{W^n_1\right\},\ImagS\left\{W^n_1\right\}\big) \\
    \cov\big(\ImagS\left\{W^n_1\right\},\RealS\left\{W^n_1\right\}\big) & \cov\big(\ImagS\left\{W^n_1\right\},\ImagS\left\{W^n_1\right\}\big) \end{array}
    \right]\stackrel{(a)}{=}\left[ \begin{array}{cc}
    {\frac{1}{2}\Imat}_n & {\Omat}_{n\times n} \\
    {\Omat}_{n\times n} & \frac{1}{2}{\Imat}_n \end{array}
    \right]=\frac{1}{2}{\Imat}_{2n},
\end{equation*}
where (a) follows since for circularly symmetric complex Normal RVs then\footnote{http://www.rle.mit.edu/rgallager/documents/CircSymGauss.pdf}
\begin{eqnarray*}
    \cov\big(\RealS\left\{W^n_1\right\},\RealS\left\{W^n_1\right\}\big)&=&\cov\big(\ImagS\left\{W^n_1\right\},\ImagS\left\{W^n_1\right\}\big)=\frac{1}{2}\RealS\big\{\cov\left(W^n_1,W^n_1\right)\big\}\\
    \cov\big(\RealS\left\{W^n_1\right\},\ImagS\left\{W^n_1\right\}\big)&=&-\cov\big(\ImagS\left\{W^n_1\right\},\RealS\left\{W^n_1\right\}\big)=-\frac{1}{2}\ImagS\big\{\cov\left(W^n_1,W^n_1\right)\big\},
\end{eqnarray*}
\and we note that as $W^n_1$ has i.i.d. elements then $\cov\left(W^n_1,W^n_1\right) = \Imat_{n}$ is a real matrix, thus  $\ImagS\left\{\cov\left(W^n_1,W^n_1\right)\right\}={\Omat}_{n\times n}$.
\else
$\cov({\oW}^{2n}_1)=\frac{1}{2}\Imat_{2n}$, see \cite[Lemma 5]{Massey:93}.
\fi
\ifextended
Next, recall from \eqref{eqnE:defweta} that
\begin{eqnarray*}
    \RealS\left\{{\eta }_1\cdot W^n_1\right\}&=&\RealS\left\{{\eta }_1\right\}\RealS\left\{W^n_1\right\}-\ImagS\left\{{\eta }_1\right\}\ImagS\left\{W^n_1\right\}\\
    \ImagS\left\{{\eta }_1\cdot W^n_1\right\}&=&\RealS\left\{{\eta }_1\right\}\ImagS\left\{W^n_1\right\}+\ImagS\left\{{\eta }_1\right\}\RealS\left\{W^n_1\right\}.
\end{eqnarray*}
It thus follows that  $\RealS\left\{{\eta }_1\cdot W^n_1\right\}$ is a zero-mean Normal RV whose covariance matrix is
\begin{eqnarray*}
    & &\!\!\!\!\!\!\!\! \!\!\!\!\E\left\{\RealS\left\{{\eta }_1\cdot W^n_1\right\}{\big(\RealS\left\{{\eta }_1\cdot W^n_1\right\}\big)}^T\right\}=\\
    & = & \E\left\{\big(\RealS\left\{{\eta }_1\right\}\RealS\left\{W^n_1\right\}-\ImagS\left\{{\eta }_1\right\}\ImagS\left\{W^n_1\right\}\big){\big(\RealS\left\{{\eta }_1\right\}\RealS\left\{W^n_1\right\}-\ImagS\left\{{\eta }_1\right\}\ImagS\left\{W^n_1\right\}\big)}^T\right\}\\
    &=&{\big(\RealS\left\{{\eta }_1\right\}\big)}^2\E\left\{\RealS\left\{W^n_1\right\}{\big(\RealS\left\{W^n_1\right\}\big)}^T\right\}+{\big(\ImagS\left\{{\eta }_1\right\}\big)}^2\ \left\{\ImagS\left\{W^n_1\right\}{\big(\ImagS\left\{W^n_1\right\}\big)}^T\right\}\\
    &=&{\big(\RealS\left\{{\eta }_1\right\}\big)}^2\frac{1}{2}{\Imat}_n+{\big(\ImagS\left\{{\eta }_1\right\}\big)}^2\frac{1}{2}{\Imat}_n\\
    &=&{\left|{\eta }_1\right|}^2\frac{1}{2}{\Imat}_n,
\end{eqnarray*}
hence $\RealS\left\{{\eta }_1\cdot W^n_1\right\}\sim \mathcal{N}\left({\bm 0},\ {\left|{\eta }_1\right|}^2\frac{1}{2}{\Imat}_n\right)$. Repeating the derivation for $\ImagS\left\{{\eta }_1\cdot W^n_1\right\}=\RealS\left\{{\eta }_1\right\}\ImagS\left\{W^n_1\right\}+\ImagS\left\{{\eta }_1\right\}\RealS\left\{W^n_1\right\}$, it directly follows that  $\ImagS\left\{{\eta }_1\cdot W^n_1\right\}\sim \mathcal{N}\left({\bm 0},{\left|{\eta }_1\right|}^2\frac{1}{2}{\Imat}_n\right)$. Finally we compute
\begin{eqnarray*}
    & &\!\!\!\!\!\!\!\! \!\!\!\!\E\left\{\RealS\left\{{\eta }_1\cdot W^n_1\right\}{\big(\ImagS\left\{{\eta }_1\cdot W^n_1\right\}\big)}^T\right\}\\
    &=& \E\left\{\big(\RealS\left\{{\eta }_1\right\}\RealS\left\{W^n_1\right\}-\ImagS\left\{{\eta }_1\right\}\ImagS\left\{W^n_1\right\}\big)\cdot {\big(\RealS\left\{{\eta }_1\right\}\ImagS\left\{W^n_1\right\}+\ImagS\left\{{\eta }_1\right\}\RealS\left\{W^n_1\right\}\big)}^T\right\}\\
    &=&\E\Big\{\RealS\left\{{\eta }_1\right\}\RealS\left\{W^n_1\right\}\RealS\left\{{\eta }_1\right\}{\big(\ImagS\left\{W^n_1\right\}\big)}^T+\RealS\left\{{\eta }_1\right\}\RealS\left\{W^n_1\right\}\ImagS\left\{{\eta }_1\right\}{\big(\RealS\left\{W^n_1\right\}\big)}^T\\
    & &\qquad-\ImagS\left\{{\eta }_1\right\}\ImagS\left\{W^n_1\right\}\RealS\left\{{\eta }_1\right\}{\big(\ImagS\left\{W^n_1\right\}\big)}^T-\ImagS\left\{{\eta }_1\right\}\ImagS\left\{W^n_1\right\}\ImagS\left\{{\eta }_1\right\}{\big(\RealS\left\{W^n_1\right\}\big)}^T\Big\}\\
    &=&\E\Big\{\big(\RealS{\left\{{\eta }_1\right\}}\big)^2\RealS\left\{W^n_1\right\}{\big(\ImagS\left\{W^n_1\right\}\big)}^T+\RealS\left\{{\eta }_1\right\}\ImagS\left\{{\eta }_1\right\}\RealS\left\{W^n_1\right\}{\big(\RealS\left\{W^n_1\right\}\big)}^T\\
    & &\qquad-\ImagS\left\{{\eta }_1\right\}\RealS\left\{{\eta }_1\right\}\ImagS\left\{W^n_1\right\}{\left(\ImagS\left\{W^n_1\right\}\right)}^T-{\left(\ImagS\left\{{\eta }_1\right\}\right)}^2\ImagS\left\{W^n_1\right\}{\left(\RealS\left\{W^n_1\right\}\right)}^T\Big\}\\
    &=&\E\left\{\RealS\left\{{\eta }_1\right\}\ImagS\left\{{\eta }_1\right\}\RealS\left\{W^n_1\right\}{\big(\RealS\left\{W^n_1\right\}\big)}^T-\ImagS\left\{{\eta }_1\right\}\RealS\left\{{\eta }_1\right\}\ImagS\left\{W^n_1\right\}{\big(\ImagS\left\{W^n_1\right\}\big)}^T\right\}\\
    &=&\RealS\left\{{\eta }_1\right\}\ImagS\left\{{\eta }_1\right\}\E\left\{\RealS\left\{W^n_1\right\}{\big(\RealS\left\{W^n_1\right\}\big)}^T-\ImagS\left\{W^n_1\right\}{\big(\ImagS\left\{W^n_1\right\}\big)}^T\right\}\\
    &=&\RealS\left\{{\eta }_1\right\}\ImagS\left\{{\eta }_1\right\}\left(\frac{1}{2}\Imat_n - \frac{1}{2}\Imat_n\right)\\
    &=&{\Omat}_{n\times n},
\end{eqnarray*}
\else
Next, from \eqref{eqnE:defweta} it directly follows that $\RealS\left\{{\eta }_1\cdot W^n_1\right\}\sim \mathcal{N}\left({\bm 0},\ {\left|{\eta }_1\right|}^2\frac{1}{2}{\Imat}_n\right)$
and $\ImagS\left\{{\eta }_1\cdot W^n_1\right\}\sim \mathcal{N}\left({\bm 0},{\left|{\eta }_1\right|}^2\frac{1}{2}{\Imat}_n\right)$,
\fi
and consequently, ${\oW}^{2n}_{\eta} \sim \mathcal{N}\left({\bm 0},\frac{{\left|{\eta }_1\right|}^2}{2}{\Imat}_{2n}\right)$.
Defining ${\oV}^{2n}_2 = \left(\big(\RealS\{V_2^n\}\big)^T,\big(\ImagS\{V_2^n\}\big)^T \right)^T$, we similarly conclude ${\oV}^{2n}_2\sim \mathcal{N}\left({\bm 0},\frac{{1-\left|v_2\right|}^2}{2}{\Imat}_{2n}\right).$
Lastly, write
\begin{eqnarray*}
    {\hmat}^{(n)}_{h_{k1}}X^n_{k\oG} & = & \left(\RealS\left\{{\hmat}^{(n)}_{h_{k1}}\right\}\RealS\left\{X^n_{k\oG}\right\}-\ImagS\left\{{\hmat}^{(n)}_{h_{k1}}\right\}\ImagS\left\{X^n_{k\oG}\right\}\right)\\
    & & \qquad\qquad  +j\cdot \left(\RealS\left\{{\hmat}^{(n)}_{h_{k1}}\right\}\ImagS\left\{X^n_{k\oG}\right\}+\ImagS\left\{{\hmat}^{(n)}_{h_{k1}}\right\}\RealS\left\{X^n_{k\oG}\right\}\right).
\end{eqnarray*}
From the above assignments it also follows that  ${\hmat}^{(n)}_{h_{k1}}X^n_{k\oG}$ is statistically equivalent to  ${\hmat}_{k1}{\oX}^{2n}_{k\oG}$, $k\in \left\{1,3\right\}$, and that
${\hmat}^{(n)}_{h_{32}}X^n_{3\oG}$ is statistically equivalent to  ${\hmat}_{32}{\oX}^{2n}_{3\oG}$.
Hence,
\begin{eqnarray}
\label{eqnE:DefEquivDiffEnt}
&  &   \dsE_{{\tH}^n_1,{\tH}^n_2}\left\{ h\left({\Hmat}^{(n)}_{h_{31}}X^n_{3\oG}+{\Hmat}^{(n)}_{h_{11}}X^n_{1\oG}+{\eta }_1\cdot W^n_1\big|{\tH}^n_1 = {\th}^n_1\right)\right\}\nonumber\\
& & \qquad \qquad\qquad \qquad
    =\dsE_{{\tH}^n_1}\left\{h\left({\hmat}_{31}{\oX}^{2n}_{3\oG}+{\hmat}_{11}{\oX}^{2n}_{1\oG}+{\oW}^{2n}_{\eta}\big| {\tH}^n_1= {\th}^n_1\right)\right\}.
\end{eqnarray}

Denote the $2n\times 2n$ covariance matrices for ${\oX}^{2n}_{3\oG}$ and ${\oX}^{2n}_{1\oG}$ by ${\Kmat}_{{\oX}_{3\oG}}$ and ${\Kmat}_{{\oX}_{1\oG}}$, respectively, and let ${\Kmat}_{{\oX}_{13\oG}}\triangleq \E\left\{{\oX}^{2n}_{1\oG}\cdot {\left({\oX}^{2n}_{3\oG}\right)}^T\right\}$. Next, as $X^n_{1\oG}$ and $X^n_{3\oG}$ are complex jointly  Gaussian and both are independent of $W^n_1$, it follows that  given
${\tH}^n_1= {\th}^n_1$, then ${\hmat}_{31}{\oX}^{2n}_{3\oG}+{\hmat}_{11}{\oX}^{2n}_{1\oG}+{\oW}^{2n}_{\eta}$ is a jointly Gaussian real random vector, whose covariance matrix, $\cov\left({\hmat}_{31}{\oX}^{2n}_{3\oG}+{\hmat}_{11}{\oX}^{2n}_{1\oG}+{\oW}^{2n}_{\eta}\right)\triangleq \Tmat_1$, is given by:
\begin{equation*}
    {\Tmat}_1={\hmat}_{31}{\Kmat}_{{\oX}_{3\oG}}{\hmat}^T_{31}+{\hmat}_{11}{\Kmat}_{{\oX}_{1\oG}}{\hmat}^T_{11}+{\hmat}_{11}{\Kmat}_{{\oX}_{13\oG}}{\hmat}^T_{31}+{\hmat}_{31}{\Kmat}^T_{{\oX}_{13\oG}}{\hmat}^T_{11}+\frac{{\left|{\eta }_1\right|}^2}{2}{\Imat}_{2n},
\end{equation*}
hence, its  differential entropy is given by \cite[Pg. 1296]{Massey:93}:
\begin{equation*}
    h\left({\hmat}_{31}{\oX}^{2n}_{3\oG}+{\hmat}_{11}{\oX}^{2n}_{1\oG}+{\oW}^{2n}_{\eta}\big|{\tH}^n_1= {\th}^n_1\right)=\frac{1}{2}\log \det \left({\left(2\pi \right)}^{2n}{\ \Tmat}_1\right).
\end{equation*}

Note that the expectation over ${\tH}^n_1$
\ifextended in \eqref{eqnE:DefEquivDiffEnt}:
\begin{eqnarray*}
    & & {\E}_{{\tH}^n_1,{\tH}^n_2}\left\{h\left({\hmat}^{(n)}_{h_{31}}X^n_{3\oG}+{\hmat}^{(n)}_{h_{11}}X^n_{1\oG}+{\eta }_1\cdot W^n_1\big|{\tH}^n_1={\th}^n_1\right)\right\}=\\
    & & \qquad \qquad \qquad \qquad \qquad {\E}_{{\tH}^n_1}\left\{h\left({\hmat}^{(n)}_{h_{31}}X^n_{3\oG}+{\hmat}^{(n)}_{h_{11}}X^n_{1\oG}+{\eta }_1\cdot W^n_1\big|{\tH}^n_1={\th}^n_1\right)\right\},
\end{eqnarray*}
\else
in \eqref{eqnE:DefEquivDiffEnt}
\fi
is taken with respect to the phases of the channel coefficients, which are mutually independent over the links, i.i.d. over time, and distributed uniformly over $\left[0,2\pi \right)$. Since $\cos(x+\pi)=-\cos(x)$ and $\sin (x+\pi)=-\sin(x)$, then replacing ${\hmat}_{11}$ with ${-\hmat}_{11}$ is equivalent to shifting the phases of all time realizations of $H_{11}$, i.e., all coefficients in ${\hmat}^{(n)}_{h_{11}}$ by $\pi $. As the complex exponential is periodic with a period of $2\pi $, and the expectation spans a continuous intervals of $2\pi $ radians, a constant phase shift to all elements of the diagonal matrix ${\hmat}^{(n)}_{11}$ does not affect the expectation, see proof of \cite[Thm. 8]{Kramer:05}, and consequently:
\begin{eqnarray}
    & & {\E}_{{\tH}^n_1}\left\{h\left({\hmat}^{(n)}_{h_{31}}X^n_{3\oG}+{\hmat}^{(n)}_{h_{11}}X^n_{1\oG}+{\eta }_1\cdot W^n_1\big|{\tH}^n_1={\th}^n_1\right)\right\}\nonumber\\
    & & \qquad \qquad \qquad \qquad \qquad = {\E}_{{\tH}^n_1}\left\{h\left({\hmat}^{(n)}_{h_{31}}X^n_{3\oG}-{\hmat}^{(n)}_{h_{11}}X^n_{1\oG}+{\eta }_1\cdot W^n_1\big|{\tH}^n_1={\th}^n_1\right)\right\}\nonumber\\
    & & \qquad \qquad \qquad \qquad \qquad = {\E}_{{\tH}^n_1}\left\{h\left({\hmat}_{31} \oX^{2n}_{3\oG}- {\hmat}_{11}\oX^{2n}_{1\oG}+ \oW^{2n}_{\eta}\big|{\tH}^n_1={\th}^n_1\right)\right\}.\label{eq:AppENeweq1}
\end{eqnarray}
Let the covariance matrix for ${\hmat}_{31} \oX^{2n}_{3\oG}- {\hmat}_{11}\oX^{2n}_{1\oG}+ \oW^{2n}_{\eta}$ be denoted with $\Tmat_2 \triangleq \cov\left({\hmat}_{31} \oX^{2n}_{3\oG}- {\hmat}_{11}\oX^{2n}_{1\oG}+\oW^{2n}_{\eta} \right)$. $\Tmat_2$ can be written as
\begin{eqnarray*}
    {\Tmat}_2 & = & {\hmat}_{31}{\Kmat}_{{\oX}_{3\oG}}{\hmat}^T_{31}+{\hmat}_{11}{\Kmat}_{{\oX}_{1\oG}}{\hmat}^T_{11}
    %&& \qquad \qquad
		-{\hmat}_{11}{\Kmat}_{{\oX}_{13\oG}}{\hmat}^T_{31}-{\hmat}_{31}{\Kmat}^T_{{\oX}_{13\oG}}{\hmat}^T_{11}+\frac{{\left|{\eta }_1\right|}^2}{2}{\Imat}_{2n}.
\end{eqnarray*}

\noindent
Using $\Tmat_1$ and $\Tmat_2$  we can write
\begin{eqnarray*}
    & & \hspace{-2cm}{\E}_{{\tH}^n_1}\left\{h\left({\hmat}^{(n)}_{h_{31}}X^n_{3\oG}+{\hmat}^{(n)}_{h_{11}}X^n_{1\oG}+{\eta }_1\cdot W^n_1\big|{\tH}^n_1={\th}^n_1\right)\right\}\\
    &=&\frac{1}{2}{\E}_{{\tH}^n_1}\left\{h\left({\hmat}^{(n)}_{h_{31}}X^n_{3\oG}+{\hmat}^{(n)}_{h_{11}}X^n_{1\oG}+{\eta }_1\cdot W^n_1\big|{\tH}^n_1={\th}^n_1\right)\right\} \\
    & & \qquad \qquad +\frac{1}{2}{\E}_{{\tH}^n_1}\left\{h\left({\hmat}^{(n)}_{h_{31}}X^n_{3\oG}-{\hmat}^{(n)}_{h_{11}}X^n_{1\oG}+{\eta }_1\cdot W^n_1\big|{\tH}^n_1={\th}^n_1\right)\right\}\\
    &=&\frac{1}{2}{\E}_{{\tH}^n_1}\left\{\frac{1}{2}\log\det \left({\left(2\pi \right)}^{2n}{\ \Tmat}_1\right)\right\}+\frac{1}{2}{\E}_{{\tH}^n_1}\left\{\frac{1}{2}\log\det \left({\left(2\pi \right)}^{2n}{\ \Tmat}_2\right)\right\}\\
    &\stackrel{(a)}{\le}& {\E}_{{\tH}^n_1}\left\{\frac{1}{2}\log\det \left(\frac{1}{2}{\left(2\pi \right)}^{2n}{\Tmat}_1+\frac{1}{2}{\left(2\pi \right)}^{2n}{\Tmat}_2\right)\right\}\\
    &=&{\E}_{{\tH}^n_1}\left\{\frac{1}{2}\log\det \left({\left(2\pi \right)}^{2n}\right)+\frac{1}{2}\log\det \bigg({\hmat}_{31}{\Kmat}_{{\oX}_{3\oG}}{\hmat}^T_{31}+{\hmat}_{11}{\Kmat}_{{\oX}_{1\oG}}{\hmat}^T_{11}+\frac{{\left|{\eta }_1\right|}^2}{2}{\Imat}_{2n}\bigg)\right\},
\end{eqnarray*}

\noindent
where the inequality (a) follows from the concavity of the logdet function in space of p.d symmetric matrices \cite[Section 3.1.5]{ConvexOpt}. Observe that the inequality (a) is obtained with equality when ${\Kmat}_{{\oX}_{13\oG}}={\Omat}_{2n \times 2n}$. As ${\oX}^{2n}_{3\oG}$ and ${\oX}^{2n}_{1\oG}$, are jointly Gaussian, then zero cross-correlation implies ${\oX}^{2n}_{3\oG}$ and ${\oX}^{2n}_{1\oG}$ are mutually independent. Note that ${\E}_{{\tH}^n_1,{\tH}^n_2}\left\{h\left({\hmat}^{(n)}_{h_{32}}X^n_{3\oG}+V^n_2\big|{\tH}^n_2={\th}^n_2\right)\right\}={\E}_{{\tH}^n_2}\left\{h\left({\hmat}^{(n)}_{h_{32}}X^n_{3\oG}
+V^n_2\big|{\tH}^n_2={\th}^n_2\right)\right\}$ is not affected by the correlation between ${\oX}^{2n}_{3\oG}$ and ${\oX}^{2n}_{1\oG}$, hence, we conclude that \eqref{eq:NewAppEeq0} is maximized by mutually independent ${\oX}^{2n}_{3\oG}$ and ${\oX}^{2n}_{1\oG}$, and henceforth we shall proceed with this assumption.

Next, define $\hat{\Kmat} \triangleq \cov\left\{{\hmat}_{31}{\oX}^{2n}_{3\oG}+{\hmat}_{11}{\oX}^{2n}_{1\oG}\Big|{\tH}^n_1={\th}^n_1\right\}={\hmat}_{31}{\Kmat}_{{\oX}_{3\oG}}{\hmat}^T_{31}+{\hmat}_{11}{\Kmat}_{{\oX}_{1\oG}}{\hmat}^T_{11}$. As ${\Kmat}_{{\oX}_{3\oG}}$ and ${\Kmat}_{{\oX}_{1\oG}}$ are both covariance matrices for real random vectors, they are both
\ifextended
symmetric\footnote{Note that this follows since for any two RVs, $X$ and $Y$ we have that $\E\{X\cdot Y\}=\E\{Y\cdot X\}$.},
\else
symmetric,
\fi
hence, ${\hat{\Kmat}}$ and $\cov\left({\hmat}_{32}{\oX}^{2n}_{3\oG}\big|{\tH}^n_2={\th}^n_2\right)\triangleq {\hmat}_{32}{\Kmat}_{{\oX}_{3\oG}}{\hmat}^T_{32}$ are both real symmetric matrices whose eigenvalues are real and non-negative.
%Thus, for both ${\hat{\Kmat}}$ and ${\hmat}_{32}{\Kmat}_{{\oX}_{3\oG}}{\hmat}^T_{32}$ there exist eigenvalue representations in which {\em the respective eigenvectors can be collected into orthogonal matrices}
For a square  real matrix $\Bmat$, we use $\eig\left(\Bmat\right)$ to denote a diagonal matrix whose diagonal elements are the eigenvalues of $\Bmat$, when the eigenvectors are orthogonal and have a unit norm, and the eigenvalues appear in ascending order in $\eig\left(\Bmat\right)$, i.e., $\eig\left(\Bmat\right)_{i,i} \le \eig\left(\Bmat\right)_{k,k}$ for $i \le k$,
\cite[Pg. 549]{Meyer}.
In the following we will refer only to eigenvalue decompositions of this form.
We next write the eigenvalue representations for ${\hmat}_{32}{\Kmat}_{{\oX}_{3\oG}}{\hmat}^T_{32}$ and  for ${\hat{\Kmat}}$ as ${\hmat}_{32}{\Kmat}_{{\oX}_{3\oG}}{\hmat}^T_{32}={\hat{\Umat}}_{32\oG}{\hat{\Dmat}}_{32\oG}{\hat{\Umat}}^T_{32\oG}$, and ${\hat{\Kmat}}={\hat{\Umat}}{\hat{\Dmat}}{\hat{\Umat}}^T$.
\noindent
As the determinant of an orthogonal matrix is 1, we can write
\begin{eqnarray*}
    & &\hspace{-2cm}h\left({\hmat}^{(n)}_{h_{31}}X^n_{3\oG}+{\hmat}^{(n)}_{h_{11}}X^n_{1\oG}+{\eta }_1\cdot W^n_1\big|{\tH}^n_1={\th}^n_1\right)-h\left({\hmat}^{(n)}_{h_{32}}X^n_{3\oG}+V^n_2\big|{\tH}^n_2={\th}^n_2\right)\\
    &= &\frac{1}{2}\log \det \left({\hat{\Kmat}}+\frac{{\left|{\eta }_1\right|}^2}{2}{\Imat}_{2n}\right)-\frac{1}{2}\log\det \left({\hmat}_{32}{\Kmat}_{{\oX}_{3\oG}}{\hmat}^T_{32}+\frac{1-{\left|{\CORRN}_2\right|}^2}{2}{\Imat}_{2n}\right)\\
    \ifextended
    &=&\frac{1}{2}\log\det \left({\hat{\Umat}}{\hat{\Dmat}}{\hat{\Umat}}^T+\frac{{\left|{\eta }_1\right|}^2}{2}{\Imat}_{2n}\right)-\frac{1}{2}\log\det \left({\hat{\Umat}}_{32\oG}{\hat{\Dmat}}_{32\oG}{\hat{\Umat}}^T_{32\oG}+\frac{1-{\left|{\CORRN}_2\right|}^2}{2}{\Imat}_{2n}\right)\\
    \fi
    &=&\frac{1}{2}\log\det \left({\hat{\Dmat}}+\frac{{\left|{\eta }_1\right|}^2}{2}{\Imat}_{2n}\right)-\frac{1}{2}\log\det \left({\hat{\Dmat}}_{32\oG}+\frac{1-{\left|{\CORRN}_2\right|}^2}{2}{\Imat}_{2n}\right).
\end{eqnarray*}
%For a square and symmetric real matrix $\Bmat$, we use $\eig\left(\Bmat\right)$ to denote a diagonal matrix whose diagonal elements are the eigenvalues of $\Bmat$, when the eigenvectors are orthogonal, and the eigenvalues appear in ascending order in $\eig\left(\Bmat\right)$.
Next, we state the following basic fact: For a real symmetric matrix $\Amat$, and a real orthogonal matrix $\Qmat$ then $\eig\left(\Amat\right)=\eig\left(\Qmat\Amat{\Qmat}^T\right)$. Note that for orthogonal $\Qmat$ then $\Amat$ and $\Qmat\Amat{\Qmat}^T$ are called similar and the simple fact stated above is also referred to as ``similarity preserves eigenvalues'', see \cite[Pg. 508 and Pg. 549]{Meyer}.

\ifextended
This fact can be easily confirmed as the eigenvalue decomposition for an $n \times n$ real symmetric matrix can be written as $\Amat=\Vmat\Dmat{\Vmat}^T$, where $\Vmat$ is an orthogonal matrix and $\Dmat$ is a square real diagonal matrix, s.t.  ${\left[\Dmat\right]}_{i,i}\le {\left[\Dmat\right]}_{k,k},\ \forall 1\le i\le k\le n$. Then, $\Qmat\Amat{\Qmat}^T=\Qmat\Vmat\Dmat{\Vmat}^T{\Qmat}^T$. Since $\Qmat$ and $\Vmat$ are orthogonal it immediately follows that $\Qmat\Vmat{\left(\Qmat\Vmat\right)}^T=\Qmat\Vmat{\Vmat}^T{\Qmat}^T=\Imat$, thus $\Qmat\Vmat$ is orthogonal. Hence, for $\Qmat\Amat{\Qmat}^T$, we conclude that  $\Qmat\Vmat\Dmat{\left(\Qmat\Vmat\right)}^T$ is the eigenvalue decomposition of $\Qmat\Amat{\Qmat}^T$, in which the eigenvectors are orthogonal and the eigenvalues appear in ascending order, and the corresponding eigenvalues matrix is $\Dmat$.
\fi

\noindent
As ${\Kmat}_{{\oX}_{3\oG}}$ is a real $2n\times 2n$ symmetric matrix, we can express ${\Kmat}_{{\oX}_{3\oG}}={\Amat}_{3\oG}{\Dmat}_{3\oG}{\Amat}^T_{3\oG}$, where ${\Amat}_{3\oG}$ is an orthogonal
matrix and ${\Dmat}_{3\oG}$ is a diagonal matrix, whose diagonal elements are in ascending order:
$d_{3\oG,i} \triangleq {\left[\Dmat_{3\oG}\right]}_{i,i}\le {\left[\Dmat_{3\oG}\right]}_{k,k},\ \forall 1\le i\le k\le 2n$. Next, using similarity we write
\begin{equation*}
    {\hat{\Dmat}}_{32\oG}=\eig\left({\hmat}_{32}{\Kmat}_{{\oX}_{3\oG}}{\hmat}^T_{32}\right)=\eig\left(\sqrt{\SNR_{32}}{\Umat}_{32}{\Amat}_{3\oG}{\Dmat}_{3\oG}{\Amat}^T_{3\oG}{\Umat}^T_{32}\sqrt{\SNR_{32}}\right)=\SNR_{32}{\Dmat}_{3\oG}.
\end{equation*}
Note that
\begin{equation}
    \label{eq:NewAppEeq2}
    {\hmat}_{32}{\Dmat}_{3\oG}{\hmat}^T_{32}=\sqrt{\SNR_{32}}{\Umat}_{32}{\Dmat}_{3\oG}{\Umat}^T_{32}\sqrt{\SNR_{32}}={\Umat}_{32}{\hat{\Dmat}}_{32\oG}{\Umat}^T_{32}.
\end{equation}
From \eqref{eq:NewAppEeq2} we obtain that
%for any ${\Kmat}_{{\oX}_{3\oG}}$ we obtain that the rate $R_3$ satisfies [Number 2]
\ifextended
\begin{eqnarray}
    \frac{1}{2}\log\det \left({\hmat}_{32}{\Kmat}_{{\oX}_{3\oG}}{\hmat}^T_{32}+\frac{1-{\left|{\CORRN}_2\right|}^2}{2}{\Imat}_{2n}\right)
    &=&\frac{1}{2}\log\det \left({\hat{\Dmat}}_{32\oG}+\frac{1-{\left|{\CORRN}_2\right|}^2}{2}{\Imat}_{2n}\right)\nonumber\\
    &=&  \frac{1}{2}\log\det \left(\SNR_{32}{\Dmat}_{3\oG}+\frac{1-{\left|{\CORRN}_2\right|}^2}{2}{\Imat}_{2n}\right).\label{eq:NewAppEeq3}
\end{eqnarray}
\else\begin{eqnarray}
    \frac{1}{2}\log\det \left({\hmat}_{32}{\Kmat}_{{\oX}_{3\oG}}{\hmat}^T_{32}+\frac{1-{\left|{\CORRN}_2\right|}^2}{2}{\Imat}_{2n}\right)
    %&=&\frac{1}{2}\log\det \left({\hat{\Dmat}}_{32\oG}+\frac{1-{\left|{\CORRN}_2\right|}^2}{2}{\Imat}_{2n}\right)\nonumber\\
    &=&  \frac{1}{2}\log\det \left(\SNR_{32}{\Dmat}_{3\oG}+\frac{1-{\left|{\CORRN}_2\right|}^2}{2}{\Imat}_{2n}\right).\label{eq:NewAppEeq3}
\end{eqnarray}
\fi
From \eqref{eq:NewAppEeq2} and \eqref{eq:NewAppEeq3} we conclude that given ${\Kmat}_{{\oX}_{3\oG}}$, the same differential entropy can be achieved by replacing ${\oX}^{2n}_{3\oG}$ with an ${\tilde{X}}^{2n}_{3\oG}\sim \mathcal{N}\left({\bm 0},{\Dmat}_{3\oG\ }\right).$ We also recall that from the power constraint $P_{k,i}\triangleq\E\big\{|X_{k,i}|^2\big\}\le 1$, $k\in\{1,2,3\}$  we have $\tr\left\{{\Dmat}_{3\oG}\right\}=\tr\left\{{\Kmat}_{{\oX}_{3\oG}}\right\}\le n.$

Next, using matrix similarity we write
\begin{eqnarray}
    & & \eig\left({\hat{\Kmat}} + \frac{{\left|{\eta }_1\right|}^2}{2}{\Imat}_{2n}\right)\nonumber\\
    %%%%%%%%%%%%%%%%%%%%%%%%
    & & =\eig\left({\hmat}_{31}{\Kmat}_{{\oX}_{3\oG}}{\hmat}^T_{31}+{\hmat}_{11}{\Kmat}_{{\oX}_{1\oG}}{\hmat}^T_{11} + \frac{{\left|{\eta }_1\right|}^2}{2}{\Imat}_{2n}\right)\nonumber\\
    %%%%%%%%%%%%%%%%%%%%%%%%
    & & =\eig\left(\SNR_{31}{\Umat}_{31}{\Amat}_{3\oG}{\Dmat}_{3\oG}{\Amat}^T_{3\oG}{\Umat}^T_{31}+\SNR_{11}{\Umat}_{11}{\Kmat}_{{\oX}_{1\oG}}{\Umat}^T_{11}+ \frac{{\left|{\eta }_1\right|}^2}{2}{\Imat}_{2n}\right)\nonumber\\
    %%%%%%%%%%%%%%%%%%%%%%%%
    & & = \eig\left(\!{\Umat}_{31}{\Amat}_{3\oG}\!\left(\!\SNR_{31}{\Dmat}_{3\oG}+
    \SNR_{11}{\Amat}^T_{3\oG}{\Umat}^T_{31}{\Umat}_{11}{\Kmat}_{{\oX}_{1\oG}}{\Umat}^T_{11}{\Umat}_{31}{\Amat}_{3\oG}+ \frac{{\left|{\eta }_1\right|}^2}{2}{\Amat}^T_{3\oG}{\Umat}^T_{31}{\Imat}_{2n}{\Umat}_{31}{\Amat}_{3\oG}\!\right)\!{\Amat}^T_{3\oG}{\Umat}^T_{31}\!\right)
    \nonumber\\
    %%%%%%%%%%%%%%%%%%%%%%%%
    \label{eqnE:eigen2}
    & & = \eig\left(\SNR_{31}{\Dmat}_{3\oG}+\SNR_{11}{\Amat}^T_{3\oG}{\Umat}^T_{31}{\Umat}_{11}{\Kmat}_{{\oX}_{1\oG}}{\Umat}^T_{11}{\Umat}_{31}{\Amat}_{3\oG}+ \frac{{\left|{\eta }_1\right|}^2}{2}{\Imat}_{2n}\right).
\end{eqnarray}

Let $\avec_{3\oG,i}\triangleq{\left[{\Amat}_{3\oG}\right]}_{\left(1\dots 2n\right),i}$ denote that $i$'th column of ${\Amat}_{3\oG}$. Applying Hadamard's inequality we obtain
\begin{eqnarray}
    & &\det \Bigg(\SNR_{11}{\Amat}^T_{3\oG}{\Umat}^T_{31}{\Umat}_{11}{\Kmat}_{{\oX}_{1\oG}}{\Umat}^T_{11}{\Umat}_{31}{\Amat}_{3\oG}
 +\SNR_{31}{\Dmat}_{3\oG}+\frac{{\left|{\eta }_1\right|}^2}{2}{\Imat}_{2n}\Bigg)\nonumber\\
%    \label{eqnE:Hadamardbound}
%    & &\qquad \qquad \prod^{2n}_{i=1}{\left(\SNR_{11}{\left({\left[{\Amat}_{3\oG}\right]}_{\left(1\dots 2n\right),i}\right)}^T{\Umat}^T_{31}{\Umat}_{11}{\Kmat}_{{\oX}_{1\oG}}{\Umat}^T_{11}{\Umat}_{31}{\left[{\Amat}_{3\oG}\right]}_{\left(1\dots 2n\right),i}+\SNR_{31}d_{3\oG,i}+\frac{{\left|{\eta }_1\right|}^2}{2}\right)}.
   \label{eqnE:Hadamardbound}
    & &\qquad \qquad\le\prod^{2n}_{i=1}\Bigg(\SNR_{11}{\left(\avec_{3\oG,i}\right)}^T
    {\Umat}^T_{31}{\Umat}_{11}{\Kmat}_{{\oX}_{1\oG}}{\Umat}^T_{11}{\Umat}_{31}\avec_{3\oG,i}  +\SNR_{31}d_{3\oG,i}+\frac{{\left|{\eta }_1\right|}^2}{2}\Bigg).
\end{eqnarray}
Combining \eqref{eq:NewAppEeq3}, \eqref{eqnE:eigen2}, \eqref{eqnE:Hadamardbound} and the fact that if matrices have the same eigenvalues their determinants are identical, we obtain the following upper bound on
%
%${\E}_{{\tH}^n_1,{\tH}^n_2}\bigg\{h\left({\Hmat}^{(n)}_{h_{31}}X^n_{3\oG}+{\Hmat}^{(n)}_{h_{11}}X^n_{1\oG}+{\eta }_1\cdot W^n_1\big|{\tH}^n_1={\th}^n_1\right) -$ $h\left({\Hmat}^{(n)}_{h_{32}}X^n_{3\oG}+V^n_2\big|{\tH}^n_2={\th}^n_2\right)\bigg\}$,
\eqref{eq:NewAppEeq0}:
%which is stated in  \eqref{eqnE:upperbound onE1_stepa},
%\ifextended
%\begin{eqnarray*}
%    & &h\left({\Hmat}^{(n)}_{h_{31}}X^n_{3\oG}+{\Hmat}^{(n)}_{h_{11}}X^n_{1\oG}+{\eta }_1\cdot W^n_1\big|{\tH}^n_1={\th}^n_1\right)-h\left({\Hmat}^{(n)}_{h_{32}}X^n_{3\oG}+V^n_2\big|{\tH}^n_2={\th}^n_2\right)\le\\
%    & &\qquad\sum^{2n}_{i=1}\frac{1}{2}\Bigg({\log \left(\SNR_{11}{\left({\left[{\Amat}_{3\oG}\right]}_{i,\left(1\dots 2n\right)}\right)}^T{\Umat}^T_{31}{\Umat}_{11}{\Kmat}_{{\oX}_{1\oG}}{\Umat}^T_{11}{\Umat}_{31}{\left[{\Amat}_{3\oG}\right]}_{i,\left(1\dots 2n\right)}\ +\SNR_{31}d_{3\oG,i}+\frac{{\left|{\eta }_1\right|}^2}{2}\right)\ }\\
%    & &\qquad \qquad \qquad {-\log \left(\SNR_{32}d_{3\oG,i}+\frac{1-{\left|\CORRN_2\right|}^2}{2}\right)\ }\Bigg).
%\end{eqnarray*}
%Hence,
%\fi
\begin{eqnarray}
    & &\hspace{-1.5cm} {\E}_{{\tH}^n_1,{\tH}^n_2}\bigg\{h\left({\Hmat}^{(n)}_{h_{31}}X^n_{3\oG}+{\Hmat}^{(n)}_{h_{11}}X^n_{1\oG}+{\eta }_1\cdot W^n_1\big|{\tH}^n_1={\th}^n_1\right)  -h\left({\Hmat}^{(n)}_{h_{32}}X^n_{3\oG}+V^n_2\big|{\tH}^n_2={\th}^n_2\right)\bigg\}\nonumber\\
    %%%%%%%%%%%%%%%%%%%%%%%%%%%%%%%%%%%%%%%%%
    &\le& {\E}_{{\tH}^n_1,{\tH}^n_2}\Bigg\{\sum^{2n}_{i=1}\frac{1}{2}\Bigg(\log \left(\SNR_{11}{\left(\avec_{3\oG,i}\right)}^T{\Umat}^T_{31}{\Umat}_{11}{\Kmat}_{{\oX}_{1\oG}}{\Umat}^T_{11}{\Umat}_{31}\avec_{3\oG,i} +\SNR_{31}d_{3\oG,i}+\frac{{\left|{\eta }_1\right|}^2}{2}\right)\nonumber \\
    & &\qquad \qquad \qquad \qquad \qquad \qquad \qquad -\log \left(\SNR_{32}d_{3\oG,i}+\frac{1-{\left|{\CORRN}_2\right|}^2}{2}\right)\ \Bigg)\Bigg\}\nonumber\\
    %%%%%%%%%%%%%%%%%%%%%%%%%%%%%%%%%%%%%%%%%
     &\stackrel{(a)}{\le} &\sum^{2n}_{i=1}\Bigg(\frac{1}{2}\log \left(\SNR_{11}{\left(\avec_{3\oG,i}\right)}^T{\E}_{{\tH}^n_1,{\tH}^n_2}\left\{{\Umat}^T_{31}{\Umat}_{11}{\Kmat}_{{\oX}_{1\oG}}{\Umat}^T_{11}{\Umat}_{31}\right\}\avec_{3\oG,i} +\SNR_{31}d_{3\oG,i}+\frac{{\left|{\eta }_1\right|}^2}{2}\right)\nonumber \\
     \label{eqnE:upperbound onE1_stepa}
     & &\qquad \qquad \qquad \qquad \qquad \qquad \qquad-\frac{1}{2}\log \left(\SNR_{32}d_{3\oG,i}+\frac{1-{\left|{\CORRN}_2\right|}^2}{2}\right)\Bigg)
\end{eqnarray}
where (a) is due to the fact that the logarithm function is a concave function and the application of Jensen's inequality.

\noindent
Next, we define $\Gmat\triangleq {\Umat}^T_{31}{\Umat}_{11}{\Kmat}_{{\oX}_{1\oG}}{\Umat}^T_{11}{\Umat}_{31}$,  and compute the expectation
\[
\bar{\Gmat}={\E}_{{\tH}^n_1,{\tH}^n_2}\big\{\Gmat\big\}={\E}_{{\tH}^n_1,{\tH}^n_2}\Big\{{\Umat}^T_{31}{\Umat}_{11}{\Kmat}_{{\oX}_{1\oG}}{\Umat}^T_{11}{\Umat}_{31}\Big\}.
\]
To that aim, define
\begin{eqnarray*}
    {\Kmat}^{RR}_{{\oX}_{1\oG}}&=&\E\left\{\RealS\left\{X^n_{1\oG}\right\}  \cdot\big(\RealS{\left\{X^n_{1\oG}\right\}}\big)^T\right\},\\
    {\Kmat}^{RI}_{{\oX}_{1\oG}} &= &\E\left\{\RealS\left\{X^n_{1\oG}\right\}\cdot\big(\ImagS{\left\{X^n_{1\oG}\right\}}\big)^T\right\},\\ {\Kmat}^{IR}_{{\oX}_{1\oG}}&=&\E\left\{\ImagS\left\{X^n_{1\oG}\right\}  \cdot\big(\RealS{\left\{X^n_{1\oG}\right\}}\big)^T\right\},\\ {\Kmat}^{II}_{{\oX}_{1\oG}}&=&\E\left\{\ImagS\left\{X^n_{1\oG}\right\}  \cdot\big(\ImagS{\left\{X^n_{1\oG}\right\}}\big)^T\right\},
\end{eqnarray*}
and recall the definitions  ${\Umat}_{31}=\left[ \begin{array}{cc}
{\Cmat}_{31} & {-\Smat}_{31} \\
{\Smat}_{31} & {\Cmat}_{31} \end{array}
\right]$ and $\ {\Umat}_{11}=\left[ \begin{array}{cc}
{\Cmat}_{11} & {-\Smat}_{11} \\
{\Smat}_{11} & {\Cmat}_{11} \end{array}
\right]$. Using these definitions we write
\ifextended
\begin{equation*}
    {\Umat}^T_{31}{\Umat}_{11}=\left[ \begin{array}{cc}
    {\Cmat}_{31} & {\Smat}_{31} \\
    -{\Smat}_{31} & {\Cmat}_{31} \end{array}
    \right]\left[ \begin{array}{cc}
    {\Cmat}_{11} & {-\Smat}_{11} \\
    {\Smat}_{11} & {\Cmat}_{11} \end{array}
    \right]=\left[ \begin{array}{cc}
    {\Cmat}_{31}{\Cmat}_{11}+{\Smat}_{31}{\Smat}_{11} & -{\Cmat}_{31}{\Smat}_{11}+{\Smat}_{31}{\Cmat}_{11} \\
    -{\Smat}_{31}{\Cmat}_{11}+{\Cmat}_{31}{\Smat}_{11} & {\Smat}_{31}{\Smat}_{11}+{\Cmat}_{31{\Cmat}_{11}} \end{array}
    \right].
\end{equation*}
Hence,
\begin{eqnarray*}
    & &\hspace{-0.5cm}{\Umat}^T_{31}{\Umat}_{11}{\Kmat}_{{\oX}_{1\oG}}{\Umat}^T_{11}{\Umat}_{31}\\
    &=&\left[ \begin{array}{cc}
    {\Cmat}_{31}{\Cmat}_{11}+{\Smat}_{31}{\Smat}_{11} & -{\Cmat}_{31}{\Smat}_{11}+{\Smat}_{31}{\Cmat}_{11} \\
    -{\Smat}_{31}{\Cmat}_{11}+{\Cmat}_{31}{\Smat}_{11} & {\Smat}_{31}{\Smat}_{11}+{\Cmat}_{31}{\Cmat}_{11} \end{array}
    \right]\left[ \begin{array}{cc}
    {\Kmat}^{RR}_{{\oX}_{1\oG}} & {\Kmat}^{RI}_{{\oX}_{1\oG}} \\
    {\Kmat}^{IR}_{{\oX}_{1\oG}} & {\Kmat}^{II}_{{\oX}_{1\oG}} \end{array}
    \right]\left[ \begin{array}{cc}
    {\Cmat}_{31}{\Cmat}_{11}+{\Smat}_{31}{\Smat}_{11} & -{\Smat}_{31}{\Cmat}_{11}+{\Cmat}_{31}{\Smat}_{11} \\
    -{\Cmat}_{31}{\Smat}_{11}+{\Smat}_{31}{\Cmat}_{11} & {\Smat}_{31}{\Smat}_{11}+{\Cmat}_{31}{\Cmat}_{11} \end{array}
    \right]\\
    &=&\left[ \begin{array}{cc}
    \left({\Cmat}_{31}{\Cmat}_{11}+{\Smat}_{31}{\Smat}_{11}\right){\Kmat}^{RR}_{{\oX}_{1\oG}}+\left(-{\Cmat}_{31}{\Smat}_{11}+{\Smat}_{31}{\Cmat}_{11}\right){\Kmat}^{IR}_{{\oX}_{1\oG}} & \left({\Cmat}_{31}{\Cmat}_{11}+{\Smat}_{31}{\Smat}_{11}\right){\Kmat}^{RI}_{{\oX}_{1\oG}}+\left(-{\Cmat}_{31}{\Smat}_{11}+{\Smat}_{31}{\Cmat}_{11}\right){\Kmat}^{II}_{{\oX}_{1\oG}} \\
    \left(-{\Smat}_{31}{\Cmat}_{11}+{\Cmat}_{31}{\Smat}_{11}\right){\Kmat}^{RR}_{{\oX}_{1\oG}}+\left({\Smat}_{31}{\Smat}_{11}+{\Cmat}_{31}{\Cmat}_{11}\right){\Kmat}^{IR}_{{\oX}_{1\oG}} & \left(-{\Smat}_{31}{\Cmat}_{11}+{\Cmat}_{31}{\Smat}_{11}\right){\Kmat}^{RI}_{{\oX}_{1\oG}}+\left({\Smat}_{31}{\Smat}_{11}+{\Cmat}_{31}{\Cmat}_{11}\right){\Kmat}^{II}_{{\oX}_{1\oG}} \end{array}
    \right]\\
    & &\left[ \begin{array}{cc}
    {\Cmat}_{31}{\Cmat}_{11}+{\Smat}_{31}{\Smat}_{11} & -{\Smat}_{31}{\Cmat}_{11}+{\Cmat}_{31}{\Smat}_{11} \\
    -{\Cmat}_{31}{\Smat}_{11}+{\Smat}_{31}{\Cmat}_{11} & {\Smat}_{31}{\Smat}_{11}+{\Cmat}_{31}{\Cmat}_{11} \end{array}
    \right]
\end{eqnarray*}
Using the above derivations, we write the matrix $\Gmat$ as a block matrix: $\Gmat={\Umat}^T_{31}{\Umat}_{11}{\Kmat}_{{\oX}_{1\oG}}{\Umat}^T_{11}{\Umat}_{31}=\left[ \begin{array}{cc}
{\Gmat}_{11} & {\Gmat}_{12} \\
{\Gmat}_{21} & {\Gmat}_{22} \end{array}
\right]$, where
\begin{eqnarray*}
    {\Gmat}_{11}&=&\left({\Cmat}_{31}{\Cmat}_{11}+{\Smat}_{31}{\Smat}_{11}\right){\Kmat}^{RR}_{{\oX}_{1\oG}}\left({\Cmat}_{31}{\Cmat}_{11}+{\Smat}_{31}{\Smat}_{11}\right)+\left(-{\Cmat}_{31}{\Smat}_{11}+{\Smat}_{31}{\Cmat}_{11}\right){\Kmat}^{IR}_{{\oX}_{1\oG}}\left({\Cmat}_{31}{\Cmat}_{11}+{\Smat}_{31}{\Smat}_{11}\right)\\
    & & \qquad \qquad +\left({\Cmat}_{31}{\Cmat}_{11}+{\Smat}_{31}{\Smat}_{11}\right){\Kmat}^{RI}_{{\oX}_{1\oG}}\left(-{\Cmat}_{31}{\Smat}_{11}+{\Smat}_{31}{\Cmat}_{11}\right)\\
    & & \qquad \qquad +\left(-{\Cmat}_{31}{\Smat}_{11}+{\Smat}_{31}{\Cmat}_{11}\right){\Kmat}^{II}_{{\oX}_{1\oG}}\left(-{\Cmat}_{31}{\Smat}_{11}+{\Smat}_{31}{\Cmat}_{11}\right)\\
    {\Gmat}_{22}&=&\left(-{\Smat}_{31}{\Cmat}_{11}+{\Cmat}_{31}{\Smat}_{11}\right){\Kmat}^{RR}_{{\oX}_{1\oG}}\left(-{\Smat}_{31}{\Cmat}_{11}+{\Cmat}_{31}{\Smat}_{11}\right)+\left({\Smat}_{31}{\Smat}_{11}+{\Cmat}_{31}{\Cmat}_{11}\right){\Kmat}^{IR}_{{\oX}_{1\oG}}\left(-{\Smat}_{31}{\Cmat}_{11}+{\Cmat}_{31}{\Smat}_{11}\right)\\
    & & \qquad \qquad +\left(-{\Smat}_{31}{\Cmat}_{11}+{\Cmat}_{31}{\Smat}_{11}\right){\Kmat}^{RI}_{{\oX}_{1\oG}}\left({\Smat}_{31}{\Smat}_{11}+{\Cmat}_{31}{\Cmat}_{11}\right)\\
    & & \qquad \qquad +\left({\Smat}_{31}{\Smat}_{11}+{\Cmat}_{31}{\Cmat}_{11}\right){\Kmat}^{II}_{{\oX}_{1\oG}}\left({\Smat}_{31}{\Smat}_{11}+{\Cmat}_{31}{\Cmat}_{11}\right)
		\end{eqnarray*}
	\begin{eqnarray*}
    {\Gmat}_{21}&=&\left(-{\Smat}_{31}{\Cmat}_{11}+{\Cmat}_{31}{\Smat}_{11}\right){\Kmat}^{RR}_{{\oX}_{1\oG}}\left({\Cmat}_{31}{\Cmat}_{11}+{\Smat}_{31}{\Smat}_{11}\right)+\left({\Smat}_{31}{\Smat}_{11}+{\Cmat}_{31}{\Cmat}_{11}\right){\Kmat}^{IR}_{{\oX}_{1\oG}}\left({\Cmat}_{31}{\Cmat}_{11}+{\Smat}_{31}{\Smat}_{11}\right)\\
    & & \qquad \qquad +\left(-{\Smat}_{31}{\Cmat}_{11}+{\Cmat}_{31}{\Smat}_{11}\right){\Kmat}^{RI}_{{\oX}_{1\oG}}\left(-{\Cmat}_{31}{\Smat}_{11}+{\Smat}_{31}{\Cmat}_{11}\right)\\
    & & \qquad \qquad +\left({\Smat}_{31}{\Smat}_{11}+{\Cmat}_{31}{\Cmat}_{11}\right){\Kmat}^{II}_{{\oX}_{1\oG}}\left(-{\Cmat}_{31}{\Smat}_{11}+{\Smat}_{31}{\Cmat}_{11}\right)\\
    {\Gmat}_{12}&=&\left({\Cmat}_{31}{\Cmat}_{11}+{\Smat}_{31}{\Smat}_{11}\right){\Kmat}^{RR}_{{\oX}_{1\oG}}\left(-{\Smat}_{31}{\Cmat}_{11}+{\Cmat}_{31}{\Smat}_{11}\right)+\left(-{\Cmat}_{31}{\Smat}_{11}+{\Smat}_{31}{\Cmat}_{11}\right){\Kmat}^{IR}_{{\oX}_{1\oG}}\left(-{\Smat}_{31}{\Cmat}_{11}+{\Cmat}_{31}{\Smat}_{11}\right)\\
    & & \qquad \qquad +\left({\Cmat}_{31}{\Cmat}_{11}+{\Smat}_{31}{\Smat}_{11}\right){\Kmat}^{RI}_{{\oX}_{1\oG}}\left({\Smat}_{31}{\Smat}_{11}+{\Cmat}_{31}{\Cmat}_{11}\right)\\
    & & \qquad \qquad +\left(-{\Cmat}_{31}{\Smat}_{11}+{\Smat}_{31}{\Cmat}_{11}\right){\Kmat}^{II}_{{\oX}_{1\oG}}\left({\Smat}_{31}{\Smat}_{11}+{\Cmat}_{31}{\Cmat}_{11}\right)\\
                &=&-{\Cmat}_{31}{\Cmat}_{11}{\Kmat}^{RR}_{{\oX}_{1\oG}}{\Smat}_{31}{\Cmat}_{11}+{\Cmat}_{31}{\Cmat}_{11}{\Kmat}^{RR}_{{\oX}_{1\oG}}{\Cmat}_{31}{\Smat}_{11}-{\Smat}_{31}{\Smat}_{11}{\Kmat}^{RR}_{{\oX}_{1\oG}}{\Smat}_{31}{\Cmat}_{11}+{\Smat}_{31}{\Smat}_{11}{\Kmat}^{RR}_{{\oX}_{1\oG}}{\Cmat}_{31}{\Smat}_{11}\\
                & & \qquad \qquad +{\Cmat}_{31}{\Smat}_{11}{\Kmat}^{IR}_{{\oX}_{1\oG}}{\Smat}_{31}{\Cmat}_{11}-{\Cmat}_{31}{\Smat}_{11}{\Kmat}^{IR}_{{\oX}_{1\oG}}{\Cmat}_{31}{\Smat}_{11}-{\Smat}_{31}{\Cmat}_{11}{\Kmat}^{IR}_{{\oX}_{1\oG}}{\Smat}_{31}{\Cmat}_{11}\\
                & & \qquad \qquad +{\Smat}_{31}{\Cmat}_{11}{\Kmat}^{IR}_{{\oX}_{1\oG}}{\Cmat}_{31}{\Smat}_{11}+{\Cmat}_{31}{\Cmat}_{11}{\Kmat}^{RI}_{{\oX}_{1\oG}}{\Smat}_{31}{\Smat}_{11}+{\Cmat}_{31}{\Cmat}_{11}{\Kmat}^{RI}_{{\oX}_{1\oG}}{\Cmat}_{31}{\Cmat}_{11}\\
                & & \qquad \qquad +{\Smat}_{31}{\Smat}_{11}{\Kmat}^{RI}_{{\oX}_{1\oG}}{\Smat}_{31}{\Smat}_{11}+{\Smat}_{31}{\Smat}_{11}{\Kmat}^{RI}_{{\oX}_{1\oG}}{\Cmat}_{31}{\Cmat}_{11}-{\Cmat}_{31}{\Smat}_{11}{\Kmat}^{II}_{{\oX}_{1\oG}}{\Smat}_{31}{\Smat}_{11}\\
                & & \qquad \qquad -{\Cmat}_{31}{\Smat}_{11}{\Kmat}^{II}_{{\oX}_{1\oG}}{\Cmat}_{31}{\Cmat}_{11}+{\Smat}_{31}{\Cmat}_{11}{\Kmat}^{II}_{{\oX}_{1\oG}}{\Smat}_{31}{\Smat}_{11}+{\Smat}_{31}{\Cmat}_{11}{\Kmat}^{II}_{{\oX}_{1\oG}}{\Cmat}_{31}{\Cmat}_{11}.
\end{eqnarray*}
\else
 the product as a block matrix $\Gmat={\Umat}^T_{31}{\Umat}_{11}{\Kmat}_{{\oX}_{1\oG}}{\Umat}^T_{11}{\Umat}_{31}=\left[ \begin{array}{cc}
{\Gmat}_{11} & {\Gmat}_{12} \\
{\Gmat}_{21} & {\Gmat}_{22} \end{array}
\right]$.
\fi

To compute ${\bar{\Gmat}}_{12}\triangleq {\E}_{{\tH}^n_1,{\tH}^n_2}\left\{{\Gmat}_{12}\right\}$, where averaging is carried out over the phase variables $\left\{\theta_{31,l}, \theta_{11,l}\right\}_{l=1}^n$,
\ifextended
we note that unless a term has two identical matrices in the product, its expectation is necessarily zero. We first write explicitly
\else
we first write explicitly
\fi
\begin{eqnarray*}
\ifextended
    {\bar{\Gmat}}_{12} &\triangleq & {\E}_{{\tH}^n_1,{\tH}^n_2}\left\{{\Gmat}_{12}\right\}\\
\else
{\bar{\Gmat}}_{12}
\fi
    &=&{\E}_{{\tH}^n_1,{\tH}^n_2}\left\{-{\Cmat}_{31}{\Smat}_{11}{\Kmat}^{IR}_{{\oX}_{1\oG}}{\Cmat}_{31}{\Smat}_{11}-{\Smat}_{31}{\Cmat}_{11}{\Kmat}^{IR}_{{\oX}_{1\oG}}{\Smat}_{31}{\Cmat}_{11}+{\Cmat}_{31}{\Cmat}_{11}{\Kmat}^{RI}_{{\oX}_{1\oG}}{\Cmat}_{31}{\Cmat}_{11}+{\Smat}_{31}{\Smat}_{11}{\Kmat}^{RI}_{{\oX}_{1\oG}}{\Smat}_{31}{\Smat}_{11}\right\}.
\end{eqnarray*}
Next, we note that the products of the matrices ${\Cmat}_{lm}$ and ${\Smat}_{uv}$ matrices are diagonal matrices.
Thus, for $l\neq m$ ${\left[{\bar{\Gmat}}_{12}\right]}_{l,m}$ correspond to averaging over phases from different times and are thus zero.
It therefore remains to consider the diagonal elements of ${\bar{\Gmat}}_{12}$:
\begin{eqnarray*}
    {\left[{\bar{\Gmat}}_{12}\right]}_{l,l}&=&{\E}_{{\tH}^n_1,{\tH}^n_2}\Bigg\{-{\cos \left({\theta }_{31,l}\right)\ }{\sin \left({\theta }_{11,l}\right)\ }{\left[{\Kmat}^{IR}_{{\oX}_{1\oG}}\right]}_{l,l}{\cos \left({\theta }_{31,l}\right)\ }{\sin \left({\theta }_{11,l}\right)\ }\\
        & &\qquad\qquad\qquad\qquad-{\sin \left({\theta }_{31,l}\right)\ }{\cos \left({\theta }_{11,l}\right)\ }{\left[{\Kmat}^{IR}_{{\oX}_{1\oG}}\right]}_{l,l}{\sin \left({\theta }_{31,l}\right)\ }{\cos \left({\theta }_{11,l}\right)\ }\\
        & &\qquad\qquad\qquad\qquad+{\cos \left({\theta }_{31,l}\right)\ }{\cos \left({\theta }_{11,l}\right)\ }{\left[{\Kmat}^{RI}_{{\oX}_{1\oG}}\right]}_{l,l}{\cos \left({\theta }_{31,l}\right)\ }{\cos \left({\theta }_{11,l}\right)\ }\\
        & &\qquad\qquad\qquad\qquad+{\sin \left({\theta }_{31,l}\right)\ }{\sin \left({\theta }_{11,l}\right)\ }{\left[{\Kmat}^{RI}_{{\oX}_{1\oG}}\right]}_{l,l}{\sin \left({\theta }_{31,l}\right)\ }{\sin \left({\theta }_{11,l}\right)\ }\Bigg\}\\
    %%%%%%%%%%%%%%%%%%%%%%%%
    &=&{\E}_{{\tH}^n_1,{\tH}^n_2}\Bigg\{-{\cos^{2} \left({\theta }_{31,l}\right)\ }{\sin^{2} \left({\theta }_{11,l}\right)\ }{\left[{\Kmat}^{IR}_{{\oX}_{1\oG}}\right]}_{l,l}-{\sin^{2} \left({\theta }_{31,l}\right)\ }{\cos^{2} \left({\theta }_{11,l}\right)\ }{\left[{\Kmat}^{IR}_{{\oX}_{1\oG}}\right]}_{l,l}\\
    & &\qquad\qquad\qquad\qquad  +{\cos^{2} \left({\theta }_{31,l}\right)\ }{\cos^{2} \left({\theta }_{11,l}\right)\ }{\left[{\Kmat}^{RI}_{{\oX}_{1\oG}}\right]}_{l,l}+{\sin^{2} \left({\theta }_{31,l}\right)\ }{\sin^{2} \left({\theta }_{11,l}\right)\ }{\left[{\Kmat}^{RI}_{{\oX}_{1\oG}}\right]}_{l,l}\Bigg\}\\
    %%%%%%%%%%%%%%%%%%%%%%%%
    &=&-\frac{1}{4}{\left[{\Kmat}^{IR}_{{\oX}_{1\oG}}\right]}_{l,l}-\frac{1}{4}{\left[{\Kmat}^{IR}_{{\oX}_{1\oG}}\right]}_{l,l}+\frac{1}{4}{\left[{\Kmat}^{RI}_{{\oX}_{1\oG}}\right]}_{l,l}+\frac{1}{4}{\left[{\Kmat}^{RI}_{{\oX}_{1\oG}}\right]}_{l,l}.
\end{eqnarray*}
Lastly, we note that ${\Kmat}^{IR}_{{\oX}_{1\oG}}={\left({\Kmat}^{RI}_{{\oX}_{1\oG}}\right)}^T$, hence ${\left[{\Kmat}^{IR}_{{\oX}_{1\oG}}\right]}_{l,l}={\left[{\Kmat}^{RI}_{{\oX}_{1\oG}}\right]}_{l,l}$, and consequently ${\left[{\bar{\Gmat}}_{12}\right]}_{l,l}=0$. We conclude that ${\bar{\Gmat}}_{12}=\Omat_{n\times n}$. Similarly, we obtain ${\bar{\Gmat}}_{21}\triangleq {\E}_{{\tH}^n_1,{\tH}^n_2}\left\{{\Gmat}_{21}\right\}=\Omat_{n\times n}$.
\ifextended
Next, we compute
${\bar{\Gmat}}_{11}\triangleq {\E}_{{\tH}^n_1,{\tH}^n_2}\left\{{\Gmat}_{11}\right\}$: First, we write explicitly:
\begin{eqnarray*}
    {\Gmat}_{11}&=&\left({\Cmat}_{31}{\Cmat}_{11}+{\Smat}_{31}{\Smat}_{11}\right){\Kmat}^{RR}_{{\oX}_{1\oG}}\left({\Cmat}_{31}{\Cmat}_{11}+{\Smat}_{31}{\Smat}_{11}\right)+\left(-{\Cmat}_{31}{\Smat}_{11}+{\Smat}_{31}{\Cmat}_{11}\right){\Kmat}^{IR}_{{\oX}_{1\oG}}\left({\Cmat}_{31}{\Cmat}_{11}+{\Smat}_{31}{\Smat}_{11}\right)\\
    & & \qquad \qquad +\left({\Cmat}_{31}{\Cmat}_{11}+{\Smat}_{31}{\Smat}_{11}\right){\Kmat}^{RI}_{{\oX}_{1\oG}}\left(-{\Cmat}_{31}{\Smat}_{11}+{\Smat}_{31}{\Cmat}_{11}\right)\\
    & & \qquad \qquad +\left(-{\Cmat}_{31}{\Smat}_{11}+{\Smat}_{31}{\Cmat}_{11}\right){\Kmat}^{II}_{{\oX}_{1\oG}}\left(-{\Cmat}_{31}{\Smat}_{11}+{\Smat}_{31}{\Cmat}_{11}\right)\\
                &=&{\Cmat}_{31}{\Cmat}_{11}{\Kmat}^{RR}_{{\oX}_{1\oG}}{\Cmat}_{31}{\Cmat}_{11}+{\Cmat}_{31}{\Cmat}_{11}{\Kmat}^{RR}_{{\oX}_{1\oG}}{\Smat}_{31}{\Smat}_{11}+{\Smat}_{31}{\Smat}_{11}{\Kmat}^{RR}_{{\oX}_{1\oG}}{\Cmat}_{31}{\Cmat}_{11}+{\Smat}_{31}{\Smat}_{11}{\Kmat}^{RR}_{{\oX}_{1\oG}}{\Smat}_{31}{\Smat}_{11}\\
                & & \qquad \qquad -{\Cmat}_{31}{\Smat}_{11}{\Kmat}^{IR}_{{\oX}_{1\oG}}{\Cmat}_{31}{\Cmat}_{11}-{\Cmat}_{31}{\Smat}_{11}{\Kmat}^{IR}_{{\oX}_{1\oG}}{\Smat}_{31}{\Smat}_{11}+{\Smat}_{31}{\Cmat}_{11}{\Kmat}^{IR}_{{\oX}_{1\oG}}{\Cmat}_{31}{\Cmat}_{11}\\
                & & \qquad \qquad +{\Smat}_{31}{\Cmat}_{11}{\Kmat}^{IR}_{{\oX}_{1\oG}}{\Smat}_{31}{\Smat}_{11}-{\Cmat}_{31}{\Cmat}_{11}{\Kmat}^{RI}_{{\oX}_{1\oG}}{\Cmat}_{31}{\Smat}_{11}+{\Cmat}_{31}{\Cmat}_{11}{\Kmat}^{RI}_{{\oX}_{1\oG}}{\Smat}_{31}{\Cmat}_{11}\\
                & & \qquad \qquad -{\Smat}_{31}{\Smat}_{11}{\Kmat}^{RI}_{{\oX}_{1\oG}}{\Cmat}_{31}{\Smat}_{11}+{\Smat}_{31}{\Smat}_{11}{\Kmat}^{RI}_{{\oX}_{1\oG}}{\Smat}_{31}{\Cmat}_{11}+{\Cmat}_{31}{\Smat}_{11}{\Kmat}^{II}_{{\oX}_{1\oG}}{\Cmat}_{31}{\Smat}_{11}\\
                & & \qquad \qquad -{\Cmat}_{31}{\Smat}_{11}{\Kmat}^{II}_{{\oX}_{1\oG}}{\Smat}_{31}{\Cmat}_{11}-{\Smat}_{31}{\Cmat}_{11}{\Kmat}^{II}_{{\oX}_{1\oG}}{\Cmat}_{31}{\Smat}_{11}+{\Smat}_{31}{\Cmat}_{11}{\Kmat}^{II}_{{\oX}_{1\oG}}{\Smat}_{31}{\Cmat}_{11}.
\end{eqnarray*}
\else
Finally, we compute
${\bar{\Gmat}}_{11}\triangleq {\E}_{{\tH}^n_1,{\tH}^n_2}\left\{{\Gmat}_{11}\right\}$.
\fi
Again, we note that unless a term has two identical matrices in the product, its expectation is necessarily zero:
\begin{eqnarray*}
\ifextended
    {\bar{\Gmat}}_{11}&\triangleq& {\E}_{{\tH}^n_1,{\tH}^n_2}\left\{{\Gmat}_{11}\right\}\\
\else
    {\bar{\Gmat}}_{11}
\fi
    &=&{\E}_{{\tH}^n_1,{\tH}^n_2}\left\{{\Cmat}_{31}{\Cmat}_{11}{\Kmat}^{RR}_{{\oX}_{1\oG}}{\Cmat}_{31}{\Cmat}_{11}+{\Smat}_{31}{\Smat}_{11}{\Kmat}^{RR}_{{\oX}_{1\oG}}{\Smat}_{31}{\Smat}_{11}+{\Cmat}_{31}{\Smat}_{11}{\Kmat}^{II}_{{\oX}_{1\oG}}{\Cmat}_{31}{\Smat}_{11}+{\Smat}_{31}{\Cmat}_{11}{\Kmat}^{II}_{{\oX}_{1\oG}}{\Smat}_{31}{\Cmat}_{11}\right\}.
\end{eqnarray*}
Similarly to the computation of ${\bar{\Gmat}}_{12}$, we note that the products of the matrices ${\Cmat}_{lm}$ and ${\Smat}_{uv}$ matrices are diagonal matrices. Thus, for $l\neq m$ ${\left[{\bar{\Gmat}}_{11}\right]}_{l,m}$ correspond to averaging over phases from different times and are thus zero. Hence,
it remains to consider the diagonal elements of ${\bar{\Gmat}}_{11}$:
\begin{eqnarray*}
    {\left[{\bar{\Gmat}}_{11}\right]}_{l,l}&=&{\E}_{{\tH}^n_1,{\tH}^n_2}\Big\{{\cos \left({\theta }_{31,l}\right)\ }{\cos \left({\theta }_{11,l}\right)\ }{\left[{\Kmat}^{RR}_{{\oX}_{1\oG}}\right]}_{l,l}{\cos \left({\theta }_{31,l}\right)\ }{\cos \left({\theta }_{11,l}\right)\ }\\
        & &\qquad\qquad\qquad\qquad+{\sin \left({\theta }_{31,l}\right)\ }{\sin \left({\theta }_{11,l}\right)\ }{\left[{\Kmat}^{RR}_{{\oX}_{1\oG}}\right]}_{l,l}{\sin \left({\theta }_{31,l}\right)\ }{\sin \left({\theta }_{11,l}\right)\ }\\
        & &\qquad\qquad\qquad\qquad+{\cos \left({\theta }_{31,l}\right)\ }{\sin \left({\theta }_{11,l}\right)\ }{\left[{\Kmat}^{II}_{{\oX}_{1\oG}}\right]}_{l,l}{\cos \left({\theta }_{31,l}\right)\ }{\sin \left({\theta }_{11,l}\right)\ }\\
        & &\qquad\qquad\qquad\qquad+{\sin \left({\theta }_{31,l}\right)\ }{\cos \left({\theta }_{11,l}\right)\ }{\left[{\Kmat}^{II}_{{\oX}_{1\oG}}\right]}_{l,l}{\sin \left({\theta }_{31,l}\right)\ }{\cos \left({\theta }_{11,l}\right)\ }\Big\}\\
%%%%%%%%%%%%%%%%%%%%
\ifextended
    &=&{\E}_{{\tH}^n_1,{\tH}^n_2}\Big\{{\cos^{2} \left({\theta }_{31,l}\right)\ }{\cos^{2} \left({\theta }_{11,l}\right)\ }{\left[{\Kmat}^{RR}_{{\oX}_{1\oG}}\right]}_{l,l}\\
        & &\qquad\qquad\qquad\qquad+{\sin^{2} \left({\theta }_{31,l}\right)\ }{\sin^{2} \left({\theta }_{11,l}\right)\ }{\left[{\Kmat}^{RR}_{{\oX}_{1\oG}}\right]}_{l,l}\\
        & &\qquad\qquad\qquad\qquad+{\cos^{2} \left({\theta }_{31,l}\right)\ }{\sin^{2} \left({\theta }_{11,l}\right)\ }{\left[{\Kmat}^{II}_{{\oX}_{1\oG}}\right]}_{l,l}\\
        & &\qquad\qquad\qquad\qquad+{\sin^{2} \left({\theta }_{31,l}\right)\ }{\cos^{2} \left({\theta }_{11,l}\right)\ }{\left[{\Kmat}^{II}_{{\oX}_{1\oG}}\right]}_{l,l}\Big\}\\
\fi
%%%%%%%%%%%%%%%%%%%%
    &=&\frac{1}{4}{\left[{\Kmat}^{RR}_{{\oX}_{1\oG}}\right]}_{l,l}+\frac{1}{4}{\left[{\Kmat}^{RR}_{{\oX}_{1\oG}}\right]}_{l,l}+\frac{1}{4}{\left[{\Kmat}^{II}_{{\oX}_{1\oG}}\right]}_{l,l}+\frac{1}{4}{\left[{\Kmat}^{II}_{{\oX}_{1\oG}}\right]}_{l,l}\\
    &=&\frac{1}{2}\left({\left[{\Kmat}^{RR}_{{\oX}_{1\oG}}\right]}_{l,l}+{\left[{\Kmat}^{II}_{{\oX}_{1\oG}}\right]}_{l,l}\right)\\
    &\le& \frac{1}{2},
\end{eqnarray*}
where the last inequality follows from the per-symbol power constraint $P_{k,i}\triangleq\E\big\{|X_{k,i}|^2\big\}\le 1$, $k\in\{1,2,3\}$, and the condition on the optimal covariance matrix
of Lemma \ref{lem:lemma8}.
Following similar steps we obtain that ${\bar{\Gmat}}_{22}$ is a diagonal matrix with non-negative elements, each smaller than $\frac{1}{2}$.

We conclude that ${\E}_{{\tH}^n_1,{\tH}^n_2}\left\{{\Umat}^T_{31}{\Umat}_{11}{\Kmat}_{{\oX}_{1\oG}}{\Umat}^T_{11}{\Umat}_{31}\right\}={\E}_{{\tH}^n_1,{\tH}^n_2}\left\{\left[ \begin{array}{cc}
{\Gmat}_{11} & {\Gmat}_{12} \\
{\Gmat}_{21} & {\Gmat}_{22} \end{array}
\right]\right\}=\bar{\Gmat}$ is a diagonal matrix whose diagonal elements are all non-negative and less than  $\frac{1}{2}$:
\begin{equation*}
    0\le {\left[\bar{\Gmat}\right]}_{k,k}\triangleq {\E}_{{\tH}^n_1,{\tH}^n_2}\left\{{\left[\Gmat\right]}_{k,k}\right\}\le \frac{1}{2}.
\end{equation*}
Lastly, observe that as $\bar{\Gmat}$ is a diagonal matrix, we can write
\begin{eqnarray*}
    \left(\avec_{3\oG,i}\right)^T
    & & {\E}_{{\tH}^n_1,{\tH}^n_2}\left\{{\Umat}^T_{31}{\Umat}_{11}{\Kmat}_{{\oX}_{1\oG}}{\Umat}^T_{11}{\Umat}_{31}\right\}\avec_{3\oG,i}\\
    %%%%%%%%%%%%%%%%%%%%%%%
    & &\qquad \qquad \qquad \qquad \qquad \qquad \equiv {\left(\avec_{3\oG,i}\right)}^T\cdot \bar{\Gmat}\cdot \avec_{3\oG,i}\\
    %%%%%%%%%%%%%%%%%%%%%%%
    & &\qquad \qquad \qquad \qquad \qquad \qquad = \sum^{2n}_{k=1}{\left({\left[{\Amat}_{3\oG}\right]}_{k,i}\right)}^2{\bar{\Gmat}}_{k,k}.
\end{eqnarray*}

Proceeding with the bounding, we define
\begin{equation*}
d_{m,i}\triangleq \sum^{2n}_{k=1}{{\left({\left[{\Amat}_{3\oG}\right]}_{k,i}\right)}^2{\bar{\Gmat}}_{k,k}}.
\end{equation*}
The facts that $0\le{\left[\bar{\Gmat}\right]}_{k,k}\le \frac{1}{2}$ and that ${\Amat}_{3\oG}$ is orthogonal, imply that
\begin{equation*}
    d_{m,i}\triangleq \sum^{2n}_{k=1}{{\left({\left[{\Amat}_{3\oG}\right]}_{k,i}\right)}^2{\bar{\Gmat}}_{k,k}}\le \frac{1}{2}\sum^{2n}_{k=1}{{\left({\left[{\Amat}_{3\oG}\right]}_{k,i}\right)}^2}=\frac{1}{2},
\end{equation*}
and also that $d_{m,i}\ge 0$. Using the above definition of $d_{m,i}$ in \eqref{eqnE:upperbound onE1_stepa}, we can upper bound \eqref{eq:NewAppEeq0} as
\begin{eqnarray*}
& & {\E}_{{\tH}^n_1,{\tH}^n_2}\left\{h\left({\hmat}^{(n)}_{h_{31}}X^n_{3\oG}+{\hmat}^{(n)}_{h_{11}}X^n_{1\oG}+{\eta }_1\cdot W^n_1\big|{\tH}^n_1={\th}^n_1\right)-h\left({\hmat}^{(n)}_{h_{32}}X^n_{3\oG}+V^n_2\big|{\tH}^n_2={\th}^n_2\right)\right\} \\
 & &\qquad\qquad\qquad \le\sum^{2n}_{i=1}{\left({\frac{1}{2}\log \left(\SNR_{11}d_{m,i}\ +\SNR_{31}d_{3\oG,i}+\frac{{\left|{\eta }_1\right|}^2}{2}\right)\ }-\frac{1}{2}{\log \left(\SNR_{32}d_{3\oG,i}+\frac{1-{\left|{\CORRN}_2\right|}^2}{2}\right)\ }\right),}
\end{eqnarray*}
where $0\le d_{m,i}\le \frac{1}{2}\ $ and $\sum^{2n}_{i=1}{d_{3\oG,i}}\equiv \tr\left\{{\Dmat}_{3\oG}\right\}\le n$, $d_{3\oG,i}\ge 0$.

We conclude that we can write the upper bound on \eqref{eq:NewAppEeq0} as follows:
%As ${\left\{d_{3\oG,i}\right\}}^{2n}_{i=1}$ and ${\left\{d_{m,i}\right\}}^{2n}_{i=1}$ satisfy  ${\sum }^{2n}_{i=1}d_{3\oG,i}\le n$ and  $d_{m,i}\le \frac{1}{2}$, respectively, then, given a %realization ${\tH}^n_1,\ {\tH}^n_2$, and defining ${\oN}_1\triangleq \frac{{\left|{\eta }_1\right|}^2}{2}$ and ${\oN}_2\triangleq \frac{1-{\left|{\CORRN}_2\right|}^2}{2}$, we upper bound on
%\begin{equation*}
%  h\left({\Hmat}^{(n)}_{h_{31}}X^n_{3\oG}+{\Hmat}^{(n)}_{h_{11}}X^n_{1\oG}+{\eta }_1\cdot W^n_1\big|{\tH}^n_1,{\th}^n_1\right)-h\left({\Hmat}^{(n)}_{h_{32}}X^n_{3\oG}+V^n_2\big|{\tH}^n_2={\th}^n_2\right)
%\end{equation*}
%as follows:
\begin{eqnarray}
    & &\hspace{-0.8cm} {\E}_{{\tH}^n_1,{\tH}^n_2}\bigg\{h\left({\Hmat}^{(n)}_{h_{31}}X^n_{3\oG}+{\Hmat}^{(n)}_{h_{11}}X^n_{1\oG}+{\eta }_1\cdot W^n_1\big|{\tH}^n_1={\th}^n_1\right)  - h\left({\Hmat}^{(n)}_{h_{32}}X^n_{3\oG}+V^n_2\big|{\tH}^n_2={\th}^n_2\right)\bigg\}\nonumber\\
    \label{eqn:AppE_stmnt_optim}
    & &\le \!\!\!\mathop{\max}_{\left\{\!\substack{{\left\{d_{3\oG,i}\right\}}^{2n}_{i=1},\\{\left\{d_{m,i}\right\}}^{2n}_{i=1}}\!\right\}}\! \sum^{2n}_{i=1}\!\frac{1}{2}\Bigg(\!{\log \left(\SNR_{11}d_{m,i} \!+\!\SNR_{31}d_{3\oG,i}\!+\!\frac{{\left|{\eta }_1\right|}^2}{2}\!\right) }
     -{\log \left(\SNR_{32}d_{3\oG,i}+\frac{1-{\left|{\CORRN}_2\right|}^2}{2}\right)}\!\Bigg),\\
&&\hspace{0.9cm}\mbox{subject to:}\mbox{ \phantom{x}}     \sum^{2n}_{i=1}{d_{3\oG,i}} \le n, \qquad d_{3\oG,i}\ge0,\ \  0\le d_{m,i}\le \frac{1}{2},\ \  i=1,2,\dots ,2n.\nonumber
\end{eqnarray}

\bigskip

\phantomsection
\label{phn:defforappndxE}
Next, defining ${\oN}_1\triangleq \frac{{\left|{\eta }_1\right|}^2}{2}$ and ${\oN}_2\triangleq \frac{1-{\left|{\CORRN}_2\right|}^2}{2}$, we can write
\begin{eqnarray*}
    & &\hspace{-3cm}{\log \left(\SNR_{11}d_{m,i}\ +\SNR_{31}d_{3\oG,i}+{\oN}_1\right) }-{\log \left(\SNR_{32}d_{3\oG,i}+{\oN}_2\right)\ }\\
    &&\qquad\qquad\qquad = {\log \left(\frac{\SNR_{11}d_{m,i}\ +\SNR_{31}d_{3\oG,i}+{\oN}_1}{\SNR_{32}d_{3\oG,i}+{\oN}_2}\right)\ }\ \\
    && \qquad \qquad\qquad=\frac{1}{{\ln 2\ }}{\ln \left(\frac{\SNR_{11}d_{m,i}\ +\SNR_{31}d_{3\oG,i}+{\oN}_1}{\SNR_{32}d_{3\oG,i}+{\oN}_2}\right)}.
\end{eqnarray*}
The function $\frac{{\scriptsize \SNR}_{11} d_{m,i}\ +{\scriptsize \SNR}_{31} d_{3\oG,i}+{\oN}_1}{{\scriptsize \SNR}_{32} d_{3\oG,i}+{\oN}_2}$ is a linear fractional function in $\left(d_{m,i},d_{3\oG,i}\right)$ with positive denominator, and is thus a quasilinear function, see \cite[Example 3.32]{ConvexOpt}. Consequently, depending on $d_{m,i}$, it is either a monotone increasing function of $d_{3\oG,i}$ or a monotone decreasing function of $d_{3\oG,i}$. To find this threshold, we differentiate $\frac{{\scriptsize \SNR}_{11} d_{m,i}\ +{\scriptsize \SNR}_{31} d_{3\oG,i}+{\oN}_1}{{\scriptsize \SNR}_{32} d_{3\oG,i}+{\oN}_2}$ w.r.t $d_{3\oG,i}$:
\begin{eqnarray*}
    & &\hspace{-2cm}\frac{\partial }{\partial d_{3\oG,i}}\left\{\frac{\SNR_{11}d_{m,i}\ +\SNR_{31}d_{3\oG,i}+{\oN}_1}{\SNR_{32}d_{3\oG,i}+{\oN}_2}\right\}\\
    &=&\frac{\SNR_{31}\left(\SNR_{32}d_{3\oG,i}+{\oN}_2\right)-\SNR_{32}\left(\SNR_{11}d_{m,i}\ +\SNR_{31}d_{3\oG,i}+{\oN}_1\right)}{{\left(\SNR_{32}d_{3\oG,i}+{\oN}_2\right)}^2}\\
    &=&\frac{{\oN}_2\cdot \SNR_{31}-\SNR_{32}\SNR_{11}d_{m,i}-\SNR_{32}\cdot {\oN}_1}{{\left(\SNR_{32}d_{3\oG,i}+{\oN}_2\right)}^2}.
\end{eqnarray*}

For  $\frac{{\scriptsize \SNR}_{11}d_{m,i}\ +{\scriptsize \SNR}_{31}d_{3\oG,i}+{\oN}_1}{{\scriptsize \SNR}_{32}d_{3\oG,i}+{\oN}_2}$ to be monotone increasing in $d_{3\oG,i}$, the derivative must be positive. This occurs if
\begin{eqnarray*}
    & &{\oN}_2\cdot \SNR_{31}-\SNR_{32}\SNR_{11}d_{m,i}-\SNR_{32}\cdot {\oN}_1>0.
%    & &\Longrightarrow d_{m,i}<\frac{{\oN}_2\cdot \SNR_{31}-\SNR_{32}\cdot {\oN}_1}{\SNR_{32}\SNR_{11}}\triangleq t_0,
\end{eqnarray*}
%Note that from the assumption on ${\left|{\eta }_1\right|}^2$ and ${\left|{\CORRN}_2\right|}^2$  stated in Eqns. \eqref{eqn:derivation_iii inputs_part9_0} we conclude that $t_0\le 1$.
As the feasible $d_{m,i}$ must satisfy $d_{m,i}\le \frac{1}{2}$, it  directly follows from \eqref{eqn:derivation_iii inputs_part9_2} that indeed for all feasible $d_{m,i}$ it holds that
\begin{equation}
\label{eqnE_conditions for opt}
    \!\!\!\SNRA_{32}|\eta_1|^2\! < \!\SNRA_{31}\big(1\!-\!|\CORRN_2|^2\big)\!-\!2d_{m,i}\SNRA_{32}\SNRA_{11}.
\end{equation}
 Therefore, we conclude that for all feasible $d_{m,i}$, the objective increases with $d_{3\oG,i}$, hence, the objective is concave in the feasible region. It thus follows that there is a unique solution to the optimization problem \eqref{eqn:AppE_stmnt_optim} (see \cite[Chapter 3.4]{ConvexOpt}), and that this solution corresponds to the global maxima.
Next, in order to maximize the upper bound \eqref{eqn:AppE_stmnt_optim},
% and explicitly solve for the maximizing $\left\{d_{3\oG,i}\right\}_{i=1}^{2n}$ and $\left\{d_{m,i}\right\}_{i=1}^{2n}$.
we define ${\mathcal{N}}_2\triangleq \left\{1,2,\dots 2n\right\}$ and rewrite the problem as
\begin{eqnarray*}
    \mathop{\max}_{{\left\{d_{3\oG,i}\right\}}^{2n}_{i=1},{\left\{d_{m,i}\right\}}^{2n}_{i=1}} &&\sum^{2n}_{i=1}{\frac{1}{2}\Big({\log \left(\SNR_{11}d_{m,i}\ +\SNR_{31}d_{3\oG,i}+{\oN}_1\right)\ }-{\log \left(\SNR_{32}d_{3\oG,i}+{\oN}_2\right)\ }\Big)}\ \\
    && \\
   && \hspace{-2cm}\mbox{subject to: \phantom{xx}} \sum^{2n}_{i=1}{d_{3\oG,i}}-n\le 0,\ \ \ \ -d_{3\oG,i}\le 0,\ \ \ -d_{m,i}\le 0,\ \ d_{m,i}-\frac{1}{2}\le 0,\ \ \ i\in {\mathcal{N}}_2.
\end{eqnarray*}
First, note that the objective is monotone increasing with $d_{m,i}$, hence it is maximized by
letting $d_{m,i}=\frac{1}{2},\ \  i\in {\mathcal{N}}_2$. Thus, the optimization problem can be written as
\begin{eqnarray*}
    \mathop{\max}_{{\left\{d_{3\oG,i}\right\}}^{2n}_{i=1}} &&\sum^{2n}_{i=1}{\frac{1}{2}\left({\log \left(\SNR_{11}\frac{1}{2}\ \ +\SNR_{31}d_{3\oG,i}+{\oN}_1\right)\ }-{\log \left(\SNR_{32}d_{3\oG,i}+{\oN}_2\right)\ }\right)}\\
    && \\
     && \hspace{-1.5cm}\mbox{subject to: \phantom{xx}} \sum^{2n}_{i=1}{d_{3\oG,i}}-n\le 0,\ \ \ \ -d_{3\oG,i}\le 0,\ \ \ \ \ i\in {\mathcal{N}}_2.
\end{eqnarray*}

\phantomsection
\label{phn:KKT}
\noindent
The Lagrangian for the above optimization function is
\begin{eqnarray*}
    L\left({\left\{d_{3\oG,i},{\varphi }_{3,i}\right\}}^{2n}_{i=1},{\psi }_3\right)&\triangleq& \sum^{2n}_{i=1}{\frac{1}{2}\Bigg({\log \left(\SNR_{11}\frac{1}{2}\ +\SNR_{31}d_{3\oG,i}+{\oN}_1\right)\ }-{\log \left(\SNR_{32}d_{3\oG,i}+{\oN}_2\right)}\Bigg)}\ \\
    & &\qquad\qquad\qquad\qquad -{\psi }_3\left(\sum^{2n}_{i=1}{d_{3\oG,i}}-n\right)-\sum^{2n}_{i=0}{{\varphi }_{1,i}(-d_{m,i})}.
\end{eqnarray*}
To find the optimal solution, we first write the Karush–-Kuhn–-Tucker (KKT)  necessary conditions for optimality  \cite[Chapter 5.5.3]{ConvexOpt}, and solve for the maximizing
$\left(\left\{d_{3\oG,i},{\varphi }_{3,i}\right\}_{i=1}^{2n},{\psi }_3\right)$. As \eqref{eqnE_conditions for opt} guarantees  concavity of the objective function, the optimal solution is the unique
maximum. The KKT conditions for the above problem are:\begin{subequations}
\begin{eqnarray}
\label{eqn:KKT_diff_final}
    &&\frac{\partial }{\partial d_{3\oG,i}}L\left({\left\{d_{3\oG,i},{\varphi }_{3,i},{\psi }_3\right\}}^{2n}_{i=1}\right)=0\\
\label{eqn:KKT_slack_final}
     &&{\varphi }_{3,i}\left(-d_{3\oG,i}\right)=0,\ \ \ {\varphi }_{3,i}\ge 0, \ \ \ i\in {\mathcal{N}}_2,\\
     &&{\psi }_3\left(\sum^{2n}_{i=1}{d_{3\oG,i}}-n\right)=0,\ \ \ \ \ \ {\psi }_3\ge 0.
\end{eqnarray}
\end{subequations}
Evaluating  \eqref{eqn:KKT_diff_final} explicitly we obtain
\begin{eqnarray*}
    & & \frac{\partial }{\partial d_{3\oG,i}}L\left({\left\{d_{3\oG,i},{\varphi }_{3,i}\right\}}^{2n}_{i=1},{\psi }_3\right)\\
		& &\qquad\qquad  =\frac{1}{{\ln 2\ }}\cdot \ \frac{\SNR_{32}d_{3\oG,i}+{\oN}_2}{\SNR_{11}\frac{1}{2}\ +\SNR_{31}d_{3\oG,i}+{\oN}_1}\frac{{\oN}_2\cdot \SNR_{31}-\SNR_{32}\SNR_{11}\frac{1}{2}-\SNR_{32}\cdot {\oN}_1}{{\left(\SNR_{32}d_{3\oG,i}+{\oN}_2\right)}^2}-{\psi }_3+{\varphi }_{3,i}\\
    & &\qquad\qquad  =\frac{1}{{\ln 2\ }}\cdot \frac{{\oN}_2\cdot \SNR_{31}-\SNR_{32}\SNR_{11}\frac{1}{2}-\SNR_{32}\cdot {\oN}_1}{\left(\SNR_{11}\frac{1}{2}\ +\SNR_{31}d_{3\oG,i}+{\oN}_1\right)\left(\SNR_{32}d_{3\oG,i}+{\oN}_2\right)}-{\psi }_3+{\varphi }_{3,i}\\
		&&\qquad\qquad =0
\end{eqnarray*}
%\begin{equation*}
%    \ \ {\varphi }_{3,i}\left(-d_{3\oG,i}\right)=0,\ \ \ \ \ i\in {\mathcal{N}}_2,\ \ \ \ \ \ \ \ \ {\psi }_3\left(\sum^{2n}_{i=1}{d_{3\oG,i}}-n\right)=0,\ \ \ \ \ \ {\psi }_3\ge 0,\ \ \ %{\varphi }_{3,i}\ge 0,\ \ \ \ \ i\in {\mathcal{N}}_2.
%\end{equation*}
As the objective is a monotone increasing function of $d_{3\oG,i}$,
and the first derivative decreases as $d_{3\oG,i}$ increases,
then $d_{3\oG,i}>0$, and thus, from \eqref{eqn:KKT_slack_final} it follows that ${\varphi }_{3,i}=0$.

From %the assumption on ${\left|{\eta }_1\right|}^2$ and ${\left|{\CORRN}_2\right|}^2$  stated in
Eqn. \eqref{eqnE_conditions for opt},
% $h\left(S^n_1\big|{\tH}^n_1=\th_1^n\right)-h\left(Y^n_2\right|X^n_2,S^n_2,{\tH}^n_2=\th_2^n)$,
it follows that ${\oN}_2\cdot \SNR_{31}-\SNR_{32}\SNR_{11}\frac{1}{2}-\SNR_{32}\cdot {\oN}_1> 0$.
With ${\varphi }_{3,i}=0$ and $d_{3\oG,i}> 0$, $\forall i\in {\mathcal{N}}_2$, we obtain
\begin{equation*}
    {\psi }_3=\frac{1}{{\ln 2\ }}\cdot \frac{{\oN}_2\cdot \SNR_{31}-\SNR_{32}\SNR_{11}\frac{1}{2}-\SNR_{32}\cdot {\oN}_1}{\left(\SNR_{11}\frac{1}{2}\ +\SNR_{31}d_{3\oG,i}+{\oN}_1\right)\left(\SNR_{32}d_{3\oG,i}+{\oN}_2\right)}>0.
\end{equation*}
It remains to see what values of $d_{3\oG,i}$  can be considered. From the above equation we conclude that the denominator has to be a positive constant:
\begin{equation*}
    \left(\SNR_{11}\frac{1}{2}\ +\SNR_{31}d_{3\oG,i}+{\oN}_1\right)\left(\SNR_{32}d_{3\oG,i}+{\oN}_2\right)=C_0.
\end{equation*}
Hence, we obtain the quadratic equation for $d_{3\oG,i}$:
\begin{eqnarray*}
    & &\SNR_{32}\cdot \SNR_{31}\cdot d^2_{3\oG,i}+\left(\SNR_{32}\cdot {\oN}_1+\frac{1}{2}\SNR_{32}\cdot \SNR_{11}+\SNR_{31}\cdot {\oN}_2\right)d_{3\oG,i}\\
& &\qquad\qquad\qquad\qquad\qquad\qquad\qquad\qquad\qquad\qquad\qquad\qquad+{\oN}_2\left(\frac{1}{2}\SNR_{11}\ +{\oN}_1\right)-C_0=0,
\end{eqnarray*}
from which we conclude that since the coefficient of $d_{3\oG,i}$ is positive, then the equation has at most one positive root.
It thus follows that for $i\in {\mathcal{N}}_2$, $d_{3\oG,i}=d_3$ is a constant, and hence, from the sum condition $\sum^{2n}_{i=1}{d_{3\oG,i}}=n$ we conclude that $d_{3\oG,i}=d_3=\frac{1}{2},\ \ i\in {\mathcal{N}}_2$.
The maximal objective is therefore:
\begin{eqnarray}
    & &\sum^{2n}_{i=1}{\frac{1}{2}\Bigg({\log \left(\SNR_{11}\frac{1}{2} +\SNR_{31}\frac{1}{2}+{\oN}_1\right)}-{\log \left(\SNR_{32}\frac{1}{2}+{\oN}_2\right) }\Bigg)}=\nonumber\\
    & &\qquad \qquad \qquad n\bigg({\log \left(\SNR_{11} +\SNR_{31}+{\left|{\eta }_1\right|}^2\right) }-{\log \left(\SNR_{32}+\left(1-{\left|v_2\right|}^2\right)\right) }\bigg).\label{eq:NewAppEmaxObj}
\end{eqnarray}
We note that the maximum of the objective can be achieved with equality by letting ${\oX}^{2n}_{3\oG}$ and ${\oX}^{2n}_{1\oG}$  be mutually independent real Gaussian vectors with i.i.d. elements, each with zero mean and variance of $\frac{1}{2}$. With this assignment ${\Kmat}_{{\oX}_{1\oG}}={\Kmat}_{{\oX}_{3\oG}}=\frac{1}{2}{\Imat}_{2n}$. The resulting complex variables are circularly symmetric complex Normal $X^n_{1\oG},X^n_{3\oG}\sim \mathcal{C}\mathcal{N}\left({\bm 0},{\Imat}_n\right)$.
\ifextended
Accordingly the difference of entropies is
\begin{eqnarray*}
    & & \hspace{-2cm} {\E}_{{\tH}^n_1,{\tH}^n_2}\bigg\{h\left({\hmat}^{(n)}_{h_{31}}X^n_{3\oG}+{\hmat}^{(n)}_{h_{11}}X^n_{1\oG}+{\eta }_1\cdot W^n_1\big|{\tH}^n_1={\th}^n_1\right)-h\left({\hmat}^{(n)}_{h_{32}}X^n_{3\oG}+V^n_2\big|{\tH}^n_2={\th}^n_2\right)\bigg\}\\
    &=&{\E}_{{\tH}^n_1,{\tH}^n_2}\Bigg\{\log\det \left({\hmat}^{(n)}_{h_{31}}{\Imat}_n{\left({\hmat}^{(n)}_{h_{31}}\right)}^H+{\hmat}^{(n)}_{h_{11}}{\Imat}_n{\left({\hmat}^{(n)}_{h_{11}}\right)}^H+{\left|{\eta }_1\right|}^2{\Imat}_n\right)\\
    & &\qquad\qquad\qquad\qquad\qquad\qquad\qquad\qquad\qquad-\log\det \left({\hmat}^{(n)}_{h_{32}}{\Imat}_n{\left({\hmat}^{(n)}_{h_{32}}\right)}^H+\left(1-{\left|v_2\right|}^2\right){\Imat}_n\right)\Bigg\}\\
    &=&{\E}_{{\tH}^n_1,{\tH}^n_2}\bigg\{\log\det \left(\SNR_{31}\cdot {\Imat}_n+\SNR_{11}\cdot {\Imat}_n+{\left|{\eta }_1\right|}^2{\Imat}_n\right)\\
    & &\qquad\qquad\qquad\qquad\qquad\qquad\qquad\qquad\qquad-\log\det \left(\SNR_{3n}\cdot {\Imat}_n+\left(1-{\left|v_2\right|}^2\right){\Imat}_n\right)\bigg\}\\
    &=&\log\det \left(\SNR_{31}\cdot {\Imat}_n+\SNR_{11}\cdot {\Imat}_n+{\left|{\eta }_1\right|}^2{\Imat}_n\right)-\log\det \left(\SNR_{3n}\cdot {\Imat}_n+\left(1-{\left|v_2\right|}^2\right){\Imat}_n\right)\\
    &=&n\Bigg({\log \left(\SNR_{11}\ \ +\SNR_{31}+{\left|{\eta }_1\right|}^2\right)\ }-{\log \left(\SNR_{32}+\Big(1-{\left|v_2\right|}^2\Big)\right)\ }\Bigg),
\end{eqnarray*}
which indeed coincides with the maximum of the objective in \eqref{eq:NewAppEmaxObj}.
\fi

\phantomsection
\label{endE}

\label{pginc_start}

\setboolean{@twoside}{false}


\section*{Supplementary material}

# Appendix : Analytic Derivation of the Integral
# in Step (a) after Eqn.(G .1) via ``Mathematica''

In the following we present the outcome of the Mathematica code used for evaluating the definite integral in Step (a) on Pg .45 of the manuscript after Eqn.(G .1).

```
(*The integrand of Appendix D. Variable 'r' denotes the absolue value of the
  correlation coefficient, variable 'x' denotes theta_1 and variable 'y' denotes theta_2*)
f = Log[2 * π *
      (c3 + ((c1 * c4 + c2 * c4 * r * Cos[x] − c5 * c5 − (c1 + c2 * r * Cos[x]) * (2 * c5 * Cos[y])) /
          (c1 + c2 * r * Cos[x] + c4)))]
```

$\text{Log}\left[2\,\pi\left(\text{c3} + \left(\text{c1 c4} - \text{c5}^2 + \text{c2 c4 r Cos}[x] - 2\,\text{c5 (c1} + \text{c2 r Cos}[x])\,\text{Cos}[y]\right) \big/ (\text{c1} + \text{c4} + \text{c2 r Cos}[x])\right)\right]$

(*————————————————————————————————————————————————*)

(*_Phase 1: Calculating the indefinite_
   _integral. (The result was removed due its length)_*)

```
Integ = Integrate[f, x]
```

(*————————————————————————————————————————————————*)

(*Simplifing integral result − Used later for verification (written as a sum of 15 elements)*)

```
Simplify[Integ]
```

$x \, \text{Log}[2\,\pi] +$

$x \, \text{Log}\left[\left(\text{c1 c3} + \text{c1 c4} + \text{c3 c4} - \text{c5}^2 - 2\,\text{c1 c5 Cos}[y] + \text{c2 r Cos}[x]\,(\text{c3} + \text{c4} - 2\,\text{c5 Cos}[y])\right) \big/\right.$

$\qquad (\text{c1} + \text{c4} + \text{c2 r Cos}[x]) \Big] + x \, \text{Log}\left[-\dfrac{\sqrt{-\text{c1} - \text{c4} - \text{c2 r}}}{\sqrt{-\text{c1} + \text{c4} - \text{c2 r}}} + \text{Tan}\left[\dfrac{x}{2}\right]\right] +$

$i \, \text{Log}\left[\dfrac{\sqrt{\text{c1} + \text{c4} - \text{c2 r}}\ \left(-i + \text{Tan}\left[\frac{x}{2}\right]\right)}{\sqrt{-\text{c1} - \text{c4} - \text{c2 r}}\ -\ i\,\sqrt{\text{c1} + \text{c4} - \text{c2 r}}}\right] \text{Log}\left[-\dfrac{\sqrt{-\text{c1} - \text{c4} - \text{c2 r}}}{\sqrt{\text{c1} + \text{c4} - \text{c2 r}}} + \text{Tan}\left[\dfrac{x}{2}\right]\right] -$

$i \, \text{Log}\left[\dfrac{\sqrt{\text{c1} + \text{c4} - \text{c2 r}}\ \left(i + \text{Tan}\left[\frac{x}{2}\right]\right)}{\sqrt{-\text{c1} - \text{c4} - \text{c2 r}}\ +\ i\,\sqrt{\text{c1} + \text{c4} - \text{c2 r}}}\right] \text{Log}\left[-\dfrac{\sqrt{-\text{c1} - \text{c4} - \text{c2 r}}}{\sqrt{\text{c1} + \text{c4} - \text{c2 r}}} + \text{Tan}\left[\dfrac{x}{2}\right]\right] +$

$x \, \text{Log}\left[\dfrac{\sqrt{-\text{c1} - \text{c4} - \text{c2 r}}}{\sqrt{\text{c1} + \text{c4} - \text{c2 r}}} + \text{Tan}\left[\dfrac{x}{2}\right]\right] +$

$i \, \text{Log}\left[-\dfrac{\sqrt{\text{c1} + \text{c4} - \text{c2 r}}\ \left(-i + \text{Tan}\left[\frac{x}{2}\right]\right)}{\sqrt{-\text{c1} - \text{c4} - \text{c2 r}}\ +\ i\,\sqrt{\text{c1} + \text{c4} - \text{c2 r}}}\right] \text{Log}\left[\dfrac{\sqrt{-\text{c1} - \text{c4} - \text{c2 r}}}{\sqrt{\text{c1} + \text{c4} - \text{c2 r}}} + \text{Tan}\left[\dfrac{x}{2}\right]\right] -$

$i \, \text{Log}\left[-\dfrac{\sqrt{\text{c1} + \text{c4} - \text{c2 r}}\ \left(i + \text{Tan}\left[\frac{x}{2}\right]\right)}{\sqrt{-\text{c1} - \text{c4} - \text{c2 r}}\ -\ i\,\sqrt{\text{c1} + \text{c4} - \text{c2 r}}}\right] \text{Log}\left[\dfrac{\sqrt{-\text{c1} - \text{c4} - \text{c2 r}}}{\sqrt{\text{c1} + \text{c4} - \text{c2 r}}} + \text{Tan}\left[\dfrac{x}{2}\right]\right] -$

$x \, \text{Log}\left[-\left(\sqrt{\left(-\text{c3 c4} - \text{c1 (c3} + \text{c4}) + \text{c5}^2 - \text{c2 c3 r} - \text{c2 c4 r} + 2\,\text{c5 (c1} + \text{c2 r})\,\text{Cos}[y]\right)}\right) \big/\right.$

$\qquad \left(\sqrt{\left(\text{c3 c4} + \text{c1 (c3} + \text{c4}) - \text{c5}^2 - \text{c2 c3 r} - \text{c2 c4 r} - 2\,\text{c5 (c1} - \text{c2 r})\,\text{Cos}[y]\right)}\right) + \text{Tan}\left[\dfrac{x}{2}\right]\right] +$

$$i \, \text{Log}\Big[\Big(\sqrt{\big(c3\,c4 + c1\,(c3+c4) - c5^2 - c2\,c3\,r - c2\,c4\,r - 2\,c5\,(c1 - c2\,r)\big)\cos[y]}$$

$$\Big(1 - i\,\text{Tan}\Big[\tfrac{x}{2}\Big]\Big)\Big)\Big/$$

$$\Big(\sqrt{\big(c3\,c4 + c1\,(c3+c4) - c5^2 - c2\,c3\,r - c2\,c4\,r - 2\,c5\,(c1 - c2\,r)\big)\cos[y]} -$$

$$i\,\sqrt{\big(-c3\,c4 - c1\,(c3+c4) + c5^2 - c2\,c3\,r - c2\,c4\,r + 2\,c5\,(c1 + c2\,r)\big)\cos[y]}\Big)\Big]$$

$$\text{Log}\Big[-\Big(\sqrt{\big(-c3\,c4 - c1\,(c3+c4) + c5^2 - c2\,c3\,r - c2\,c4\,r + 2\,c5\,(c1 + c2\,r)\big)\cos[y]}\Big)\Big/$$

$$\Big(\sqrt{\big(c3\,c4 + c1\,(c3+c4) - c5^2 - c2\,c3\,r - c2\,c4\,r - 2\,c5\,(c1 - c2\,r)\big)\cos[y]}\Big) + \text{Tan}\Big[\tfrac{x}{2}\Big]\Big] -$$

$$i \, \text{Log}\Big[\Big(\sqrt{\big(c3\,c4 + c1\,(c3+c4) - c5^2 - c2\,c3\,r - c2\,c4\,r - 2\,c5\,(c1 - c2\,r)\big)\cos[y]}$$

$$\Big(1 + i\,\text{Tan}\Big[\tfrac{x}{2}\Big]\Big)\Big)\Big/$$

$$\Big(\sqrt{\big(c3\,c4 + c1\,(c3+c4) - c5^2 - c2\,c3\,r - c2\,c4\,r - 2\,c5\,(c1 - c2\,r)\big)\cos[y]} +$$

$$i\,\sqrt{\big(-c3\,c4 - c1\,(c3+c4) + c5^2 - c2\,c3\,r - c2\,c4\,r + 2\,c5\,(c1 + c2\,r)\big)\cos[y]}\Big)\Big]$$

$$\text{Log}\Big[-\Big(\sqrt{\big(-c3\,c4 - c1\,(c3+c4) + c5^2 - c2\,c3\,r - c2\,c4\,r + 2\,c5\,(c1 + c2\,r)\big)\cos[y]}\Big)\Big/$$

$$\Big(\sqrt{\big(c3\,c4 + c1\,(c3+c4) - c5^2 - c2\,c3\,r - c2\,c4\,r - 2\,c5\,(c1 - c2\,r)\big)\cos[y]}\Big) + \text{Tan}\Big[\tfrac{x}{2}\Big]\Big] -$$

$$x \, \text{Log}\Big[\Big(\sqrt{\big(-c3\,c4 - c1\,(c3+c4) + c5^2 - c2\,c3\,r - c2\,c4\,r + 2\,c5\,(c1 + c2\,r)\big)\cos[y]}\Big)\Big/$$

$$\Big(\sqrt{\big(c3\,c4 + c1\,(c3+c4) - c5^2 - c2\,c3\,r - c2\,c4\,r - 2\,c5\,(c1 - c2\,r)\big)\cos[y]}\Big) + \text{Tan}\Big[\tfrac{x}{2}\Big]\Big] +$$

$$i \, \text{Log}\Big[\Big(\sqrt{\big(c3\,c4 + c1\,(c3+c4) - c5^2 - c2\,c3\,r - c2\,c4\,r - 2\,c5\,(c1 - c2\,r)\big)\cos[y]}$$

$$\Big(1 - i\,\text{Tan}\Big[\tfrac{x}{2}\Big]\Big)\Big)\Big/$$

$$\Big(\sqrt{\big(c3\,c4 + c1\,(c3+c4) - c5^2 - c2\,c3\,r - c2\,c4\,r - 2\,c5\,(c1 - c2\,r)\big)\cos[y]} +$$

$$i\,\sqrt{\big(-c3\,c4 - c1\,(c3+c4) + c5^2 - c2\,c3\,r - c2\,c4\,r + 2\,c5\,(c1 + c2\,r)\big)\cos[y]}\Big)\Big]$$

$$\text{Log}\Big[\Big(\sqrt{\big(-c3\,c4 - c1\,(c3+c4) + c5^2 - c2\,c3\,r - c2\,c4\,r + 2\,c5\,(c1 + c2\,r)\big)\cos[y]}\Big)\Big/$$

$$\Big(\sqrt{\big(c3\,c4 + c1\,(c3+c4) - c5^2 - c2\,c3\,r - c2\,c4\,r - 2\,c5\,(c1 - c2\,r)\big)\cos[y]}\Big) + \text{Tan}\Big[\tfrac{x}{2}\Big]\Big] -$$

$$i \, \text{Log}\Big[\Big(\sqrt{\big(c3\,c4 + c1\,(c3+c4) - c5^2 - c2\,c3\,r - c2\,c4\,r - 2\,c5\,(c1 - c2\,r)\big)\cos[y]}$$

$$\Big(1 + i\,\text{Tan}\Big[\tfrac{x}{2}\Big]\Big)\Big)\Big/$$

$$\Big(\sqrt{\big(c3\,c4 + c1\,(c3+c4) - c5^2 - c2\,c3\,r - c2\,c4\,r - 2\,c5\,(c1 - c2\,r)\big)\cos[y]} -$$

$$i\,\sqrt{\big(-c3\,c4 - c1\,(c3+c4) + c5^2 - c2\,c3\,r - c2\,c4\,r + 2\,c5\,(c1 + c2\,r)\big)\cos[y]}\Big)\Big]$$

$$\text{Log}\Big[\Big(\sqrt{\big(-c3\,c4 - c1\,(c3+c4) + c5^2 - c2\,c3\,r - c2\,c4\,r + 2\,c5\,(c1 + c2\,r)\big)\cos[y]}\Big)\Big/$$

$$\Big(\sqrt{\big(c3\,c4 + c1\,(c3+c4) - c5^2 - c2\,c3\,r - c2\,c4\,r - 2\,c5\,(c1 - c2\,r)\big)\cos[y]}\Big) + \text{Tan}\Big[\tfrac{x}{2}\Big]\Big] +$$

$$i \, \text{PolyLog}\Big[2, \frac{\sqrt{-c1 - c4 - c2\,r} - \sqrt{c1 + c4 - c2\,r}\,\text{Tan}\Big[\tfrac{x}{2}\Big]}{\sqrt{-c1 - c4 - c2\,r} - i\,\sqrt{c1 + c4 - c2\,r}}\Big] -$$

$$i \, \text{PolyLog}\Big[2, \frac{\sqrt{-c1 - c4 - c2\,r} - \sqrt{c1 + c4 - c2\,r}\,\text{Tan}\Big[\tfrac{x}{2}\Big]}{\sqrt{-c1 - c4 - c2\,r} + i\,\sqrt{c1 + c4 - c2\,r}}\Big] -$$

$$i \, \text{PolyLog}\Big[2, \frac{\sqrt{-c1 - c4 - c2\,r} + \sqrt{c1 + c4 - c2\,r}\,\text{Tan}\Big[\tfrac{x}{2}\Big]}{\sqrt{-c1 - c4 - c2\,r} - i\,\sqrt{c1 + c4 - c2\,r}}\Big] +$$

$$\texttt{i PolyLog}\Big[2, \frac{\sqrt{-c1 - c4 - c2\ r} + \sqrt{c1 + c4 - c2\ r}\ \texttt{Tan}\big[\frac{x}{2}\big]}{\sqrt{-c1 - c4 - c2\ r} + \texttt{i}\ \sqrt{c1 + c4 - c2\ r}}\Big] +$$

$$\texttt{i PolyLog}\Big[2,$$
$$-\Big(\texttt{i}\ \Big(\sqrt{\big(-c3\ c4 - c1\ (c3 + c4) + c5^2 - c2\ c3\ r - c2\ c4\ r + 2\ c5\ (c1 + c2\ r)\ \texttt{Cos}[y]\big)} - \sqrt{\big(c3\ c4}+$$
$$c1\ (c3 + c4) - c5^2 - c2\ c3\ r - c2\ c4\ r - 2\ c5\ (c1 - c2\ r)\ \texttt{Cos}[y]\big)\ \texttt{Tan}\big[\frac{x}{2}\big]\Big)\Big)\Big/$$
$$\Big(\sqrt{\big(c3\ c4 + c1\ (c3 + c4) - c5^2 - c2\ c3\ r - c2\ c4\ r - 2\ c5\ (c1 - c2\ r)\ \texttt{Cos}[y]\big)} -$$
$$\texttt{i}\ \sqrt{\big(-c3\ c4 - c1\ (c3 + c4) + c5^2 - c2\ c3\ r - c2\ c4\ r + 2\ c5\ (c1 + c2\ r)\ \texttt{Cos}[y]\big)}\Big)\Big] -$$

$$\texttt{i PolyLog}\Big[2, \Big(\texttt{i}\ \Big(\sqrt{\big(-c3\ c4 - c1\ (c3 + c4) + c5^2 - c2\ c3\ r - c2\ c4\ r + 2\ c5\ (c1 + c2\ r)\ \texttt{Cos}[y]\big)} -$$
$$\sqrt{\big(c3\ c4 + c1\ (c3 + c4) - c5^2 - c2\ c3\ r - c2\ c4\ r - 2\ c5\ (c1 - c2\ r)\ \texttt{Cos}[y]\big)}\ \texttt{Tan}\big[\frac{x}{2}\big]\Big)\Big)\Big/$$
$$\Big(\sqrt{\big(c3\ c4 + c1\ (c3 + c4) - c5^2 - c2\ c3\ r - c2\ c4\ r - 2\ c5\ (c1 - c2\ r)\ \texttt{Cos}[y]\big)} +$$
$$\texttt{i}\ \sqrt{\big(-c3\ c4 - c1\ (c3 + c4) + c5^2 - c2\ c3\ r - c2\ c4\ r + 2\ c5\ (c1 + c2\ r)\ \texttt{Cos}[y]\big)}\Big)\Big] - \texttt{i}$$

$$\texttt{PolyLog}\Big[2, -\Big(\texttt{i}\ \Big(\sqrt{\big(-c3\ c4 - c1\ (c3 + c4) + c5^2 - c2\ c3\ r - c2\ c4\ r + 2\ c5\ (c1 + c2\ r)\ \texttt{Cos}[y]\big)} +$$
$$\sqrt{\big(c3\ c4 + c1\ (c3 + c4) - c5^2 - c2\ c3\ r - c2\ c4\ r - 2\ c5\ (c1 - c2\ r)\ \texttt{Cos}[y]\big)}$$
$$\texttt{Tan}\big[\frac{x}{2}\big]\Big)\Big)\Big/$$
$$\Big(\sqrt{\big(c3\ c4 + c1\ (c3 + c4) - c5^2 - c2\ c3\ r - c2\ c4\ r - 2\ c5\ (c1 - c2\ r)\ \texttt{Cos}[y]\big)} -$$
$$\texttt{i}\ \sqrt{\big(-c3\ c4 - c1\ (c3 + c4) + c5^2 - c2\ c3\ r - c2\ c4\ r + 2\ c5\ (c1 + c2\ r)\ \texttt{Cos}[y]\big)}\Big)\Big] +$$

$$\texttt{i PolyLog}\Big[2, \Big(\texttt{i}\ \Big(\sqrt{\big(-c3\ c4 - c1\ (c3 + c4) + c5^2 - c2\ c3\ r - c2\ c4\ r + 2\ c5\ (c1 + c2\ r)\ \texttt{Cos}[y]\big)} +$$
$$\sqrt{\big(c3\ c4 + c1\ (c3 + c4) - c5^2 - c2\ c3\ r - c2\ c4\ r - 2\ c5\ (c1 - c2\ r)\ \texttt{Cos}[y]\big)}\ \texttt{Tan}\big[\frac{x}{2}\big]\Big)\Big)\Big/$$
$$\Big(\sqrt{\big(c3\ c4 + c1\ (c3 + c4) - c5^2 - c2\ c3\ r - c2\ c4\ r - 2\ c5\ (c1 - c2\ r)\ \texttt{Cos}[y]\big)} +$$
$$\texttt{i}\ \sqrt{\big(-c3\ c4 - c1\ (c3 + c4) + c5^2 - c2\ c3\ r - c2\ c4\ r + 2\ c5\ (c1 + c2\ r)\ \texttt{Cos}[y]\big)}\Big)\Big]$$

```
(*--------------------------------------------------------*)
```

## *Phase 2: Checking the calcultion of the indefinite integral. (The result was removed due its length)*

```
Deri = D[Integ, x]
```

```
(*--------------------------------------------------------*)
```

**(∗Checking that indeed the derivative of Integ is equal to the main integrand above∗)**
```
DeriSim = FullSimplify[Deri, TimeConstraint → Infinity]
```

$$\texttt{Log}\Big[\big(2\ \pi\ \big(c1\ c3 + (c1 + c3)\ c4 - c5^2 - 2\ c1\ c5\ \texttt{Cos}[y] + c2\ r\ \texttt{Cos}[x]\ (c3 + c4 - 2\ c5\ \texttt{Cos}[y])\big)\big)\Big/$$
$$(c1 + c4 + c2\ r\ \texttt{Cos}[x])\Big]$$

```
(*--------------------------------------------------------*)
```

(∗*Phase 3: Calculating the value of the definite Integral*
 *(i.e., calculating the expected value of the entropy)*∗)

```
Integ2 = Simplify[((Integ /. x → 2 * π) - (Integ /. x → 0)) / (2 π)]
```

$$\text{Log}[2\,\pi] + \text{Log}\left[-\frac{\sqrt{-c1 - c4 - c2\,r}}{\sqrt{c1 + c4 - c2\,r}}\right] + \text{Log}\left[\frac{\sqrt{-c1 - c4 - c2\,r}}{\sqrt{c1 + c4 - c2\,r}}\right] +$$

$$\text{Log}\left[\frac{1}{c1 + c4 + c2\,r}\,(c3\,c4 + c1\,(c3 + c4) - c5^2 + c2\,c3\,r + c2\,c4\,r - 2\,c5\,(c1 + c2\,r)\,\text{Cos}[y])\right] -$$

$$\text{Log}\left[-\left(\sqrt{-c3\,c4 - c1\,(c3 + c4) + c5^2 - c2\,c3\,r - c2\,c4\,r + 2\,c5\,(c1 + c2\,r)\,\text{Cos}[y]}\right)\right/$$
$$\left(\sqrt{c3\,c4 + c1\,(c3 + c4) - c5^2 - c2\,c3\,r - c2\,c4\,r - 2\,c5\,(c1 - c2\,r)\,\text{Cos}[y]}\right)\right] -$$
$$\text{Log}\left[\left(\sqrt{-c3\,c4 - c1\,(c3 + c4) + c5^2 - c2\,c3\,r - c2\,c4\,r + 2\,c5\,(c1 + c2\,r)\,\text{Cos}[y]}\right)\right/$$
$$\left(\sqrt{c3\,c4 + c1\,(c3 + c4) - c5^2 - c2\,c3\,r - c2\,c4\,r - 2\,c5\,(c1 - c2\,r)\,\text{Cos}[y]}\right)\right]$$

(∗————————————————————————————————————∗)

(∗**This is the simplified result of the expected value**∗)

```
AA = Log[
```
$$2\,\pi * \frac{1}{c1 + c4 - c2\,r}\,(c3\,c4 + c1\,(c3 + c4) - c5^2 - c2\,c3\,r - c2\,c4\,r - 2\,c5\,(c1 - c2\,r)\,\text{Cos}[y])\right]$$

(∗————————————————————————————————————∗)

(∗**The argument of the Log above**∗)

```
Integ3 = 
```
$$\frac{1}{c1 + c4 - c2\,r}\,(c3\,c4 + c1\,(c3 + c4) - c5^2 - c2\,c3\,r - c2\,c4\,r - 2\,c5\,(c1 - c2\,r)\,\text{Cos}[y])$$

$$\frac{1}{c1 + c4 - c2\,r}\,(c3\,c4 + c1\,(c3 + c4) - c5^2 - c2\,c3\,r - c2\,c4\,r - 2\,c5\,(c1 - c2\,r)\,\text{Cos}[y])$$

(∗————————————————————————————————————∗)

(∗**Checking that indeed the derivative of the**
 **integrand is indeed negative and hence rho=0 is optimal**∗)

```
Der2 = Simplify[D[Integ3, r]]
```

$$-\frac{c2\,\left(c4^2 + c5^2 - 2\,c4\,c5\,\text{Cos}[y]\right)}{(c1 + c4 - c2\,r)^2}$$

(∗*Phase 4:*
    *Manual verification of the indefinite integral (Second verification):*
    *calculating the derivative of the integral and checking that indeed*
    *it is equal to the integrand. The verification is done in 3 steps*∗)

(∗————————————————————————————————————————————————∗)

(∗**Step 1: Splitting the integral result into a sum of 15 elements: f1,f2,...,f15**∗)

**f1 = x Log[2 π]**

x Log[2 π]

**f2 =**
   **+x Log$\left[\left(\text{c1 c3} + \text{c1 c4} + \text{c3 c4} - \text{c5}^2 - 2\,\text{c1 c5 Cos}[y] + \text{c2 r Cos}[x]\,(\text{c3} + \text{c4} - 2\,\text{c5 Cos}[y])\right) \big/ \right.$**
      **$\left.(\text{c1} + \text{c4} + \text{c2 r Cos}[x])\right]$**

x Log$\left[\left(\text{c1 c3} + \text{c1 c4} + \text{c3 c4} - \text{c5}^2 - 2\,\text{c1 c5 Cos}[y] + \text{c2 r Cos}[x]\,(\text{c3} + \text{c4} - 2\,\text{c5 Cos}[y])\right) \big/\right.$
   $\left.(\text{c1} + \text{c4} + \text{c2 r Cos}[x])\right]$

**f3 = +x Log$\left[-\dfrac{\sqrt{-\text{c1} - \text{c4} - \text{c2 r}}}{\sqrt{\text{c1} + \text{c4} - \text{c2 r}}} + \text{Tan}\left[\dfrac{x}{2}\right]\right]$ + x Log$\left[\dfrac{\sqrt{-\text{c1} - \text{c4} - \text{c2 r}}}{\sqrt{\text{c1} + \text{c4} - \text{c2 r}}} + \text{Tan}\left[\dfrac{x}{2}\right]\right]$**

x Log$\left[-\dfrac{\sqrt{-\text{c1} - \text{c4} - \text{c2 r}}}{\sqrt{\text{c1} + \text{c4} - \text{c2 r}}} + \text{Tan}\left[\dfrac{x}{2}\right]\right]$ + x Log$\left[\dfrac{\sqrt{-\text{c1} - \text{c4} - \text{c2 r}}}{\sqrt{\text{c1} + \text{c4} - \text{c2 r}}} + \text{Tan}\left[\dfrac{x}{2}\right]\right]$

**f4 = +i Log$\left[\dfrac{\sqrt{\text{c1} + \text{c4} - \text{c2 r}}\,\left(-i + \text{Tan}\left[\frac{x}{2}\right]\right)}{\sqrt{-\text{c1} - \text{c4} - \text{c2 r}} - i\,\sqrt{\text{c1} + \text{c4} - \text{c2 r}}}\right]$ Log$\left[-\dfrac{\sqrt{-\text{c1} - \text{c4} - \text{c2 r}}}{\sqrt{\text{c1} + \text{c4} - \text{c2 r}}} + \text{Tan}\left[\dfrac{x}{2}\right]\right]$ +**

   **i Log$\left[-\dfrac{\sqrt{\text{c1} + \text{c4} - \text{c2 r}}\,\left(-i + \text{Tan}\left[\frac{x}{2}\right]\right)}{\sqrt{-\text{c1} - \text{c4} - \text{c2 r}} + i\,\sqrt{\text{c1} + \text{c4} - \text{c2 r}}}\right]$ Log$\left[\dfrac{\sqrt{-\text{c1} - \text{c4} - \text{c2 r}}}{\sqrt{\text{c1} + \text{c4} - \text{c2 r}}} + \text{Tan}\left[\dfrac{x}{2}\right]\right]$**

i Log$\left[\dfrac{\sqrt{\text{c1} + \text{c4} - \text{c2 r}}\,\left(-i + \text{Tan}\left[\frac{x}{2}\right]\right)}{\sqrt{-\text{c1} - \text{c4} - \text{c2 r}} - i\,\sqrt{\text{c1} + \text{c4} - \text{c2 r}}}\right]$ Log$\left[-\dfrac{\sqrt{-\text{c1} - \text{c4} - \text{c2 r}}}{\sqrt{\text{c1} + \text{c4} - \text{c2 r}}} + \text{Tan}\left[\dfrac{x}{2}\right]\right]$ +

   i Log$\left[-\dfrac{\sqrt{\text{c1} + \text{c4} - \text{c2 r}}\,\left(-i + \text{Tan}\left[\frac{x}{2}\right]\right)}{\sqrt{-\text{c1} - \text{c4} - \text{c2 r}} + i\,\sqrt{\text{c1} + \text{c4} - \text{c2 r}}}\right]$ Log$\left[\dfrac{\sqrt{-\text{c1} - \text{c4} - \text{c2 r}}}{\sqrt{\text{c1} + \text{c4} - \text{c2 r}}} + \text{Tan}\left[\dfrac{x}{2}\right]\right]$

$$f5 = -i \, Log\left[\frac{\sqrt{c1 + c4 - c2\,r}\,\left(i + Tan\left[\frac{x}{2}\right]\right)}{\sqrt{-c1 - c4 - c2\,r} + i\,\sqrt{c1 + c4 - c2\,r}}\right] Log\left[-\frac{\sqrt{-c1 - c4 - c2\,r}}{\sqrt{c1 + c4 - c2\,r}} + Tan\left[\frac{x}{2}\right]\right] -$$

$$i \, Log\left[-\frac{\sqrt{c1 + c4 - c2\,r}\,\left(i + Tan\left[\frac{x}{2}\right]\right)}{\sqrt{-c1 - c4 - c2\,r} - i\,\sqrt{c1 + c4 - c2\,r}}\right] Log\left[\frac{\sqrt{-c1 - c4 - c2\,r}}{\sqrt{c1 + c4 - c2\,r}} + Tan\left[\frac{x}{2}\right]\right]$$

$$-i \, Log\left[\frac{\sqrt{c1 + c4 - c2\,r}\,\left(i + Tan\left[\frac{x}{2}\right]\right)}{\sqrt{-c1 - c4 - c2\,r} + i\,\sqrt{c1 + c4 - c2\,r}}\right] Log\left[-\frac{\sqrt{-c1 - c4 - c2\,r}}{\sqrt{c1 + c4 - c2\,r}} + Tan\left[\frac{x}{2}\right]\right] -$$

$$i \, Log\left[-\frac{\sqrt{c1 + c4 - c2\,r}\,\left(i + Tan\left[\frac{x}{2}\right]\right)}{\sqrt{-c1 - c4 - c2\,r} - i\,\sqrt{c1 + c4 - c2\,r}}\right] Log\left[\frac{\sqrt{-c1 - c4 - c2\,r}}{\sqrt{c1 + c4 - c2\,r}} + Tan\left[\frac{x}{2}\right]\right]$$

$$f6 = -x$$
$$Log\left[-\left(\sqrt{\left(-c3\,c4 - c1\,(c3 + c4) + c5^2 - c2\,c3\,r - c2\,c4\,r + 2\,c5\,(c1 + c2\,r)\,Cos[y]\right)}\right) \Big/ \left(\sqrt{\left(c3\right.}\right.$$
$$\left.\left. c4 + c1\,(c3 + c4) - c5^2 - c2\,c3\,r - c2\,c4\,r - 2\,c5\,(c1 - c2\,r)\,Cos[y]\right)}\right) + Tan\left[\frac{x}{2}\right]\right] -$$
$$x \, Log\left[\left(\sqrt{\left(-c3\,c4 - c1\,(c3 + c4) + c5^2 - c2\,c3\,r - c2\,c4\,r + 2\,c5\,(c1 + c2\,r)\,Cos[y]\right)}\right) \Big/$$
$$\left(\sqrt{\left(c3\,c4 + c1\,(c3 + c4) - c5^2 - c2\,c3\,r - c2\,c4\,r - 2\,c5\,(c1 - c2\,r)\,Cos[y]\right)}\right) + Tan\left[\frac{x}{2}\right]\right]$$

$$-x \, Log\left[-\left(\sqrt{\left(-c3\,c4 - c1\,(c3 + c4) + c5^2 - c2\,c3\,r - c2\,c4\,r + 2\,c5\,(c1 + c2\,r)\,Cos[y]\right)}\right) \Big/$$
$$\left(\sqrt{\left(c3\,c4 + c1\,(c3 + c4) - c5^2 - c2\,c3\,r - c2\,c4\,r - 2\,c5\,(c1 - c2\,r)\,Cos[y]\right)}\right) + Tan\left[\frac{x}{2}\right]\right] -$$
$$x \, Log\left[\left(\sqrt{\left(-c3\,c4 - c1\,(c3 + c4) + c5^2 - c2\,c3\,r - c2\,c4\,r + 2\,c5\,(c1 + c2\,r)\,Cos[y]\right)}\right) \Big/$$
$$\left(\sqrt{\left(c3\,c4 + c1\,(c3 + c4) - c5^2 - c2\,c3\,r - c2\,c4\,r - 2\,c5\,(c1 - c2\,r)\,Cos[y]\right)}\right) + Tan\left[\frac{x}{2}\right]\right]$$

$$f7 = i \, Log\left[\left(\sqrt{\left(c3\,c4 + c1\,(c3 + c4) - c5^2 - c2\,c3\,r - c2\,c4\,r - 2\,c5\,(c1 - c2\,r)\,Cos[y]\right)}\right.\right.$$
$$\left.\left(1 - i\,Tan\left[\frac{x}{2}\right]\right)\right) \Big/$$
$$\left(\sqrt{\left(c3\,c4 + c1\,(c3 + c4) - c5^2 - c2\,c3\,r - c2\,c4\,r - 2\,c5\,(c1 - c2\,r)\,Cos[y]\right)} - \right.$$
$$\left.\left. i\,\sqrt{\left(-c3\,c4 - c1\,(c3 + c4) + c5^2 - c2\,c3\,r - c2\,c4\,r + 2\,c5\,(c1 + c2\,r)\,Cos[y]\right)}\right)\right]$$
$$Log\left[-\left(\sqrt{\left(-c3\,c4 - c1\,(c3 + c4) + c5^2 - c2\,c3\,r - c2\,c4\,r + 2\,c5\,(c1 + c2\,r)\,Cos[y]\right)}\right) \Big/$$
$$\left(\sqrt{\left(c3\,c4 + c1\,(c3 + c4) - c5^2 - c2\,c3\,r - c2\,c4\,r - 2\,c5\,(c1 - c2\,r)\,Cos[y]\right)}\right) + Tan\left[\frac{x}{2}\right]\right]$$

$$i \, Log\Big[$$
$$\left(\sqrt{\left(c3\,c4 + c1\,(c3 + c4) - c5^2 - c2\,c3\,r - c2\,c4\,r - 2\,c5\,(c1 - c2\,r)\,Cos[y]\right)}\,\left(1 - i\,Tan\left[\frac{x}{2}\right]\right)\right) \Big/$$
$$\left(\sqrt{\left(c3\,c4 + c1\,(c3 + c4) - c5^2 - c2\,c3\,r - c2\,c4\,r - 2\,c5\,(c1 - c2\,r)\,Cos[y]\right)} - \right.$$
$$\left. i\,\sqrt{\left(-c3\,c4 - c1\,(c3 + c4) + c5^2 - c2\,c3\,r - c2\,c4\,r + 2\,c5\,(c1 + c2\,r)\,Cos[y]\right)}\right)\Big]$$
$$Log\left[-\left(\sqrt{\left(-c3\,c4 - c1\,(c3 + c4) + c5^2 - c2\,c3\,r - c2\,c4\,r + 2\,c5\,(c1 + c2\,r)\,Cos[y]\right)}\right) \Big/$$
$$\left(\sqrt{\left(c3\,c4 + c1\,(c3 + c4) - c5^2 - c2\,c3\,r - c2\,c4\,r - 2\,c5\,(c1 - c2\,r)\,Cos[y]\right)}\right) + Tan\left[\frac{x}{2}\right]\right]$$

**f8** $= -i \, \text{Log}\Big[\Big(\sqrt{\big(c3\,c4 + c1\,(c3+c4) - c5^2 - c2\,c3\,r - c2\,c4\,r - 2\,c5\,(c1 - c2\,r)\,\text{Cos}[y]\big)}$

$\Big(1 + i \, \text{Tan}\big[\tfrac{x}{2}\big]\Big)\Big)\Big/$

$\Big(\sqrt{\big(c3\,c4 + c1\,(c3+c4) - c5^2 - c2\,c3\,r - c2\,c4\,r - 2\,c5\,(c1 - c2\,r)\,\text{Cos}[y]\big)} +$

$i \, \sqrt{\big(-c3\,c4 - c1\,(c3+c4) + c5^2 - c2\,c3\,r - c2\,c4\,r + 2\,c5\,(c1 + c2\,r)\,\text{Cos}[y]\big)}\Big)\Big]$

$\text{Log}\Big[-\big(\sqrt{\big(-c3\,c4 - c1\,(c3+c4) + c5^2 - c2\,c3\,r - c2\,c4\,r + 2\,c5\,(c1 + c2\,r)\,\text{Cos}[y]\big)}\big)\Big/$

$\Big(\sqrt{\big(c3\,c4 + c1\,(c3+c4) - c5^2 - c2\,c3\,r - c2\,c4\,r - 2\,c5\,(c1 - c2\,r)\,\text{Cos}[y]\big)} + \text{Tan}\big[\tfrac{x}{2}\big]\Big)\Big]$

$-\,i \, \text{Log}\Big[$

$\Big(\sqrt{\big(c3\,c4 + c1\,(c3+c4) - c5^2 - c2\,c3\,r - c2\,c4\,r - 2\,c5\,(c1 - c2\,r)\,\text{Cos}[y]\big)}\,\Big(1 + i\,\text{Tan}\big[\tfrac{x}{2}\big]\Big)\Big)\Big/$

$\Big(\sqrt{\big(c3\,c4 + c1\,(c3+c4) - c5^2 - c2\,c3\,r - c2\,c4\,r - 2\,c5\,(c1 - c2\,r)\,\text{Cos}[y]\big)} +$

$i \, \sqrt{\big(-c3\,c4 - c1\,(c3+c4) + c5^2 - c2\,c3\,r - c2\,c4\,r + 2\,c5\,(c1 + c2\,r)\,\text{Cos}[y]\big)}\Big)\Big]$

$\text{Log}\Big[-\big(\sqrt{\big(-c3\,c4 - c1\,(c3+c4) + c5^2 - c2\,c3\,r - c2\,c4\,r + 2\,c5\,(c1 + c2\,r)\,\text{Cos}[y]\big)}\big)\Big/$

$\Big(\sqrt{\big(c3\,c4 + c1\,(c3+c4) - c5^2 - c2\,c3\,r - c2\,c4\,r - 2\,c5\,(c1 - c2\,r)\,\text{Cos}[y]\big)} + \text{Tan}\big[\tfrac{x}{2}\big]\Big)\Big]$

**f9** $= i \, \text{Log}\Big[\Big(\sqrt{\big(c3\,c4 + c1\,(c3+c4) - c5^2 - c2\,c3\,r - c2\,c4\,r - 2\,c5\,(c1 - c2\,r)\,\text{Cos}[y]\big)}$

$\Big(1 - i \, \text{Tan}\big[\tfrac{x}{2}\big]\Big)\Big)\Big/$

$\Big(\sqrt{\big(c3\,c4 + c1\,(c3+c4) - c5^2 - c2\,c3\,r - c2\,c4\,r - 2\,c5\,(c1 - c2\,r)\,\text{Cos}[y]\big)} +$

$i \, \sqrt{\big(-c3\,c4 - c1\,(c3+c4) + c5^2 - c2\,c3\,r - c2\,c4\,r + 2\,c5\,(c1 + c2\,r)\,\text{Cos}[y]\big)}\Big)\Big]$

$\text{Log}\Big[\big(\sqrt{\big(-c3\,c4 - c1\,(c3+c4) + c5^2 - c2\,c3\,r - c2\,c4\,r + 2\,c5\,(c1 + c2\,r)\,\text{Cos}[y]\big)}\big)\Big/$

$\Big(\sqrt{\big(c3\,c4 + c1\,(c3+c4) - c5^2 - c2\,c3\,r - c2\,c4\,r - 2\,c5\,(c1 - c2\,r)\,\text{Cos}[y]\big)} + \text{Tan}\big[\tfrac{x}{2}\big]\Big)\Big]$

$i \, \text{Log}\Big[$

$\Big(\sqrt{\big(c3\,c4 + c1\,(c3+c4) - c5^2 - c2\,c3\,r - c2\,c4\,r - 2\,c5\,(c1 - c2\,r)\,\text{Cos}[y]\big)}\,\Big(1 - i\,\text{Tan}\big[\tfrac{x}{2}\big]\Big)\Big)\Big/$

$\Big(\sqrt{\big(c3\,c4 + c1\,(c3+c4) - c5^2 - c2\,c3\,r - c2\,c4\,r - 2\,c5\,(c1 - c2\,r)\,\text{Cos}[y]\big)} +$

$i \, \sqrt{\big(-c3\,c4 - c1\,(c3+c4) + c5^2 - c2\,c3\,r - c2\,c4\,r + 2\,c5\,(c1 + c2\,r)\,\text{Cos}[y]\big)}\Big)\Big]$

$\text{Log}\Big[\big(\sqrt{\big(-c3\,c4 - c1\,(c3+c4) + c5^2 - c2\,c3\,r - c2\,c4\,r + 2\,c5\,(c1 + c2\,r)\,\text{Cos}[y]\big)}\big)\Big/$

$\Big(\sqrt{\big(c3\,c4 + c1\,(c3+c4) - c5^2 - c2\,c3\,r - c2\,c4\,r - 2\,c5\,(c1 - c2\,r)\,\text{Cos}[y]\big)} + \text{Tan}\big[\tfrac{x}{2}\big]\Big)\Big]$

$$f10 = -i \, \text{Log}\Big[\Big(\sqrt{(c3\,c4 + c1\,(c3+c4) - c5^2 - c2\,c3\,r - c2\,c4\,r - 2\,c5\,(c1-c2\,r)\,\text{Cos}[y])}$$

$$\Big(1 + i\,\text{Tan}\Big[\tfrac{x}{2}\Big]\Big)\Big)\Big/$$

$$\Big(\sqrt{(c3\,c4 + c1\,(c3+c4) - c5^2 - c2\,c3\,r - c2\,c4\,r - 2\,c5\,(c1-c2\,r)\,\text{Cos}[y])} -$$

$$i\,\sqrt{(-c3\,c4 - c1\,(c3+c4) + c5^2 - c2\,c3\,r - c2\,c4\,r + 2\,c5\,(c1+c2\,r)\,\text{Cos}[y])}\Big)\Big]$$

$$\text{Log}\Big[\Big(\sqrt{(-c3\,c4 - c1\,(c3+c4) + c5^2 - c2\,c3\,r - c2\,c4\,r + 2\,c5\,(c1+c2\,r)\,\text{Cos}[y])}\Big)\Big/$$

$$\Big(\sqrt{(c3\,c4 + c1\,(c3+c4) - c5^2 - c2\,c3\,r - c2\,c4\,r - 2\,c5\,(c1-c2\,r)\,\text{Cos}[y])}\Big) + \text{Tan}\Big[\tfrac{x}{2}\Big]\Big]\Big]$$

$$- i\,\text{Log}\Big[$$

$$\Big(\sqrt{(c3\,c4 + c1\,(c3+c4) - c5^2 - c2\,c3\,r - c2\,c4\,r - 2\,c5\,(c1-c2\,r)\,\text{Cos}[y])}\,\Big(1 + i\,\text{Tan}\Big[\tfrac{x}{2}\Big]\Big)\Big)\Big/$$

$$\Big(\sqrt{(c3\,c4 + c1\,(c3+c4) - c5^2 - c2\,c3\,r - c2\,c4\,r - 2\,c5\,(c1-c2\,r)\,\text{Cos}[y])} -$$

$$i\,\sqrt{(-c3\,c4 - c1\,(c3+c4) + c5^2 - c2\,c3\,r - c2\,c4\,r + 2\,c5\,(c1+c2\,r)\,\text{Cos}[y])}\Big)\Big]$$

$$\text{Log}\Big[\Big(\sqrt{(-c3\,c4 - c1\,(c3+c4) + c5^2 - c2\,c3\,r - c2\,c4\,r + 2\,c5\,(c1+c2\,r)\,\text{Cos}[y])}\Big)\Big/$$

$$\Big(\sqrt{(c3\,c4 + c1\,(c3+c4) - c5^2 - c2\,c3\,r - c2\,c4\,r - 2\,c5\,(c1-c2\,r)\,\text{Cos}[y])}\Big) + \text{Tan}\Big[\tfrac{x}{2}\Big]\Big]$$

$$f11 = +i \, \text{PolyLog}\Big[2, \frac{\sqrt{-c1-c4-c2\,r} - \sqrt{c1+c4-c2\,r}\,\text{Tan}\Big[\tfrac{x}{2}\Big]}{\sqrt{-c1-c4-c2\,r} - i\,\sqrt{c1+c4-c2\,r}}\Big] -$$

$$i \, \text{PolyLog}\Big[2, \frac{\sqrt{-c1-c4-c2\,r} - \sqrt{c1+c4-c2\,r}\,\text{Tan}\Big[\tfrac{x}{2}\Big]}{\sqrt{-c1-c4-c2\,r} + i\,\sqrt{c1+c4-c2\,r}}\Big] -$$

$$i \, \text{PolyLog}\Big[2, \frac{\sqrt{-c1-c4-c2\,r} + \sqrt{c1+c4-c2\,r}\,\text{Tan}\Big[\tfrac{x}{2}\Big]}{\sqrt{-c1-c4-c2\,r} - i\,\sqrt{c1+c4-c2\,r}}\Big] +$$

$$i \, \text{PolyLog}\Big[2, \frac{\sqrt{-c1-c4-c2\,r} + \sqrt{c1+c4-c2\,r}\,\text{Tan}\Big[\tfrac{x}{2}\Big]}{\sqrt{-c1-c4-c2\,r} + i\,\sqrt{c1+c4-c2\,r}}\Big]$$

$$i \, \text{PolyLog}\Big[2, \frac{\sqrt{-c1-c4-c2\,r} - \sqrt{c1+c4-c2\,r}\,\text{Tan}\Big[\tfrac{x}{2}\Big]}{\sqrt{-c1-c4-c2\,r} - i\,\sqrt{c1+c4-c2\,r}}\Big] -$$

$$i \, \text{PolyLog}\Big[2, \frac{\sqrt{-c1-c4-c2\,r} - \sqrt{c1+c4-c2\,r}\,\text{Tan}\Big[\tfrac{x}{2}\Big]}{\sqrt{-c1-c4-c2\,r} + i\,\sqrt{c1+c4-c2\,r}}\Big] -$$

$$i \, \text{PolyLog}\Big[2, \frac{\sqrt{-c1-c4-c2\,r} + \sqrt{c1+c4-c2\,r}\,\text{Tan}\Big[\tfrac{x}{2}\Big]}{\sqrt{-c1-c4-c2\,r} - i\,\sqrt{c1+c4-c2\,r}}\Big] +$$

$$i \, \text{PolyLog}\Big[2, \frac{\sqrt{-c1-c4-c2\,r} + \sqrt{c1+c4-c2\,r}\,\text{Tan}\Big[\tfrac{x}{2}\Big]}{\sqrt{-c1-c4-c2\,r} + i\,\sqrt{c1+c4-c2\,r}}\Big]$$

$$f12 = \mathbf{i}\, \text{PolyLog}\Big[2,$$
$$-\Big(\mathbf{i}\,\Big(\sqrt{\big(-c3\,c4 - c1\,(c3+c4) + c5^2 - c2\,c3\,r - c2\,c4\,r + 2\,c5\,(c1+c2\,r)\,\text{Cos}[y]\big)} - \sqrt{\big(c3\,c4 +$$
$$c1\,(c3+c4) - c5^2 - c2\,c3\,r - c2\,c4\,r - 2\,c5\,(c1-c2\,r)\,\text{Cos}[y]\big)}\,\text{Tan}\big[\tfrac{x}{2}\big]\Big)\Big)\Big/$$
$$\Big(\sqrt{\big(c3\,c4 + c1\,(c3+c4) - c5^2 - c2\,c3\,r - c2\,c4\,r - 2\,c5\,(c1-c2\,r)\,\text{Cos}[y]\big)} -$$
$$\mathbf{i}\,\sqrt{\big(-c3\,c4 - c1\,(c3+c4) + c5^2 - c2\,c3\,r - c2\,c4\,r + 2\,c5\,(c1+c2\,r)\,\text{Cos}[y]\big)}\Big)\Big]$$

$$\mathbf{i}\,\text{PolyLog}\Big[2, -\Big(\mathbf{i}\,\Big(\sqrt{\big(-c3\,c4 - c1\,(c3+c4) + c5^2 - c2\,c3\,r - c2\,c4\,r + 2\,c5\,(c1+c2\,r)\,\text{Cos}[y]\big)} -$$
$$\sqrt{\big(c3\,c4 + c1\,(c3+c4) - c5^2 - c2\,c3\,r - c2\,c4\,r - 2\,c5\,(c1-c2\,r)\,\text{Cos}[y]\big)}\,\text{Tan}\big[\tfrac{x}{2}\big]\Big)\Big)\Big/$$
$$\Big(\sqrt{\big(c3\,c4 + c1\,(c3+c4) - c5^2 - c2\,c3\,r - c2\,c4\,r - 2\,c5\,(c1-c2\,r)\,\text{Cos}[y]\big)} -$$
$$\mathbf{i}\,\sqrt{\big(-c3\,c4 - c1\,(c3+c4) + c5^2 - c2\,c3\,r - c2\,c4\,r + 2\,c5\,(c1+c2\,r)\,\text{Cos}[y]\big)}\Big)\Big]$$

$$f13 = -\mathbf{i}\,\text{PolyLog}\Big[2,$$
$$\Big(\mathbf{i}\,\Big(\sqrt{\big(-c3\,c4 - c1\,(c3+c4) + c5^2 - c2\,c3\,r - c2\,c4\,r + 2\,c5\,(c1+c2\,r)\,\text{Cos}[y]\big)} - \sqrt{\big(c3\,c4 +$$
$$c1\,(c3+c4) - c5^2 - c2\,c3\,r - c2\,c4\,r - 2\,c5\,(c1-c2\,r)\,\text{Cos}[y]\big)}\,\text{Tan}\big[\tfrac{x}{2}\big]\Big)\Big)\Big/$$
$$\Big(\sqrt{\big(c3\,c4 + c1\,(c3+c4) - c5^2 - c2\,c3\,r - c2\,c4\,r - 2\,c5\,(c1-c2\,r)\,\text{Cos}[y]\big)} +$$
$$\mathbf{i}\,\sqrt{\big(-c3\,c4 - c1\,(c3+c4) + c5^2 - c2\,c3\,r - c2\,c4\,r + 2\,c5\,(c1+c2\,r)\,\text{Cos}[y]\big)}\Big)\Big]$$

$$-\mathbf{i}\,\text{PolyLog}\Big[2, \Big(\mathbf{i}\,\Big(\sqrt{\big(-c3\,c4 - c1\,(c3+c4) + c5^2 - c2\,c3\,r - c2\,c4\,r + 2\,c5\,(c1+c2\,r)\,\text{Cos}[y]\big)} -$$
$$\sqrt{\big(c3\,c4 + c1\,(c3+c4) - c5^2 - c2\,c3\,r - c2\,c4\,r - 2\,c5\,(c1-c2\,r)\,\text{Cos}[y]\big)}\,\text{Tan}\big[\tfrac{x}{2}\big]\Big)\Big)\Big/$$
$$\Big(\sqrt{\big(c3\,c4 + c1\,(c3+c4) - c5^2 - c2\,c3\,r - c2\,c4\,r - 2\,c5\,(c1-c2\,r)\,\text{Cos}[y]\big)} +$$
$$\mathbf{i}\,\sqrt{\big(-c3\,c4 - c1\,(c3+c4) + c5^2 - c2\,c3\,r - c2\,c4\,r + 2\,c5\,(c1+c2\,r)\,\text{Cos}[y]\big)}\Big)\Big]$$

$$f14 = -\mathbf{i}\,\text{PolyLog}\Big[2,$$
$$-\Big(\mathbf{i}\,\Big(\sqrt{\big(-c3\,c4 - c1\,(c3+c4) + c5^2 - c2\,c3\,r - c2\,c4\,r + 2\,c5\,(c1+c2\,r)\,\text{Cos}[y]\big)} + \sqrt{\big(c3\,c4 +$$
$$c1\,(c3+c4) - c5^2 - c2\,c3\,r - c2\,c4\,r - 2\,c5\,(c1-c2\,r)\,\text{Cos}[y]\big)}\,\text{Tan}\big[\tfrac{x}{2}\big]\Big)\Big)\Big/$$
$$\Big(\sqrt{\big(c3\,c4 + c1\,(c3+c4) - c5^2 - c2\,c3\,r - c2\,c4\,r - 2\,c5\,(c1-c2\,r)\,\text{Cos}[y]\big)} -$$
$$\mathbf{i}\,\sqrt{\big(-c3\,c4 - c1\,(c3+c4) + c5^2 - c2\,c3\,r - c2\,c4\,r + 2\,c5\,(c1+c2\,r)\,\text{Cos}[y]\big)}\Big)\Big]$$

$$-\mathbf{i}$$
$$\text{PolyLog}\Big[2, -\Big(\mathbf{i}\,\Big(\sqrt{\big(-c3\,c4 - c1\,(c3+c4) + c5^2 - c2\,c3\,r - c2\,c4\,r + 2\,c5\,(c1+c2\,r)\,\text{Cos}[y]\big)} +$$
$$\sqrt{\big(c3\,c4 + c1\,(c3+c4) - c5^2 - c2\,c3\,r - c2\,c4\,r - 2\,c5\,(c1-c2\,r)\,\text{Cos}[y]\big)}\,\text{Tan}\big[\tfrac{x}{2}\big]\Big)\Big)\Big/$$
$$\Big(\sqrt{\big(c3\,c4 + c1\,(c3+c4) - c5^2 - c2\,c3\,r - c2\,c4\,r - 2\,c5\,(c1-c2\,r)\,\text{Cos}[y]\big)} -$$
$$\mathbf{i}\,\sqrt{\big(-c3\,c4 - c1\,(c3+c4) + c5^2 - c2\,c3\,r - c2\,c4\,r + 2\,c5\,(c1+c2\,r)\,\text{Cos}[y]\big)}\Big)\Big]$$

**f15 = i PolyLog[2,**

$$\left(i\left(\sqrt{\left(-c3\ c4 - c1\ (c3+c4) + c5^2 - c2\ c3\ r - c2\ c4\ r + 2\ c5\ (c1+c2\ r)\ \text{Cos}[y]\right)} + \sqrt{\left(c3\ c4 + c1\ (c3+c4) - c5^2 - c2\ c3\ r - c2\ c4\ r - 2\ c5\ (c1-c2\ r)\ \text{Cos}[y]\right)}\ \text{Tan}\left[\frac{x}{2}\right]\right)\right) \Big/$$

$$\left(\sqrt{\left(c3\ c4 + c1\ (c3+c4) - c5^2 - c2\ c3\ r - c2\ c4\ r - 2\ c5\ (c1-c2\ r)\ \text{Cos}[y]\right)} + i\ \sqrt{\left(-c3\ c4 - c1\ (c3+c4) + c5^2 - c2\ c3\ r - c2\ c4\ r + 2\ c5\ (c1+c2\ r)\ \text{Cos}[y]\right)}\right)\Big]$$

i PolyLog[2, $\left(i\left(\sqrt{\left(-c3\ c4 - c1\ (c3+c4) + c5^2 - c2\ c3\ r - c2\ c4\ r + 2\ c5\ (c1+c2\ r)\ \text{Cos}[y]\right)} + \right.\right.$

$$\left.\left.\sqrt{\left(c3\ c4 + c1\ (c3+c4) - c5^2 - c2\ c3\ r - c2\ c4\ r - 2\ c5\ (c1-c2\ r)\ \text{Cos}[y]\right)}\ \text{Tan}\left[\frac{x}{2}\right]\right)\right) \Big/$$

$$\left(\sqrt{\left(c3\ c4 + c1\ (c3+c4) - c5^2 - c2\ c3\ r - c2\ c4\ r - 2\ c5\ (c1-c2\ r)\ \text{Cos}[y]\right)} + \right.$$
$$\left. i\ \sqrt{\left(-c3\ c4 - c1\ (c3+c4) + c5^2 - c2\ c3\ r - c2\ c4\ r + 2\ c5\ (c1+c2\ r)\ \text{Cos}[y]\right)}\right)\Big]$$

(*————————————————————————————————————————*)

(*Step 2: Calculating the derivative of each of 15 elements —
 The results here were checked manually*)

**ff1 = D[f1, x]**

Log[2 π]

**ff2 = D[f2, x]**

$$\text{Log}\left[\left(c1\ c3 + c1\ c4 + c3\ c4 - c5^2 - 2\ c1\ c5\ \text{Cos}[y] + c2\ r\ \text{Cos}[x]\ (c3+c4-2\ c5\ \text{Cos}[y])\right) \Big/ \right.$$
$$\left.(c1+c4+c2\ r\ \text{Cos}[x])\right] +$$
$$\left(x\ (c1+c4+c2\ r\ \text{Cos}[x])\ \left(-\frac{c2\ r\ (c3+c4-2\ c5\ \text{Cos}[y])\ \text{Sin}[x]}{c1+c4+c2\ r\ \text{Cos}[x]} + \right.\right.$$
$$\left(c2\ r\ \left(c1\ c3 + c1\ c4 + c3\ c4 - c5^2 - 2\ c1\ c5\ \text{Cos}[y] + c2\ r\ \text{Cos}[x]\ (c3+c4-2\ c5\ \text{Cos}[y])\right)\right.$$
$$\left.\left.\left.\text{Sin}[x]\right) \Big/ (c1+c4+c2\ r\ \text{Cos}[x])^2\right)\right) \Big/$$
$$\left(c1\ c3 + c1\ c4 + c3\ c4 - c5^2 - 2\ c1\ c5\ \text{Cos}[y] + c2\ r\ \text{Cos}[x]\ (c3+c4-2\ c5\ \text{Cos}[y])\right)$$

**ff3 = D[f3, x]**

$$\text{Log}\left[-\frac{\sqrt{-c1-c4-c2\ r}}{\sqrt{c1+c4-c2\ r}} + \text{Tan}\left[\frac{x}{2}\right]\right] + \text{Log}\left[\frac{\sqrt{-c1-c4-c2\ r}}{\sqrt{c1+c4-c2\ r}} + \text{Tan}\left[\frac{x}{2}\right]\right] +$$
$$\frac{x\ \text{Sec}\left[\frac{x}{2}\right]^2}{2\left(-\frac{\sqrt{-c1-c4-c2\ r}}{\sqrt{c1+c4-c2\ r}} + \text{Tan}\left[\frac{x}{2}\right]\right)} + \frac{x\ \text{Sec}\left[\frac{x}{2}\right]^2}{2\left(\frac{\sqrt{-c1-c4-c2\ r}}{\sqrt{c1+c4-c2\ r}} + \text{Tan}\left[\frac{x}{2}\right]\right)}$$

**ff4 = D[f4, x]**

$$\frac{i \, \text{Log}\left[-\frac{\sqrt{-c1-c4-c2\, r}}{\sqrt{c1+c4-c2\, r}} + \text{Tan}\left[\frac{x}{2}\right]\right] \text{Sec}\left[\frac{x}{2}\right]^2}{2\left(-i + \text{Tan}\left[\frac{x}{2}\right]\right)} + \frac{i \, \text{Log}\left[\frac{\sqrt{-c1-c4-c2\, r}}{\sqrt{c1+c4-c2\, r}} + \text{Tan}\left[\frac{x}{2}\right]\right] \text{Sec}\left[\frac{x}{2}\right]^2}{2\left(-i + \text{Tan}\left[\frac{x}{2}\right]\right)} +$$

$$\frac{i \, \text{Log}\left[\frac{\sqrt{c1+c4-c2\, r}\,\left(-i+\text{Tan}\left[\frac{x}{2}\right]\right)}{\sqrt{-c1-c4-c2\, r}\,-i\,\sqrt{c1+c4-c2\, r}}\right] \text{Sec}\left[\frac{x}{2}\right]^2}{2\left(-\frac{\sqrt{-c1-c4-c2\, r}}{\sqrt{c1+c4-c2\, r}} + \text{Tan}\left[\frac{x}{2}\right]\right)} + \frac{i \, \text{Log}\left[-\frac{\sqrt{c1+c4-c2\, r}\,\left(-i+\text{Tan}\left[\frac{x}{2}\right]\right)}{\sqrt{-c1-c4-c2\, r}\,+i\,\sqrt{c1+c4-c2\, r}}\right] \text{Sec}\left[\frac{x}{2}\right]^2}{2\left(\frac{\sqrt{-c1-c4-c2\, r}}{\sqrt{c1+c4-c2\, r}} + \text{Tan}\left[\frac{x}{2}\right]\right)}$$

**ff5 = D[f5, x]**

$$-\frac{i \, \text{Log}\left[-\frac{\sqrt{-c1-c4-c2\, r}}{\sqrt{c1+c4-c2\, r}} + \text{Tan}\left[\frac{x}{2}\right]\right] \text{Sec}\left[\frac{x}{2}\right]^2}{2\left(i + \text{Tan}\left[\frac{x}{2}\right]\right)} - \frac{i \, \text{Log}\left[\frac{\sqrt{-c1-c4-c2\, r}}{\sqrt{c1+c4-c2\, r}} + \text{Tan}\left[\frac{x}{2}\right]\right] \text{Sec}\left[\frac{x}{2}\right]^2}{2\left(i + \text{Tan}\left[\frac{x}{2}\right]\right)} -$$

$$\frac{i \, \text{Log}\left[\frac{\sqrt{c1+c4-c2\, r}\,\left(i+\text{Tan}\left[\frac{x}{2}\right]\right)}{\sqrt{-c1-c4-c2\, r}\,+i\,\sqrt{c1+c4-c2\, r}}\right] \text{Sec}\left[\frac{x}{2}\right]^2}{2\left(-\frac{\sqrt{-c1-c4-c2\, r}}{\sqrt{c1+c4-c2\, r}} + \text{Tan}\left[\frac{x}{2}\right]\right)} - \frac{i \, \text{Log}\left[-\frac{\sqrt{c1+c4-c2\, r}\,\left(i+\text{Tan}\left[\frac{x}{2}\right]\right)}{\sqrt{-c1-c4-c2\, r}\,-i\,\sqrt{c1+c4-c2\, r}}\right] \text{Sec}\left[\frac{x}{2}\right]^2}{2\left(\frac{\sqrt{-c1-c4-c2\, r}}{\sqrt{c1+c4-c2\, r}} + \text{Tan}\left[\frac{x}{2}\right]\right)}$$

**ff6 = D[f6, x]**

$$-\text{Log}\left[-\left(\sqrt{\left(-c3\, c4 - c1\,(c3+c4) + c5^2 - c2\, c3\, r - c2\, c4\, r + 2\, c5\,(c1+c2\, r)\, \text{Cos}[y]\right)}\right) \Big/\right.$$

$$\left.\left(\sqrt{\left(c3\, c4 + c1\,(c3+c4) - c5^2 - c2\, c3\, r - c2\, c4\, r - 2\, c5\,(c1-c2\, r)\, \text{Cos}[y]\right)}\right) + \text{Tan}\left[\frac{x}{2}\right]\right] -$$

$$\text{Log}\left[\left(\sqrt{\left(-c3\, c4 - c1\,(c3+c4) + c5^2 - c2\, c3\, r - c2\, c4\, r + 2\, c5\,(c1+c2\, r)\, \text{Cos}[y]\right)}\right) \Big/\right.$$

$$\left.\left(\sqrt{\left(c3\, c4 + c1\,(c3+c4) - c5^2 - c2\, c3\, r - c2\, c4\, r - 2\, c5\,(c1-c2\, r)\, \text{Cos}[y]\right)}\right) + \text{Tan}\left[\frac{x}{2}\right]\right] -$$

$$\frac{x \, \text{Sec}\left[\frac{x}{2}\right]^2}{2\left(-\frac{\sqrt{-c3\, c4 - c1\,(c3+c4) + c5^2 - c2\, c3\, r - c2\, c4\, r + 2\, c5\,(c1+c2\, r)\, \text{Cos}[y]}}{\sqrt{c3\, c4 + c1\,(c3+c4) - c5^2 - c2\, c3\, r - c2\, c4\, r - 2\, c5\,(c1-c2\, r)\, \text{Cos}[y]}} + \text{Tan}\left[\frac{x}{2}\right]\right)} -$$

$$\frac{x \, \text{Sec}\left[\frac{x}{2}\right]^2}{2\left(\frac{\sqrt{-c3\, c4 - c1\,(c3+c4) + c5^2 - c2\, c3\, r - c2\, c4\, r + 2\, c5\,(c1+c2\, r)\, \text{Cos}[y]}}{\sqrt{c3\, c4 + c1\,(c3+c4) - c5^2 - c2\, c3\, r - c2\, c4\, r - 2\, c5\,(c1-c2\, r)\, \text{Cos}[y]}} + \text{Tan}\left[\frac{x}{2}\right]\right)}$$

**ff7 = D[f7, x]**

$$\left(\text{Log}\left[-\left(\sqrt{\left(-c3\;c4 - c1\;(c3+c4) + c5^2 - c2\;c3\;r - c2\;c4\;r + 2\;c5\;(c1+c2\;r)\;\text{Cos}[y]\right)}\right)\Big/\right.\right.$$
$$\left(\sqrt{\left(c3\;c4 + c1\;(c3+c4) - c5^2 - c2\;c3\;r - c2\;c4\;r - 2\;c5\;(c1-c2\;r)\;\text{Cos}[y]\right)}\right) +$$
$$\left.\text{Tan}\left[\frac{x}{2}\right]\right]\;\text{Sec}\left[\frac{x}{2}\right]^2\right)\Big/\left(2\left(1 - i\;\text{Tan}\left[\frac{x}{2}\right]\right)\right) +$$
$$\left(i\;\text{Log}\left[\left(\sqrt{\left(c3\;c4 + c1\;(c3+c4) - c5^2 - c2\;c3\;r - c2\;c4\;r - 2\;c5\;(c1-c2\;r)\;\text{Cos}[y]\right)}\right.\right.\right.$$
$$\left.\left(1 - i\;\text{Tan}\left[\frac{x}{2}\right]\right)\right)\Big/$$
$$\left(\sqrt{\left(c3\;c4 + c1\;(c3+c4) - c5^2 - c2\;c3\;r - c2\;c4\;r - 2\;c5\;(c1-c2\;r)\;\text{Cos}[y]\right)} -\right.$$
$$\left.\left.i\;\sqrt{\left(-c3\;c4 - c1\;(c3+c4) + c5^2 - c2\;c3\;r - c2\;c4\;r + 2\;c5\;(c1+c2\;r)\;\text{Cos}[y]\right)}\right)\right]$$
$$\text{Sec}\left[\frac{x}{2}\right]^2\right)\Big/\left(2\left(-\left(\sqrt{\left(-c3\;c4 - c1\;(c3+c4) + c5^2 - c2\;c3\;r - c2\;c4\;r +\right.}\right.\right.\right.$$
$$\left.2\;c5\;(c1+c2\;r)\;\text{Cos}[y]\right)\Big/\left(\sqrt{\left(c3\;c4 + c1\;(c3+c4) - c5^2 -\right.}\right.$$
$$\left.\left.\left.c2\;c3\;r - c2\;c4\;r - 2\;c5\;(c1-c2\;r)\;\text{Cos}[y]\right)}\right) + \text{Tan}\left[\frac{x}{2}\right]\right)\right)$$

**ff8 = D[f8, x]**

$$\left(\text{Log}\left[-\left(\sqrt{\left(-c3\;c4 - c1\;(c3+c4) + c5^2 - c2\;c3\;r - c2\;c4\;r + 2\;c5\;(c1+c2\;r)\;\text{Cos}[y]\right)}\right)\Big/\right.\right.$$
$$\left(\sqrt{\left(c3\;c4 + c1\;(c3+c4) - c5^2 - c2\;c3\;r - c2\;c4\;r - 2\;c5\;(c1-c2\;r)\;\text{Cos}[y]\right)}\right) +$$
$$\left.\text{Tan}\left[\frac{x}{2}\right]\right]\;\text{Sec}\left[\frac{x}{2}\right]^2\right)\Big/\left(2\left(1 + i\;\text{Tan}\left[\frac{x}{2}\right]\right)\right) -$$
$$\left(i\;\text{Log}\left[\left(\sqrt{\left(c3\;c4 + c1\;(c3+c4) - c5^2 - c2\;c3\;r - c2\;c4\;r - 2\;c5\;(c1-c2\;r)\;\text{Cos}[y]\right)}\right.\right.\right.$$
$$\left(1 + i\;\text{Tan}\left[\frac{x}{2}\right]\right)\right)\Big/$$
$$\left(\sqrt{\left(c3\;c4 + c1\;(c3+c4) - c5^2 - c2\;c3\;r - c2\;c4\;r - 2\;c5\;(c1-c2\;r)\;\text{Cos}[y]\right)} +\right.$$
$$\left.\left.i\;\sqrt{\left(-c3\;c4 - c1\;(c3+c4) + c5^2 - c2\;c3\;r - c2\;c4\;r + 2\;c5\;(c1+c2\;r)\;\text{Cos}[y]\right)}\right)\right]$$
$$\text{Sec}\left[\frac{x}{2}\right]^2\right)\Big/\left(2\left(-\left(\sqrt{\left(-c3\;c4 - c1\;(c3+c4) + c5^2 - c2\;c3\;r - c2\;c4\;r +\right.}\right.\right.\right.$$
$$\left.2\;c5\;(c1+c2\;r)\;\text{Cos}[y]\right)\Big/\left(\sqrt{\left(c3\;c4 + c1\;(c3+c4) - c5^2 -\right.}\right.$$
$$\left.\left.\left.c2\;c3\;r - c2\;c4\;r - 2\;c5\;(c1-c2\;r)\;\text{Cos}[y]\right)}\right) + \text{Tan}\left[\frac{x}{2}\right]\right)\right)$$

**ff9 = D[f9, x]**

$$\Bigg( \text{Log}\Big[ \big(\sqrt{\left(-c3\ c4 - c1\ (c3 + c4) + c5^2 - c2\ c3\ r - c2\ c4\ r + 2\ c5\ (c1 + c2\ r)\ \text{Cos}[y]\right)}\big) \Big/$$
$$\big(\sqrt{\left(c3\ c4 + c1\ (c3 + c4) - c5^2 - c2\ c3\ r - c2\ c4\ r - 2\ c5\ (c1 - c2\ r)\ \text{Cos}[y]\right)}\big) +$$
$$\text{Tan}\Big[\frac{x}{2}\Big]\Big]\ \text{Sec}\Big[\frac{x}{2}\Big]^2\Big) \Big/ \Big(2\ \big(1 - i\ \text{Tan}\Big[\frac{x}{2}\Big]\big)\Big) +$$
$$\Big(i\ \text{Log}\Big[\big(\sqrt{\left(c3\ c4 + c1\ (c3 + c4) - c5^2 - c2\ c3\ r - c2\ c4\ r - 2\ c5\ (c1 - c2\ r)\ \text{Cos}[y]\right)}$$
$$\big(1 - i\ \text{Tan}\Big[\frac{x}{2}\Big]\big)\big)\Big/$$
$$\big(\sqrt{\left(c3\ c4 + c1\ (c3 + c4) - c5^2 - c2\ c3\ r - c2\ c4\ r - 2\ c5\ (c1 - c2\ r)\ \text{Cos}[y]\right)} +$$
$$i\ \sqrt{\left(-c3\ c4 - c1\ (c3 + c4) + c5^2 - c2\ c3\ r - c2\ c4\ r + 2\ c5\ (c1 + c2\ r)\ \text{Cos}[y]\right)}\big)\Big]\ \text{Sec}\Big[$$
$$\frac{x}{2}\Big]^2\Big) \Big/$$
$$\Big(2\ \Big(\big(\sqrt{\left(-c3\ c4 - c1\ (c3 + c4) + c5^2 - c2\ c3\ r - c2\ c4\ r + 2\ c5\ (c1 + c2\ r)\ \text{Cos}[y]\right)}\big) \Big/$$
$$\big(\sqrt{\left(c3\ c4 + c1\ (c3 + c4) - c5^2 - c2\ c3\ r - c2\ c4\ r - 2\ c5\ (c1 - c2\ r)\ \text{Cos}[y]\right)}\big) + \text{Tan}\Big[$$
$$\frac{x}{2}\Big]\Big)\Big)\Bigg)$$

**ff10 = D[f10, x]**

$$\Bigg( \text{Log}\Big[ \big(\sqrt{\left(-c3\ c4 - c1\ (c3 + c4) + c5^2 - c2\ c3\ r - c2\ c4\ r + 2\ c5\ (c1 + c2\ r)\ \text{Cos}[y]\right)}\big) \Big/$$
$$\big(\sqrt{\left(c3\ c4 + c1\ (c3 + c4) - c5^2 - c2\ c3\ r - c2\ c4\ r - 2\ c5\ (c1 - c2\ r)\ \text{Cos}[y]\right)}\big) +$$
$$\text{Tan}\Big[\frac{x}{2}\Big]\Big]\ \text{Sec}\Big[\frac{x}{2}\Big]^2\Big) \Big/ \Big(2\ \big(1 + i\ \text{Tan}\Big[\frac{x}{2}\Big]\big)\Big) -$$
$$\Big(i\ \text{Log}\Big[\big(\sqrt{\left(c3\ c4 + c1\ (c3 + c4) - c5^2 - c2\ c3\ r - c2\ c4\ r - 2\ c5\ (c1 - c2\ r)\ \text{Cos}[y]\right)}$$
$$\big(1 + i\ \text{Tan}\Big[\frac{x}{2}\Big]\big)\big)\Big/$$
$$\big(\sqrt{\left(c3\ c4 + c1\ (c3 + c4) - c5^2 - c2\ c3\ r - c2\ c4\ r - 2\ c5\ (c1 - c2\ r)\ \text{Cos}[y]\right)} -$$
$$i\ \sqrt{\left(-c3\ c4 - c1\ (c3 + c4) + c5^2 - c2\ c3\ r - c2\ c4\ r + 2\ c5\ (c1 + c2\ r)\ \text{Cos}[y]\right)}\big)\Big]\ \text{Sec}\Big[$$
$$\frac{x}{2}\Big]^2\Big) \Big/$$
$$\Big(2\ \Big(\big(\sqrt{\left(-c3\ c4 - c1\ (c3 + c4) + c5^2 - c2\ c3\ r - c2\ c4\ r + 2\ c5\ (c1 + c2\ r)\ \text{Cos}[y]\right)}\big) \Big/$$
$$\big(\sqrt{\left(c3\ c4 + c1\ (c3 + c4) - c5^2 - c2\ c3\ r - c2\ c4\ r - 2\ c5\ (c1 - c2\ r)\ \text{Cos}[y]\right)}\big) + \text{Tan}\Big[$$
$$\frac{x}{2}\Big]\Big)\Big)\Bigg)$$

**ff11 = D[f11, x]**

$$\left(\mathtt{i}\ \sqrt{\mathtt{c1}+\mathtt{c4}-\mathtt{c2}\ \mathtt{r}}\ \text{Log}\left[1-\frac{\sqrt{-\mathtt{c1}-\mathtt{c4}-\mathtt{c2}\ \mathtt{r}}-\sqrt{\mathtt{c1}+\mathtt{c4}-\mathtt{c2}\ \mathtt{r}}\ \text{Tan}\left[\frac{\mathtt{x}}{2}\right]}{\sqrt{-\mathtt{c1}-\mathtt{c4}-\mathtt{c2}\ \mathtt{r}}-\mathtt{i}\ \sqrt{\mathtt{c1}+\mathtt{c4}-\mathtt{c2}\ \mathtt{r}}}\right]\text{Sec}\left[\frac{\mathtt{x}}{2}\right]^2\right)\Big/$$

$$\left(2\left(\sqrt{-\mathtt{c1}-\mathtt{c4}-\mathtt{c2}\ \mathtt{r}}-\sqrt{\mathtt{c1}+\mathtt{c4}-\mathtt{c2}\ \mathtt{r}}\ \text{Tan}\left[\frac{\mathtt{x}}{2}\right]\right)\right)-$$

$$\left(\mathtt{i}\ \sqrt{\mathtt{c1}+\mathtt{c4}-\mathtt{c2}\ \mathtt{r}}\ \text{Log}\left[1-\frac{\sqrt{-\mathtt{c1}-\mathtt{c4}-\mathtt{c2}\ \mathtt{r}}-\sqrt{\mathtt{c1}+\mathtt{c4}-\mathtt{c2}\ \mathtt{r}}\ \text{Tan}\left[\frac{\mathtt{x}}{2}\right]}{\sqrt{-\mathtt{c1}-\mathtt{c4}-\mathtt{c2}\ \mathtt{r}}+\mathtt{i}\ \sqrt{\mathtt{c1}+\mathtt{c4}-\mathtt{c2}\ \mathtt{r}}}\right]\text{Sec}\left[\frac{\mathtt{x}}{2}\right]^2\right)\Big/$$

$$\left(2\left(\sqrt{-\mathtt{c1}-\mathtt{c4}-\mathtt{c2}\ \mathtt{r}}-\sqrt{\mathtt{c1}+\mathtt{c4}-\mathtt{c2}\ \mathtt{r}}\ \text{Tan}\left[\frac{\mathtt{x}}{2}\right]\right)\right)+$$

$$\left(\mathtt{i}\ \sqrt{\mathtt{c1}+\mathtt{c4}-\mathtt{c2}\ \mathtt{r}}\ \text{Log}\left[1-\frac{\sqrt{-\mathtt{c1}-\mathtt{c4}-\mathtt{c2}\ \mathtt{r}}+\sqrt{\mathtt{c1}+\mathtt{c4}-\mathtt{c2}\ \mathtt{r}}\ \text{Tan}\left[\frac{\mathtt{x}}{2}\right]}{\sqrt{-\mathtt{c1}-\mathtt{c4}-\mathtt{c2}\ \mathtt{r}}-\mathtt{i}\ \sqrt{\mathtt{c1}+\mathtt{c4}-\mathtt{c2}\ \mathtt{r}}}\right]\text{Sec}\left[\frac{\mathtt{x}}{2}\right]^2\right)\Big/$$

$$\left(2\left(\sqrt{-\mathtt{c1}-\mathtt{c4}-\mathtt{c2}\ \mathtt{r}}+\sqrt{\mathtt{c1}+\mathtt{c4}-\mathtt{c2}\ \mathtt{r}}\ \text{Tan}\left[\frac{\mathtt{x}}{2}\right]\right)\right)-$$

$$\left(\mathtt{i}\ \sqrt{\mathtt{c1}+\mathtt{c4}-\mathtt{c2}\ \mathtt{r}}\ \text{Log}\left[1-\frac{\sqrt{-\mathtt{c1}-\mathtt{c4}-\mathtt{c2}\ \mathtt{r}}+\sqrt{\mathtt{c1}+\mathtt{c4}-\mathtt{c2}\ \mathtt{r}}\ \text{Tan}\left[\frac{\mathtt{x}}{2}\right]}{\sqrt{-\mathtt{c1}-\mathtt{c4}-\mathtt{c2}\ \mathtt{r}}+\mathtt{i}\ \sqrt{\mathtt{c1}+\mathtt{c4}-\mathtt{c2}\ \mathtt{r}}}\right]\text{Sec}\left[\frac{\mathtt{x}}{2}\right]^2\right)\Big/$$

$$\left(2\left(\sqrt{-\mathtt{c1}-\mathtt{c4}-\mathtt{c2}\ \mathtt{r}}+\sqrt{\mathtt{c1}+\mathtt{c4}-\mathtt{c2}\ \mathtt{r}}\ \text{Tan}\left[\frac{\mathtt{x}}{2}\right]\right)\right)$$

**ff12 = D[f12, x]**

$$\left(\mathtt{i}\ \sqrt{\left(\mathtt{c3}\ \mathtt{c4}+\mathtt{c1}\ (\mathtt{c3}+\mathtt{c4})-\mathtt{c5}^2-\mathtt{c2}\ \mathtt{c3}\ \mathtt{r}-\mathtt{c2}\ \mathtt{c4}\ \mathtt{r}-2\ \mathtt{c5}\ (\mathtt{c1}-\mathtt{c2}\ \mathtt{r})\ \text{Cos}[\mathtt{y}]\right)}\right.$$

$$\text{Log}\left[1+\left(\mathtt{i}\ \left(\sqrt{\left(-\mathtt{c3}\ \mathtt{c4}-\mathtt{c1}\ (\mathtt{c3}+\mathtt{c4})+\mathtt{c5}^2-\mathtt{c2}\ \mathtt{c3}\ \mathtt{r}-\mathtt{c2}\ \mathtt{c4}\ \mathtt{r}+2\ \mathtt{c5}\ (\mathtt{c1}+\mathtt{c2}\ \mathtt{r})\ \text{Cos}[\mathtt{y}]\right)}-\right.\right.\right.$$

$$\sqrt{\left(\mathtt{c3}\ \mathtt{c4}+\mathtt{c1}\ (\mathtt{c3}+\mathtt{c4})-\mathtt{c5}^2-\mathtt{c2}\ \mathtt{c3}\ \mathtt{r}-\mathtt{c2}\ \mathtt{c4}\ \mathtt{r}-2\ \mathtt{c5}\ (\mathtt{c1}-\mathtt{c2}\ \mathtt{r})\ \text{Cos}[\mathtt{y}]\right)}$$

$$\left.\left.\text{Tan}\left[\frac{\mathtt{x}}{2}\right]\right)\right)\Big/$$

$$\left(\sqrt{\left(\mathtt{c3}\ \mathtt{c4}+\mathtt{c1}\ (\mathtt{c3}+\mathtt{c4})-\mathtt{c5}^2-\mathtt{c2}\ \mathtt{c3}\ \mathtt{r}-\mathtt{c2}\ \mathtt{c4}\ \mathtt{r}-2\ \mathtt{c5}\ (\mathtt{c1}-\mathtt{c2}\ \mathtt{r})\ \text{Cos}[\mathtt{y}]\right)}-\right.$$

$$\left.\left.\mathtt{i}\ \sqrt{\left(-\mathtt{c3}\ \mathtt{c4}-\mathtt{c1}\ (\mathtt{c3}+\mathtt{c4})+\mathtt{c5}^2-\mathtt{c2}\ \mathtt{c3}\ \mathtt{r}-\mathtt{c2}\ \mathtt{c4}\ \mathtt{r}+2\ \mathtt{c5}\ (\mathtt{c1}+\mathtt{c2}\ \mathtt{r})\ \text{Cos}[\mathtt{y}]\right)}\right)\right]$$

$$\text{Sec}\left[\frac{\mathtt{x}}{2}\right]^2\right)\Big/\left(2\left(\sqrt{\left(-\mathtt{c3}\ \mathtt{c4}-\mathtt{c1}\ (\mathtt{c3}+\mathtt{c4})+\mathtt{c5}^2-\mathtt{c2}\ \mathtt{c3}\ \mathtt{r}-\mathtt{c2}\ \mathtt{c4}\ \mathtt{r}+\right.\right.\right.$$

$$\left.2\ \mathtt{c5}\ (\mathtt{c1}+\mathtt{c2}\ \mathtt{r})\ \text{Cos}[\mathtt{y}]\right)-$$

$$\left.\left.\sqrt{\left(\mathtt{c3}\ \mathtt{c4}+\mathtt{c1}\ (\mathtt{c3}+\mathtt{c4})-\mathtt{c5}^2-\mathtt{c2}\ \mathtt{c3}\ \mathtt{r}-\mathtt{c2}\ \mathtt{c4}\ \mathtt{r}-2\ \mathtt{c5}\ (\mathtt{c1}-\mathtt{c2}\ \mathtt{r})\ \text{Cos}[\mathtt{y}]\right)}\ \text{Tan}\left[\frac{\mathtt{x}}{2}\right]\right)\right)$$

**ff13 = D[f13, x]**

$$-\left(\mathbf{i}\,\sqrt{c3\,c4 + c1\,(c3+c4) - c5^2 - c2\,c3\,r - c2\,c4\,r - 2\,c5\,(c1-c2\,r)\,\cos[y]}\right.$$

$$\mathrm{Log}\Big[1 - \Big(\mathbf{i}\,\Big(\sqrt{\left(-c3\,c4 - c1\,(c3+c4) + c5^2 - c2\,c3\,r - c2\,c4\,r + 2\,c5\,(c1+c2\,r)\,\cos[y]\right)} - $$

$$\sqrt{\left(c3\,c4 + c1\,(c3+c4) - c5^2 - c2\,c3\,r - c2\,c4\,r - 2\,c5\,(c1-c2\,r)\,\cos[y]\right)}$$

$$\mathrm{Tan}\Big[\tfrac{x}{2}\Big]\Big)\Big)\Big/$$

$$\Big(\sqrt{\left(c3\,c4 + c1\,(c3+c4) - c5^2 - c2\,c3\,r - c2\,c4\,r - 2\,c5\,(c1-c2\,r)\,\cos[y]\right)} + $$

$$\mathbf{i}\,\sqrt{\left(-c3\,c4 - c1\,(c3+c4) + c5^2 - c2\,c3\,r - c2\,c4\,r + 2\,c5\,(c1+c2\,r)\,\cos[y]\right)}\Big)\Big]$$

$$\mathrm{Sec}\Big[\tfrac{x}{2}\Big]^2\Big)\Big/\Big(2\,\Big(\sqrt{\left(-c3\,c4 - c1\,(c3+c4) + c5^2 - c2\,c3\,r - c2\,c4\,r + \right.}$$

$$2\,c5\,(c1+c2\,r)\,\cos[y]\Big) - $$

$$\sqrt{\left(c3\,c4 + c1\,(c3+c4) - c5^2 - c2\,c3\,r - c2\,c4\,r - 2\,c5\,(c1-c2\,r)\,\cos[y]\right)}\,\mathrm{Tan}\Big[\tfrac{x}{2}\Big]\Big)\Big)$$

**ff14 = D[f14, x]**

$$\Big(\mathbf{i}\,\sqrt{c3\,c4 + c1\,(c3+c4) - c5^2 - c2\,c3\,r - c2\,c4\,r - 2\,c5\,(c1-c2\,r)\,\cos[y]}$$

$$\mathrm{Log}\Big[1 + \Big(\mathbf{i}\,\Big(\sqrt{\left(-c3\,c4 - c1\,(c3+c4) + c5^2 - c2\,c3\,r - c2\,c4\,r + 2\,c5\,(c1+c2\,r)\,\cos[y]\right)} + $$

$$\sqrt{\left(c3\,c4 + c1\,(c3+c4) - c5^2 - c2\,c3\,r - c2\,c4\,r - 2\,c5\,(c1-c2\,r)\,\cos[y]\right)}$$

$$\mathrm{Tan}\Big[\tfrac{x}{2}\Big]\Big)\Big)\Big/$$

$$\Big(\sqrt{\left(c3\,c4 + c1\,(c3+c4) - c5^2 - c2\,c3\,r - c2\,c4\,r - 2\,c5\,(c1-c2\,r)\,\cos[y]\right)} - $$

$$\mathbf{i}\,\sqrt{\left(-c3\,c4 - c1\,(c3+c4) + c5^2 - c2\,c3\,r - c2\,c4\,r + 2\,c5\,(c1+c2\,r)\,\cos[y]\right)}\Big)\Big]$$

$$\mathrm{Sec}\Big[\tfrac{x}{2}\Big]^2\Big)\Big/\Big(2\,\Big(\sqrt{\left(-c3\,c4 - c1\,(c3+c4) + c5^2 - c2\,c3\,r - c2\,c4\,r + \right.}$$

$$2\,c5\,(c1+c2\,r)\,\cos[y]\Big) + $$

$$\sqrt{\left(c3\,c4 + c1\,(c3+c4) - c5^2 - c2\,c3\,r - c2\,c4\,r - 2\,c5\,(c1-c2\,r)\,\cos[y]\right)}\,\mathrm{Tan}\Big[\tfrac{x}{2}\Big]\Big)\Big)$$

**ff15 = D[f15, x]**

$$-\left(\mathbf{i}\,\sqrt{c3\,c4 + c1\,(c3+c4) - c5^2 - c2\,c3\,r - c2\,c4\,r - 2\,c5\,(c1-c2\,r)\,\cos[y]}\right.$$

$$\mathrm{Log}\Big[1 - \Big(\mathbf{i}\,\Big(\sqrt{\left(-c3\,c4 - c1\,(c3+c4) + c5^2 - c2\,c3\,r - c2\,c4\,r + 2\,c5\,(c1+c2\,r)\,\cos[y]\right)} + $$

$$\sqrt{\left(c3\,c4 + c1\,(c3+c4) - c5^2 - c2\,c3\,r - c2\,c4\,r - 2\,c5\,(c1-c2\,r)\,\cos[y]\right)}$$

$$\mathrm{Tan}\Big[\tfrac{x}{2}\Big]\Big)\Big)\Big/$$

$$\Big(\sqrt{\left(c3\,c4 + c1\,(c3+c4) - c5^2 - c2\,c3\,r - c2\,c4\,r - 2\,c5\,(c1-c2\,r)\,\cos[y]\right)} + $$

$$\mathbf{i}\,\sqrt{\left(-c3\,c4 - c1\,(c3+c4) + c5^2 - c2\,c3\,r - c2\,c4\,r + 2\,c5\,(c1+c2\,r)\,\cos[y]\right)}\Big)\Big]$$

$$\mathrm{Sec}\Big[\tfrac{x}{2}\Big]^2\Big)\Big/\Big(2\,\Big(\sqrt{\left(-c3\,c4 - c1\,(c3+c4) + c5^2 - c2\,c3\,r - c2\,c4\,r + \right.}$$

$$2\,c5\,(c1+c2\,r)\,\cos[y]\Big) + $$

$$\sqrt{\left(c3\,c4 + c1\,(c3+c4) - c5^2 - c2\,c3\,r - c2\,c4\,r - 2\,c5\,(c1-c2\,r)\,\cos[y]\right)}\,\mathrm{Tan}\Big[\tfrac{x}{2}\Big]\Big)\Big)$$

(*———————————————————————————————————————————*)

(***Step 3: Checking that indeed the sum of 15 derivatives is equal to the initial interand**)

**Deri = ff1 + ff2 + ff3 + ff4 + ff5 + ff6 +**
  **ff7 + ff8 + ff9 + ff10 + ff11 + ff12 + ff13 + ff14 + ff15**

$\text{Log}[2\,\pi] +$

$\text{Log}\Big[\big(c1\,c3 + c1\,c4 + c3\,c4 - c5^2 - 2\,c1\,c5\,\text{Cos}[y] + c2\,r\,\text{Cos}[x]\,(c3 + c4 - 2\,c5\,\text{Cos}[y])\big)\Big/$

$\quad (c1 + c4 + c2\,r\,\text{Cos}[x])\Big] +$

$\text{Log}\Big[-\dfrac{\sqrt{-c1 - c4 - c2\,r}}{\sqrt{c1 + c4 - c2\,r}} + \text{Tan}\big[\tfrac{x}{2}\big]\Big] + \text{Log}\Big[\dfrac{\sqrt{-c1 - c4 - c2\,r}}{\sqrt{c1 + c4 - c2\,r}} + \text{Tan}\big[\tfrac{x}{2}\big]\Big] -$

$\text{Log}\Big[-\big(\sqrt{(-c3\,c4 - c1\,(c3 + c4) + c5^2 - c2\,c3\,r - c2\,c4\,r + 2\,c5\,(c1 + c2\,r)\,\text{Cos}[y])}\big)\Big/$

$\quad \big(\sqrt{(c3\,c4 + c1\,(c3 + c4) - c5^2 - c2\,c3\,r - c2\,c4\,r - 2\,c5\,(c1 - c2\,r)\,\text{Cos}[y])}\big) + \text{Tan}\big[\tfrac{x}{2}\big]\Big] -$

$\text{Log}\Big[\big(\sqrt{(-c3\,c4 - c1\,(c3 + c4) + c5^2 - c2\,c3\,r - c2\,c4\,r + 2\,c5\,(c1 + c2\,r)\,\text{Cos}[y])}\big)\Big/$

$\quad \big(\sqrt{(c3\,c4 + c1\,(c3 + c4) - c5^2 - c2\,c3\,r - c2\,c4\,r - 2\,c5\,(c1 - c2\,r)\,\text{Cos}[y])}\big) + \text{Tan}\big[\tfrac{x}{2}\big]\Big] +$

$\Big(x\,(c1 + c4 + c2\,r\,\text{Cos}[x])\,\Big(-\dfrac{c2\,r\,(c3 + c4 - 2\,c5\,\text{Cos}[y])\,\text{Sin}[x]}{c1 + c4 + c2\,r\,\text{Cos}[x]} +$

$\quad \big(c2\,r\,(c1\,c3 + c1\,c4 + c3\,c4 - c5^2 - 2\,c1\,c5\,\text{Cos}[y] + c2\,r\,\text{Cos}[x]\,(c3 + c4 - 2\,c5\,\text{Cos}[y]))$

$\quad \text{Sin}[x]\big)\Big/(c1 + c4 + c2\,r\,\text{Cos}[x])^2\Big)\Big)\Big/$

$(c1\,c3 + c1\,c4 + c3\,c4 - c5^2 - 2\,c1\,c5\,\text{Cos}[y] + c2\,r\,\text{Cos}[x]\,(c3 + c4 - 2\,c5\,\text{Cos}[y])) +$

$\Big(\text{Log}\Big[-\big(\sqrt{(-c3\,c4 - c1\,(c3 + c4) + c5^2 - c2\,c3\,r - c2\,c4\,r + 2\,c5\,(c1 + c2\,r)\,\text{Cos}[y])}\big)\Big/$

$\quad \big(\sqrt{(c3\,c4 + c1\,(c3 + c4) - c5^2 - c2\,c3\,r - c2\,c4\,r - 2\,c5\,(c1 - c2\,r)\,\text{Cos}[y])}\big) +$

$\quad \text{Tan}\big[\tfrac{x}{2}\big]\Big]\,\text{Sec}\big[\tfrac{x}{2}\big]^2\Big)\Big/\Big(2\,\big(1 - \mathbb{i}\,\text{Tan}\big[\tfrac{x}{2}\big]\big)\Big) +$

$\Big(\text{Log}\Big[\big(\sqrt{(-c3\,c4 - c1\,(c3 + c4) + c5^2 - c2\,c3\,r - c2\,c4\,r + 2\,c5\,(c1 + c2\,r)\,\text{Cos}[y])}\big)\Big/$

$\quad \big(\sqrt{(c3\,c4 + c1\,(c3 + c4) - c5^2 - c2\,c3\,r - c2\,c4\,r - 2\,c5\,(c1 - c2\,r)\,\text{Cos}[y])}\big) +$

$\quad \text{Tan}\big[\tfrac{x}{2}\big]\Big]\,\text{Sec}\big[\tfrac{x}{2}\big]^2\Big)\Big/\Big(2\,\big(1 - \mathbb{i}\,\text{Tan}\big[\tfrac{x}{2}\big]\big)\Big) +$

$\Big(\text{Log}\Big[-\big(\sqrt{(-c3\,c4 - c1\,(c3 + c4) + c5^2 - c2\,c3\,r - c2\,c4\,r + 2\,c5\,(c1 + c2\,r)\,\text{Cos}[y])}\big)\Big/$

$\quad \big(\sqrt{(c3\,c4 + c1\,(c3 + c4) - c5^2 - c2\,c3\,r - c2\,c4\,r - 2\,c5\,(c1 - c2\,r)\,\text{Cos}[y])}\big) +$

$\quad \text{Tan}\big[\tfrac{x}{2}\big]\Big]\,\text{Sec}\big[\tfrac{x}{2}\big]^2\Big)\Big/\Big(2\,\big(1 + \mathbb{i}\,\text{Tan}\big[\tfrac{x}{2}\big]\big)\Big) +$

$\Big(\text{Log}\Big[\big(\sqrt{(-c3\,c4 - c1\,(c3 + c4) + c5^2 - c2\,c3\,r - c2\,c4\,r + 2\,c5\,(c1 + c2\,r)\,\text{Cos}[y])}\big)\Big/$

$\quad \big(\sqrt{(c3\,c4 + c1\,(c3 + c4) - c5^2 - c2\,c3\,r - c2\,c4\,r - 2\,c5\,(c1 - c2\,r)\,\text{Cos}[y])}\big) + \text{Tan}\big[\tfrac{x}{2}\big]\Big]$

$\quad \text{Sec}\big[\tfrac{x}{2}\big]^2\Big)\Big/\Big(2\,\big(1 + \mathbb{i}\,\text{Tan}\big[\tfrac{x}{2}\big]\big)\Big) + \dfrac{\mathbb{i}\,\text{Log}\Big[-\dfrac{\sqrt{-c1 - c4 - c2\,r}}{\sqrt{c1 + c4 + c2\,r}} + \text{Tan}\big[\tfrac{x}{2}\big]\Big]\,\text{Sec}\big[\tfrac{x}{2}\big]^2}{2\,\big(-\mathbb{i} + \text{Tan}\big[\tfrac{x}{2}\big]\big)} +$

$$\frac{i \, \text{Log}\left[\frac{\sqrt{-c1-c4-c2\,r}}{\sqrt{c1+c4-c2\,r}} + \text{Tan}\left[\frac{x}{2}\right]\right] \text{Sec}\left[\frac{x}{2}\right]^2}{2\left(-i + \text{Tan}\left[\frac{x}{2}\right]\right)} -$$

$$\frac{i \, \text{Log}\left[-\frac{\sqrt{-c1-c4-c2\,r}}{\sqrt{c1+c4-c2\,r}} + \text{Tan}\left[\frac{x}{2}\right]\right] \text{Sec}\left[\frac{x}{2}\right]^2}{2\left(i + \text{Tan}\left[\frac{x}{2}\right]\right)} -$$

$$\frac{i \, \text{Log}\left[\frac{\sqrt{-c1-c4-c2\,r}}{\sqrt{c1+c4-c2\,r}} + \text{Tan}\left[\frac{x}{2}\right]\right] \text{Sec}\left[\frac{x}{2}\right]^2}{2\left(i + \text{Tan}\left[\frac{x}{2}\right]\right)} +$$

$$\frac{x \, \text{Sec}\left[\frac{x}{2}\right]^2}{2\left(-\frac{\sqrt{-c1-c4-c2\,r}}{\sqrt{c1+c4-c2\,r}} + \text{Tan}\left[\frac{x}{2}\right]\right)} +$$

$$\frac{i \, \text{Log}\left[\frac{\sqrt{c1+c4-c2\,r}\,\left(-i+\text{Tan}\left[\frac{x}{2}\right]\right)}{\sqrt{-c1-c4-c2\,r}\,-i\,\sqrt{c1+c4-c2\,r}}\right] \text{Sec}\left[\frac{x}{2}\right]^2}{2\left(-\frac{\sqrt{-c1-c4-c2\,r}}{\sqrt{c1+c4-c2\,r}} + \text{Tan}\left[\frac{x}{2}\right]\right)} -$$

$$\frac{i \, \text{Log}\left[\frac{\sqrt{c1+c4-c2\,r}\,\left(i+\text{Tan}\left[\frac{x}{2}\right]\right)}{\sqrt{-c1-c4-c2\,r}\,+i\,\sqrt{c1+c4-c2\,r}}\right] \text{Sec}\left[\frac{x}{2}\right]^2}{2\left(-\frac{\sqrt{-c1-c4-c2\,r}}{\sqrt{c1+c4-c2\,r}} + \text{Tan}\left[\frac{x}{2}\right]\right)} +$$

$$\frac{x \, \text{Sec}\left[\frac{x}{2}\right]^2}{2\left(\frac{\sqrt{-c1-c4-c2\,r}}{\sqrt{c1+c4-c2\,r}} + \text{Tan}\left[\frac{x}{2}\right]\right)} +$$

$$\frac{i \, \text{Log}\left[-\frac{\sqrt{c1+c4-c2\,r}\,\left(-i+\text{Tan}\left[\frac{x}{2}\right]\right)}{\sqrt{-c1-c4-c2\,r}\,+i\,\sqrt{c1+c4-c2\,r}}\right] \text{Sec}\left[\frac{x}{2}\right]^2}{2\left(\frac{\sqrt{-c1-c4-c2\,r}}{\sqrt{c1+c4-c2\,r}} + \text{Tan}\left[\frac{x}{2}\right]\right)} -$$

$$\frac{i \, \text{Log}\left[-\frac{\sqrt{c1+c4-c2\,r}\,\left(i+\text{Tan}\left[\frac{x}{2}\right]\right)}{\sqrt{-c1-c4-c2\,r}\,-i\,\sqrt{c1+c4-c2\,r}}\right] \text{Sec}\left[\frac{x}{2}\right]^2}{2\left(\frac{\sqrt{-c1-c4-c2\,r}}{\sqrt{c1+c4-c2\,r}} + \text{Tan}\left[\frac{x}{2}\right]\right)} -$$

$$\frac{x \, \text{Sec}\left[\frac{x}{2}\right]^2}{2\left(-\frac{\sqrt{-c3\,c4-c1\,(c3+c4)+c5^2-c2\,c3\,r-c2\,c4\,r+2\,c5\,(c1+c2\,r)\,\text{Cos}[y]}}{\sqrt{c3\,c4+c1\,(c3+c4)-c5^2-c2\,c3\,r-c2\,c4\,r-2\,c5\,(c1-c2\,r)\,\text{Cos}[y]}} + \text{Tan}\left[\frac{x}{2}\right]\right)} +$$

$$\left(i \, \text{Log}\left[\left(\sqrt{\left(c3\,c4+c1\,(c3+c4)-c5^2-c2\,c3\,r-c2\,c4\,r-2\,c5\,(c1-c2\,r)\,\text{Cos}[y]\right)}\right.\right.\right.$$

$$\left.\left.\left(1 - i \, \text{Tan}\left[\frac{x}{2}\right]\right)\right) \middle/$$

$$\left(\sqrt{\left(c3\,c4+c1\,(c3+c4)-c5^2-c2\,c3\,r-c2\,c4\,r-2\,c5\,(c1-c2\,r)\,\text{Cos}[y]\right)} -\right.$$

$$\left.\left. i \, \sqrt{\left(-c3\,c4-c1\,(c3+c4)+c5^2-c2\,c3\,r-c2\,c4\,r+2\,c5\,(c1+c2\,r)\,\text{Cos}[y]\right)}\right]\right]$$

$$\left. \text{Sec}\left[\frac{x}{2}\right]^2\right) \middle/ \left(2\left(-\left(\sqrt{\left(-c3\,c4-c1\,(c3+c4)+c5^2-c2\,c3\,r-c2\,c4\,r+\right.}\right.\right.\right.$$

$$\left.\left. 2\,c5\,(c1+c2\,r)\,\text{Cos}[y]\right)\right) \middle/ \left(\sqrt{\left(c3\,c4+c1\,(c3+c4)-c5^2-\right.}\right.$$

$$c2\,c3\,r - c2\,c4\,r - 2\,c5\,(c1 - c2\,r)\,\text{Cos}[y]\big)\big) + \text{Tan}\Big[\frac{x}{2}\Big]\Big]\Big]\Big) -$$

$$\bigg(i\,\text{Log}\Big[\Big(\sqrt{\big(c3\,c4 + c1\,(c3 + c4) - c5^2 - c2\,c3\,r - c2\,c4\,r - 2\,c5\,(c1 - c2\,r)\,\text{Cos}[y]}$$

$$\Big(1 + i\,\text{Tan}\Big[\frac{x}{2}\Big]\Big)\Big)\bigg/\Big(\sqrt{\big(c3\,c4 + c1\,(c3 + c4) - c5^2 - c2\,c3\,r - c2\,c4\,r -}$$

$$2\,c5\,(c1 - c2\,r)\,\text{Cos}[y]} + i\,\sqrt{\big(-c3\,c4 - c1\,(c3 + c4) + c5^2 -}$$

$$c2\,c3\,r - c2\,c4\,r + 2\,c5\,(c1 + c2\,r)\,\text{Cos}[y]}\big)\Big]\,\text{Sec}\Big[\frac{x}{2}\Big]^2\Big)\bigg/$$

$$\bigg(2\,\Big(-\Big(\sqrt{\big(-c3\,c4 - c1\,(c3 + c4) + c5^2 - c2\,c3\,r - c2\,c4\,r + 2\,c5\,(c1 + c2\,r)\,\text{Cos}[y]}\Big)\Big)\bigg/$$

$$\Big(\sqrt{\big(c3\,c4 + c1\,(c3 + c4) - c5^2 - c2\,c3\,r - c2\,c4\,r - 2\,c5\,(c1 - c2\,r)\,\text{Cos}[y]}\Big) +$$

$$\text{Tan}\Big[\frac{x}{2}\Big]\Big)\Big) - \frac{x\,\text{Sec}\Big[\frac{x}{2}\Big]^2}{2\,\Big(\frac{\sqrt{-c3\,c4 - c1\,(c3 + c4) + c5^2 - c2\,c3\,r - c2\,c4\,r + 2\,c5\,(c1 + c2\,r)\,\text{Cos}[y]}}{\sqrt{c3\,c4 + c1\,(c3 + c4) - c5^2 - c2\,c3\,r - c2\,c4\,r - 2\,c5\,(c1 - c2\,r)\,\text{Cos}[y]}} + \text{Tan}\Big[\frac{x}{2}\Big]\Big)} +$$

$$\bigg(i\,\text{Log}\Big[\Big(\sqrt{\big(c3\,c4 + c1\,(c3 + c4) - c5^2 - c2\,c3\,r - c2\,c4\,r - 2\,c5\,(c1 - c2\,r)\,\text{Cos}[y]}$$

$$\Big(1 - i\,\text{Tan}\Big[\frac{x}{2}\Big]\Big)\Big)\bigg/\Big(\sqrt{\big(c3\,c4 + c1\,(c3 + c4) - c5^2 - c2\,c3\,r - c2\,c4\,r -}$$

$$2\,c5\,(c1 - c2\,r)\,\text{Cos}[y]} + i\,\sqrt{\big(-c3\,c4 - c1\,(c3 + c4) + c5^2 -}$$

$$c2\,c3\,r - c2\,c4\,r + 2\,c5\,(c1 + c2\,r)\,\text{Cos}[y]}\big)\Big]\,\text{Sec}\Big[\frac{x}{2}\Big]^2\Big)\bigg/$$

$$\bigg(2\,\Big(\big(\sqrt{\big(-c3\,c4 - c1\,(c3 + c4) + c5^2 - c2\,c3\,r - c2\,c4\,r + 2\,c5\,(c1 + c2\,r)\,\text{Cos}[y]}\big)\big/\big(\sqrt{\big(c3\,c4 +}$$

$$c1\,(c3 + c4) - c5^2 - c2\,c3\,r - c2\,c4\,r - 2\,c5\,(c1 - c2\,r)\,\text{Cos}[y]}\big) + \text{Tan}\Big[\frac{x}{2}\Big]\Big)\Big) -$$

$$\bigg(i\,\text{Log}\Big[\Big(\sqrt{\big(c3\,c4 + c1\,(c3 + c4) - c5^2 - c2\,c3\,r - c2\,c4\,r - 2\,c5\,(c1 - c2\,r)\,\text{Cos}[y]}$$

$$\Big(1 + i\,\text{Tan}\Big[\frac{x}{2}\Big]\Big)\Big)\bigg/\Big(\sqrt{\big(c3\,c4 + c1\,(c3 + c4) - c5^2 - c2\,c3\,r - c2\,c4\,r -}$$

$$2\,c5\,(c1 - c2\,r)\,\text{Cos}[y]} - i\,\sqrt{\big(-c3\,c4 - c1\,(c3 + c4) + c5^2 -}$$

$$c2\,c3\,r - c2\,c4\,r + 2\,c5\,(c1 + c2\,r)\,\text{Cos}[y]}\big)\Big]\,\text{Sec}\Big[\frac{x}{2}\Big]^2\Big)\bigg/$$

$$\bigg(2\,\Big(\big(\sqrt{\big(-c3\,c4 - c1\,(c3 + c4) + c5^2 - c2\,c3\,r - c2\,c4\,r + 2\,c5\,(c1 + c2\,r)\,\text{Cos}[y]}\big)\big/\big(\sqrt{\big(c3\,c4 +}$$

$$c1\,(c3 + c4) - c5^2 - c2\,c3\,r - c2\,c4\,r - 2\,c5\,(c1 - c2\,r)\,\text{Cos}[y]}\big) + \text{Tan}\Big[\frac{x}{2}\Big]\Big)\Big) +$$

$$\bigg(i\,\sqrt{c1 + c4 - c2\,r}\,\text{Log}\Big[1 - \frac{\sqrt{-c1 - c4 - c2\,r} - \sqrt{c1 + c4 - c2\,r}\,\text{Tan}\Big[\frac{x}{2}\Big]}{\sqrt{-c1 - c4 - c2\,r} - i\,\sqrt{c1 + c4 - c2\,r}}\Big]\,\text{Sec}\Big[\frac{x}{2}\Big]^2\bigg)\bigg/$$

$$\bigg(2\,\Big(\sqrt{-c1 - c4 - c2\,r} - \sqrt{c1 + c4 - c2\,r}\,\text{Tan}\Big[\frac{x}{2}\Big]\Big)\bigg) -$$

$$\bigg(i\,\sqrt{c1 + c4 - c2\,r}\,\text{Log}\Big[1 - \frac{\sqrt{-c1 - c4 - c2\,r} - \sqrt{c1 + c4 - c2\,r}\,\text{Tan}\Big[\frac{x}{2}\Big]}{\sqrt{-c1 - c4 - c2\,r} + i\,\sqrt{c1 + c4 - c2\,r}}\Big]\,\text{Sec}\Big[\frac{x}{2}\Big]^2\bigg)\bigg/$$

$$\bigg(2\,\Big(\sqrt{-c1 - c4 - c2\,r} - \sqrt{c1 + c4 - c2\,r}\,\text{Tan}\Big[\frac{x}{2}\Big]\Big)\bigg) +$$

$$\bigg(i\,\sqrt{c1 + c4 - c2\,r}\,\text{Log}\Big[1 - \frac{\sqrt{-c1 - c4 - c2\,r} + \sqrt{c1 + c4 - c2\,r}\,\text{Tan}\Big[\frac{x}{2}\Big]}{\sqrt{-c1 - c4 - c2\,r} - i\,\sqrt{c1 + c4 - c2\,r}}\Big]\,\text{Sec}\Big[\frac{x}{2}\Big]^2\bigg)\bigg/$$

$$\left(2\left(\sqrt{-c1-c4-c2\ r}+\sqrt{c1+c4-c2\ r}\ \text{Tan}\left[\frac{x}{2}\right]\right)\right)-$$

$$\left(\mathbf{i}\ \sqrt{c1+c4-c2\ r}\ \text{Log}\left[1-\frac{\sqrt{-c1-c4-c2\ r}+\sqrt{c1+c4-c2\ r}\ \text{Tan}\left[\frac{x}{2}\right]}{\sqrt{-c1-c4-c2\ r}+\mathbf{i}\ \sqrt{c1+c4-c2\ r}}\right]\text{Sec}\left[\frac{x}{2}\right]^2\right)\Bigg/$$

$$\left(2\left(\sqrt{-c1-c4-c2\ r}+\sqrt{c1+c4-c2\ r}\ \text{Tan}\left[\frac{x}{2}\right]\right)\right)+$$

$$\left(\mathbf{i}\ \sqrt{c3\ c4+c1\ (c3+c4)-c5^2-c2\ c3\ r-c2\ c4\ r-2\ c5\ (c1-c2\ r)\ \text{Cos}[y]}\right.$$

$$\text{Log}\left[1+\left(\mathbf{i}\left(\sqrt{(-c3\ c4-c1\ (c3+c4)+c5^2-c2\ c3\ r-c2\ c4\ r+2\ c5\ (c1+c2\ r)\ \text{Cos}[y])}-\right.\right.\right.$$

$$\sqrt{(c3\ c4+c1\ (c3+c4)-c5^2-c2\ c3\ r-c2\ c4\ r-2\ c5\ (c1-c2\ r)\ \text{Cos}[y])}$$

$$\left.\text{Tan}\left[\frac{x}{2}\right]\right)\right)\Bigg/\left(\sqrt{(c3\ c4+c1\ (c3+c4)-c5^2-c2\ c3\ r-c2\ c4\ r-}\right.$$

$$2\ c5\ (c1-c2\ r)\ \text{Cos}[y])-\mathbf{i}\ \sqrt{(-c3\ c4-c1\ (c3+c4)+c5^2-}$$

$$\left.\left.\left.c2\ c3\ r-c2\ c4\ r+2\ c5\ (c1+c2\ r)\ \text{Cos}[y])}\right)\right)\text{Sec}\left[\frac{x}{2}\right]^2\right)\Bigg/$$

$$\left(2\left(\sqrt{(-c3\ c4-c1\ (c3+c4)+c5^2-c2\ c3\ r-c2\ c4\ r+2\ c5\ (c1+c2\ r)\ \text{Cos}[y])}-\right.\right.$$

$$\left.\left.\sqrt{(c3\ c4+c1\ (c3+c4)-c5^2-c2\ c3\ r-c2\ c4\ r-2\ c5\ (c1-c2\ r)\ \text{Cos}[y])}\ \text{Tan}\left[\frac{x}{2}\right]\right)\right)-$$

$$\left(\mathbf{i}\ \sqrt{c3\ c4+c1\ (c3+c4)-c5^2-c2\ c3\ r-c2\ c4\ r-2\ c5\ (c1-c2\ r)\ \text{Cos}[y]}\right.$$

$$\text{Log}\left[1-\left(\mathbf{i}\left(\sqrt{(-c3\ c4-c1\ (c3+c4)+c5^2-c2\ c3\ r-c2\ c4\ r+2\ c5\ (c1+c2\ r)\ \text{Cos}[y])}-\right.\right.\right.$$

$$\sqrt{(c3\ c4+c1\ (c3+c4)-c5^2-c2\ c3\ r-c2\ c4\ r-2\ c5\ (c1-c2\ r)\ \text{Cos}[y])}$$

$$\left.\text{Tan}\left[\frac{x}{2}\right]\right)\right)\Bigg/\left(\sqrt{(c3\ c4+c1\ (c3+c4)-c5^2-c2\ c3\ r-c2\ c4\ r-}\right.$$

$$2\ c5\ (c1-c2\ r)\ \text{Cos}[y])+\mathbf{i}\ \sqrt{(-c3\ c4-c1\ (c3+c4)+c5^2-}$$

$$\left.\left.\left.c2\ c3\ r-c2\ c4\ r+2\ c5\ (c1+c2\ r)\ \text{Cos}[y])}\right)\right)\text{Sec}\left[\frac{x}{2}\right]^2\right)\Bigg/$$

$$\left(2\left(\sqrt{(-c3\ c4-c1\ (c3+c4)+c5^2-c2\ c3\ r-c2\ c4\ r+2\ c5\ (c1+c2\ r)\ \text{Cos}[y])}-\right.\right.$$

$$\left.\left.\sqrt{(c3\ c4+c1\ (c3+c4)-c5^2-c2\ c3\ r-c2\ c4\ r-2\ c5\ (c1-c2\ r)\ \text{Cos}[y])}\ \text{Tan}\left[\frac{x}{2}\right]\right)\right)+$$

$$\left(\mathbf{i}\ \sqrt{c3\ c4+c1\ (c3+c4)-c5^2-c2\ c3\ r-c2\ c4\ r-2\ c5\ (c1-c2\ r)\ \text{Cos}[y]}\right.$$

$$\text{Log}\left[1+\left(\mathbf{i}\left(\sqrt{(-c3\ c4-c1\ (c3+c4)+c5^2-c2\ c3\ r-c2\ c4\ r+2\ c5\ (c1+c2\ r)\ \text{Cos}[y])}+\right.\right.\right.$$

$$\sqrt{(c3\ c4+c1\ (c3+c4)-c5^2-c2\ c3\ r-c2\ c4\ r-2\ c5\ (c1-c2\ r)\ \text{Cos}[y])}$$

$$\left.\text{Tan}\left[\frac{x}{2}\right]\right)\right)\Bigg/\left(\sqrt{(c3\ c4+c1\ (c3+c4)-c5^2-c2\ c3\ r-c2\ c4\ r-}\right.$$

$$2\ c5\ (c1-c2\ r)\ \text{Cos}[y])-\mathbf{i}\ \sqrt{(-c3\ c4-c1\ (c3+c4)+c5^2-}$$

$$\left.\left.\left.c2\ c3\ r-c2\ c4\ r+2\ c5\ (c1+c2\ r)\ \text{Cos}[y])}\right)\right)\text{Sec}\left[\frac{x}{2}\right]^2\right)\Bigg/$$

$$\left(2\left(\sqrt{(-c3\ c4-c1\ (c3+c4)+c5^2-c2\ c3\ r-c2\ c4\ r+2\ c5\ (c1+c2\ r)\ \text{Cos}[y])}+\right.\right.$$

$$\left.\left.\sqrt{(c3\ c4+c1\ (c3+c4)-c5^2-c2\ c3\ r-c2\ c4\ r-2\ c5\ (c1-c2\ r)\ \text{Cos}[y])}\ \text{Tan}\left[\frac{x}{2}\right]\right)\right)-$$

$$\left(\mathbf{i}\ \sqrt{c3\ c4+c1\ (c3+c4)-c5^2-c2\ c3\ r-c2\ c4\ r-2\ c5\ (c1-c2\ r)\ \text{Cos}[y]}\right.$$

$$\text{Log}\left[1-\left(\mathbf{i}\left(\sqrt{(-c3\ c4-c1\ (c3+c4)+c5^2-c2\ c3\ r-c2\ c4\ r+2\ c5\ (c1+c2\ r)\ \text{Cos}[y])}+\right.\right.\right.$$

$$\sqrt{\left(c3\ c4 + c1\ (c3 + c4) - c5^2 - c2\ c3\ r - c2\ c4\ r - 2\ c5\ (c1 - c2\ r)\ \text{Cos}[y]\right)}$$

$$\text{Tan}\left[\frac{x}{2}\right]\right)\right) \bigg/ \left(\sqrt{\left(c3\ c4 + c1\ (c3 + c4) - c5^2 - c2\ c3\ r - c2\ c4\ r - \right.}\right.$$

$$2\ c5\ (c1 - c2\ r)\ \text{Cos}[y]\right) + \dot{\mathbb{1}}\ \sqrt{\left(-c3\ c4 - c1\ (c3 + c4) + c5^2 - \right.}$$

$$\left.\left. c2\ c3\ r - c2\ c4\ r + 2\ c5\ (c1 + c2\ r)\ \text{Cos}[y]\right)\right)\right]\ \text{Sec}\left[\frac{x}{2}\right]^2\right) \bigg/$$

$$\left(2\ \left(\sqrt{\left(-c3\ c4 - c1\ (c3 + c4) + c5^2 - c2\ c3\ r - c2\ c4\ r + 2\ c5\ (c1 + c2\ r)\ \text{Cos}[y]\right)} + \right.\right.$$

$$\left.\left.\sqrt{\left(c3\ c4 + c1\ (c3 + c4) - c5^2 - c2\ c3\ r - c2\ c4\ r - 2\ c5\ (c1 - c2\ r)\ \text{Cos}[y]\right)}\ \text{Tan}\left[\frac{x}{2}\right]\right)\right)$$

(∗————————————————————————————————————————————————∗)

(∗**Simplifying the above sum**∗)

**FullSimplify[Deri, TimeConstraint → Infinity]**

$$\text{Log}\left[\left(2\ \pi\ \left(c1\ c3 + (c1 + c3)\ c4 - c5^2 - 2\ c1\ c5\ \text{Cos}[y] + c2\ r\ \text{Cos}[x]\ (c3 + c4 - 2\ c5\ \text{Cos}[y])\right)\right) \bigg/ \right.$$

$$\left. (c1 + c4 + c2\ r\ \text{Cos}[x])\right]$$


\label{pginc_end}

%----------------------------------------------------------------------------------------------------------------------------
%                                                       References
%----------------------------------------------------------------------------------------------------------------------------

\end{document}